\documentclass[10pt,a4paper]{article}

\usepackage{hyperref}  

\usepackage{cite}

\usepackage[T1]{fontenc}    
\usepackage{lmodern}
\usepackage{times}

\usepackage{pgf}
\usepackage{tikz}

\usepackage{doi}
\usepackage{cite}

\usepackage[fleqn]{amsmath}
\usepackage{amsfonts,amsthm,amssymb}

\numberwithin{equation}{section}

\newtheorem{theorem}{Theorem}[section]
\newtheorem{proposition}{Proposition}[section]
\newtheorem{lemma}{Lemma}[section]

\newtheorem{corollary}{Corollary}[section]

\newtheorem{identity}{Identity}
{\theoremstyle{remark}
\newtheorem{rem}{\sl Remark}

\makeatletter
\newcommand{\vast}{\bBigg@{3}}
\newcommand{\Vast}{\bBigg@{4}}
\makeatother

\newcommand{\bra}[1]{\langle\,#1\,|}
\newcommand{\ket}[1]{|\,#1\,\rangle}

\newcommand{\moy}[1]{\langle\,#1\,\rangle}

\def\tr{\operatorname{tr}}
\newcommand{\End}{\operatorname{End}}

\DeclareMathSymbol{\Alpha}{\mathalpha}{operators}{"41}
\DeclareMathSymbol{\Beta}{\mathalpha}{operators}{"42}
\DeclareMathSymbol{\Epsilon}{\mathalpha}{operators}{"45}
\DeclareMathSymbol{\Zeta}{\mathalpha}{operators}{"5A}
\DeclareMathSymbol{\Eta}{\mathalpha}{operators}{"48}
\DeclareMathSymbol{\Iota}{\mathalpha}{operators}{"49}
\DeclareMathSymbol{\Kappa}{\mathalpha}{operators}{"4B}
\DeclareMathSymbol{\Mu}{\mathalpha}{operators}{"4D}
\DeclareMathSymbol{\Nu}{\mathalpha}{operators}{"4E}
\DeclareMathSymbol{\Omicron}{\mathalpha}{operators}{"4F}
\DeclareMathSymbol{\Rho}{\mathalpha}{operators}{"50}
\DeclareMathSymbol{\Tau}{\mathalpha}{operators}{"54}
\DeclareMathSymbol{\Chi}{\mathalpha}{operators}{"58}
\DeclareMathSymbol{\omicron}{\mathord}{letters}{"6F}

\begin{document}

LPENSL-TH-02/22

\bigskip

\bigskip

\begin{center}
\textbf{\Large Correlation functions for open XXZ spin 1/2 quantum chains with unparallel
boundary magnetic fields} 


\vspace{45pt}
\end{center}

\begin{center}
{\large \textbf{G. Niccoli}\footnote{{Univ Lyon, Ens de Lyon, Univ
Claude Bernard, CNRS, Laboratoire de Physique, F-69342 Lyon, France;
giuliano.niccoli@ens-lyon.fr}} }, {\large \textbf{V. Terras}\footnote{{Université Paris-Saclay, CNRS,  LPTMS, 91405, Orsay, France; veronique.terras@universite-paris-saclay.fr}} }
\end{center}

\begin{center}
\vspace{45pt}

\today
\vspace{45pt}
\end{center}

\begin{abstract}
In this paper we continue our derivation of the correlation functions of open quantum spin 1/2 chains with unparallel magnetic fields on the edges; this time for the more involved case of the XXZ spin 1/2 chains. We develop our study in the framework of the quantum Separation of Variables (SoV), which gives us both the complete spectrum characterization and simple scalar product formulae for separate states, including transfer matrix eigenstates. Here, we leave the boundary magnetic field in the first site of the chain completely arbitrary, and we fix the boundary field in the last site $N$ of the chain to be a specific value along the $z$-direction. This is a natural first choice for the unparallel boundary magnetic fields. We prove that under these special boundary conditions, on the one side, we have a simple enough complete spectrum description in terms of homogeneous Baxter like $TQ$-equation. On the other side, we  prove a simple enough description of the action of a basis of local operators on transfer matrix eigenstates as linear combinations of separate states. Thanks to these results, we achieve our main goal to derive correlation functions for a set of local operators both for the finite and half-infinite chains, with multiple integral formulae in this last case.
\end{abstract}

\newpage

\tableofcontents

\newpage

\section{Introduction}

In \cite{Nic21}, the correlation functions of an open spin 1/2  XXX  quantum chain with unparallel boundary magnetic fields were derived. Here, we continue this program by considering the more complicated case of an open XXZ chain with unparallel boundary magnetic fields. More precisely, we explain how to compute correlation functions in the case of an open spin 1/2 XXZ quantum chain coupled with a completely arbitrary magnetic field on the first site of the chain and with a specific fixed longitudinal magnetic field on the last site of the chain. 

We develop our analysis in the framework of the quantum Separation of Variables (SoV) \cite{Skl85,Skl85a,Skl90,Skl92,Skl95,Skl96,BabBS96,Smi98a,Smi01,DerKM01,DerKM03,DerKM03b,BytT06,vonGIPS06,FraSW08,AmiFOW10,NicT10,Nic10a,Nic11,FraGSW11,GroMN12,GroN12,Nic12,Nic13,Nic13a,Nic13b,GroMN14,FalN14,FalKN14,KitMN14,NicT15,LevNT16,NicT16,KitMNT16,JiaKKS16,KitMNT17,MaiNP17,KitMNT18}. This approach, first pioneered by Sklyanin \cite{Skl85,Skl85a,Skl90,Skl92,Skl95,Skl96} in the framework of the quantum inverse scattering method (QISM) \cite{FadS78,FadST79,FadT79,Skl79,Skl79a,FadT81,Skl82,Fad82,Fad96,BogIK93L}, has been more recently reformulated in \cite{MaiN18} on the pure basis of the integrable structure of the model, hence extending its range of application even to higher rank cases\footnote{See also \cite{Skl96,Smi01,MarS16,GroLS17} for previous developments.} in \cite{MaiN18,RyaV19,MaiN19,MaiN19d,MaiN19b,MaiN19c,MaiNV20,RyaV20}.
On the one hand, and contrary to Bethe Ansatz \cite{Bet31} and its variants, the SoV approach does not rely on an Ansatz, which means that the completeness of the SoV description of the eigenvalues and eigenstates of the considered models follows by construction. On the other hand, the SoV approach has the advantage to produce naturally and quite universally determinant formulae for the scalar products of the so-called {\it separate states} \cite{GroMN12,Nic12,Nic13,Nic13a,Nic13b,GroMN14,FalN14,FalKN14,LevNT16,KitMNT16,KitMNT17,MaiNP17,KitMNT18}\footnote{This is the case for the rank one models while in \cite{MaiNV20b} it has been shown how these type of formulae extend to the gl(3) higher rank case, see also the interesting and recent papers \cite{CavGL19,GroLRV20}.}, a class of states which notably contains the eigenstates of the model. In particular, in \cite{KitMNT18} were derived scalar products formulae in determinant form for the separate states of the open spin 1/2 XXZ quantum chains under quite general boundary conditions; this is the first required step to compute correlation functions. 

In the last decades, there have been impressive progresses concerning the exact determination of correlation functions of quantum integrable models. However, the state of the art remains quite unsatisfactory with respect to the type of boundary conditions for which such results are available. In fact, most of the results that have been obtained so far concern the quantum non-linear Schr\"odinger model or the XXX/XXZ spin 1/2 quantum chains with periodic boundary conditions or, for the latter models, with open boundary conditions with longitudinal boundary magnetic fields, i.e. along the $z$-direction. 
More precisely, 
correlation functions of the XXX/XXZ spin 1/2 chains were first computed  in infinite volume and at zero temperature in the framework of the $q$-vertex operator approach in \cite{JimM95L,JimMMN92,JimM96}. These results were confirmed and generalized to the case of a non-zero global magnetic field by means of the study of the periodic chain in finite volume by algebraic Bethe Ansatz (ABA) in \cite{KitMT99,MaiT00,KitMT00,KitMST02a,KitMST05a,KitMST05b,KitKMST07}, and to the temperature cases in \cite{GohKS04,GohKS05,BooGKS07,GohS10,DugGKS15,GohKKKS17}. See also \cite{BooJMST05,BooJMST06,BooJMST06a,BooJMST06b,BooJMST06c,BooJMST07,BooJMST09,JimMS09,JimMS11,MesP14,Poz17} for relevant results on correlation functions by the use of the hidden Grassmann structure.
Still in this periodic background, analytical study in the thermodynamic limit of long distances two-point and multi-point correlation functions have been developed in \cite{KitKMST09a,KitKMST09b,KitKMST09c,KozMS11a,KozMS11b,KozT11,KitKMST11a,KitKMST11b,KitKMST12,DugGK13,KitKMT14} and numerical study of the dynamical structure factors in \cite{CauHM05,CauM05,PerSCHMWA06} have made accessible comparison with experimental settings as neutron scattering \cite{KenCTVHST02}. Correlation functions in the case of an open spin chain were considered in \cite{JimKKKM95,JimKKMW95,KitKMNST07,KitKMNST08}, but only for longitudinal boundary fields, i.e. along the $z$-direction. In fact, for different kind of boundary conditions, usual Bethe Ansatz techniques can in principle no longer be used to compute correlations.
In the case of the XXX spin chain, these limitations on the computation of correlation functions have been overcome only recently by the consideration of the model in the SoV framework: in \cite{NicPT20}, correlation functions were computed for all types of quasi-periodic (including anti-periodic) boundary conditions, and in \cite{Nic21}, correlation functions were computed in the open case  with unparallel boundary magnetic fields.
As previously announced, we will extend these results here to the case of an open XXZ spin 1/2 quantum chains with unparallel boundary magnetic fields. 

Open quantum spin chains have been recently at the center of a large research activity, see e.g. \cite{AlcBBBQ87,Skl88,GhoZ94,JimKKKM95,JimKKMW95,FanHSY96,Nep02,Nep04,CaoLSW03,YanZ07,Bas06,BasK07,KitKMNST07,KitKMNST08,CraRS10,CraRS11,FilK11,FraSW08,FraGSW11,Nic12,CaoYSW13b,FalN14,FalKN14,KitMN14,BelC13,Bel15,BelP15,BelP15b,AvaBGP15,BelP16,GriDT19} and related references, with application to the studies of classical stochastic models, as asymmetric simple exclusion models \cite{deGE05}, and also to
modelling numerous applications in quantum condensed matter physics, such as
out-of-equilibrium and transport properties in the spin chains \cite{Pro11}. Exact results on the spectrum for parallel magnetic fields have been accessible since the seminal works in coordinate Bethe ansatz \cite{AlcBBBQ87} and algebraic Bethe ansatz \cite{Skl88}. Instead, for general unparallel boundary magnetic fields, the description of the spectrum was only addressed recently, and is still not completely well understood, in particular concerning the accurate description of the Bethe roots for the ground state in the thermodynamic limit.
In fact, by using the reflection equation first introduced by Cherednik \cite{Che84}, Sklyanin \cite{Skl88} has extended the quantum inverse scattering method (QISM) to open quantum chains. Whereas this algebraic QISM framework applies {\em a priori} to both parallel and unparallel boundary magnetic fields, its ABA implementation as proposed in \cite{Skl88} holds only for diagonal boundary matrices\footnote{The boundary matrices allow to parametrize the magnetic fields at the boundaries of the quantum spin chain.},
i.e. for longitudinal parallel boundary magnetic fields. Other integrable techniques have been developed to handle the non-diagonal case, i.e. generally speaking the unparallel cases. In \cite{Nep02}, under a special constrain relating the parameters of the two boundary matrices, a first description of the spin chain spectrum with unparallel boundary fields has been obtained, by using the fusion procedure \cite{KulRS81}, in terms of polynomial solutions of some $TQ$-equation of Baxter's type \cite{Bax82L}, for the roots of unity points and later in \cite{Nep04} for general values of the anisotropy parameter. Under such a constraint, generalized Bethe eigenstates have also been constructed \cite{FanHSY96,CaoLSW03,YanZ07} by using Baxter's vertex-IRF transformations \cite{Bax73a} so as to simplify the boundary matrices. Such a constraint appears also in other integrable approaches leading to the spectrum description in terms of polynomial solutions of ordinary $TQ$-equations, such as by coordinate Bethe ansatz with elements of matrix product ansatz \cite{CraRS10,CraRS11}, $q$-Onsager algebra \cite{Bas06,BasK07} etc. This constraint has been overcome only more recently. In \cite{CaoYSW13b} a description of the spectrum of these open chains was proposed on the basis of analytic properties and functional relations satisfied by the transfer matrices. The proposed spectrum description has been there presented by polynomial solutions of inhomogeneous $TQ$-equations, i.e. equations admitting some extra term, see also \cite{Nep13}. Such a description of the spectrum in terms of an inhomogeneous $TQ$-equation also appears in the framework of the modified algebraic Bethe ansatz \cite{BelC13,Bel15,BelP15,AvaBGP15} dealing with unconstrained boundary conditions. The SoV method has been used, in particular, for open spin chains with the most general unconstrained non-diagonal boundary matrices \cite{FraSW08,FraGSW11,Nic12,FalKN14,FalN14}, and in \cite{KitMN14} SoV approach was used to prove that the complete spectrum characterization of these open chains, indeed, can be formulated in terms of polynomial solutions to functional $TQ$-equations of Baxter type which, for the most general unconstrained boundary matrices, have the aforesaid inhomogeneous extra term.

Our aim is here to derive exact expressions for the correlation functions of local/quasi-local operators in the vicinity of the first site of the chain. It is natural to expect such correlation functions to be strongly influenced by the magnetic field on this first site; in fact, we expect an explicit dependance of the result on the boundary field at site 1. Instead, for long enough chains, we expect the influence of the boundary field at the other side of the chain to be only indirect (i.e. mainly encoded in the structure of the ground state and not in the explicit expression of the correlation functions themselves, see \cite{GriDT19}). 
This type of physical argument motivates us to initiate our study by considering some special choice of the boundary magnetic field at site $N$. The idea is to adjust the magnetic field at site $N$ in such a way that we can gauge out some of the difficulties due to the consideration of non-diagonal boundaries while keeping unparallel boundary magnetic fields. In particular, we fix the magnetic field on site $N$ to a fixed value along the $z$-direction, whereas the magnetic field on site 1 is completely arbitrary. On the one hand, this enables us to have an easy description of the thermodynamic ground state distribution of the Bethe roots, which, as already mentioned, is not known in the most general case: indeed, in this special case of unparallel boundary fields that we consider here, the transfer matrix happens to be  isospectral to the transfer matrix of a chain with parallel longitudinal boundary magnetic fields.
On the other hand, this special choice of unparallel boundary magnetic fields is done here so as to simplify the computation of the action of the local/quasi local operators on the transfer matrix eigenstates. Let us recall that a reconstruction formula expressing the local operators in terms of the generators of the reflection algebra is so far missing in the literature, and that the alternative method proposed in \cite{KitKMNST07} consists in using the known reconstruction in terms of the bulk Yang-Baxter algebra, and in decomposing the boundary eigenstates in terms of bulk Bethe states ({\em boundary-bulk decomposition}) so as to be able to act with the local operators on these states. This is  possible due to the fact that the reflection algebra generators are quadratic functions of the bulk ones, and that the eigenstates are expressed in the form of Bethe states in terms of these generators. Here, we extend these constructions to the current unparallel boundary conditions. This requires first the introduction of a gauge deformation of the Yang-Baxter algebra and the generalization to it of the reconstruction formulae for the local operators. Then, we have to compute the boundary-bulk decomposition of the separate states in terms of some gauged version of bulk Bethe states. Finally, we have to compute the action of local operators first on the gauged bulk Bethe states and then on the boundary separate states by using the boundary-bulk decomposition. Here, the special choice of unparallel boundary magnetic fields is done so as to simplify the form of the boundary-bulk decomposition and to have an easier access to the formulae for the action of local operators. In this paper we develop all these technical steps which, together with the known formulae for the scalar products in the SoV framework\footnote{It is worth mentioning that, the recent and interesting results on scalar products for open chains with general boundary conditions \cite{BelPS21}, see also \cite{BelS19a,SlaZZ20,BelP16,FilK11,YanCFHHSZ11,DuvP15}, may put a basis for the computation of correlation functions in a generalized/modified Bethe Ansatz framework.}, give us access to the computation of the correlation functions here presented. 

\medskip

The paper is organized as  follows. 

Section~\ref{sec-model} provides a technical introduction to the model we consider, the XXZ spin 1/2 chain with general boundary conditions. There, we introduce the reflection algebras leading to the integrable description of the spin chains by using two kinds of symmetries which reverse the role of the plus/minus reflection algebras, and which happen to be convenient in that they enable us to make a direct use of some former results~\cite{KitMNT18}. We also introduce the gauge-transformed version of this algebra by Baxter's vertex-IRF transformation, which enables one to simplify the form of the boundary matrices. 

In section~\ref{sec-spectrum}, we recall the complete  characterization of the transfer matrix spectrum that was obtained for this model in the SoV framework. There, we particularize the SoV spectrum construction to the special boundary conditions associated to one arbitrary magnetic field in the site 1 and a fixed magnetic field oriented along the $z$-direction in the site $N$. In such a case, the transfer matrix spectrum is completely characterized by the polynomial solutions of a homogeneous $TQ$-equation of Baxter type, and the separate eigenstates can be constructed in the form of Bethe states. 
We also prove the isospectrality of this case with the case of parallel boundary magnetic fields along the $z$-direction. 

In section~\ref{sec-act-state}, we compute the action of a basis of quasi-local operators on Bethe-type boundary states. More precisely, we define a set of quasi-local operators on the first $m$ sites of the chain, which can be written as monomials in terms of the generator of the gauged Yang-Baxter algebra. We prove that for any fixed finite $m$, this operator set is an operator basis of the corresponding $m$-sites quasi-local operators and we compute the action of its elements on gauged bulk Bethe-type states. By means of the boundary-bulk decomposition of the boundary states, we rewrite this action as an action on general boundary gauged Bethe-type states. 

Finally, in section~\ref{sec-corr-fct}, this result concerning the action of local operators on transfer matrix eigenstates, as well as the known scalar product formulae (previously derived in the SoV framework \cite{KitMNT16, KitMNT18}) allow us to compute the correlation functions. Here, we explicitly present the expressions for the correlation functions of a particular class of elements of the aformentioned local operator basis, both on the finite chain and on the half-infinite chain. In the latter case, we obtain multiple integral representations that are quite similar, in their form, to the one previously obtained in \cite{KitKMNST07} for the case of a chain with diagonal boundary conditions, with however some  differences related to the different choice of boundary field at site 1 with respect to  \cite{KitKMNST07}.

The paper contains also three appendices with some technical details.
In appendix~\ref{app-gaugeYBbulk}, we gather some useful properties of the bulk gauge Yang-Baxter algebra. We show notably that the latter can in fact be described in terms of a single operator family with three continuous parameters, written as linear combinations of the ungauged Yang-Baxter generators: one parameter being the spectral parameter and the other two - the gauge parameters - defining the coefficients of the linear combination. In particular, we derive the commutation relations of this operator family and its action on gauged reference states and gauged Bethe-type states. 
Appendix~\ref{app-bound-bulk} is devoted to the boundary-bulk decomposition: we notably explain how to decompose  the transfer matrix eigenstates, and more generally the SoV separate states as linear combinations of generalized gauged Bethe states; i.e. a gauged version of the so-called boundary-bulk decomposition \cite{KitKMNST07, KitKMNST08}.
Finally, in appendix~\ref{app-diag}, we show how, taking the limit in which the boundary at site 1 becomes diagonal, we can obtain from our final result the result of \cite{KitKMNST07}.
\section{The open spin-1/2 XXZ quantum chain}
\label{sec-model}

The Hamiltonian of the integrable quantum spin-1/2 XXZ chain with open boundary conditions and general boundary fields can be written as: 
\begin{align}
H& =\sum_{n=1}^{N-1}\Big[\sigma _{n}^{x}\sigma _{n+1}^{x}+\sigma
_{n}^{y}\sigma _{n+1}^{y}+\cosh \eta \,\sigma _{n}^{z}\sigma _{n+1}^{z}\Big]
\notag \\
& \hspace{2cm}+\frac{\sinh \eta }{\sinh \varsigma _{-}}\Big[\sigma
_{1}^{z}\cosh \varsigma _{-}+2\kappa _{-}\big(\sigma _{1}^{x}\cosh \tau
_{-}+i\sigma _{1}^{y}\sinh \tau _{-}\big)\Big]  \notag \\
& \hspace{2cm}+\frac{\sinh \eta }{\sinh \varsigma _{+}}\Big[\sigma _{
N}^{z}\cosh \varsigma _{+}+2\kappa _{+}\big(\sigma _{N}^{x}\cosh \tau
_{+}+i\sigma _{N}^{y}\sinh \tau _{+}\big)\Big].  \label{H-XXZ-Non-D}
\end{align}
The operators $\sigma _{n}^{\alpha },\ \alpha \in \{x,y,z\}$, act as the corresponding Pauli matrices on the local quantum spin space $\mathcal{H}_{n}\simeq \mathbb{C}^{2}$ at site $n$, so that $H$ is an operator acting on the $2^N$-dimensional quantum space $\mathcal{H}=\otimes _{n=1}^{N}\mathcal{H}_{n}$. 
The anisotropy parameter is $\Delta =\cosh \eta $.
The boundary fields are parametrized by the six boundary parameters $\varsigma _{\pm }$, $\kappa _{\pm }$, $\tau _{\pm }$. We shall also use the following re-parametrization $\varphi _{\pm },\psi_\pm$ of the boundary parameters $\varsigma _{\pm },\kappa _{\pm }$:
\begin{equation}
\sinh \varphi _{\pm }\,\cosh \psi _{\pm }=\frac{\sinh \varsigma _{\pm }}{%
2\kappa _{\pm }},\qquad \cosh \varphi _{\pm }\,\sinh \psi _{\pm }=\frac{%
\cosh \varsigma _{\pm }}{2\kappa _{\pm }}.  \label{reparam-bords}
\end{equation}
The Hamiltonian \eqref{H-XXZ-Non-D} corresponds to the most general case of open boundary conditions with non-diagonal integrable boundary interactions. 
It is worth noting that this Hamiltonian is manifestly invariant under the
following simultaneous change of sign:
\begin{equation}\label{symm-eta}
\{\eta ,\varsigma _{\pm }\}\rightarrow \{-\eta ,-\varsigma _{\pm }\},
\end{equation}
while the others parameters $\kappa _{\pm }$ and $\tau _{\pm }$ remain
unchanged. It is also invariant under a relabelling of each site $n$ by $N-n+1$ together with the interchange of the boundary parameters $+$ and $-$:
\begin{equation}\label{symm-miroir}
  \begin{cases}
   & n\rightarrow N-n+1,\quad 1\le n\le N, \\
   &\{\varsigma_\pm,\kappa_\pm,\tau_\pm\}\rightarrow \{\varsigma_\mp,\kappa_\mp,\tau_\mp\}.
   \end{cases}
\end{equation}

\subsection{Reflection algebra and symmetries}

The open chain with Hamiltonian \eqref{H-XXZ-Non-D} can be studied in the formalism introduced in \cite{Skl88}, that we recall here.

Let us introduce the 6-vertex trigonometric solution of the Yang-Baxter equation, 
\begin{equation}
R_{12}(\lambda )=%
\begin{pmatrix}
\sinh (\lambda +\eta ) & 0 & 0 & 0 \\ 
0 & \sinh \lambda & \sinh \eta & 0 \\ 
0 & \sinh \eta & \sinh \lambda & 0 \\ 
0 & 0 & 0 & \sinh (\lambda +\eta )%
\end{pmatrix}%
\in \text{End}(\mathbb{C}^{2}\otimes \mathbb{C}^{2}),  \label{R-6V}
\end{equation}
and let us consider, on $\mathbb{C}^{2}\otimes \mathbb{C}^{2}\otimes \mathcal{H}$, the reflection equation for the so-called boundary monodromy matrix $\mathcal{U}_{-}(\lambda )\in \text{End}(\mathbb{C}^{2}\otimes 
\mathcal{H})$:
\begin{equation}
R_{21}(\lambda -\mu )\,\mathcal{U}_{-,1}(\lambda )\,R_{12}(\lambda +\mu
-\eta )\,\mathcal{U}_{-,2}(\mu )=\mathcal{U}_{-,2}(\mu )\,R_{21}(\lambda
+\mu -\eta )\,\mathcal{U}_{-,1}(\lambda )\,R_{12}(\lambda -\mu ).
\label{bYB}
\end{equation}
Here the subscripts parametrize the subspaces of $\mathbb{C}%
^{2}\otimes \mathbb{C}^{2}$ on which the corresponding operators act
non-trivially, and $R_{21}(\lambda )$ is obtained from $R_{12}(\lambda )$ \eqref{R-6V} as $R_{21}(\lambda )=P_{12}\,R_{12}(\lambda )\,P_{12}$, where $P_{12}$ is the
permutation operator on $\mathbb{C}^{2}\otimes \mathbb{C}^{2}$.

For convenience, we shall define here two different realizations of the reflection algebra, i.e. construct two different  solutions $\mathcal{U}_{-}(\lambda )$ of the reflection equation \eqref{bYB}. These realizations slightly differs from the standard one introduced by Sklyanin~\cite{Skl88} and used in our previous work \cite{KitMNT18} in that they rely on the two types of symmetries \eqref{symm-miroir} and \eqref{symm-eta} of the model. They can easily be related to each other as described below.

Let us first introduce the following monodromy matrix $T_0(\lambda)\in\End(\mathcal{H}_0\otimes\mathcal{H})$, where $\mathcal{H}_0=\mathbb{C}^2$ is the so-called auxiliary space:
\begin{equation}\label{mon-T}
T_{0}(\lambda )=R_{01}(\lambda -\xi _{1}-\eta /2)\dots R_{0N}(\lambda -\xi
_{N}-\eta /2).
\end{equation}
Note that here we have intentionally reversed the order in the quantum
sites w.r.t. the one that is often used in the literature, and in particular in our previous series of papers.
$T_{0}(\lambda )$ is a solution of the Yang-Baxter equation on $\mathbb{C}^2\otimes \mathbb{C}^2\otimes \mathcal{H}$:
\begin{equation}
R_{1,2}(\lambda -\mu )\, T_{1}(\lambda )\, T_{2}(\mu )=T_{2}(\mu )\, T_{1}(\lambda
)\, R_{1,2}(\lambda -\mu ).
\end{equation}
Let us define from \eqref{mon-T} the following boundary monodromy matrix :
\begin{equation}
\mathcal{U}_{-,0}(\lambda )=T_{0}(\lambda )\,K_{-,0}(\lambda )\,\hat{T}%
_{0}(\lambda )=%
\begin{pmatrix}
\mathcal{A}_{-}(\lambda ) & \mathcal{B}_{-}(\lambda ) \\ 
\mathcal{C}_{-}(\lambda ) & \mathcal{D}_{-}(\lambda )%
\end{pmatrix}%
,  \label{def-U+}
\end{equation}
solution of the reflection equation \eqref{bYB}, where
\begin{align}
\hat{T}_{0}(\lambda ) &=(-1)^{N}\,\sigma _{0}^{y}\,T_{0}^{t_{0}}(-\lambda
)\,\sigma _{0}^{y}  \notag \\
&=R_{0N}(\lambda +\xi _{N}-\eta /2)\dots R_{01}(\lambda +\xi _{1}-\eta /2),
\label{Mhat}
\end{align}
and
\begin{equation}\label{def-Kpm}
K_{-}(\lambda )=K(\lambda ;\varsigma _{+},\kappa _{+},\tau _{+}),\qquad
K_{+}(\lambda )=K(\lambda +\eta ;\varsigma _{-},\kappa _{-},\tau _{-}).
\end{equation}
We have here introduced the boundary parameters plus in the
minus boundary $K$-matrix and vice versa, while the definition of the $K$-matrix is unchanged with respect to our previous paper \cite{KitMNT18}:
\begin{equation}
K(\lambda ;\varsigma ,\kappa ,\tau )=\frac{1}{\sinh \varsigma }\,%
\begin{pmatrix}
\sinh (\lambda -\eta /2+\varsigma ) & \kappa e^{\tau }\sinh (2\lambda -\eta )
\\ 
\kappa e^{-\tau }\sinh (2\lambda -\eta ) & \sinh (\varsigma -\lambda +\eta
/2)%
\end{pmatrix}
.  \label{mat-K}
\end{equation}
In other words, we have here used the symmetry \eqref{symm-miroir} of the model with respect to our previous works.
We recall that \eqref{mat-K} is the most general scalar solution \cite{deVG93,deVG94,GhoZ94} of the
reflection equation \eqref{bYB} for general values of the parameters $\varsigma ,$ $\kappa $ and $\tau $.

It follows  \cite{Skl88} that the transfer matrices
, 
\begin{equation}
\mathcal{T}(\lambda )=\text{tr}_{0}\{K_{+,0}(\lambda )\,T_{0}(\lambda
)\,K_{-,0}(\lambda )\,\hat{T}_{0}(\lambda )\},  \label{transfer}
\end{equation}
form a one-parameter family of commuting operators on $\mathcal{H}$. 
The Hamiltonian \eqref{H-XXZ-Non-D} of the spin-1/2 open chain can be obtained  in the homogeneous limit  $\xi _{m}=0$, $m=1,\ldots ,N$, as the following derivative of the transfer matrix \eqref{transfer}: 
\begin{equation}
H=\frac{2\,(\sinh \eta )^{1-2N}}{\text{tr}\{K_{+}(\eta /2)\}\,\text{tr}%
\{K_{-}(\eta /2)\}}\,\frac{d}{d\lambda }\mathcal{T}(\lambda )_{\,\vrule %
height13ptdepth1pt\>{\lambda =\eta /2}\!}+\text{constant.}  \label{Ht}
\end{equation}

We also recall that the quantum determinant of $\mathcal{U}_{-}(\lambda )$,
\begin{align}
 \mathrm{det}_{q}\mathcal{U}_{-}(\lambda )
 & = \sinh (2\lambda -2\eta ) 
 \big[\mathcal{A}_{-}(\eta /2\pm \lambda )\,\mathcal{A}_{-}(\eta /2\mp \lambda )+%
\mathcal{B}_{-}(\eta /2\pm \lambda )\,\mathcal{C}_{-}(\eta /2\mp \lambda ) \big]
\notag \\
& =\sinh (2\lambda -2\eta ) 
 \big[\mathcal{D}_{-}(\eta /2\pm \lambda )\,\mathcal{D}_{-}(\eta /2\mp \lambda
)+\mathcal{C}_{-}(\eta /2\pm \lambda )\,\mathcal{B}_{-}(\eta /2\mp \lambda )\big]
\notag\\
&=\mathrm{det}_q T(\lambda)\, \mathrm{det}_q T(-\lambda)\, \mathrm{det}_q K_-(\lambda),
\label{det-U-}
\end{align}
is a central element of the reflection algebra: $[\det_{q}\mathcal{U}%
_{-}(\lambda ),\mathcal{U}_{-}(\lambda )]=0$.
In \eqref{det-U-}, $\det_q T(\lambda)$ stands for the bulk quantum determinant which can be expressed as
\begin{equation}
\mathrm{det}_{q}T(\lambda )=a(\lambda +\eta /2)\,d(\lambda -\eta /2),
\label{det-M}
\end{equation}
where 
\begin{equation}
a(\lambda )=\prod_{n=1}^{N}\sinh (\lambda -\xi _{n}+\eta /2),\qquad
d(\lambda )=\prod_{n=1}^{N}\sinh (\lambda -\xi _{n}-\eta /2),  \label{a-d}
\end{equation}
and $\mathrm{det}_q K_-(\lambda)$ stands for the quantum determinant of the scalar boundary matrix $K_{-}(\lambda) $. The quantum determinant of the scalar boundary matrices $K_{\mp
}(\lambda )$ \eqref{def-Kpm} can be expressed as 
\begin{equation}
\frac{\det_{q}K_{\mp }(\lambda )}{\sinh (2\lambda \mp 2\eta )}=\mp \frac{%
\big(\sinh ^{2}\lambda -\sinh ^{2}\varphi _{\pm }\big)\big(\sinh ^{2}\lambda
+\cosh ^{2}\psi _{\pm }\big)}{\sinh ^{2}\varphi _{\pm }\,\cosh ^{2}\psi
_{\pm }}.  \label{det-K-}
\end{equation}
The boundary quantum monodromy matrix $\mathcal{U}_{-}(\lambda )$  \eqref{def-U+} satisfies the  quantum 
inversion relation  
\begin{equation}
\mathcal{U}_{-}(\lambda +\eta /2)\,\mathcal{U}_{-}(-\lambda +\eta /2)=\frac{%
\det_{q}\mathcal{U}_{-}(\lambda )}{\sinh (2\lambda -2\eta )},  \label{inv-U-}
\end{equation}
and its elements satisfy the following relations:
\begin{align}
   &\mathcal{A}_-(\lambda)=\frac{\sinh\eta}{\sinh(2\lambda)}\,\mathcal{D}_-(\lambda)+\frac{\sinh(2\lambda-\eta)}{\sinh(2\lambda)}\,\mathcal{D}_-(-\lambda),
   \label{rel-A-D}\\
   &\mathcal{D}_-(\lambda)=\frac{\sinh\eta}{\sinh(2\lambda)}\,\mathcal{A}_-(\lambda)+\frac{\sinh(2\lambda-\eta)}{\sinh(2\lambda)}\,\mathcal{A}_-(-\lambda),
   \label{rel-D-A}\\
   &\mathcal{B}_-(-\lambda)=-\frac{\sinh(2\lambda+\eta)}{\sinh(2\lambda-\eta)}\,\mathcal{B}_-(\lambda),
   \qquad
   \mathcal{C}_-(-\lambda)=-\frac{\sinh(2\lambda+\eta)}{\sinh(2\lambda-\eta)}\,\mathcal{C}_-(\lambda).
   \label{rel-B-C}
\end{align}


Let us now introduce a second realization relying on the symmetry \eqref{symm-eta}.
We define
\begin{equation}\label{sym-par}
    \bar{\eta}=-\eta ,
    \quad
    \bar{\varsigma}_{\pm }=-\varsigma _{\pm },
    \quad
    \bar{\kappa}_{\pm }=\kappa _{\pm },
    \quad
    \bar{\tau}_{\pm }=\tau _{\pm },
\end{equation}
Let us consider the R-matrix
\begin{equation}\label{barR}
\bar R_{12}(\lambda)=%
\begin{pmatrix}
\sinh (\lambda -\eta ) & 0 & 0 & 0 \\ 
0 & \sinh \lambda & -\sinh \eta & 0 \\ 
0 & -\sinh \eta & \sinh \lambda & 0 \\ 
0 & 0 & 0 & \sinh (\lambda -\eta )%
\end{pmatrix}
=-R_{12}(-\lambda),
\end{equation}
which corresponds to a change of parameter $\eta\rightarrow \bar\eta=-\eta$ with respect to \eqref{R-6V}, and let us introduce the bulk monodromy matrix
\begin{equation}
M(\lambda )=\bar R_{0N}(\lambda -\xi _{N}+\eta /2)\dots \bar R_{01}(\lambda-\xi _{1}+\eta /2)=%
\begin{pmatrix}
A(\lambda ) & B(\lambda ) \\ 
C(\lambda ) & D(\lambda )%
\end{pmatrix},  \label{bulk-mon}
\end{equation}
which satisfies the Yang-Baxter equation w.r.t. $\bar R_{12}(\lambda )$. From it we also define the matrix 
\begin{align}
\hat{M}(\lambda ) 
     &=(-1)^{N}\,\sigma _{0}^{y}\,M^{t_{0}}(-\lambda )\,\sigma_{0}^{y} \nonumber\\
     &=\bar R_{01}(\lambda +\xi _1+\eta /2)\dots \bar R_{0N}(\lambda +\xi_N+\eta /2).
\end{align}
Note that the bulk monodromy matrices \eqref{bulk-mon} and \eqref{mon-T} are related by
\begin{equation}\label{MT}
M(-\lambda )=(-1)^{N}\hat{T}(\lambda ),
\qquad
\hat{M}(-\lambda )=(-1)^{N}T(\lambda ).
\end{equation}
Similarly, we define the analogue of \eqref{mat-K} with a change of parameter $\eta\rightarrow \bar\eta=-\eta$:
\begin{equation}
\bar K(\lambda ;\varsigma ,\kappa ,\tau )=\frac{1}{\sinh \varsigma }\,%
\begin{pmatrix}
\sinh (\lambda +\eta /2+\varsigma ) & \kappa e^{\tau }\sinh (2\lambda +\eta )
\\ 
\kappa e^{-\tau }\sinh (2\lambda +\eta ) & \sinh (\varsigma -\lambda -\eta
/2)%
\end{pmatrix},
 \label{mat-barK}
\end{equation}
and we introduce
\begin{equation}\label{def-barKpm}
 \bar K_{-}(\lambda )=\bar K(\lambda ;\bar \varsigma _-,\bar \kappa _-,\bar \tau _-),\qquad
 \bar K_{+}(\lambda )=\bar K(\lambda -\eta ;\bar \varsigma _+,\bar \kappa _+,\bar \tau _+),
\end{equation}
in terms of the parameters \eqref{sym-par}.
From \eqref{bulk-mon}-\eqref{def-barKpm}, we define the boundary monodromy matrices
\begin{align}
&\overline{\mathcal{U}}_{-}(\lambda)\equiv
\overline{\mathcal{U}}_{-}(\lambda |\bar{\eta},\bar{\varsigma}_{-},\bar{\kappa}_{-},\bar{\tau}_{-}) =M(\lambda)\, \bar K_{-}(\lambda )\,\hat{M}(\lambda), 
\label{barU-}\\
&\overline{\mathcal{U}}_{+}^{t_{0}}(\lambda)\equiv
\overline{\mathcal{U}}_{+}^{t_{0}}(\lambda |\bar{\eta},\bar{\varsigma}_{+},\bar{\kappa}_{+},\bar{\tau}_{+}) 
=M^{t_{0}}(\lambda )\, \bar K_{+}^{t_{0}}(\lambda )\,\hat{M}^{t_{0}}(\lambda).
\label{barU+}
\end{align}

These boundary matrices can easily be expressed in terms of the previous ones. More precisely, the following identities hold, which can be shown by direct computation:

\begin{proposition}
The boundary monodromy matrix \eqref{def-U+} can be expressed in terms of the bulk monodromy matrix \eqref{bulk-mon} as
\begin{equation}\label{bmon-M}
  \mathcal{U}_{-}(\lambda)
  =\hat{M}(-\lambda )\, K_{-}(\lambda )\,M(-\lambda ),
\end{equation}
whereas the boundary matrices \eqref{def-Kpm} and \eqref{def-barKpm} are related by:
\begin{align}
 &K_{-}(\lambda ) 
 =\sigma_{0}^{y}\, \bar K_{+}^{t_{0}}(\lambda )\,\sigma _{0}^{y}, 
 \label{K-barK+}\\
 &K_{+}(\lambda )
  =\sigma_{0}^{y}\, \bar K_{-}^{t_{0}}(\lambda)\,\sigma _{0}^{y}.
  \label{K+barK-}
\end{align}
As a consequence, the boundary monodromy matrix \eqref{def-U+} can be expressed in terms of the boundary matrix \eqref{barU+} as
\begin{equation}\label{U-barU+}
  \mathcal{U}_{-}(\lambda )
 =\sigma _{0}^{y}\,\overline{\mathcal{U}}_{+}^{t_{0}}(\lambda)\,\sigma _{0}^{y}.
\end{equation}
\end{proposition}

In the following, for technical reasons, it will be convenient to use the representation \eqref{bmon-M} of the boundary monodromy matrix $\mathcal{U}_-(\lambda)$ in terms of the bulk monodromy matrix $M(\lambda)$.

\subsection{Gauge transformation of the reflection algebra}

The quantum version of the Separation of Variable (SoV) approach \cite{Skl85,Skl85a,Skl92,Skl95}, see
also \cite{MaiN18,MaiN19c} for a more general version of the SoV approach, has been used to solve the spectral model of the boundary
transfer matrices \eqref{transfer}.  Even if this is not required in the latter framework, we recall here the gauge transformations introduced in \cite{CaoLSW03,FalKN14} and used in \cite{KitMNT18} to
transform the model into an effective one in which at least one of the
boundary matrices becomes triangular \cite{Nic12}.

\subsubsection{Vertex-IRF transformation}

Let us recall the definition of the trigonometric solid-on-solid (SOS) (or dynamical) $R$-matrix:
\begin{equation}
R^{\mathrm{SOS}}(\lambda |\beta )=%
\begin{pmatrix}
\sinh (\lambda +\eta ) & 0 & 0 & 0 \\ 
0 & \frac{\sinh (\eta (\beta +1))}{\sinh (\eta \beta )}\,\sinh \lambda & 
\frac{\sinh (\lambda +\eta \beta )}{\sinh (\eta \beta )}\,\sinh \eta & 0 \\ 
0 & \frac{\sinh (\eta \beta -\lambda )}{\sinh (\eta \beta )}\,\sinh \eta & 
\frac{\sinh (\eta (\beta -1))}{\sinh (\eta \beta )}\,\sinh \lambda & 0 \\ 
0 & 0 & 0 & \sinh (\lambda +\eta )%
\end{pmatrix}%
,  \label{R-SOS}
\end{equation}
which also depends on the so-called parameter $\beta $. This $R$-matrix can be related to the  6-vertex $R$-matrix \eqref{R-6V} by the so-called Vertex-IRF transformation as
\begin{equation}
R_{12}(\lambda -\mu )\,S_{1}(\lambda |\alpha ,\beta )\,S_{2}(\mu |\alpha
,\beta +\sigma _{1}^{z})=S_{2}(\mu |\alpha ,\beta )\,S_{1}(\lambda |\alpha
,\beta +\sigma _{2}^{z})\,R_{12}^{\mathrm{SOS}}(\lambda -\mu |\beta ),
\label{Vertex-IRF}
\end{equation}
or equivalently as
\begin{equation}
R_{12}(\lambda -\mu )\,S_{2}(-\mu |\alpha ,\beta )\,S_{1}(-\lambda |\alpha
,\beta +\sigma _{2}^{z})=S_{1}(-\lambda |\alpha ,\beta )\,S_{2}(-\mu |\alpha
,\beta +\sigma _{1}^{z})\,R_{21}^{\mathrm{SOS}}(\lambda -\mu |\beta ).
\label{Vertex-IRF2}
\end{equation}
The relations \eqref{Vertex-IRF} and \eqref{Vertex-IRF2} involve the Vertex-IRF transformation matrix
\begin{equation}
S(\lambda |\alpha ,\beta )=%
\begin{pmatrix}
e^{\lambda -\eta (\beta +\alpha )} & e^{\lambda +\eta (\beta -\alpha )} \\ 
1 & 1%
\end{pmatrix}%
,  \label{mat-S}
\end{equation}
which depends on the spectral parameter $\lambda$, on the dynamical parameter $\beta$, and on an arbitrary shift  $\alpha $  of the spectral parameter.

\subsubsection{Gauge transformation of the boundary monodromy matrices}

Let us introduce the following gauged transformed boundary
monodromy matrix $\mathcal{U}_{-}(\lambda )$: 
\begin{align}
\mathcal{U}_{-}(\lambda |\alpha ,\beta )& =S_{0}^{-1}(\eta /2-\lambda
|\alpha ,\beta )\ \mathcal{U}_{-}(\lambda )\ S_{0}(\lambda -\eta /2|\alpha
,\beta )  \notag \\
& =%
\begin{pmatrix}
\mathcal{A}_{-}(\lambda |\alpha ,\beta ) & \mathcal{B}_{-}(\lambda |\alpha
,\beta ) \\ 
\mathcal{C}_{-}(\lambda |\alpha ,\beta ) & \mathcal{D}_{-}(\lambda |\alpha
,\beta )%
\end{pmatrix}%
.  \label{gauged-U}
\end{align}
It is worth remarking that, by definition,
\begin{align}
\mathcal{B}_{-}(\lambda |\alpha ,\beta )
&=\mathcal{C}_{-}(\lambda |\alpha,-\beta )
\nonumber\\
&=\frac{e^{\eta \beta }[e^{\lambda -\frac{\eta }{2}+\eta (\beta
-\alpha )}\mathcal{C}_{-}(\lambda )+\mathcal{D}_{-}(\lambda )-e^{2\lambda
-\eta } \mathcal{A}_{-}(\lambda )
-e^{\lambda -\frac{\eta }{2}+\eta (\alpha -\beta )}\mathcal{B}_{-}(\lambda )]}{2\sinh (\eta
\beta )},
\end{align}
so that $\sinh (\eta \beta )\,e^{-\eta \beta }\mathcal{B}_{-}(\lambda |\alpha
,\beta )$ and $\sinh (\eta \beta )\,e^{\eta \beta }\mathcal{C}_{-}(\lambda
|\alpha ,\beta )$ depend on the external gauge parameters only through the combinations $\alpha -\beta $ and $\alpha +\beta $, respectively.
We will denote:
\begin{align}
  &\widehat{\mathcal{B}}_-(\lambda|\alpha-\beta)
  =\sinh (\eta \beta )\,e^{-\eta \beta }\, e^{-(\lambda-\eta/2)}\,
  \mathcal{B}_{-}(\lambda |\alpha,\beta ),
  \label{Bhat}\\
  &\widehat{\mathcal{C}}_-(\lambda|\alpha+\beta)
  =\sinh (\eta \beta )\,e^{\eta \beta }\, e^{-(\lambda-\eta/2)}\,\mathcal{C}_{-}(\lambda |\alpha,\beta ).
  \label{Chat}
\end{align}
Moreover, the gauged monodromy matrix satisfies the dynamical reflection equation:
\begin{multline}
R_{21}^{\mathrm{SOS}}(\lambda -\mu |\beta )\,\mathcal{U}_{-,1}(\lambda
|\alpha ,\beta +\sigma _{2}^{z})\,R_{12}^{\mathrm{SOS}}(\lambda +\mu -\eta
|\beta )\,\mathcal{U}_{-,2}(\mu |\alpha ,\beta +\sigma _{1}^{z})
\label{dyn_refl} \\
=\mathcal{U}_{-,2}(\mu |\alpha ,\beta +\sigma _{1}^{z})\,R_{21}^{\mathrm{SOS}%
}(\lambda +\mu -\eta |\beta )\,\mathcal{U}_{-,1}(\lambda |\alpha ,\beta
+\sigma _{2}^{z})\,R_{12}^{\mathrm{SOS}}(\lambda -\mu |\beta ).
\end{multline}
From the inversion relation \eqref{inv-U-} of $\mathcal{U}_{-}(\lambda )$, we derive the inversion relation for the matrix $\mathcal{U}_{-}(\lambda |\alpha,\beta )$:   
\begin{equation}
\mathcal{U}_{-}(\lambda +\eta /2|\alpha ,\beta )\,\mathcal{U}_{-}(-\lambda
+\eta /2|\alpha ,\beta )=\frac{\det_{q}\mathcal{U}_{-}(\lambda )}{\sinh
(2\lambda -2\eta )},
\end{equation}
where $\det_{q}\mathcal{U}_{-}(\lambda )$ is the quantum determinant %
\eqref{det-U-}, so that
\begin{align}
\frac{\det_{q}\mathcal{U}_{-}(\lambda )}{\sinh (2\lambda -2\eta )}
& =\mathcal{A}_{-}(\eta /2+\epsilon \lambda |\alpha ,\beta )\,\mathcal{A}%
_{-}(\eta /2-\epsilon \lambda |\alpha ,\beta )  
+\mathcal{B}_{-}(\eta /2+\epsilon \lambda |\alpha ,\beta )\,%
\mathcal{C}_{-}(\eta /2-\epsilon \lambda |\alpha ,\beta )  \notag \\
& =\mathcal{D}_{-}(\eta /2+\epsilon \lambda |\alpha ,\beta )\,\mathcal{D}%
_{-}(\eta /2-\epsilon \lambda |\alpha ,\beta )  
+\mathcal{C}_{-}(\eta /2+\epsilon \lambda |\alpha ,\beta )\,%
\mathcal{B}_{-}(\eta /2-\epsilon \lambda |\alpha ,\beta ),  \label{det-U-SOS}
\end{align}
for any $\epsilon\in\{+,-\}$.
Let us also recall the relations between the elements of \eqref{gauged-U}, which can easily be deduced from the relations \eqref{rel-A-D}-\eqref{rel-B-C}:
\begin{align}
  &\mathcal{D}_-(\lambda|\alpha,\beta+1)
  =\frac{\sinh\eta\,\sinh(2\lambda+\eta\beta)}{\sinh(\eta(\beta+1))\,\sinh(2\lambda)}\,\mathcal{A}_-(\lambda|\alpha,\beta-1)
  \nonumber\\
  &\hspace{6cm}+e^{2\lambda}\frac{\sinh(\eta\beta)\,\sinh(2\lambda-\eta)}{\sinh(\eta(\beta+1))\,\sinh(2\lambda)}\, \mathcal{A}_-(-\lambda|\alpha,\beta-1),
  \\
   &\mathcal{A}_-(\lambda|\alpha,\beta-1)
  =-\frac{\sinh\eta\,\sinh(2\lambda-\eta\beta)}{\sinh(\eta(\beta-1))\,\sinh(2\lambda)}\,\mathcal{D}_-(\lambda|\alpha,\beta+1)
  \nonumber\\
  &\hspace{6cm}+e^{2\lambda}\frac{\sinh(\eta\beta)\,\sinh(2\lambda-\eta)}{\sinh(\eta(\beta-1))\,\sinh(2\lambda)}\, \mathcal{D}_-(-\lambda|\alpha,\beta+1),
  \\
  &\mathcal{B}_-(-\lambda|\alpha,\beta)=-e^{-2\lambda}\frac{\sinh(2\lambda+\eta)}{\sinh(2\lambda-\eta)}\,\mathcal{B}_-(\lambda|\alpha,\beta),
  \\
  &\mathcal{C}_-(-\lambda|\alpha,\beta)=-e^{-2\lambda}\frac{\sinh(2\lambda+\eta)}{\sinh(2\lambda-\eta)}\,\mathcal{C}_-(\lambda|\alpha,\beta),
\end{align}

It will be convenient in the following to use the boundary-bulk decomposition \eqref{bmon-M} of the boundary monodromy matrix $\mathcal{U}_-(\lambda)$ \eqref{def-U+} in terms of the bulk monodromy matrix $M(\lambda)$ \eqref{bulk-mon}. The latter can be rewritten at the level of the gauge-transformed boundary monodromy matrix \eqref{gauged-U} as
\begin{equation}\label{bound-bulk-gauge}
\mathcal{U}_{-}(\lambda |\alpha,\beta )
=\hat{M}(-\lambda |(\gamma ,\delta ),(\alpha ,\beta))\,
K_{-}(\lambda |(\gamma ,\delta),(\gamma ^{\prime },\delta ^{\prime }))\,
M(-\lambda |(\gamma^{\prime },\delta ^{\prime }),(\alpha ,\beta )),
\end{equation}
in which we have defined
\begin{align}
M(\lambda |(\alpha ,\beta ),(\gamma ,\delta )) 
&=S^{-1}(-\eta /2-\lambda|\alpha ,\beta )\,M(\lambda )S(-\eta /2-\lambda |\gamma ,\delta ) \nonumber\\
&=
\begin{pmatrix}
A(\lambda |(\alpha ,\beta ),(\gamma ,\delta )) & B(\lambda |(\alpha ,\beta),(\gamma ,\delta )) \\ 
C(\lambda |(\alpha ,\beta ),(\gamma ,\delta )) & D(\lambda |(\alpha ,\beta),(\gamma ,\delta ))
\end{pmatrix} ,
\label{gauge-M}\\
\hat{M}(\lambda |(\alpha ,\beta ),(\gamma ,\delta ))
&=S^{-1}(\lambda +\eta/2|\gamma ,\delta )\,\hat{M}(\lambda )\,S(\lambda +\eta /2|\alpha ,\beta ).
\label{gauge-Mhat}
\end{align}
and
\begin{equation}\label{gauge-K}
K_{-}(\lambda |(\gamma ,\delta ),(\gamma ^{\prime },\delta ^{\prime}))
=S^{-1}(\eta /2-\lambda |\gamma ,\delta )\,K_{-}(\lambda)\,
S(\lambda -\eta /2|\gamma ^{\prime },\delta ^{\prime }).
\end{equation}
It is easy to see that, up to a global normalization factor, the  gauged operators $A$, $B$, $C$ and $D$ elements of \eqref{gauge-M} depend in fact only on two combinations $\alpha\pm\beta$ and $\gamma\pm\delta$ of the gauge parameters. To highlight this dependence, we will sometimes use the following notation:
\begin{equation}\label{redef-gauge-op}
M(\lambda |(\alpha ,\beta ),(\gamma ,\delta )) 
 =\frac{e^{\eta (\alpha +1/2)}}{2\sinh \eta \beta } 
 \begin{pmatrix} 
A(\lambda |\alpha -\beta ,\gamma +\delta ) & B(\lambda |\alpha -\beta,\gamma -\delta ) \\ 
C(\lambda |\alpha +\beta ,\gamma +\delta ) & D(\lambda |\alpha +\beta,\gamma -\delta ) 
\end{pmatrix}. 
\end{equation}
%
We will also use the notation:
\begin{equation}
M_{\epsilon ,\epsilon ^{\prime }}(\lambda |\alpha +(-1)^{\epsilon }\beta,\gamma -(-1)^{\epsilon ^{\prime }}\delta )
=\frac{2\sinh \eta \beta }{e^{\eta (\alpha +1/2)}}\,
M_{\epsilon ,\epsilon ^{\prime }}(\lambda |(\alpha,\beta ),(\gamma ,\delta )).
\end{equation}
Some useful properties of the bulk gauge Yang-Baxter algebra satisfied by \eqref{gauge-M}-\eqref{redef-gauge-op} are listed in Appendix~\ref{app-gaugeYBbulk}.
In particular, by using \eqref{MhatM-comp}, we can rewrite \eqref{bound-bulk-gauge} in components  as 
%
\begin{multline}
\mathcal{U}_{-}(\lambda |\alpha ,\beta )
= (-1)^N\,  \frac{e^{-\eta\gamma}\sinh (\eta \delta )}{e^{-\eta\alpha}\sinh (\eta \beta )} 
\begin{pmatrix}
D(\lambda |(\gamma-1,\delta),(\alpha-1,\beta)) & -B(\lambda |(\gamma-1,\delta),(\alpha-1,\beta)) \\ 
-C(\lambda |(\gamma-1,\delta),(\alpha-1,\beta)) & A(\lambda |(\gamma-1,\delta),(\alpha-1,\beta))
\end{pmatrix}
\\
 \times
 K_{-}(\lambda |(\gamma ,\delta ),(\gamma^{\prime },\delta ^{\prime }))\,
\begin{pmatrix}
A(-\lambda |(\gamma ^{\prime },\delta ^{\prime }),(\alpha,\beta) ) & 
B(-\lambda |(\gamma ^{\prime },\delta ^{\prime }),(\alpha ,\beta) ) \\ 
C(-\lambda |(\gamma ^{\prime },\delta ^{\prime }),(\alpha ,\beta) ) & 
D(-\lambda |(\gamma ^{\prime },\delta ^{\prime }),(\alpha ,\beta) )%
\end{pmatrix}.
\end{multline}
%

\subsubsection{Expression of the transfer matrix in terms of the gauged generators}

The transfer matrix can be expressed in terms of the elements of the
gauged monodromy matrix as follows:

\begin{proposition}\label{Diagonal-TM}
Under the following choice of the gauge parameters: 
\begin{align}
& \eta \alpha =-\tau _{-}+\frac{\epsilon _{-}^{\prime }-\epsilon _{-}}{2}%
(\varphi _{-}-\psi _{-})-\frac{\epsilon _{-}+\epsilon _{-}^{\prime }}{4}i\pi %
\mod i\pi ,  \label{Gauge-cond-A} \\
& \eta \beta =\frac{\epsilon _{-}+\epsilon _{-}^{\prime }}{2}(\varphi
_{-}-\psi _{-})+\frac{2+\epsilon _{-}-\epsilon _{-}^{\prime }}{4}i\pi \mod %
i\pi ,  \label{Gauge-cond-B}
\end{align}
for $\epsilon _{-},\epsilon _{-}^{\prime }\in \{1,-1\}$,
the transfer matrix can be expressed as 
\begin{align}
\mathcal{T}(\lambda )
& =\mathsf{\bar a}_{+}(\lambda )\frac{\sinh (2\lambda +\eta )}{\sinh 2\lambda }\,
      \mathcal{A}_{-}(\lambda |\alpha ,\beta -1)
     +\mathsf{\bar a}_{+}(-\lambda )\frac{\sinh (2\lambda -\eta )}{\sinh 2\lambda }\,
     \mathcal{A}_{-}(-\lambda |\alpha ,\beta -1),  \label{Gauge-T-decomp-A} \\
& =\mathsf{\bar d}_{+}(\lambda)\frac{\sinh (2\lambda +\eta )}{\sinh 2\lambda }\,
     \mathcal{D}_{-}(\lambda |\alpha ,\beta +1)
   +\mathsf{\bar d}_{+}(-\lambda)\frac{\sinh (2\lambda -\eta )}{\sinh 2\lambda }\,
   \mathcal{D}_{-}(-\lambda |\alpha,\beta +1),  \label{Gauge-T-decomp-D}
\end{align}
where 
\begin{align}
\mathsf{\bar a}_{+}(\lambda )& =\epsilon _{-}\,e^{-\lambda +\frac{\eta }{2}}\,%
\frac{\sinh (\lambda -\frac{\eta }{2}+\epsilon _{-}\varphi _{-})\,\cosh
(\lambda -\frac{\eta }{2}-\epsilon _{-}\psi _{-})}{\sinh \varphi _{-}\,\cosh
\psi _{-}}, \\
\mathsf{\bar d}_{+}(\lambda )& =-\epsilon _{-}\,e^{-\lambda +\frac{\eta }{2}}\,%
\frac{\sinh (\lambda -\frac{\eta }{2}-\epsilon _{-}\varphi _{-})\,\cosh
(\lambda -\frac{\eta }{2}+\epsilon _{-}\psi _{-})}{\sinh \varphi _{-}\,\cosh
\psi _{-}}.
\end{align}
\end{proposition}

\begin{proof}
This is just a rewriting of our known results, see for example \cite%
{KitMNT18}.
\end{proof}

\section{Transfer matrix spectrum and eigenstates by SoV}
\label{sec-spectrum}

The transfer matrix $\mathcal{T}(\lambda )$ is a
polynomial in $\sinh ^{2}\lambda $ of degree $N+2$  which satisfies the following centrality conditions:
\begin{equation}
\mathcal{T}(\lambda )\underset{\lambda \rightarrow \pm \infty }{\sim }\frac{%
\kappa _{+}\kappa _{-}\,\cosh (\tau _{+}-\tau _{-})}{2^{2N+1}\,\sinh
\varsigma _{+}\sinh \varsigma _{-}}\,e^{\pm 2(N+2)\lambda },
\label{T-asympt}
\end{equation}
and
\begin{align}
& \mathcal{T}(\eta /2)=2\,(-1)^{N}\cosh \eta \, {\det }_{q}T(0),
\label{T-value1} \\
& \mathcal{T}(\eta /2+i\pi /2)=-2\,\cosh \eta \,\coth \varsigma _{+}\coth
\varsigma _{-}\,{\det }_{q}T(i\pi /2),  \label{T-value2}
\end{align}
plus the quantum determinant identity: 
\begin{equation}
\mathcal{T}(\xi _{n}+\eta /2)\, \mathcal{T}(\xi _{n}-\eta /2)=-\frac{%
\det_{q}K_{+}(\xi _{n})\,\det_{q}\mathcal{U}_{-}(\xi _{n})}{\sinh (2\xi
_{n}+\eta )\,\sinh (2\xi _{n}-\eta )},\quad \forall n\in \{1,\ldots ,N\}.
\end{equation}
A basis of the space of states (the SoV basis) which separates the variables
for the $\mathcal{T}(\lambda )$-spectral problem at particular values of the
spectral parameter $\lambda $ (related to the inhomogeneity parameters of
the model) was constructed in \cite{Nic12,FalKN14} in the framework of a
generalization of the Sklyanin's SoV approach. There, the complete spectrum
characterization of the transfer matrix was derived. This characterization
was rewritten in \cite{KitMN14} in terms of polynomial solutions of some
functional $TQ$-equation: when the boundary parameters satisfy a particular
constraint (see \eqref{homog-cond} below), the $TQ$-equation is a usual one,
otherwise the $TQ$-equation proposed in \cite{KitMN14} contains an additional
inhomogeneous term which modifies accordingly the resulting Bethe equations.
Let us also mention that, in \cite{MaiN19c}, the construction of \cite%
{Nic12,FalKN14} was generalized beyond the pseudo-diagonalizability of the
gauged $\mathcal{B}_{-}$-operator, i.e. beyond the applicability of the
generalized Sklyanin's SoV approach.

Here, we briefly recall the main results concerning the characterization of the transfer matrix spectrum in this framework.
We more particularly focus on the case with a constraint, for which the SoV description of the spectrum of \cite{Nic12,FalKN14} can be reformulated in terms of solutions of a usual (homogeneous) $TQ$-equation and the eigenstates as generalized Bethe states  \cite{KitMN14}.

\subsection{On the SoV basis}

In this subsection, we recall briefly the SoV basis both in the generalized
Sklyanin's approach 
\cite{Nic12,FalKN14} and in the new the SoV schema 
\cite{MaiN19}, which are at the
basis of the SoV characterization of the transfer matrix spectrum presented in the
next subsections.

Let us introduce some notations. We define\footnote{%
Note these are the coefficients that enters in the triangularization of the $%
K_{-}(\lambda )\,$\ matrix by gauge transformation.}
\begin{equation}
\mathsf{b}_{-}(\lambda |\alpha ,\beta )=\mathsf{c}_{-}(\lambda |\alpha
,-\beta )=\frac{e^{\lambda +\eta \beta -\eta /2}\sinh (2\lambda -\eta )}{%
\sinh (\eta \beta )\,}\,\mathsf{\bar{b}}_{-}(\alpha -\beta ),
\label{b-lambda}
\end{equation}
where
\begin{align}
\mathsf{\bar{b}}_{-}(x) 
&=\frac{-1\,}{2\sinh \varsigma _{+}}\left[ 2\kappa_{+}\sinh (x\eta +\tau _{+})+e^{\varsigma _{+}}\right]  \\
&=-\frac{\kappa _{+}\,}{\sinh \varsigma _{+}}\left[ \sinh (x\eta +\tau_{+})+\sinh (\varphi _{+}+\psi _{+})\right] ,
\end{align}
and from them the following product along the chain:
\begin{equation}
N_{-}(\alpha ,\beta ,\{\xi \})\equiv \prod_{j=1}^{N}\left[ \frac{\mathsf{b}%
_{-}(\frac{\eta }{2}-\xi _{j}|\alpha ,\beta +1+N-2j)}{\mathsf{g}_{-}(\frac{%
\eta }{2}-\xi _{j})}\frac{\sinh (\eta (\beta +1+N-2j))}{\sinh (\eta (j-\beta
+N))}\right] ^{1/2}.
\end{equation}
Here, $\mathsf{g}_{-}(\lambda )$ is a function such that 
\begin{equation}
\mathsf{g}_{-}(\lambda +\eta /2)\,\mathsf{g}_{-}(-\lambda +\eta /2)=\frac{%
\det_{q}K_{-}(\lambda )}{\sinh (2\lambda -2\eta )}.  \label{g-}
\end{equation}
Then, we also introduce the coefficients: 
\begin{equation}
\mathbf{A}_{\boldsymbol{\varepsilon}}(\lambda )=(-1)^{N}\,\frac{\sinh
(2\lambda +\eta )}{\sinh (2\lambda )}\,\mathbf{a}_{\boldsymbol{\varepsilon}%
}(\lambda )\,a(\lambda )\,d(-\lambda ),\label{DefFullA}
\end{equation}
where 
\begin{align}\label{a_eps}
\mathbf{a}_{\boldsymbol{\varepsilon}}(\lambda )&=\frac{\sinh (\lambda -\frac{%
\eta }{2}+\epsilon _{\varphi _{+}}\varphi _{+})\,\cosh (\lambda -\frac{\eta 
}{2}+\epsilon _{\psi _{+}}\psi _{+})}{\sinh (\epsilon _{\varphi _{+}}\varphi
_{+})\,\cosh (\epsilon _{\psi _{+}}\psi _{+})}\notag\\
&\times\frac{\sinh (\lambda -\frac{\eta }{2}+\epsilon _{\varphi _{-}}\varphi _{-})\,\cosh (\lambda -\frac{\eta 
}{2}-\epsilon _{\psi _{-}}\psi _{-})}{\sinh (\epsilon _{\varphi _{-}}\varphi
_{-})\,\cosh (\epsilon _{\psi _{-}}\psi _{-})},
\end{align}
for any choice of $\boldsymbol{\varepsilon}\equiv (\epsilon _{\varphi
_{+}},\epsilon _{\varphi _{-}},\epsilon _{\psi _{+}},\epsilon _{\psi
_{-}})\in \{-1,1\}^{4}$ such that $\epsilon _{\varphi _{+}}\epsilon
_{\varphi _{-}}\epsilon _{\psi _{+}}\epsilon _{\psi _{-}}=1$, which satisfy
the quantum determinant condition:
\begin{equation}
\mathbf{A}_{\boldsymbol{\varepsilon}}(\lambda +\eta /2)\,\mathbf{A}_{%
\boldsymbol{\varepsilon}}(-\lambda +\eta /2)=-\frac{\det_{q}K_{+}(\lambda
)\,\det_{q}\mathcal{U}_{-}(\lambda )}{\sinh (2\lambda +\eta )\,\sinh
(2\lambda -\eta )}.
\end{equation}
Finally, we define, for any $N$-tuple of variables $(x_{1},\ldots ,x_{N})$,
the generalized Vandermonde determinant $\widehat{V}(x_{1},\ldots ,x_{N})$
as 
\begin{equation}
\widehat{V}(x_{1},\ldots ,x_{N})=\det_{1\leq i,j\leq N}\left[ \sinh
^{2(j-1)}x_{i}\right] =\prod_{j<k}(\sinh ^{2}x_{k}-\sinh ^{2}x_{j}),
\end{equation}
and by them the following function of the inhomogeneities
\begin{equation}
N(\{\xi \})=\widehat{V}(\xi _{1},\ldots ,\xi _{N})\,\frac{\widehat{V}(\xi
_{1}^{(0)},\ldots ,\xi _{N}^{(0)})}{\widehat{V}(\xi _{1}^{(1)},\ldots ,\xi
_{N}^{(1)})},
\end{equation}
where 
\begin{equation}\label{def-xi-shift}
\xi _{n}^{(h)}=\xi _{n}+\eta /2-h\eta ,\qquad 1\leq n\leq N,\quad h\in
\{0,1\}.
\end{equation}

\subsubsection{Generalized Sklyanin's SoV basis}

Following standard notations, let us define
\begin{alignat}{2}
  &\bra{0} =\otimes _{n=1}^{\mathsf{N}}\left( 1,0\right) _{n},
  & \qquad
  &\bra{\underline{0} }=\otimes _{n=1}^{\mathsf{N}}\left( 0,1\right)_{n}, 
  \label{state0}\\
  &\ket{0} =\otimes _{n=1}^{\mathsf{N}}\begin{pmatrix} \,1\, \\ 0 \end{pmatrix} _{\! n},
  & \qquad
  &\ket{\underline{0} } =\otimes _{n=1}^{\mathsf{N}}\begin{pmatrix}\, 0\, \\ 1 \end{pmatrix}_{\! n}.
  \label{state0bar}
\end{alignat}
%
For each choice of $\boldsymbol{\varepsilon}\in \{-1,1\}^{4}$ and each $N$-tuple 
$\mathbf{h}\equiv (h_{1},\ldots ,h_{N})\in \{0,1\}^{N}$, we define the
following states\footnote{%
Up to a different normalization and the use of the symmetry \eqref{symm-miroir}, the states \eqref{SOVstate-R}
correspond to the states defined in eq. (4.6) of \cite{KitMNT18} multiplied
on the right by $S_{1\dots N}(\{\xi \}|\alpha ,\beta )$, whereas the states %
\eqref{SOVstate-L} correspond to the states defined in eq. (4.7) of \cite%
{KitMNT18} multiplied on the left by $S_{1\ldots N}(\{\xi \}|\alpha ,\beta
)^{-1}$.}
\begin{align}
& \ket{ \mathbf{h},\alpha ,\beta +1 }_\mathrm{Sk}
   =\frac{1}{N_{-}(\alpha ,\beta ,\{\xi \})}\,
   \prod_{j=1}^{N}\left( \frac{\mathcal{D}_{-}(\xi _{j}+\eta /2|\beta +1)}{k_{j}\,\mathsf{A}_{-}(\eta/2-\xi _{j})}\right) ^{\! h_{j}}
   S_{1\ldots N}(\{\xi \}|\alpha ,\beta)\, \ket{\underline{0} },
\label{SOVstate-R} \\
& {}_\mathrm{Sk\!}\bra{\alpha ,\beta -1,\mathbf{h} }
   =\frac{1}{N_{-}(\alpha ,\beta ,\{\xi \})}\,
     \bra{0}\,S_{1\ldots N}(\{\xi \}|\alpha ,\beta )^{-1}
     \prod_{j=1}^{N}\left( \frac{\mathcal{A}_{-}(\eta /2-\xi_{j}|\beta -1)}{\mathsf{A}_{-}(\eta /2-\xi _{j})}\right) ^{\! 1-h_{j}},
\label{SOVstate-L}
\end{align}
in which $h_j$, $j\in\{1,\ldots,N\}$, denotes the $j$-th component of the $N$-tuple $\mathbf{h}$.
Here we have defined the following product of local gauge matrices \eqref{mat-S} on $\mathcal{H}=\otimes _{n=1}^{N}\mathcal{H}_{n}$: 
\begin{align}
S_{1\dots N}(\{\xi \}|\alpha ,\beta )
  & =S_{1}(-\xi _{1}|\alpha ,\beta)\,S_{2}(-\xi _{2}|\alpha ,\beta +\sigma _{1}^{z})\ldots 
  S_{N}(-\xi_{N}|\alpha ,\beta +\sigma _{1}^{z}+\ldots +\sigma _{N-1}^{z})  \notag \\
  & =\prod_{n=1\rightarrow N}S_{n}\Bigg(-\xi _{n}\,\Big|\,\alpha ,\beta+\sum_{j=1}^{n-1}\sigma _{j}^{z}\Bigg),  \label{Sq}
\end{align}
where the arrow indicates in which order we have to consider the product of the non-commuting operators. 
The normalization coefficients $k_j$ and $\mathsf{A}_{-}(\lambda )$ are chosen as
\begin{equation}
k_{j}=\frac{\sinh (2\xi _{j}+\eta )}{\sinh (2\xi _{j}-\eta )},\qquad \mathsf{A}_{-}(\lambda )=\mathsf{g}_{-}(\lambda )\,a(\lambda )\,d(-\lambda ),
\label{A-}
\end{equation}
with here $\mathsf{g}_-$ given in terms of $\boldsymbol{\varepsilon}$ and of the boundary parameters $\varphi_\pm,\psi_\pm$ as
%
\begin{multline}\label{g-eps}
   \mathsf{g}_-(\lambda+\eta/2)=\epsilon_+\,\epsilon_{\varphi_-}\,(-1)^N\,
   \frac{\sinh(\lambda+\epsilon_{\varphi_+}\varphi_+)\,\cosh(\lambda+\epsilon_{\psi_+}\psi_+)}{\sinh(\epsilon_{\varphi_+}\varphi_+)\,\cosh(\epsilon_{\psi_+}\psi_+)}
   \\
   \times
   \frac{\sinh(\lambda+\epsilon_{\varphi_-}\varphi_-)\,\cosh(\lambda-\epsilon_{\psi_-}\psi_-)}{\sinh(\lambda+\epsilon_{+}\varphi_-)\,\cosh(\lambda-\epsilon_{+}\psi_-)},
\end{multline}
for some fixed $\epsilon_+\in\{-1,1\}$.
Then, the following proposition holds:

\begin{proposition}
Let us suppose that the inhomogeneity parameters are generic, i.e.
\begin{equation}\label{cond-inh}
\xi _{j},\xi _{j}\pm \xi _{k}\notin \{0,-\eta ,\eta \}\text{ mod}(i\pi
),\quad \forall j,k\in \{1,\ldots ,N\},\ j\neq k,
\end{equation}
and that the following nonzero conditions,
\begin{equation}
\prod_{j=1}^{N}\mathsf{\bar{b}}_{-}(\alpha -(\beta +1+N-2j))\neq 0,
\label{Sklyanin-range}
\end{equation}
are satisfied.
Then, $\mathcal{B}_{-}(\lambda |\alpha ,\beta )$ is right and
left pseudo-diagonalizable with right and left pseudo-eigenstates \eqref{SOVstate-R} and \eqref{SOVstate-L}, and its action on these states is given as
%
\begin{align}
& \mathcal{B}_{-}(\lambda |\alpha ,\beta -1)\,\ket{\mathbf{h},\alpha ,\beta-1}_\mathrm{Sk}
=(-1)^{N}a_{\mathbf{h}}(\lambda )\,a_{\mathbf{h}}(-\lambda)\,  \notag \\
& \hspace{1.5cm}\times \mathsf{b}_{-}(\lambda |\alpha,\beta -N-1)\,\frac{\sinh
(\eta (\beta -N-1))}{\sinh (\eta (\beta -1))}\,
\frac{N_{-}(\alpha ,\beta ,\{\xi \})}{N_{-}(\alpha ,\beta-2 ,\{\xi \})}\,\, \ket{\mathbf{h},\alpha ,\beta+1 }_\mathrm{Sk},  
   \label{act-BR} \\
& _\mathrm{Sk\!}\bra{\alpha ,\beta +1,\mathbf{h}}\,\mathcal{B}_{-}(\lambda|\alpha ,\beta +1)
=(-1)^{N}a_{\mathbf{h}}(\lambda )\,a_{\mathbf{h}}(-\lambda)\,  \notag \\
& \hspace{1.5cm}\times \mathsf{b}_{-}(\lambda |\alpha,\beta +N+1)\,\frac{\sinh
(\eta \beta )}{\sinh (\eta (\beta +N))}\,
\frac{N_{-}(\alpha ,\beta ,\{\xi \})}{N_{-}(\alpha ,\beta+2 ,\{\xi \})}\,\, _\mathrm{Sk\!}\bra{\alpha ,\beta -1,\mathbf{h}},  
\label{act-BL}
\end{align}
where 
\begin{equation}
a_{\mathbf{h}}(\lambda )=\prod_{n=1}^{N}\sinh (\lambda -\xi _{n}-\eta
/2+h_{n}\eta ).  \label{a-h}
\end{equation}
Moreover, the two left and right basis satisfy the following orthogonal conditions:
\begin{equation}\label{orth-SoV-Sk}
\langle\, \alpha ,\beta -1,\mathbf{h}\,|\,\mathbf{k,}\alpha ,\beta +1\,\rangle
=\delta _{\mathbf{h},\mathbf{k}}\, 
\frac{N(\{\xi\})\, e^{2\sum_{j=1}^{N}h_{j}\xi _{j}}}{\widehat{V}(\xi _{1}^{(h_{1})},\ldots
,\xi _{N}^{(h_{N})})}.
\end{equation}
\end{proposition}

In \cite{Nic12,FalKN14}, it has been shown that such basis are SoV basis of
generalized Sklyanin's type for the spectral problem of the open chain
transfer matrix under the choice \eqref{Gauge-cond-A} and \eqref{Gauge-cond-B} of the gauge parameters. The main reason
for these basis to be SoV basis is that under the choice \eqref{Gauge-cond-A} and \eqref{Gauge-cond-B} the transfer matrix become diagonal in $\mathcal{A}_{-}(\lambda |\alpha ,\beta -1)$ and $\mathcal{D}_{-}(\lambda |\alpha ,\beta +1)$ (see Proposition \ref{Diagonal-TM}) and these gauged operators act as shift operators on the
pseudospectrum of $\mathcal{B}_{-}(\lambda |\alpha ,\beta \pm 1)$,
respectively, on the left and right basis elements.

\subsubsection{The new general SoV basis}

One of the interesting achievements of the quantum SoV approach, in particular of the new SoV approach developed in \cite{MaiN19}, is its ability to describe in an universal way the fundamental objects of different models. This is for example the case for SoV basis, for which mainly the same characterization derived in Proposition 3.1 of \cite{Nic21} for the XXX spin chain holds true also for the XXZ chain:

\begin{proposition}
For a given co-vector $\bra{ S } $, let us define the following co-vectors
\begin{equation}
   \bra{\mathbf{h} }
   \equiv \bra{ S } 
   \prod_{n=1}^{N}\left( \frac{\mathcal{T}(\xi _{n}-\eta /2)}{\mathbf{A}_{\boldsymbol{\varepsilon}}(\eta /2-\xi _{n})}\right) ^{\! 1-h_{n}},
   \qquad
   \mathbf{h}\equiv (h_{1},\ldots ,h_{N})\in \{0,1\}^{N},
    \label{SoV-Basis-Open-L}
\end{equation}
and vectors
\begin{equation}
  \ket{ \mathbf{h} } \equiv \prod_{n=1}^{N}\left( \frac{\mathcal{T}(\xi
_{n}+\eta /2)}{\mathfrak{t}_{n}\mathbf{A}_{\boldsymbol{\varepsilon}}(\eta /2-\xi _{n})}%
\right) ^{\! h_{n}}
   \ket{ R } ,
   \qquad \mathbf{h}\in \{0,1\}^{N},
\label{SoV-Basis-Open-R}
\end{equation}
where
\begin{equation}
\mathfrak{t}_{n}=e^{-2\xi _{n}}\frac{\sinh2(\xi _{n}+\eta )}{\sinh2(\xi _{n}-\eta )},
\end{equation}
and where $\ket{R}$ is such that
\begin{equation}\label{def-R}
    \bra{\mathbf{h} }\,R\,\rangle =\delta _{\mathbf{h},\mathbf{0}}\,
    \frac{\widehat{V}(\xi _{1},\ldots ,\xi _{N})}{\widehat{V}(\xi _{1}^{(1)},\ldots,\xi _{N}^{(1)})}.
\end{equation}
Let us moreover assume that the boundary matrices $K_{-}(\lambda )$ and $K_{+}(\lambda )$ are
not both proportional to the identity.
Then, for almost any choice of the co-vector $\bra{S}$ and of the inhomogeneity parameters satisfying \eqref{cond-inh}, the vector $ \ket{R} $ is uniquely defined by \eqref{def-R} and the families of co-vectors \eqref{SoV-Basis-Open-L} and of vectors \eqref{SoV-Basis-Open-R} are basis of $\mathcal{H^*}$ and of $\mathcal{H}$ respectively, which moreover satisfy the following orthogonality conditions: 
\begin{equation}
\langle\, \mathbf{h}\,|\,\mathbf{h}^{\prime }\,\rangle 
=\delta _{\mathbf{h},\mathbf{h}^{\prime }}\,
\frac{N(\{\xi \})\, e^{2\sum_{j=1}^{N}h_{j}\xi _{j}}}{\widehat{V}(\xi _{1}^{(h_{1})},\ldots ,\xi _{N}^{(h_{N})})}.
\label{Ortho-norm}
\end{equation}
\end{proposition}

The fact that these basis are SoV basis is here mainly a clear consequence of
their very definitions. Indeed, they imply the factorized form of the
transfer matrix wave functions in terms of the transfer matrix eigenvalues.
In Theorem 3.2 of \cite{MaiN19}, it was shown that the complete
characterization of the spectrum follows by imposing the quantum determinant
condition together with the known polynomial form in $\cosh 2\lambda $ to
the transfer matrix eigenvalues.

\subsection{The transfer matrix spectrum and eigenstates} 

From the previous SoV basis follows a characterization of the transfer matrix spectrum in terms of discrete equations \cite{Nic12,FalKN14,KitMNT18,MaiN19}. 
Here, following \cite{KitMN14,KitMNT18,MaiN19}, we present instead its equivalent
characterization in terms of $TQ$-functional equations, and we more particularly insist on the special case of interest for the present paper,
in which one constrain is imposed between the boundary parameters of the sites $1$ and $N$, so that the corresponding $TQ$-equation is a {\em homogeneous} equation with {\em polynomial} Q-solutions. 

\subsubsection{SoV characterization of the transfer matrix spectrum and eigenstates in terms of solutions of a functional $TQ$-equation}

Let us start by recalling some further notations: 
\begin{align}
\mathsf{u}_{n}& =\frac{\sinh (2\xi _{n}-\eta )}{\sinh (2\xi _{n}+\eta )}%
\frac{a(\xi _{n}+\eta /2)\,d(-\xi _{n}-\eta /2)}{a(-\xi _{n}+\eta /2)\,d(\xi
_{n}-\eta /2)}  \notag \\
& =-\prod_{j\not=n}\frac{\sinh (\xi _{n}-\xi _{j}+\eta )\,\sinh (\xi
_{n}+\xi _{j}+\eta )}{\sinh (\xi _{n}+\xi _{j}-\eta )\,\sinh (\xi _{n}-\xi
_{j}-\eta )},
\end{align}
and 
\begin{equation}
\mathsf{v}_{n,\boldsymbol{\varepsilon}}=\frac{\mathbf{a}_{%
\boldsymbol{\varepsilon}}(\xi _{n}+\frac{\eta }{2})}{\mathbf{a}_{%
\boldsymbol{\varepsilon}}(-\xi _{n}+\frac{\eta }{2})}=\frac{\mathbf{a}_{%
\boldsymbol{\varepsilon}}(\xi _{n}+\frac{\eta }{2})}{\mathbf{a}_{-%
\boldsymbol{\varepsilon}}(\xi _{n}+\frac{\eta }{2})},
\end{equation}
so that 
\begin{equation}
\frac{\sinh (2\xi _{n}-2\eta )}{\sinh (2\xi _{n}+2\eta )}\,\frac{\mathbf{A}_{%
\boldsymbol{\varepsilon}}(\xi _{n}+\frac{\eta }{2})}{\mathbf{A}_{%
\boldsymbol{\varepsilon}}(-\xi _{n}+\frac{\eta }{2})}=\mathsf{u}_{n}\,%
\mathsf{v}_{n,\boldsymbol{\varepsilon}}.
\end{equation}
Moreover, we denote by $\Sigma _{Q}^{M}$ the set of $Q(\lambda )$
polynomials in $\cosh (2\lambda )$ of degree $M$ of the form 
\begin{equation}\label{Q-form}
Q(\lambda )=\prod_{j=1}^{M}\frac{\cosh (2\lambda )-\cosh (2\lambda _{j})}{2}%
=\prod_{j=1}^{M}\left( \sinh ^{2}\lambda -\sinh ^{2}\lambda _{j}\right) ,
\end{equation}
with 
\begin{equation}\label{cond-roots}
\cosh (2\lambda _{j})\not=\cosh (2\xi _{n}^{(h)}),\quad \forall \,(j,n,h)\in
\{1,\ldots ,M\}\times \{1,\ldots ,N\}\times \{0,1\}.
\end{equation}
Finally, for $r\in \mathbb{N}$ and $\boldsymbol{\varepsilon}\equiv (\epsilon
_{\varphi _{+}},\epsilon _{\varphi _{-}},\epsilon _{\psi _{+}},\epsilon
_{\psi _{-}})\in \{-1,1\}^{4}$, we define 
\begin{equation}
\mathfrak{f}_{\boldsymbol{\varepsilon}}^{(r)}\equiv 
\mathfrak{f}_{\boldsymbol{\varepsilon}}^{(r)}(\tau _{+},\tau _{-},\varphi
_{+},\varphi _{-},\psi _{+},\psi _{-})=\frac{2\kappa _{+}\kappa _{-}}{\sinh
\varsigma _{+}\,\sinh \varsigma _{-}}\,\mathfrak{g}_{\boldsymbol{\varepsilon}%
}^{(r)}(\tau _{+},\tau _{-},\varphi _{+},\varphi _{-},\psi _{+},\psi _{-}),
\end{equation}
with 
\begin{multline}
\mathfrak{g}_{\boldsymbol{\varepsilon}}^{(r)}(\tau _{+},\tau _{-},\varphi_{+},\varphi _{-},\psi _{+},\psi _{-}) 
= \cosh (\tau _{+}-\tau_{-}) \\
  -\epsilon _{\varphi _{+}}\epsilon _{\varphi _{-}}\cosh (\epsilon_{\varphi _{+}}\varphi _{+}+\epsilon _{\varphi _{-}}\varphi _{-}+\epsilon_{\psi _{+}}\psi _{+}-\epsilon _{\psi _{-}}\psi _{-}+(N-1-2r)\eta ).
\end{multline}
\begin{theorem}[\cite{KitMN14,KitMNT18}]\label{th-sprectrum-inhom}
Let the two boundary matrices be not both proportional to the identity matrix, the inhomogeneity parameters be generic, and the following identity be satisfied:
\begin{equation}\label{cond-inhom}
   \forall r\in\{0,\ldots,N-1\},\qquad \mathfrak{f}_{\boldsymbol{\varepsilon}}^{(r)}(\tau _{+},\tau _{-},\varphi
_{+},\varphi _{-},\psi _{+},\psi _{-})\not=0.
\end{equation}
with $\boldsymbol{\varepsilon}\equiv (\epsilon _{\varphi _{+}},\epsilon _{\varphi
_{-}},\epsilon _{\psi _{+}},\epsilon _{\psi _{-}})\in \{-1,1\}^{4}$ and $\epsilon _{\varphi _{+}}\epsilon _{\varphi _{-}}\epsilon _{\psi
_{+}}\epsilon _{\psi _{-}}=1$.
Then, the transfer matrix $\mathcal{T}(\lambda)$ is diagonalizable
with simple spectrum, and the set $\Sigma _{\mathcal{T}}$ of its eigenvalues  is given by the set of entire functions $\tau(\lambda)$ such that there exists a polynomial $Q(\lambda )\in \Sigma _{Q}^{N}$ satisfying with $\tau(\lambda)$ the $TQ$-equation
\begin{equation}
\tau (\lambda )\, Q(\lambda )
=\mathbf{A}_{\boldsymbol{\varepsilon}}(\lambda)\,Q(\lambda -\eta )
+\mathbf{A}_{\boldsymbol{\varepsilon}}(-\lambda)\,Q(\lambda +\eta )
+\mathbf{F}_{\boldsymbol{\varepsilon}}(\lambda),  \label{inhom-TQ}
\end{equation}
with  inhomogeneous term
\begin{equation}
  \mathbf{F}_{\boldsymbol{\varepsilon}}(\lambda)
  =\mathfrak{f}_{\boldsymbol{\varepsilon}}^{(N)}\, a(\lambda)\, a(-\lambda)\, d(\lambda)\, d(-\lambda)\, 
  [\cosh^2(2\lambda)-\cosh^2\eta].
\end{equation}
Moreover, in that case, the
corresponding $Q(\lambda )\in \Sigma _{Q}^{N}$ satisfying \eqref{inhom-TQ}
with $\tau (\lambda )$ is unique, and the unique (up overall constants) left
and right $\mathcal{T}(\lambda )$ eigenstates can be expressed as
\begin{align}
 &\ket{Q}=\sum_{\mathbf{h}\in \{0,1\}^{N}}
                 \prod_{n=1}^{N}\!\frac{Q(\xi _{n}^{(h_{n})})}{Q(\xi _{n}^{(0)})}\ 
                 e^{-\sum_{j}h_{j}\xi _{j}}\,\widehat{V}(\xi _{1}^{(h_{1})},\ldots ,\xi _{N}^{(h_{N})})\ 
                 \ket{\mathbf{h} } , \label{eigenR}\\
 &\bra{Q}=\sum_{\mathbf{h}\in \{0,1\}^{N}}
                 \prod_{n=1}^{N}\left[ (\mathsf{u}_{n}\,\mathsf{v}_{n,\boldsymbol{\varepsilon}})^{h_{n}}\ 
                 \frac{Q(\xi_{n}^{(h_{n})})}{Q(\xi _{n}^{(0)})}\right] 
                 e^{-\sum_{j}h_{j}\xi _{j}}\,\widehat{V}(\xi _{1}^{(h_{1})},\ldots ,\xi _{N}^{(h_{N})})\ 
                 \bra{\mathbf{h} }.
                 \label{eigenL}
\end{align}
\end{theorem}

Note that the characterization of the transfer matrix spectrum and eigenstates of Theorem~\ref{th-sprectrum-inhom} is complete.

However, as announced above, we are more particularly interested here
in the special case in which the following type of
constraint on the boundary parameters is satisfied: 
\begin{equation}
\mathfrak{f}_{\boldsymbol{\varepsilon}}^{(r)}(\tau _{+},\tau _{-},\varphi
_{+},\varphi _{-},\psi _{+},\psi _{-})=0,  \label{homog-cond}
\end{equation}
for some $r\in\{0,\ldots,N\}$, for which the transfer matrix spectrum and eigenstates can be (at least partially) characterized in terms of solutions of some usual (homegeneous) functional $TQ$-equation of Baxter type.
More precisely, we can state the following proposition:

\begin{proposition}[\cite{KitMN14,KitMNT18}]\label{prop-spectrum1}
Let us suppose that the inhomogeneity parameters are generic, and that the
two boundary matrices are not both proportional to the identity matrix. We
moreover suppose that the condition \eqref{homog-cond} is satisfied for a
given $r=M\in \{0,\ldots ,N\}$ and a given choice of $\boldsymbol{%
\varepsilon}\equiv (\epsilon _{\varphi _{+}},\epsilon _{\varphi
_{-}},\epsilon _{\psi _{+}},\epsilon _{\psi _{-}})\in \{-1,1\}^{4}$ such
that $\epsilon _{\varphi _{+}}\epsilon _{\varphi _{-}}\epsilon _{\psi
_{+}}\epsilon _{\psi _{-}}=1$. Then, the transfer matrix $\mathcal{T}(\lambda)$ is diagonalizable
with simple spectrum, and any entire function $\tau (\lambda )$ such that
there exists $Q(\lambda )\in \Sigma _{Q}^{M}$ satisfying the following
homogeneous $TQ$-equation, 
\begin{equation}
\tau (\lambda )\, Q(\lambda )=\mathbf{A}_{\boldsymbol{\varepsilon}}(\lambda
)\,Q(\lambda -\eta )+\mathbf{A}_{\boldsymbol{\varepsilon}}(-\lambda
)\,Q(\lambda +\eta ),  \label{hom-TQ}
\end{equation}
is an eigenvalue of $\mathcal{T}(\lambda )$ (we write $%
\tau (\lambda )\in \Sigma _{\mathcal{T}}$). Moreover, in that case, the
corresponding $Q(\lambda )\in \Sigma _{Q}^{M}$ satisfying \eqref{hom-TQ}
with $\tau (\lambda )$ is unique, and the unique (up overall constants) left
and right $\mathcal{T}(\lambda )$ eigenstates can be expressed in terms of $Q$ as in \eqref{eigenR}-\eqref{eigenL}.
\end{proposition}

Contrary to Theorem~\ref{th-sprectrum-inhom}, the above proposition does not {\em a priori} provide a complete characterization of the spectrum: as
explained in \cite{KitMN14,KitMNT18}, a part of the
spectrum is associated to solutions of the  inhomogeneous $TQ$-equation \eqref{inhom-TQ} if
the condition \eqref{homog-cond} is satisfied only for a
fixed $r=M<N$, that is if \eqref{homog-cond} is satisfied for 
$r=M<N$ being 
\begin{equation}
\mathfrak{g}_{\boldsymbol{\varepsilon}}^{(M)}(\tau _{+},\tau _{-},\varphi
_{+},\varphi _{-},\psi _{+},\psi _{-})=0,
\end{equation}
while holding
\begin{equation}
\frac{\kappa _{+}\kappa _{-}}{\sinh \varsigma _{+}\,\sinh \varsigma _{-}}%
\neq 0.
\end{equation}
This gives instead a complete characterization of the spectrum in terms
of only homogeneous $TQ$-equations when we have either
\begin{equation}
\mathfrak{g}_{\boldsymbol{\varepsilon}}^{(N)}(\tau _{+},\tau _{-},\varphi
_{+},\varphi _{-},\psi _{+},\psi _{-})=0,  \label{Full-homo-1}
\end{equation}
or
\begin{equation}
\frac{\kappa _{+}\kappa _{-}}{\sinh \varsigma _{+}\,\sinh \varsigma _{-}}=0.
\label{Fullhomo-2}
\end{equation}
More precisely, we have in that case the following result:

\begin{theorem}[\cite{KitMN14,KitMNT18}]\label{th-spectrum2}
Let the two boundary matrices be not both proportional to the identity matrix, the inhomogeneity parameters be generic, and one of the two conditions \eqref{Full-homo-1} and \eqref{Fullhomo-2} be satisfied. Then, for almost any choice of the boundary parameters, the transfer matrix $\mathcal{T}(\lambda)$ is diagonalizable with simple spectrum. 
Moreover, defined
\begin{equation}
\Sigma _{Q} 
     =\begin{cases}
       \Sigma _{Q}^{N}    
               &\text{if }\ \frac{\kappa _{+}\kappa _{-}}{\sinh \varsigma_{+}\,\sinh \varsigma _{-}}\neq 0, \\ 
       \cup _{n=0}^{N}\Sigma _{Q}^{n}
               &\text{if }\ \frac{\kappa _{+}\kappa _{-}}{\sinh \varsigma _{+}\,\sinh \varsigma _{-}}=0,
       \end{cases}
\end{equation}
the set $\Sigma _{\mathcal{T}}$ of the transfer matrix eigenvalues  is given by the set of entire functions $\tau(\lambda)$ such that there exists a polynomial $Q(\lambda )\in \Sigma _{Q}$ satisfying with $\tau(\lambda)$ the homogeneous $TQ$-equation  \eqref{hom-TQ}. For $\tau(\lambda)\in\Sigma _{\mathcal{T}}$, the corresponding $Q(\lambda )\in \Sigma _{Q}$ solution of  \eqref{hom-TQ} is unique,
and the corresponding unique (up overall constants) left and right eigenstates of $\mathcal{T}(\lambda)$ are respectively given by \eqref{eigenL} and \eqref{eigenR}.
\end{theorem}

\subsubsection{On the ABA rewriting of separate states: the Sklyanin's case}

Some remarks are in order here. In the previous results, we have used
the new SoV basis to provide the SoV characterization of the
transfer matrix eigenstates. Indeed, this SoV characterization holds for a
larger set of boundary conditions than the generalized Sklyanin's
one. However, we should comment that, when the generalized
Sklyanin's SoV approach works, the transfer matrix eigenstates can alternatively be expressed by the same formulas as \eqref{eigenR}-\eqref{eigenL} by using the basis \eqref{SOVstate-R}-\eqref{SOVstate-L} instead of \eqref{SoV-Basis-Open-L}-\eqref{SoV-Basis-Open-R}, i.e. as
\begin{align}
   &\ket{Q}_\mathrm{Sk} =\sum_{\mathbf{h}\in \{0,1\}^{N}}
                 \prod_{n=1}^{N}\!\frac{Q(\xi _{n}^{(h_{n})})}{Q(\xi _{n}^{(0)})}\ 
                 e^{-\sum_{j}h_{j}\xi _{j}}\,\widehat{V}(\xi _{1}^{(h_{1})},\ldots ,\xi _{N}^{(h_{N})})\ 
                 \ket{\mathbf{h},\alpha ,\beta +1} _\mathrm{Sk}, \label{eigenR-Sk}\\
    &_\mathrm{Sk}\bra{ Q}=\sum_{\mathbf{h}\in \{0,1\}^{N}}
                 \prod_{n=1}^{N}\left[ (\mathsf{u}_{n}\,\mathsf{v}_{n,\boldsymbol{\varepsilon}})^{h_{n}}\ 
                 \frac{Q(\xi_{n}^{(h_{n})})}{Q(\xi _{n}^{(0)})}\right] 
                 e^{-\sum_{j}h_{j}\xi _{j}}\,\widehat{V}(\xi _{1}^{(h_{1})},\ldots ,\xi _{N}^{(h_{N})})\ 
    _\mathrm{Sk}\bra{ \mathbf{h},\alpha ,\beta -1},
    \label{eigenL-Sk}
\end{align}
once the gauge parameters are fixed by $\left( \ref{Gauge-cond-A}\right) $
and $\left( \ref{Gauge-cond-B}\right) $. Moreover, due to the simplicity of the
transfer matrix spectrum, these two SoV
representations (Sklyanin and non-Sklyanin ones) of the eigenstates must
coincide up to nonzero normalization, i.e. it must holds\footnote{%
That is, the above proportionality holds independently from the precise
relation between the two SoV basis.}:
%
\begin{equation}\label{prop-SoVeigen}
   \ket{Q}_\mathrm{Sk}=\mathsf{c}^\mathrm{Sk}_Q\,\ket{Q} ,
   \quad
   _\mathrm{Sk}\bra{Q}=\bra{Q}/\mathsf{c}^\mathrm{Sk}_Q,
\end{equation}
for some normalization coefficient $\mathsf{c}^\mathrm{Sk}_Q\neq 0$, and this is always the case for any choice of
the origin states $\bra{S}$ and $\ket{R}$, used to construct the
SoV bases \eqref{SoV-Basis-Open-L} and \eqref{SoV-Basis-Open-R}.

Moreover, in the range of validity of the
generalized Sklyanin's approach, separate states of the form \eqref{eigenR-Sk}-\eqref{eigenL-Sk} for any $Q\in\Sigma_Q^M$ 
admit an Algebraic Bethe Ansatz form. More precisely, defining 
\begin{align}
 \ket{\Omega _{\alpha ,\beta +1} } 
 &=\frac{1}{N(\{\xi \})}\sum_{\mathbf{h}\in \{0,1\}^{N}}
 e^{-\sum_{j}h_{j}\xi _{j}}\,\widehat{V}(\xi_{1}^{(h_{1})},\ldots ,\xi _{N}^{(h_{N})})\,
 \ket{\mathbf{h},\alpha ,\beta+1}_\mathrm{Sk},  
 \label{ref-state-R} \\
 \bra{\Omega _{\alpha ,\beta -1}}
  &=\frac{1}{N(\{\xi \})}\sum_{\mathbf{h}\in \{0,1\}^{N}}
  \prod_{n=1}^{N}(\mathsf{u}_{n}\,\mathsf{v}_{n,\boldsymbol{\varepsilon}})^{h_{n}}\,
  e^{-\sum_{j}h_{j}\xi _{j}}\,\widehat{V}(\xi _{1}^{(h_{1})},\ldots ,\xi _{N}^{(h_{N})})\ 
  {}_\mathrm{Sk}\bra{\alpha,\beta -1,\mathbf{h}},
  \label{ref-state-L}
\end{align}
and using the following shortcut notation for a product of gauged boundary operators \eqref{Bhat}:
\begin{align}\label{product-Bhat}
  \underline{\widehat{\mathcal{B}}}_{-,M}(\{\lambda _{i}\}_{i=1}^{M}|\alpha -\beta+1)
  &= 
  \widehat{\mathcal{B}}_{-}(\lambda _{1}|\alpha -\beta +1)\cdots 
  \widehat{\mathcal{B}}_{-}(\lambda _{M}|\alpha -\beta +2M-1)
  \nonumber\\
  &= \prod_{j=1\to M} \widehat{\mathcal{B}}_{-}(\lambda_j|\alpha -\beta +2j-1),
\end{align}
we can easily show from \eqref{act-BR}-\eqref{act-BL} that, under the condition \eqref{Sklyanin-range}, any separate state of the form \eqref{eigenR-Sk}-\eqref{eigenL-Sk} for $Q\in\Sigma_Q^M$ (not necessarily solution of a $TQ$-equation) can be rewritten as
\begin{align}
  & \ket{Q}_\mathrm{Sk} = 
  \mathsf{c}^{(R)}_{Q,\text{ABA}}\,
   \underline{\widehat{\mathcal{B}}}_{-,M}(\{\lambda_i\}_{i=1}^{M}|\alpha -\beta+1)\,
   \ket{\Omega_{\alpha ,\beta +1-2M }},
   \label{separate-ABA-R}\\
   &_\mathrm{Sk}\bra{Q}=\mathsf{c}^{(L)}_{Q,\text{ABA}}\, \bra{\Omega _{\alpha ,\beta -1+2M}}\,
   \underline{\widehat{\mathcal{B}}}_{-,M}(\{\lambda_i\}_{i=1}^{M}|\alpha -\beta+1-2M),
    \label{separate-ABA-L}
\end{align}
with normalization coefficients
\begin{align}
  &\mathsf{c}^{(R)}_{Q,\text{ABA}}
  =\frac{N(\{\xi\})\, (-1)^{NM}\, e^{NM\eta} }{
  \prod_{j=1}^M a_{\boldsymbol{0}}(\lambda_j)\, a_{\boldsymbol{0}}(-\lambda_j)\,\sinh(2\lambda_j-\eta)\, \mathsf{\bar{b}}_{-}(\alpha -\beta+2j+N-1)}
  \frac{N_{-}(\alpha ,\beta-2M ,\{\xi \})}{N_{-}(\alpha ,\beta ,\{\xi \})},\label{nomr-ABA-R}\\
  &\mathsf{c}^{(L)}_{Q,\text{ABA}}
  =\frac{N(\{\xi\})\,(-1)^{NM}\, e^{-NM\eta}}{\prod_{j=1}^M a_{\boldsymbol{0}}(\lambda_j)\, a_{\boldsymbol{0}}(-\lambda_j)\,\sinh(2\lambda_j-\eta)\,\mathsf{\bar{b}}_{-}(\alpha -\beta-2j-N+1)}
  \nonumber\\
 & \hspace{6cm}\times
  \prod_{j=0}^{2M-1}\frac{\sinh(\eta(\beta+N+j))}{\sinh(\eta(\beta+j))}\,
  \frac{N_{-}(\alpha ,\beta ,\{\xi \})}{N_{-}(\alpha ,\beta-2M ,\{\xi \})}.
\label{nomr-ABA-L}
\end{align}
Here $\lambda_1,\ldots,\lambda_M$ label the roots of $Q$ similarly as in \eqref{Q-form}.
We underline once again that \eqref{separate-ABA-R}-\eqref{separate-ABA-L} hold for arbitrary $Q$ of the form \eqref{Q-form}, i.e. not necessarily solution of a $TQ$-equation.

This ABA representation of the separate states enables us to reformulate Theorem~\ref{th-spectrum2} in the Sklyanin's case as follows:

\begin{corollary}\label{cor-Skl-ABA}
Let us suppose that the generalized Sklyanin's approach is applicable%
\footnote{%
That is, let us suppose that the inhomogeneity parameters are generic and
that the nonzero condition \eqref{Sklyanin-range} is
satisfied imposing \eqref{Gauge-cond-A} and \eqref{Gauge-cond-B}.} 
and that one of the two conditions \eqref{Full-homo-1} and \eqref{Fullhomo-2} is satisfied.
Then the transfer matrix $\mathcal{T}(\lambda)$  is diagonalizable with simple spectrum and 
the set $\Sigma _{\mathcal{T}}$ of its eigenvalues  is given by the set of entire functions $\tau(\lambda)$ such that there exists a polynomial $Q(\lambda )\in \Sigma _{Q}$ satisfying with $\tau(\lambda)$ the homogeneous $TQ$-equation  \eqref{hom-TQ}. For $\tau(\lambda)\in\Sigma _{\mathcal{T}}$, the corresponding $Q(\lambda )\in \Sigma _{Q}$ solution of  \eqref{hom-TQ} is unique,
and the corresponding unique (up overall constants) left and right eigenstates of $\mathcal{T}(\lambda)$  admit the following ABA representations:
\begin{equation}\label{Bethe-st-r}
   \underline{\widehat{\mathcal{B}}}_{-,N_{K_\pm}}(\{\lambda _{i}\}_{i=1}^{N_{K_\pm}}|\alpha -\beta+1)\,
   \ket{\Omega_{\alpha ,\beta +1-2N_{K_{\pm }}} } ,
\end{equation}
and 
\begin{equation}\label{Bethe-st-l}
   \bra{\Omega _{\alpha ,\beta -1+2N_{K_{\pm }}} }\,
   \underline{\widehat{\mathcal{B}}}_{-,N_{K_\pm}}(\{\lambda _{i}\}_{i=1}^{N_{K_\pm}}|\alpha -\beta+1-2 N_{K_\pm}).
\end{equation}
Here $N_{K_{\pm }}$ denotes the degree of the polynomial $Q\in \Sigma _{Q}$, $\lambda_1,\ldots,\lambda_{N_{K_\pm}}$ label its roots similarly as in \eqref{Q-form},
and the gauge parameters are fixed in terms of the boundary parameters by \eqref{Gauge-cond-A} and \eqref{Gauge-cond-B}.
\end{corollary}

Let us comment that the left and right \textit{reference states} $\bra{\Omega _{\alpha ,\beta -1} }$ and $\ket{\Omega _{\alpha ,\beta +1} }$, with slightly different notations and normalizations, have been presented in our previous
paper \cite{KitMNT18}. 
Here, we are interested in reformulating these left and right reference states by using, under special conditions, states introduced in \cite{FadT79,CaoLSW03} in the context of algebraic Bethe Ansatz. In order
to do so, let us first introduce some notations that will be used throughout the rest of the paper:
\begin{equation}\label{ref-x}
   \bra{ x,\eta } \equiv \otimes _{n=1}^N \left( -1,e^{-(N-n+x)\eta -\xi _{n}}\right) _{\! n},
   \qquad
   \ket{ \eta ,x } \equiv \otimes_{n=1}^N
                                   \begin{pmatrix}
                                   e^{-(n-N+x)\eta -\xi _{n}} \\ 
                                   1%
                                    \end{pmatrix}_{\! n}.
\end{equation}
It is easy to verify that such states result from the action of the Vertex-IRF transformation \eqref{mat-S} or of its inverse on the states \eqref{state0} and \eqref{state0bar}:
%
\begin{align}
  &\bra{\alpha -\beta ,\eta } =e^{-\sum_n\xi_n}\, N_{\alpha ,\beta ,\eta }\,
          \bra{0}\prod_{n=1}^N S_{n}^{-1}(-\xi _{n}|\alpha ,\beta -N+n), 
          \label{Left-B-ref}
          \\
  &\ket{\eta ,\alpha +\beta }=\prod_{n=1}^N S_{n}(-\xi_{n}|\alpha ,\beta +n-N)\, \ket{0} ,
  \label{Right-C-ref}
\end{align}
with
\begin{equation}
N_{\alpha ,\beta ,\eta }=2^{N}e^{-\alpha N\eta
}\prod_{n=1}^N\sinh (\eta(n-N+\beta )) .
\end{equation}
By using these states, we can present the announced closed form of the
left and right \textit{reference states}.

\begin{proposition}
\label{Ref-States}
Let us suppose that the generalized Sklyanin's approach is applicable, with $\alpha$ and $\beta $ fixed by the gauge conditions \eqref{Gauge-cond-A}-\eqref{Gauge-cond-B} and
\begin{equation}
\epsilon _{\varphi _{+}}=\epsilon _{\psi _{+}}, 
\text{ \ \ }
\epsilon_{\varphi _{-}}=\epsilon _{\psi _{-}}=\epsilon_-=\epsilon'_-. 
\label{E-homo-cond}
\end{equation}
Then, under the condition
\begin{equation} \label{+homo-cond}
   \eta ( \alpha+\beta+N-1-2M )+ \tau_+ =-\epsilon_{\varphi_+}(\varphi_++\psi_+)+\frac{1-\epsilon_{\varphi_+}}2 i\pi \mod 2\pi i,
\end{equation}
for $M\leq N$, 
the following identity holds:
\begin{equation}\label{Ref-SoV}
 \ket{\eta,\alpha+\beta+N-1-2M}=\mathsf{c}^{(R)}_{M,\textnormal{ref}}\,\ket{\Omega_{\alpha,\beta-2M+1}},
\end{equation}
where $\mathsf{c}^{(R)}_{M,\text{ref}}$ is a nonzero scalar factor which only depends on $M$.

Let us suppose that the generalized Sklyanin's approach is applicable, with $\alpha$ and $\beta $ fixed by the gauge conditions \eqref{Gauge-cond-A}-\eqref{Gauge-cond-B} and
\begin{equation}
\epsilon _{\varphi _{+}}=\epsilon _{\psi _{+}}, 
\text{ \ \ }
\epsilon_{\varphi _{-}}=\epsilon _{\psi _{-}}=-\epsilon_-=-\epsilon'_-. 
\label{E-homo-cond-bis}
\end{equation}
Then, under the condition
\begin{equation}
   \eta (\alpha +\beta -N+1+2M)+\tau _{+}=\epsilon_{\varphi_+}(\varphi_{+}+\psi _{+})+\frac{1+\epsilon_{\varphi_+}}{2} i\pi \mod 2\pi i,  \label{+homo-cond-l}
\end{equation}
for $M\leq N$, the following identity holds:
\begin{equation}
   \bra{\alpha +\beta -N+2M+1,\eta }= \mathsf{c}^{(L)}_{M,\textnormal{ref}}\,\bra{\Omega _{\alpha ,\beta +2M-1} } ,  \label{Ref-SoV-l}
\end{equation}
where $\mathsf{c}^{(L)}_{M,\textnormal{ref}}$ is a nonzero scalar factor which only depends on $M$.
\end{proposition}



\begin{proof}
It follows from \eqref{+homo-cond} that the condition \eqref{homog-cond} is satisfied
for $M\leq N$ when we impose the condition \eqref{E-homo-cond}.
Let $Q(\lambda )\in \Sigma _{Q}^{M}$ solution of the homogeneous
$TQ$-equation \eqref{hom-TQ} for some entire function $\tau (\lambda )$, i.e. such that the parameters $\lambda_1,\ldots,\lambda_M$ labelling the roots of $Q$ satisfy the corresponding Bethe equations.
Then, the associated right transfer matrix eigenstate can be written from corollary~\ref{cor-Skl-ABA} as
\begin{equation}\label{Bethe-state-1}
  \underline{\widehat{\mathcal B}}_{-,M}(\{\lambda_j\} _{j=1}^M |\alpha -\beta +1)\,
  \ket{\Omega _{\alpha ,\beta+1-2M}}.  
\end{equation}
On the other hand, it follows from \eqref{+homo-cond} and from Lemma~\ref{lemme-eigen-A} that the state $\ket{\eta,\alpha+\beta+N-1-2M}$ is an eigenstate of the operator $\mathcal{A}_{-}(\lambda |\alpha ,\beta -1)$:
\begin{multline}
   \mathcal{A}_{-}(\lambda |\alpha ,\beta -1)\, \ket{\eta,\alpha+\beta+N-1-2M}
   = (-1)^N\, a(\lambda)\, d(-\lambda)\, 
   \\
   \times
   e^{\lambda-\eta/2}\frac{\sinh(\lambda-\eta+\epsilon_{\varphi_+}\varphi_+)\, \cosh(\lambda-\eta+\epsilon_{\varphi_+}\psi_+)}{\sinh(\epsilon_{\varphi_+}\varphi_+)\,\cosh\psi_+}\, \ket{\eta,\alpha+\beta+N-1-2M}.
\end{multline}
From the decomposition \eqref{Gauge-T-decomp-A} of
the transfer matrix 
in terms of the gauged boundary
operators $\mathcal{A}_{-}(\pm \lambda |\alpha ,\beta -1)$, with $\alpha $
and $\beta $ fixed by \eqref{Gauge-cond-A} and \eqref{Gauge-cond-B}, one can then show by direct computation, using the commutation relations between $\mathcal{A}_{-}(\lambda |\alpha ,\beta )$ and $\mathcal{B}_{-}(\mu |\alpha ,\beta )$, that the state
\begin{equation}\label{Bethe-state-2}
   \underline{\widehat{\mathcal{B}}}_{-,M}(\{\lambda_j\} _{j=1}^M |\alpha -\beta +1)\,
   \ket{\eta,\alpha+\beta+N-1-2M}
\end{equation}
is also an eigenstate of the transfer matrix $\mathcal{T}(\lambda )$ with the same eigenvalue $\tau(\lambda)$ as \eqref{Bethe-state-1}. Due to the simplicity of the transfer matrix spectrum, the two states \eqref{Bethe-state-1} and \eqref{Bethe-state-2} are therefore collinear:
\begin{multline}
   \underline{\widehat{\mathcal{B}}}_{-,M}(\{\lambda_j\} _{j=1}^M |\alpha -\beta +1)\,
   \ket{\eta,\alpha+\beta+N-1-2M}
   \\
   =\mathsf{c}_{Q,\text{ref}}\,
   \underline{\widehat{\mathcal B}}_{-,M}(\{\lambda_j\} _{j=1}^M |\alpha -\beta +1)\,
  \ket{\Omega _{\alpha ,\beta+1-2M}},
\end{multline}
for some constant $\mathsf{c}_{Q,\text{ref}}$ which depends a priori on the particular choice of $Q$, i.e. of the set of Bethe roots $\{\lambda_j\}_{j=1}^M$.
From the fact that the operators $\widehat{\mathcal B}_-(\lambda_j|\alpha-\beta+2j-1)$ are invertible when evaluated at any admissible Bethe root $\lambda_j$, i.e. satisfying \eqref{cond-roots}, we get that
\begin{equation}
  \ket{\eta,\alpha+\beta+N-1-2M}
   =\mathsf{c}_{Q,\text{ref}}\,
  \ket{\Omega _{\alpha ,\beta+1-2M}}.
\end{equation}
This relation should be true for any $Q\in\Sigma_Q^M$ solution of the $TQ$-equation, so that the proportionality constant $\mathsf{c}_{Q,\text{ref}}$ should in fact coincide for all such $Q\in\Sigma_Q^M$.

The second assertion can be proven similarly.
\end{proof}

\begin{rem}
The two conditions \eqref{+homo-cond} and \eqref{+homo-cond-l} are in general not compatible, and we made different choices for the gauges parameters in \eqref{Ref-SoV} and in \eqref{Ref-SoV-l}, so that we cannot in principle use the ABA form of the right separate states with reference state \eqref{Ref-SoV} and of the left separate states with reference state \eqref{Ref-SoV-l} simultaneously. In the following, in our study of correlation functions, we shall in fact only use the ABA form of the right separate state with condition \eqref{+homo-cond}.
\end{rem}

\subsubsection{Spectrum and eigenstates for the case of a generic magnetic field on site 1 and a
special z-oriented on site $N$}

Here, we write explicitly the form that takes the spectrum and eigenstates of the transfer
matrix under the particular case considered in this paper: when the most general boundary conditions are considered on the site 1, i.e. general values of the boundary parameters $\varsigma _{-}$, $\kappa _{-}$ and $\tau_{-}$, whereas the following very special choice of the boundary conditions is made on the
site $N$:
\begin{align}
K_{-,0}(\lambda ) 
&=K_{-,0}(\lambda ;\varsigma _{+}=-\infty ,\kappa_{+},\tau _{+})  \label{Special-K+} \\
&=%
\begin{pmatrix}
e^{(\eta /2-\lambda )} & 0 \\ 
0 & e^{(\lambda -\eta /2)}%
\end{pmatrix}%
_{0}=e^{(\eta /2-\lambda )\sigma _{0}^{z}}.
\end{align}
Then the transfer matrix spectrum characterization reads:

\begin{proposition}
\label{prop-special+BC}
Let us fix the boundary condition at site $N$ by imposing \eqref{Special-K+}, and let us suppose that the inhomogeneity parameters are generic, i.e. satisfy \eqref{cond-inh}.
Let us moreover define
\begin{equation}
\Sigma _{Q}=\cup _{n=0}^{N}\Sigma _{Q}^{n}.
\end{equation}

Then, for almost any choice of the boundary parameters $\varsigma _{-}$, $\kappa
_{-}$ and $\tau _{-}$ at site 1, the transfer matrix $\mathcal{T}(\lambda)$ is diagonalizable with simple
spectrum. Moreover, the set $\Sigma _{\mathcal{T}}$ of its eigenvalues  is given by the set of entire functions $\tau(\lambda)$ such that there exists a polynomial $Q(\lambda )\in \Sigma _{Q}$ satisfying with $\tau(\lambda)$ the homogeneous $TQ$-equation  \eqref{hom-TQ}
with coefficient $\mathbf{A}_{\boldsymbol{\varepsilon}}(\lambda)$ defined as in \eqref{DefFullA} in terms  of 
\begin{equation}
\mathbf{a}_{\boldsymbol{\varepsilon}}(\lambda )
=\frac{\sinh (\lambda -\frac{\eta }{2}+\epsilon _{\varphi _{-}}\varphi _{-})\,
          \cosh (\lambda -\frac{\eta }{2}-\epsilon _{\psi _{-}}\psi _{-})}
          {\sinh (\epsilon _{\varphi _{-}}\varphi_{-})\,\cosh (\epsilon _{\psi _{-}}\psi _{-})},  \label{Reduced form}
\end{equation}
for $\epsilon _{\varphi _{-}}=\epsilon _{\psi _{-}}\in\{+1,-1\}$.
For $\tau(\lambda)\in\Sigma _{\mathcal{T}}$, the corresponding $Q(\lambda )\in \Sigma _{Q}$ solution of  \eqref{hom-TQ} is unique,
and the corresponding unique (up overall constants) left and right eigenstates of $\mathcal{T}(\lambda)$ are respectively given by the SoV states $\ket{Q}$ \eqref{eigenR} and $\bra{Q}$ \eqref{eigenL}.

Moreover, if $\lambda_1,\ldots,\lambda_M$ label the roots of $Q(\lambda)$ similarly as in \eqref{Q-form},
then the Bethe states
\begin{equation}\label{Bethe-state-r}
 \underline{\widehat{\mathcal{B}}}_{-,M}(\{\lambda _{i}\}_{i=1}^{M}|\alpha -\beta+1)\,
 \ket{\eta ,\alpha +\beta +N-2M-1} ,
\end{equation}
for $\alpha$ and $\beta$ fixed in terms of the boundary parameters by \eqref{Gauge-cond-A} and \eqref{Gauge-cond-B} with $\epsilon_-=\epsilon'_-=\epsilon_{\varphi _{-}}=\epsilon _{\psi _{-}}$, and 
\begin{equation}\label{Bethe-state-l}
  \bra{\alpha +\beta +2M-N+1,\eta}\, 
  \underline{\widehat{\mathcal{B}}}_{-,M}(\{\lambda _{i}\}_{i=1}^{M}|\alpha -\beta+1-2M),
\end{equation}
for $\alpha$ and $\beta$ fixed in terms of the boundary parameters by \eqref{Gauge-cond-A} and \eqref{Gauge-cond-B} with $\epsilon_-=\epsilon'_-=-\epsilon_{\varphi _{-}}=-\epsilon _{\psi _{-}}$,
are respectively collinear to $\ket{Q}$ and $\bra{Q}$. 

Finally, the transfer matrix is isospectral to the transfer matrix of an open spin chain with  diagonal boundary conditions with  boundary  parameters $\varsigma _{\pm }^{(D)}$ given as
\begin{equation}
\varsigma _{\epsilon }^{(D)}=\epsilon _{\varphi _{-}}\varphi _{-},\qquad
\varsigma _{-\epsilon }^{(D)}=-\epsilon  _{\varphi _{-}}\psi _{-}+i\pi /2,
\label{Id-boundary}
\end{equation}
for $\epsilon _{\varphi _{-}} =1$ or $-1.$
\end{proposition}

\begin{proof}
The condition \eqref{Special-K+}
implies that $K_{-,0}(\lambda )$ is non-proportional to the identity so that
the requirements of Theorem~\ref{th-spectrum2} are satisfied. Therefore the above
characterization is a corollary of Theorem~\ref{th-spectrum2} once we remark that the
coefficients \eqref{Reduced form} are compatible with the
quantum determinant condition:
\begin{equation}
\det_{q}K_{-,0}(\lambda )=1,
\end{equation}
and that they can be obtained as the limit $\varsigma _{+}\rightarrow -\infty $
of the general coefficients: 
\begin{equation}
\mathbf{\lim_{\varsigma _{+}\rightarrow -\infty }}\frac{\sinh (\lambda -%
\frac{\eta }{2}+\epsilon _{\varphi _{+}}\varphi _{+})\,\cosh (\lambda -\frac{%
\eta }{2}+\epsilon _{\psi _{+}}\psi _{+})}{\sinh (\epsilon _{\varphi
_{+}}\varphi _{+})\,\cosh (\epsilon _{\psi _{+}}\psi _{+})}=1,
\end{equation}
for the choice $\epsilon _{\varphi _{+}}=\epsilon _{\psi _{+}}$. Indeed, the limit $\varsigma _{+}\rightarrow -\infty$ can be achieved by taking
\begin{equation}\label{lim-infty}
  \varphi_{+}\rightarrow -\infty ,\quad \psi _{+}\rightarrow +\infty
  \qquad
  \text{with}\quad \varphi _{+}-\psi _{+}- \varsigma _{+}\ \text{finite}.
\end{equation}

In the limit $\varsigma _{+}\rightarrow -\infty $, both conditions \eqref{Gauge-B+} and \eqref{cond-eigen-D+} of Lemma~\ref{lemme-eigen-A} are satisfied, so that one can show similarly as in Proposition~\ref{prop-special+BC} by direct ABA computations that the states \eqref{Bethe-state-r} and \eqref{Bethe-state-l} are eigenstates of the transfer matrix with the same eigenvalue as $\ket{Q}$ and $\bra{Q}$, so that they should be collinear to the latter due to the simplicity of the transfer matrix spectrum.

The second part of the proposition, about isospectrality, follows once we
observe that under the identification of the boundary parameters given in $%
\left( \ref{Reduced form}\right) $ the corresponding diagonal transfer
matrix and the original non-diagonal one have their eigenvalues completely
characterized by the same homogeneous $TQ$-equation.
\end{proof}

In the following, for the computation of correlation functions for the chain with such special boundary conditions, we will use the fact that boundary Bethe states of the form \eqref{Bethe-state-r} can be decomposed into bulk Bethe states, similarly as what was done in \cite{KitKMNST07}, as explained in Appendix~\ref{app-boundary-bulk}. This will enable us to compute the action of local operators on such states as done in \cite{KitKMNST07}. The computation of such an action is the purpose of the next section.

\section{Action of local operators on boundary states}
\label{sec-act-state}

In this section, we compute the action of a basis of local operators on boundary states of ABA form \eqref{Bethe-state-r}. The strategy is similar to the one used in the diagonal case \cite{KitKMNST07}, i.e.
\begin{enumerate}
\item we decompose the boundary Bethe states of ABA form in terms of bulk Bethe states;
\item we act with local operators on the bulk Bethe states by means of the solution of the bulk inverse problem;
\item we reconstruct the resulting states in terms of boundary states.
\end{enumerate}
However, here, all this process should be done in terms of the gauged boundary and bulk operators.

The boundary-bulk decomposition of arbitrary states of the form \eqref{Bethe-state-r} in terms of gauged bulk Bethe states is done in Appendix~\ref{app-boundary-bulk} for the particular boundary conditions that we consider in this paper, i.e. with the boundary matrix \eqref{Special-K+}. From Proposition~\ref{prop-boundary-bulk}, it takes the form
\begin{multline}\label{Boundary-bulk-state}
   \underline{\widehat{\mathcal B}}_{-,M}(\{\lambda_i\}_{i=1}^M | \alpha-\beta+1)\, \ket{\eta,\alpha +\beta +N-2M-1}
   = \mathsf{h}_M(x;\alpha,\beta) 
   \sum_{\sigma_1=\pm,\ldots,\sigma_M=\pm} \hspace{-4mm}
   H_{\sigma_1,\ldots,\sigma_M}(\{\lambda_i\}_{i=1}^M)\, 
   \\
   \times
   \underline{B}_M(\{\sigma_i\lambda_i\}_{i=1}^M|x-1,\alpha-\beta)\, \ket{\eta,\alpha +\beta +N-M-1},
\end{multline}
where the coefficient $H_{\sigma_1,\ldots,\sigma_M}$ is given by \eqref{H_sigma}, that we recall here:
\begin{equation}\label{H_sig}
   H_{\sigma_1,\ldots,\sigma_M}(\{\lambda_i\}_{i=1}^M)
   =\prod_{j=1}^M\left[ \sigma_j\, a(-\lambda_j^{(\sigma)})\,
                                    \frac{\sinh(2\lambda_j-\eta)}{\sinh(2\lambda_j)} \right]
     \prod_{1\le i<j\le M}\frac{\sinh(\lambda_i^{(\sigma)}+\lambda_j^{(\sigma)}+\eta)}{\sinh(\lambda_i^{(\sigma)}+\lambda_j^{(\sigma)})},
\end{equation}
and where $\mathsf{h}_M(x;\alpha,\beta)\equiv h_M(x,\alpha-\beta,\alpha +\beta +N-2M-1)$ is an overall non-zero constant depending exclusively on the gauge parameters $x,\alpha,\beta$ and on the number $M$ of $\widehat{\mathcal B}_-$ operators (see \eqref{hM}):
\begin{align}\label{hMalpha}
  \mathsf{h}_M(x;\alpha,\beta)
 &=(-1)^{M(N+1)} \, e^{M\eta\frac{x+\alpha-\beta+N}2}\, \prod_{j=1}^M\frac{\sinh(\eta(\beta-M-j) )}{2\sinh(\eta\frac{-x+\alpha+\beta+N-2M-1+2j}2)}.
\end{align}
In \eqref{H_sig} and in the following, we use the shortcut notations $\lambda_j^{(\sigma)}\equiv \sigma_j\lambda_j$, $1\le j\le M$, and
\begin{multline}\label{bulk-state}
   \underline{B}_M(\{\mu_i\}_{i=1}^M|x-1,\alpha-\beta)\, \ket{\eta,\alpha +\beta +N-M-1}
   \\
   =\prod_{j=1\to M} \hspace{-2mm} B(\mu_j|x-j,\alpha-\beta+j-1)\, \ket{\eta,\alpha +\beta +N-M-1}.
\end{multline}
To compute correlation functions, we therefore need to act with local operators on the bulk gauged states of the form \eqref{bulk-state}. In subsection~\ref{sec-reconstr}, we explain how to reconstruct local operators in terms of the elements of the gauged bulk monodromy matrix \eqref{gauge-M}. 
In subsection~\ref{sec-act-bulk}, we identify a basis of local operators whose action on states of the form \eqref{bulk-state} is relatively simple, and we explicitly compute this action. Finally, in subsection~\ref{sec-act-boundary}, we compute the resulting action on the boundary states \eqref{Boundary-bulk-state}.

\subsection{Reconstruction of local operators in terms of bulk gauged Yang-Baxter generators}
\label{sec-reconstr}

Let us start remarking that, using the ordinary bulk inverse relation, 
\begin{equation}
\hat{M}(\pm \lambda -\eta /2)=(-1)^{N}\,M^{-1}(\mp \lambda -\eta/2)\,\det_{q}M(\mp \lambda ),
\end{equation}
we obtain the following inversion relation for the gauge bulk monodromy matrix \eqref{gauge-M}:
\begin{equation}
\hat{M}(\pm \lambda -\eta /2|(\alpha,\beta ),(\gamma,\delta))
=(-1)^{N}\,M^{-1}(\mp \lambda-\eta/2 |(\alpha ,\beta ),(\gamma ,\delta))\,\det_{q}M^{t_{0}}(\mp \lambda ).
\end{equation}
Let us now define, for $i,j\in\{1,2\}$, the following local operators at the site $n$ of the chain:
\begin{equation}\label{gauge-op}
E_{n}^{i,j}(\lambda |(\alpha ,\beta ),(\gamma ,\delta ))
=S_{n}(-\lambda|\gamma ,\delta )\, E_{n}^{i,j}\, S_{n}^{-1}(-\lambda |\alpha ,\beta ),
\end{equation}
where $E^{i,j}\in\End(\mathbb{C}^2)$ is such that $(E^{i,j})_{k,\ell}=\delta_{i,k}\,\delta_{j,\ell}$.
Then, by the ungauged results \cite{KitMT99,MaiT00,KitKMNST07}, one can show the following reconstruction:

\begin{proposition}
The local operators \eqref{gauge-op} can be reconstructed in terms of the elements of the bulk monodromy matrix \eqref{gauge-M} as
\begin{align}
  E_{n}^{i,j}(\xi _{a}|(\alpha ,\beta ),(\gamma ,\delta ))
  &=\prod_{k=1}^{n-1}t(\xi _k-\eta /2)\
  M_{j,i}(\xi _{n}-\eta /2|(\alpha ,\beta ),(\gamma ,\delta ))\,
  \prod_{k=1}^{n} \big[ t(\xi_k-\eta/2)\big]^{-1},
  \label{reconstr-1}
\end{align}
and
\begin{align}
  E_{n}^{i,j}(\xi _{a}|(\gamma ,\delta ),(\alpha ,\beta ))
  &=(-1)^N \prod_{k=1}^n t(\xi_k-\eta /2)\ 
      \frac{\hat{M}_{j,i}(-\eta /2-\xi _{n}|(\alpha ,\beta),(\gamma ,\delta ))}{\det_q M(\xi_n)}\,
      \prod_{k=1}^{n-1} \big[ t(\xi_k-\eta/2)\big]^{-1}
  \nonumber\\
  &=(-1)^{i-j}\,\frac{\det S(-\xi_n|\alpha,\beta)}{\det S(-\xi_n|\gamma,\delta)}\,  
       \prod_{k=1}^n t(\xi_k-\eta /2)\nonumber\\
   &\quad\times
       \frac{M_{3-i,3-j}(\xi_n+\eta/2|(\alpha-1,\beta),(\gamma-1,\delta))}{\det_q M(\xi_n)}\,
      \prod_{k=1}^{n-1} \big[ t(\xi_k-\eta/2)\big]^{-1},
  \label{reconstr-2}
\end{align}
where $t(\lambda )$ is the bulk transfer matrix:
\begin{equation}\label{tramsfer-bulk}
t(\lambda )=\tr_{0}[M_{0}(\lambda )],
\end{equation}
and
\begin{align}
   &M_{j,i}(\lambda|(\alpha ,\beta ),(\gamma ,\delta ))
   =\tr_{0}\!\left[ E_{0}^{i,j}\, M_{0}(\lambda |(\alpha ,\beta ),(\gamma ,\delta ))\right] ,
   \\
   &\hat{M}_{j,i}(\lambda|(\alpha ,\beta),(\gamma ,\delta ))
   =\tr_{0}\!\left[ E_{0}^{i,j}\, \hat{M}_{0}(\lambda|(\alpha ,\beta),(\gamma ,\delta ))\right] .
\end{align}
\end{proposition}

\begin{rem}
It will be convenient, for further computations, to express the product of inverse transfer matrices $\big[ t(\xi_k-\eta/2)\big]^{-1}$ in terms of the transfer matrices $t(\xi_k+\eta/2)$ thanks to the quantum determinant condition:
\begin{equation}
 \big[ t(\xi_k-\eta/2)\big]^{-1}=\frac{t(\xi_k+\eta /2)}{\det_q M(\xi_k)}.
\end{equation}

\end{rem}

\begin{proof}
We use  the known reconstruction \cite{KitMT99,MaiT00,KitKMNST07} for the local operators $E_n^{i,j}(\xi_n|(\alpha,\beta),(\gamma,\delta))$:
\begin{align}
   E_n^{i,j}(\xi_n|(\alpha,\beta),(\gamma,\delta))
   &=\prod_{k=1}^{n-1}t(\xi _k-\eta /2)\,  \tr_{0}\!\left[ E_{0}^{i,j}(\xi_n|(\alpha,\beta),(\gamma,\delta))\, M_{0}(\xi _{n}-\eta /2)\right]  \big[ t(\xi_k-\eta/2)\big]^{-1}
   \nonumber\\
   &=\prod_{k=1}^{n}t(\xi _k-\eta /2)\, \frac{\tr_{0}\!\left[ E_{0}^{i,j}(\xi _{a}|(\alpha ,\beta ),(\gamma,\delta ))\,
   \sigma _{0}^{y} \, M^{t_{0}}(\xi _{n}+\eta/2)\,\sigma _{0}^{y}\right]}{\det_q M(\xi_n)}\,
   \nonumber\\
   &\hspace{6cm}\times
   \prod_{k=1}^{n-1} \big[ t(\xi_k-\eta/2)\big]^{-1},
\end{align}
and remark that
\begin{equation}
  \tr_{0}\!\left[ E_{0}^{i,j}(\xi _{a}|(\alpha ,\beta),(\gamma ,\delta ))\,M_{0}(\xi _{n}-\eta /2)\right]
  =\tr_{0}\!\left[ E_{0}^{i,j}\, M_{0}(\xi _{n}-\eta /2|(\alpha ,\beta ),(\gamma,\delta ))\right] , 
\end{equation}
and
\begin{align}
  &\tr_{0}\!\left[ E_{0}^{i,j}(\xi _{a}|(\gamma,\delta ),(\alpha ,\beta ))\,
   \sigma _{0}^{y} \, M^{t_{0}}(\xi _{n}+\eta/2)\,\sigma _{0}^{y}\right]
   =(-1)^N \tr_{0}\!\left[ E_{0}^{i,j}\,\hat{M}_{0}(-\eta /2-\xi _{n}|(\alpha ,\beta),(\gamma ,\delta ))\right]
   \nonumber\\
   &\quad=\frac{\det S(\lambda +\eta /2|\alpha ,\beta )}{ \det S(\lambda +\eta /2|\gamma ,\delta )}\,
   \tr_{0}\!\left[ E_{0}^{i,j}\,\sigma_{0}^{y}\, M^{t_{0}}(\xi_n+\eta/2 |(\alpha -1,\beta ),(\gamma -1,\delta ))\, \sigma_{0}^{y} \right]
    ,
\end{align}
in which we have used \eqref{Mhat-M}.
\end{proof}

As a consequence, one can show, similarly as for the ungauged case, the following identities for the gauged monodromy matrix elements:

\begin{corollary}
For any $\epsilon,\epsilon',\bar\epsilon'\in\{1,2\}$, and any choice of parameters $\alpha,\beta,\gamma,\delta,\gamma',\delta'$ such that $\sinh(\eta\beta)\not=0$, $\sinh(\eta\delta)\not=0$ and $\sinh(\eta\delta')\not=0$, the gauged Yang-Baxter algebra generators satisfy the following annihilation
relations:
\begin{align}
  &M_{\epsilon ,\epsilon ^{\prime }}(\xi _{n}-\eta /2|(\alpha ,\beta ),(\gamma,\delta ))\,
  M_{\epsilon,\bar{\epsilon}^{\prime }}(\xi _{n}+\eta/2|(\alpha -1,\beta ),(\gamma' -1,\delta' )) =0, 
  \label{MM=0-1}\\
  &M_{\epsilon ^{\prime },\epsilon }(\xi _{n}+\eta /2|(\gamma -1,\delta),(\alpha -1,\beta ))\,
  M_{\bar{\epsilon}^{\prime },\epsilon}(\xi_{n}-\eta /2|(\gamma' ,\delta' ),(\alpha ,\beta )) =0.
  \label{MM=0-2}
\end{align}
One also has, for any $\epsilon,\bar\epsilon,\epsilon',\bar\epsilon'\in\{1,2\}$, and any choice of parameters $\alpha,\beta,\alpha',\beta',\gamma,\delta,\gamma',\delta'$ such that $\sinh(\eta\beta)\not=0$, $\sinh(\eta\beta')\not=0$, $\sinh(\eta\delta)\not=0$ and $\sinh(\eta\delta')\not=0$,
\begin{align}
  &M_{\epsilon,\epsilon'}(\xi _{n}-\eta /2|(\alpha ,\beta),(\gamma ,\delta ))\,
  M_{3-\epsilon,\bar\epsilon'}(\xi_n+\eta/2|(\alpha-1,\beta),(\gamma'-1,\delta'))
  =(-1)^{\bar\epsilon-\epsilon}  \frac{e^{\eta\alpha}\sinh(\eta\beta')}{e^{\eta\alpha'}\sinh(\eta\beta)}
  \nonumber\\
  &\hspace{2cm}\times
  M_{\bar\epsilon,\epsilon'}(\xi _{n}-\eta /2|(\alpha' ,\beta'),(\gamma ,\delta ))\,
  M_{3-\bar\epsilon,\bar\epsilon'}(\xi_n+\eta/2|(\alpha'-1,\beta'),(\gamma'-1,\delta')),
  \label{MM=MM-1}
  \\
  &
  M_{\epsilon',3-\epsilon}(\xi_n+\eta/2|(\gamma-1,\delta),(\alpha-1,\beta))\,
  M_{\bar\epsilon',\epsilon}(\xi_n-\eta/2|(\gamma',\delta'),(\alpha,\beta))
  =(-1)^{\bar\epsilon-\epsilon}\,\frac{e^{\eta\alpha'}\sinh(\eta\beta)}{e^{\eta\alpha}\sinh(\eta\beta')}\,
  \nonumber\\
  &\hspace{2cm}\times
  M_{\epsilon',3-\bar\epsilon}(\xi_n+\eta/2|(\gamma-1,\delta),(\alpha'-1,\beta'))\,
  M_{\bar\epsilon',\bar\epsilon}(\xi_n-\eta/2|(\gamma',\delta'),(\alpha',\beta')).
  \label{MM=MM-2}
\end{align}
\end{corollary}

\begin{proof}
The proof works as for the ungauged case. 
From the definition \eqref{gauge-op} and the reconstruction formulas \eqref{reconstr-1} and \eqref{reconstr-2}, we have:
\begin{align}
  &S_{n}(-\xi_n|\gamma ,\delta )\, E_{n}^{i,j}\, E_{n}^{k,\ell}\, S_{n}^{-1}(-\xi_n |\gamma' ,\delta' )
  =E_{n}^{i,j}(\xi _{n}|(\alpha ,\beta ),(\gamma,\delta )) \cdot  E_{n}^{k,\ell}(\xi _{n}|(\gamma' ,\delta' ),(\alpha,\beta ))
\nonumber\\
  &\quad
  =(-1)^{k-\ell} \frac{\det S(-\xi_n|\alpha,\beta)}{\det S(-\xi_n|\gamma',\delta')\, \det_q M(\xi_n)}\,
  \prod_{a=1}^{n-1}t(\xi _{a}-\eta/2)\, 
  \nonumber\\
   &\quad\times
  M_{j,i}(\xi _{n}-\eta /2|(\alpha ,\beta),(\gamma ,\delta ))\,
  M_{3-k,3-\ell}(\xi_n+\eta/2|(\alpha-1,\beta),(\gamma'-1,\delta'))\,
  \prod_{a=1}^{n-1}t(\xi _{a}-\eta/2)^{-1},
  \nonumber\\
  &\quad
  =(-1)^{i-j}\frac{\det S(-\xi_n|\gamma,\delta)}{\det S(-\xi_n|\alpha,\beta)\,\det_q M(\xi_n)}\,
  \prod_{a=1}^n t(\xi_a-\eta/2)\,
  \nonumber\\
  &\quad\times
  M_{3-i,3-j}(\xi_n+\eta/2|(\gamma-1,\delta),(\alpha-1,\beta))\,
  M_{\ell,k}(\xi_n-\eta/2|(\gamma',\delta'),(\alpha,\beta))\,
  \prod_{a=1}^n t(\xi_a-\eta/2)^{-1}.
  \label{EE}
\end{align}
To prove \eqref{MM=0-1} and \eqref{MM=0-2}, it is then enough to notice that \eqref{EE} vanishes when $k\not=j$, i.e. when $j=3-k$, which implies the cancellation of the corresponding product of monodromy matrix elements. 

\eqref{MM=MM-1} and \eqref{MM=MM-2} follow from the fact that
\begin{align}
   E_n^{i,j}\, E_n^{j,\ell}=E_n^{i,k}\, E_n^{k,\ell},
\end{align}
which implies
\begin{multline}
  E_{n}^{i,j}(\xi _{n}|(\alpha ,\beta ),(\gamma,\delta )) \cdot  E_{n}^{j,\ell}(\xi _{n}|(\gamma' ,\delta' ),(\alpha,\beta ))
  \\
  =E_{n}^{i,k}(\xi _{n}|(\alpha' ,\beta' ),(\gamma,\delta )) \cdot  E_{n}^{k,\ell}(\xi _{n}|(\gamma' ,\delta' ),(\alpha',\beta' )),
\end{multline}
and one uses the first  identity of \eqref{EE}, respectively the second identity of \eqref{EE}.
%
%
\end{proof}

As a particular case of \eqref{MM=MM-1} for $\epsilon=\bar\epsilon'=2,\epsilon'=\bar\epsilon=1$, we notably have 
\begin{multline}\label{Alternative-q-det}
 C(\xi _{n}-\eta /2|(\alpha ,\beta),(\gamma ,\delta ))\,
  B(\xi_n+\eta/2|(\alpha-1,\beta),(\gamma'-1,\delta'))
  \\  
  = -\frac{e^{\eta\alpha}\sinh(\eta\beta')}{e^{\eta\alpha'}\sinh(\eta\beta)}\,
  A(\xi _{n}-\eta /2|(\alpha' ,\beta'),(\gamma ,\delta ))\,
  D(\xi_n+\eta/2|(\alpha'-1,\beta'),(\gamma'-1,\delta')),
\end{multline}

Note also that \eqref{MM=0-1}, \eqref{MM=0-2} can be extended to the case in which the two operators belong to a larger product of operators:

\begin{corollary}\label{cor-prodM=0}
Let $n\le m$ and let $\boldsymbol{\alpha}\equiv(\alpha_n,\ldots,\alpha_m)$, $\boldsymbol{\beta}\equiv(\beta_n,\ldots,\beta_m)$, $\boldsymbol{\alpha'}\equiv(\alpha'_n,\ldots,\alpha'_m)$, $\boldsymbol{\beta'}\equiv(\beta'_n,\ldots,\beta'_m)$, $\boldsymbol{\gamma}\equiv(\gamma_n,\ldots,\gamma_m)$ and $\boldsymbol{\delta}\equiv(\delta_n,\ldots,\delta_m)$ be arbitrary $(m-n+1)$-tuples of parameters such that all operators
\begin{equation}\label{prod-S}
   \mathfrak{S}^{n,m} _{(\boldsymbol{\omega},\boldsymbol{\omega'})}=\prod_{k=n}^{m} S_{k}(-\xi_k|\omega _{k},\omega'_{k}),
\end{equation}
for $(\boldsymbol{\omega},\boldsymbol{\omega'})\in\{ (\boldsymbol{\alpha},\boldsymbol{\beta}),(\boldsymbol{\alpha'},\boldsymbol{\beta'}),(\boldsymbol{\gamma},\boldsymbol{\delta})\}$ are invertible.
Then, the product of operators
\begin{equation}
   \prod_{k=n\to m} M_{\epsilon_k,\epsilon'_k}(\xi_k-\eta/2|(\gamma_k,\delta_k),(\alpha_k,\beta_k))\,
   \prod_{k=m\to n} M_{\bar\epsilon_k,\bar\epsilon'_k}(\xi_k+\eta/2|(\gamma_k-1,\delta_k),(\alpha'_k-1,\beta'_k))
\end{equation}
vanishes as soon as there exists some $k\in\{n,\ldots,m\}$ such that $\epsilon_k=\bar\epsilon_k$.
\end{corollary}

\begin{proof}
This is once again a direct consequence of the reconstruction formulas \eqref{reconstr-1} and \eqref{reconstr-2}. We have
\begin{align}
   &\prod_{k=n}^m E_k^{\epsilon'_k,\epsilon_k}(\xi_k|(\gamma_k,\delta_k),(\alpha_k,\beta_k))
   \prod_{k=n}^m E_k^{3-\bar\epsilon_k,3-\bar\epsilon'_k}(\xi_k|(\alpha'_k,\beta'_k),(\gamma_k,\delta_k))
   \nonumber\\
   &\qquad
   =\mathfrak{S}^{n,m} _{(\boldsymbol{\alpha},\boldsymbol{\beta})}
     \prod_{k=n}^m E_k^{\epsilon'_k,\epsilon_k} \prod_{k=n}^m E_k^{3-\bar\epsilon_k,3-\bar\epsilon'_k}
    \left[ \mathfrak{S}^{n,m} _{(\boldsymbol{\alpha'},\boldsymbol{\beta'})}\right]^{-1}
    \label{prod-E}\\
   &\qquad
   =\prod_{k=n}^m\left[(-1)^{\bar\epsilon'_k-\bar\epsilon_k}
   \frac{\det S(-\xi_k|\gamma_k,\delta_k)}{\det S(-\xi_k|\alpha'_k,\beta'_k)}\right]
   \prod_{k=1}^{n-1} t(\xi_k-\eta) 
   \prod_{k=n\to m} \!\!\!\! M_{\epsilon_k,\epsilon'_k}(\xi_k-\eta/2|(\gamma_k,\delta_k),(\alpha_k,\beta_k))\,
   \nonumber\\
   &\qquad\hspace{2cm}\times
   \prod_{k=m\to n}
   \frac{M_{\bar\epsilon_k,\bar\epsilon'_k}(\xi_k+\eta/2|(\gamma_k-1,\delta_k),(\alpha'_k-1,\beta'_k))}{\det_q M(\xi_k)} 
   \prod_{k=1}^{n-1} t(\xi_k-\eta)^{-1},
\end{align}
in which we have used \eqref{reconstr-1} to reconstruct the first product of operators and \eqref{reconstr-2} to reconstruct the second product. \eqref{prod-E} obviously vanishes if $\epsilon_k=\bar\epsilon_k$ for some $k$, which achieves the proof.
\end{proof}

We can now use the above results to produce the reconstruction of a basis of
quasi-local operators:

\begin{proposition}\label{prop-prod-gauge}
Let us consider the following monomials of local operators:
\begin{align}\label{gauge-m-basis}
\prod_{n=1}^m E_n^{\epsilon'_n,\epsilon_n}(\xi_n|(\alpha'_n,\beta'_n),(\alpha_n,\beta_n))
=\mathfrak{S}^{1,m}_{(\boldsymbol{\alpha},\boldsymbol{\beta})}\,
\prod_{n=1}^{m}E_{n}^{\epsilon _{n}^{\prime },\epsilon_{n}}\
\left[\mathfrak{S}^{1,m}_{(\boldsymbol{\alpha' },\boldsymbol{\beta'})}\right]^{-1},
\end{align}
obtained from the elementary basis of local operators $\prod_{n=1}^{m}E_{n}^{\epsilon _{n}^{\prime },\epsilon _{n}}$, in which we have denoted $\boldsymbol{\epsilon}=(\epsilon_1,\ldots,\epsilon_n)$, $\boldsymbol{\epsilon'}=(\epsilon'_1,\ldots,\epsilon'_n)$, by means of the tensor product like operators $\mathfrak{S}^{1,m}_{(\boldsymbol{\alpha},\boldsymbol{\beta})}$ and $\mathfrak{S}^{1,m}_{(\boldsymbol{\alpha' },\boldsymbol{\beta'})}$ defined as in \eqref{prod-S}.
Here $\boldsymbol{\alpha}\equiv(\alpha_1,\ldots,\alpha_m)$, $\boldsymbol{\beta}\equiv(\beta_1,\ldots,\beta_m)$, $\boldsymbol{\alpha'}\equiv(\alpha'_1,\ldots,\alpha'_m)$ and $\boldsymbol{\beta'}\equiv(\beta'_1,\ldots,\beta'_m)$ are arbitrary $m$-tuples of parameters such that the operators
$\mathfrak{S}^{1,m}_{(\boldsymbol{\alpha},\boldsymbol{\beta})}$ and $\mathfrak{S}^{1,m}_{(\boldsymbol{\alpha' },\boldsymbol{\beta'})}$
are invertible. 
Then, the elements \eqref{gauge-m-basis} admit the
following reconstruction in terms of the gauge
transformed Yang-Baxter generators:
\begin{multline}\label{recontr-m-op}
  \prod_{i=1}^m E_i^{\epsilon'_i,\epsilon_i}(\xi_i|(\alpha'_i,\beta'_i),(\alpha_i,\beta_i))
  =\prod_{i=1}^m\frac{\det S(-\xi_i|\gamma_i,\delta_i)}{\det S(-\xi_i|\alpha_i',\beta_i')}\
  \prod_{i=1\to m}\hspace{-2mm} M_{\epsilon_i,\epsilon'_i}(\xi_i-\eta/2|(\gamma_i,\delta_i),(\alpha_i,\beta_i))
  \\
  \times
  \prod_{i=m\to 1}\hspace{-2mm}\frac{M_{3-\epsilon_i,3-\epsilon_i}(\xi _i+\eta /2|(\gamma _i-1,\delta _i),(\alpha' _i-1,\beta'_i))}{\det_q M(\xi_i)},
\end{multline}
for any $m$-tuples of arbitrary parameters $\boldsymbol{\gamma}\equiv(\gamma_1,\ldots,\gamma_m)$ and $\boldsymbol{\delta}\equiv(\delta_1,\ldots,\delta_m)$ such that the operator $\mathfrak{S}^{1,m}_{(\boldsymbol{\gamma},\boldsymbol{\delta})}$, defined similarly as in \eqref{prod-S}, is invertible.\end{proposition}

\begin{proof}
This reconstruction follows from the identities
\begin{equation}
     E_n^{\epsilon'_n ,\epsilon_n}\, E_n^{\epsilon_n,\epsilon_n}=E_n^{\epsilon'_n ,\epsilon_n},
\end{equation}
for any $\epsilon_n,\epsilon'_n\in\{1,2\}$,
which enable us to write
\begin{align}
&\prod_{n=1}^m E_n^{\epsilon'_n,\epsilon_n}(\xi_n|(\alpha'_n,\beta'_n),(\alpha_n,\beta_n))
   =\mathfrak{S}^{1,m}_{(\boldsymbol{\alpha},\boldsymbol{\beta})}\,
\prod_{n=1}^{m}E_{n}^{\epsilon _{n}^{\prime },\epsilon_{n}}\
\left[\mathfrak{S}^{1,m}_{(\boldsymbol{\gamma },\boldsymbol{\delta})}\right]^{-1}\,
\mathfrak{S}^{1,m}_{(\boldsymbol{\gamma},\boldsymbol{\delta})}
\prod_{n=1}^m E_n^{\epsilon_n,\epsilon_n}
\left[\mathfrak{S}^{1,m}_{(\boldsymbol{\alpha' },\boldsymbol{\beta'})}\right]^{-1}
\nonumber\\
   &\hspace{3cm}=\prod_{n=1}^m E_n^{\epsilon'_n,\epsilon_n}(\xi_n|(\gamma_n,\delta_n),(\alpha_n,\beta_n))\,
   \prod_{n=1}^m E_n^{\epsilon_n,\epsilon_n}(\xi_n|(\alpha'_n,\beta'_n),(\gamma_n,\delta_n)).
\end{align}
Then we use \eqref{reconstr-1} to reconstruct the first product of operators:
\begin{equation}
   \prod_{i=1}^m E_i^{\epsilon'_i,\epsilon_i}(\xi_i |(\gamma_i,\delta_i),(\alpha_i,\beta_i))
   =
   \prod_{i=1\to m} \hspace{-2mm} M_{\epsilon_i,\epsilon'_i}(\xi_i-\eta /2|(\gamma_i,\delta_i),(\alpha _i,\beta_i))\,
\prod_{i=1}^m\big[ t(\xi_i-\eta/2)\big]^{-1},
\end{equation}
and \eqref{reconstr-2} to reconstruct the second product of operators:
\begin{multline}
   \prod_{i=1}^m E_i^{\epsilon_i,\epsilon_i}(\xi_i |(\alpha'_i,\beta'_i),(\gamma_i,\delta_i))
   =\prod_{i=1}^m \frac{\det S(-\xi_i |\gamma_i,\delta_i)}{\det S(-\xi_i |\alpha'_i,\beta'_i)}\,
   \prod_{i=1}^m t(\xi_i-\eta/2)\,
   \\
   \times 
   \prod_{i=m\to 1}\hspace{-2mm}
   \frac{M_{3-\epsilon_i,3-\epsilon_i}(\xi_i+\eta/2|(\gamma_i-1,\delta_i),(\alpha'_i-1,\beta'_i))}{\det_q M(\xi_i)},
\end{multline}
to obtain the reconstruction \eqref{recontr-m-op}.
\end{proof}

\subsection{Action of a basis of local operators on gauged bulk states}
\label{sec-act-bulk}

Here, we identify a basis of local operators and we compute the action of
its elements on bulk gauged Bethe like states. 


\begin{proposition}\label{prop-local-basis}
For two given gauge parameters $\alpha$ and $\beta$, 
the set of monomial of local operators 
\begin{equation}
  \mathbb{E}_m(\alpha,\beta)
  =\left\{ \prod_{n=1}^{m}E_{n}^{\epsilon _{n}^{\prime},\epsilon _{n}}(\xi _{n}|(a_n,b_{n}),(\bar a_n,\bar{b}_{n}))\ \mid\
  \boldsymbol{\epsilon},\boldsymbol{\epsilon'}\in\{1,2\}^m
  \right\},  \label{Local-Basis}
\end{equation}
where the gauge parameters $a_n,\bar a_n, b_n,\bar b_n$, $1\le n\le m$, are fixed in terms of $\alpha,\beta$ and of the $m$-tuples $\boldsymbol{\epsilon}\equiv(\epsilon_1,\ldots,\epsilon_m)$ and $\boldsymbol{\epsilon'}\equiv(\epsilon'_1,\ldots,\epsilon'_m)$ as
\begin{alignat}{2}
   &a_n =\alpha+1,\qquad & 
   &b_n=\beta-\sum_{r=1}^n (-1)^{\epsilon_r},   
   \label{Gauge.Basis-1} \\
   &\bar{a}_n =\alpha-1,  
   \qquad &
   &\bar{b}_n=\beta+\sum_{r=n+1}^{m}(-1)^{\epsilon'_r}-\sum_{r=1}^{m}(-1)^{\epsilon_r}
   =b_n+2\tilde m_{n+1},  
   \label{Gauge.Basis-2}
\end{alignat}
with
\begin{align}\label{tilde-m_n}
   &\tilde m_n=\sum_{r=n}^{m}(\epsilon'_r-\epsilon_r)=\sum_{r=n}^{m}\frac{( -1) ^{\epsilon'_r}-( -1) ^{\epsilon_r}}{2}.
\end{align}
%
defines a basis of $\End(\otimes_{n=1}^m\mathcal{H}_n)$, except for a finite numbers of values of $\beta\!\!\mod 2\pi/\eta$.
\end{proposition}

In other words, this means that any local operator acting on the first $m$ sites of the chain can be expressed as a linear combination of elements of \eqref{Local-Basis}.

\begin{proof}
Thanks to the tensor product form of the elements of \eqref{Local-Basis}, we have just
to prove that, for any fixed $n\in \{1,\ldots,m\}$ the four operators 
$E_{n}^{\epsilon _{n}^{\prime},\epsilon _{n}}(\xi _{n}|(a_n,b_{n}),(\bar a_n,\bar{b}_{n}))$
associated with the four possible choices of $(\epsilon _{n},\epsilon _{n}^{\prime})\in \{1,2\}^{2}$,
form a basis independently from the choice of the other $\epsilon_j,\epsilon'_j$ for $j\not= n$. 
Here, the main point is that the gauge parameters $b_{n}$ and $\bar{b}_{n}$, defined in \eqref{Gauge.Basis-1}-\eqref{Gauge.Basis-2}, depend not only on $(\epsilon _{n},\epsilon _{n}^{\prime })\in\{1,2\}^{2}$ but also on the choice of the other $\epsilon_j,\epsilon'_j$ for $j\not= n$. 
We evidence this dependence as follows:
\begin{alignat}{2}
   &b_{n}=\check{b}_{n}-( -1) ^{\epsilon _{n}}, \quad &
   &\text{with}\quad
   \check{b}_{n}=b-\sum_{r=1}^{n-1}(-1)^{\epsilon_r}, \\
   &\bar{b}_{n} =\check{\bar{b}}_{n}-(-1)^{\epsilon _{n}}, \quad &
   &\text{with}\quad
   \check{\bar{b}}_{n}=\check{b}_{n}+ 2\tilde m_{n+1},
\end{alignat}
so that $\check{b}_n$ and $\check{\hat b}_n$ do not depend on $\epsilon_n,\epsilon'_n$, but only on the other $\epsilon_j,\epsilon'_j$, $j\not= n$.
Let us 
%
%
define the coefficients $E_{n,i,j}^{\epsilon _{n}^{\prime},\epsilon _{n}}(\xi _{n}|(a_n,b_{n}),(\bar a_n,\bar{b}_{n}))$, $1\le i,j\le 2$, by the decomposition
\begin{equation}
    E_{n}^{\epsilon _{n}^{\prime },\epsilon _{n}}(\xi _{n}|(a_n,b_{n}),(\bar a_n,\bar{b}_{n}))
    =\sum_{i=1}^{2}\sum_{j=1}^{2}
    E_{n,i,j}^{\epsilon_{n}^{\prime },\epsilon _{n}}(\xi _{n}|(a_n,b_{n}),(\bar a_n,\bar{b}_{n}))\,
    E_{n}^{i,j},
\end{equation}
in the natural basis $E_{n}^{i,j}$, $1\le i,j\le 2$, of local operators on the site $n$.
The determinant
\begin{equation}
  \det_{\substack{(i,j)\in\{1,2\}^2 \\ (\epsilon'_n.\epsilon_n)\in\{1,2\}^2}} \left[ E_{n,i,j}^{\epsilon _{n}^{\prime },\epsilon _{n}}(\xi_{n}|(a_n,b_{n}),(\bar a_n,\bar{b}_{n})) \right]
\end{equation}
can easily be computed, and one can remark that it is a rational function of $e^{-2\eta\beta}$ with coefficients being integers powers of $e^{-2\eta}$. Hence this 
determinant is always nonzero, up to a finite number of values of $\beta$ (modulo the
periodicity $i\pi /\eta $), so that the local operators $E_{n}^{\epsilon'_{n},\epsilon _{n}}(\xi _{n}|(a_n,b_{n}),\bar a_n,\bar{b}_{n}))$ form a basis for any choice of $\epsilon_j,\epsilon'_j$, $j\not= n$.
\end{proof}

The choice of this basis is due to the fact that the action of its elements on the gauged Bethe-type bulk states \eqref{bulk-state} has a relatively simple form, as shown in the remaining part of this section.
To compute this action,  we use the following reconstruction of the local operators elements of \eqref{Local-Basis} (see Proposition~\ref{prop-prod-gauge}): 
\begin{multline}\label{reconst-basis}
   \prod_{n=1}^{m}E_{n}^{\epsilon _{n}^{\prime },\epsilon _{n}}(\xi_{n}|(a_{n},b_{n}),(\bar{a}_{n},\bar{b}_{n}))
    = \prod_{n=1}^m\frac{\det S(-\xi_n|c_n,d_n)}{\det S(-\xi_n| a_n, b_n)}
    \prod_{n=1\to m} \hspace{-2mm}
    M_{\epsilon_n,\epsilon'_n}(\xi_n-\eta /2|(c_n,d_n),(\bar{a}_n,\bar{b}_n)) 
    \\
    \times
    \prod_{n=m\to 1} 
    \frac{M_{3-\epsilon_n,3-\epsilon'_n}(\xi_n+\eta /2|(c_n-1,d_n),(a_n-1,b_n))}{\det_q M(\xi_n)},
\end{multline}
with the following choice for the internal gauge parameters:
\begin{align}
   &c_n=\frac{x+\alpha+\beta+N-1}2-M, \\ 
   &d_n=\frac{-x+\alpha+\beta+N-1}2-M-\sum_{r=1}^n (-1)^{\epsilon_r},
\end{align}
for any arbitrary choice of $x$ which leaves $\mathfrak{S}^{(1,m)}_{(\mathbf{c},\mathbf{d})}$ invertible. 

In the remaining part of this subsection,  we derive the action of \eqref{reconst-basis} on the Bethe-type states \eqref{bulk-state} in full details.
In order to do so we introduce the following notations.

Let us associate, to two given $m$-tuples $\boldsymbol{\epsilon}\equiv(\epsilon_1,\ldots,\epsilon_m), \boldsymbol{\epsilon'}\equiv(\epsilon'_1,\ldots,\epsilon'_m)\in \{1,2\}^{m}$, the following sets of integers, for any $n\in\{1,\ldots,m\}$:
\begin{alignat}{2}
   &\{i_{p}^{(n) }\}_{p\in \{1,\ldots ,s_{(n)}\}},\quad  &
   &\text{with}\quad   i_{k}^{(n)}<i_{h}^{(n)}\quad
   \text{for}\quad   0<k<h\leq s_{(n)},  
   \label{i_p-Def0} \\
  &\{i_{p}^{(n)}\}_{p\in \{s_{(n)}+1,\ldots ,s_{(n)}+s'_{(n)}\}},\quad &
  &\text{with}\quad  i_{k}^{(n)}>i_{h}^{(n)}\quad
  \text{for}\quad  s_{(n)}<k<h\leq s_{(n)}+s_{(n)}^{\prime },  
  \label{i_p-Def1}
\end{alignat}
defined by the conditions
\begin{alignat}{2}
  &j\in \{i_{p}^{(n)}\}_{p\in \{1,\ldots ,s_{(n)}\}} \qquad & 
  &\text{iff}\qquad n\le j\le m\quad\text{and}\quad\epsilon _{j}=2,
  \\
  &j\in \{i_{p}^{(n)}\}_{p\in \{s_{(n) }+1,\ldots ,s_{(n) }+s_{(n) }^{\prime }\}} \qquad & 
  &\text{iff}\qquad n\le j\le m\quad\text{and}\quad \epsilon'_{j}=1. 
\end{alignat}
Note that we have
\begin{equation}
   s_{(n)}=\sum_{j=n}^m(\epsilon_j-1),\qquad
   s'_{(n)}=\sum_{j=n}^m(2-\epsilon'_j),\qquad
   s_{(n)}+s'_{(n)}=m-n+1+\sum_{j=n}^m (\epsilon_j-\epsilon'_j).
\end{equation}
We use the following simplified notations for the case $n=1$,
\begin{equation}\label{def-s-s'-i_p}
    s=s_{( n=1) },\quad s^{\prime}=s_{( n=1) }^{\prime },
    \qquad\text{and}\qquad
    i_{p}=i_{p}^{( n=1) }\quad\forall p\in \{1,\ldots,s+s^{\prime }\}.
\end{equation}
We can relate these sets for two different values of $n$ by
\begin{align}
    &i_{1}^{( n-1) } =(n-1)\delta _{\epsilon _{n-1},2}+(1-\delta_{\epsilon _{n-1},2})\, i_{1}^{( n) },
    \\
    &i_{p+\delta _{\epsilon_{n-1},2}}^{( n-1) }=i_{p}^{( n) },\hspace{6.7mm}
    \forall p\in\{1,\ldots ,s_{( n) }\}, 
    \\
    &i_{s_{( n-1) }+p}^{( n-1) } =i_{s_{( n)}+p}^{( n) },\quad \forall p\in \{1,\ldots ,s_{( n)}^{\prime }\},
    \\
    &i_{s_{( n-1) }+s_{( n-1)}^{\prime }}^{( n-1) }=n\, \delta _{\epsilon _{n-1}^{\prime},1}+(1-\delta _{\epsilon _{n-1}^{\prime },1})\, i_{s_{( n)}+s_{( n) }^{\prime }}^{( n) },
\end{align}
with
\begin{equation}
       s_{( n-1) }=\delta _{\epsilon _{n-1},2}+s_{( n) },
       \qquad
       s_{( n-1) }^{\prime }=\delta _{\epsilon _{n-1}^{\prime},1}+s_{( n) }^{\prime }.
\end{equation}
Then, we can derive the following result:

\begin{proposition}\label{prop-rec-action}
For any given $m$-tuples $\boldsymbol{\epsilon}\equiv(\epsilon_1,\ldots,\epsilon_m), \boldsymbol{\epsilon'}\equiv(\epsilon'_1,\ldots,\epsilon'_m)\in \{1,2\}^{m}$
and gauge parameters $\alpha,\beta,x$, let us consider, for each $n\in\{1,\ldots,m\}$, the following monomials of elements of the bulk gauged monodromy matrix:
\begin{multline}
  \underline{M}^{(n,m)}_{\boldsymbol{\epsilon},\boldsymbol{\epsilon'}}(\alpha,\beta,x)
  =\prod_{k=n\to m} M_{\epsilon_k,\epsilon'_k}(\xi_k-\eta/2|(c_k,d_k),(\bar a_k,\bar b_k))\,
  \\
  \times
  \prod_{k=m\to n} \frac{M_{3-\epsilon_k,3-\epsilon_k}(\xi_k+\eta/2 |(c_k-1,d_k),( a_k-1,b_k))}{\det_q M(\xi_k)},
\end{multline}
%
%
where we have defined
\begin{alignat}{2}
    &a_k = \alpha+1,\qquad & & b_k =\beta-\sum_{r=1}^k (-1)^{\epsilon_r},\\
    &\bar{a}_k = \alpha-1, \qquad & 
    & \bar{b}_k =\beta+\sum_{r=k+1}^{m} \! (-1)^{\epsilon'_r}-\sum_{r=1}^{m}(-1)^{\epsilon_r}
                  = b_k +2\tilde m_{k+1},
    \\
    &c_k = \frac{x+\alpha+\beta+N-1}2-M, \qquad &   &d_k = \frac{-x+\alpha+\beta+N-1}2-M-\sum_{r=1}^k (-1)^{\epsilon_r}.
\end{alignat}
with
\begin{align} 
   &\tilde m_k=\sum_{r=k}^{m}(\epsilon'_r-\epsilon_r)=\sum_{r=k}^{m}\frac{( -1) ^{\epsilon'_r}-( -1) ^{\epsilon_r}}{2}.
\end{align}
Then, the action of $\underline{M}^{(n,m)}_{\boldsymbol{\epsilon},\boldsymbol{\epsilon'}}(\alpha,\beta,x)$
on the following gauged Bethe-like bulk states
\begin{equation}\label{BulkBLS}
   \underline{B}_M(\{\mu_i\}_{i=1}^M | x_n-1,z_n)\,\ket{\eta,y_n}
\end{equation}
where
\begin{align}
   &x_n=x+\sum_{r=1}^{n-1}( -1)^{\epsilon_r}=c_n-d_n-(-1)^{\epsilon_n},\\
   &z_n=\alpha -\beta+\sum_{r=1}^{n-1}( -1)^{\epsilon_r}=a_n-b_n-(-1)^{\epsilon_n}-1,\\
   &y_n=\alpha+\beta+N-M-1 -\sum_{r=1}^{n-1}( -1) ^{\epsilon _r} =a_n+b_n+(-1)^{\epsilon_n}+N-M-2
   \nonumber\\
   &\hphantom{y_n=\alpha+\beta+N-M-1 -\sum_{r=1}^{n-1}( -1) ^{\epsilon _r}}
 =c_n+d_n+M+(-1)^{\epsilon_n},
\end{align}
%
%
is
\begin{multline}\label{action_Mn}
   \underline{M}^{(n,m)}_{\boldsymbol{\epsilon},\boldsymbol{\epsilon'}}(\alpha,\beta,x)\
   \underline{B}_M(\{\mu_i\}_{i=1}^M | x_n-1,z_n)\,\ket{\eta,y_n}
   = 
   \sum_{\mathsf{B}_{\boldsymbol{\epsilon},\boldsymbol{\epsilon'}}^{(n)}}
   \mathcal{F}_{\mathsf{B}_{\boldsymbol{\epsilon},\boldsymbol{\epsilon'}}^{(n)}}(\{\mu_j\}_{j=1}^{M},\{\xi_j^{(1)}\}_{j=n}^{m} \mid  \alpha,\beta,x)\\
   \times
   \underline{B}_{M+\tilde m_n}   (\{\mu_j\}_{j\in\mathsf{A}_{\boldsymbol{\epsilon},\boldsymbol{\epsilon'}}^{(n)}} | x_n-1,z_n-2\tilde m_n)\, \ket{\eta,y_n+\tilde{m}_n},
\end{multline}
where we have defined $\mu _{M+j}\equiv\xi _{m+1-j}^{( 1) }$ for $j\in \{1,\ldots ,m+1-n\}$.
In \eqref{action_Mn}, the sum runs over all the possible sets of integers 
$\mathsf{B}^{( n) }_{\boldsymbol{\epsilon},\boldsymbol{\epsilon'}}
=\{\text{\textsc{b}}_{1}^{( n) },\ldots ,\text{\textsc{b}}_{s_{(n)}+s'_{(n)}}^{( n) }\}$ whose elements satisfy the conditions 
\begin{equation}
\begin{cases}
   \text{\textsc{b}}_{p}^{( n) }\in 
   \{1,\ldots ,M\}\setminus \{\text{\textsc{b}}_{1}^{( n) },\ldots ,\text{\textsc{b}}_{p-1}^{(n) }\}
   \qquad & \text{for}\quad 0<p\leq s_{( n) }, 
   \\ 
   \text{\textsc{b}}_{p}^{( n) }\in \{1,\ldots ,M+m+1-i_{p}^{(n) }\}\setminus \{\text{\textsc{b}}_{1}^{( n) },\ldots ,\text{\textsc{b}}_{p-1}^{( n) }\}
   \quad  & \text{for}\quad s<p\leq s_{( n) }+s_{( n) }^{\prime },
\end{cases}
\label{Def-Bss}
\end{equation}
whereas
\begin{equation}\label{Def-Ass}
   \mathsf{A}^{(n)}_{\boldsymbol{\epsilon},\boldsymbol{\epsilon'}}
   \equiv \big\{\text{\textsc{a}}_{1}^{( n) },\ldots,\text{\textsc{a}}_{M+\tilde m_n}^{( n)}\big\}
   =\{1,\ldots ,M+m+1-n\}\setminus \mathsf{B}^{( n) }_{\boldsymbol{\epsilon},\boldsymbol{\epsilon'}}.
\end{equation}
Finally
\begin{align}\label{Def-F-rec}
     &\mathcal{F}_{\mathsf{B}_{\boldsymbol{\epsilon},\boldsymbol{\epsilon'}}^{(n)}}(\{\mu_j\}_{j=1}^{M},\{\xi_j^{(1)}\}_{j=n}^{m} \mid \alpha,\beta,x)
     = \prod_{k=n}^m 
     \frac{e^{\eta (c_k-a_k+1)}\, }{\sinh(\eta d_k)}\,   
     f^{(n,m)}_{\boldsymbol{\epsilon},\boldsymbol{\epsilon'}}(\alpha,\beta,x)\, 
     \nonumber\\
     &\qquad\times
     \frac{\prod\limits_{j=1}^{s_{(n)}+s'_{(n)}} \bigg[ d(\mu_{\text{\textsc{b}}_j^{(n)}}) \,           
             \frac{\prod_{k=1}^M\sinh(\mu_k-\mu_{\text{\textsc{b}}_j^{(n)}}-\eta)}
                    {\prod_{\substack{k=1 \\ k\neq \text{\textsc{b}}_j^{(n)}}}^M\sinh(\mu_k-\mu_{\text{\textsc{b}}_j^{(n)}})}\bigg] }
           {\prod\limits_{j=n}^m \bigg[ d(\xi_j^{(1)}) \,
            \prod\limits_{k=1}^M\frac{\sinh(\mu_k-\xi_j^{(1)}-\eta)}{\sinh(\mu_k-\xi_j^{(1)})}
            \bigg]}  
           \,
     \prod_{1\leq i<j\leq s_{(n)}+s_{(n)}^{\prime }}
     \frac{\sinh (\mu _{\text{\textsc{b}}_i^{(n)}}-\mu_{\text{\textsc{b}}_j^{(n)}})}
            {\sinh (\mu _{\text{\textsc{b}}_i^{(n)}}-\mu_{\text{\textsc{b}}_j^{(n) }}-\eta )}
     \nonumber\\
     &\qquad\times
     \prod_{p=1}^{s_{(n)}}\vast[ 
     \sinh(\xi_{i_p^{(n)}}^{(1)}-\mu_{\text{\textsc{b}}_p^{(n)}}+\eta(1+b_{i_p^{(n)}}))\,
     \frac{\prod_{k=i_p^{(n)}+1}^m \sinh(\mu_{\text{\textsc{b}}_p^{(n)}}-\xi_k^{(1)}-\eta)}
            {\prod_{k=i_p^{(n)}}^m \sinh(\mu_{\text{\textsc{b}}_p^{(n)}}-\xi_k^{(1)})} 
     \vast]    
     \nonumber\\
     &\qquad\times
     \prod_{p=s_{(n)}+1}^{s_{(n)}+s'_{(n)}}\vast[ 
     \sinh(\xi_{i_p^{(n)}}^{(1)}-\mu_{\text{\textsc{b}}_p^{(n)}}-\eta(1-\bar b_{i_p^{(n)}}))\,
     \frac{\prod_{k=i_p^{(n)}+1}^m \sinh(\xi_k^{(1)}-\mu_{\text{\textsc{b}}_p^{(n)}}-\eta)}
            {\prod_{\substack{k=i_p^{(n)} \\ k\not= M+m+1-{\text{\textsc{b}}_p^{(n)}}}}^m \sinh(\xi_k^{(1)}-\mu_{\text{\textsc{b}}_p^{(n)}})} 
     \vast]      ,     
\end{align}
where
\begin{equation}\label{Def-f-rec}
  f^{(n,m)}_{\boldsymbol{\epsilon},\boldsymbol{\epsilon'}}(\alpha,\beta,x)
   = \begin{cases}
     \displaystyle \prod\limits_{k=1}^{\tilde m_n}\frac{e^{(\frac{x+\alpha+\beta+N}2-M)\eta}\,\sinh(\eta(b_{n-1}+k-M-1))}{2\sinh(\eta(d_{n-1}+k+M))} &\text{if}\quad \tilde m_n>0,\\
     1  &\text{if}\quad \tilde m_n=0,\\
    \displaystyle \prod\limits_{k=1}^{|\tilde m_n|}\frac{2\sinh(\eta(d_{n-1}-k+M+1))}{e^{(\frac{x+\alpha+\beta+N}2-M)\eta}\,\sinh(\eta(b_{n-1}-k-M))} &\text{if}\quad \tilde m_n<0.
   \end{cases}
\end{equation}
%
%
%
\end{proposition}

\begin{proof}
Let us prove this result by induction. 
We can easily prove that it holds for $n=m$ using the actions \eqref{Action-A-0}, \eqref{Action-D-0} and the identity \eqref{Alternative-q-det} according to the required case. 

So let us assume that the stated form of the action holds for  $2\leq n+1\leq m$ and let us prove it for $n$. We distinguish the four possible choices of $(\epsilon _{n},\epsilon _{n}^{\prime })\in
\{1,2\}^{2}$. 


A) Let us first consider the case $(\epsilon_n,\epsilon'_n)=(1,1)$. We have to act on the state \eqref{BulkBLS} with the monomial
\begin{multline}   \label{mon-11}
    \frac{e^{2\eta c_n}}{4\sinh^{2}(\eta d_{n})}\, A(\xi_n^{(1)} | c_n-d_n, \bar{a}_n+\bar{b}_n)\ 
    \underline{M}^{(n+1,m)}_{\boldsymbol{\epsilon},\boldsymbol{\epsilon'}}(\alpha,\beta,x)\
    \frac{D(\xi_n^{(0)} | c_n+d_n-1, a_n-b_n-1)}{\det_q M(\xi_n)}
    \\
 =    \frac{e^{2\eta c_n}}{4\sinh^{2}(\eta d_{n})}\,
 A(\xi_n^{(1)} |x_{n+1},\tilde{y}_{n+1}-N+ M+\tilde m_{n+1})\
    \underline{M}^{(n+1,m)}_{\boldsymbol{\epsilon},\boldsymbol{\epsilon'}}(\alpha,\beta,x)\ 
    \frac{D(\xi_n^{(0)} | y_n-M,z_n-1)}{a(\xi_n^{(0)})\, d(\xi_n^{(1)})}    ,
\end{multline}
in which we have used that
\begin{align}
   &a_n-b_n-1=z_n+(-1)^{\epsilon_n},\label{id:a-b}\\
   &c_n+d_n-1=y_n-M-(-1)^{\epsilon_n}-1,\label{id:c+d}\\
   &c_n-d_n=x_n+(-1)^{\epsilon_n}=x_{n+1},\label{id:c-d}\\
   &\bar{a}_n+\bar{b}_n=y_{n+1}-N+M+2\tilde m_{n+1} 
   =\tilde y_{n+1}-N+(M +\tilde m_{n+1}),\label{id:abar+bbar}
\end{align}
with the shortcut notations
\begin{align}\label{ztilde-ytilde}
   &\tilde{z}_n = z_n-2\tilde m_n,
   \qquad
   \tilde{y}_n=   y_n+\tilde m_n.
\end{align}
We can use \eqref{Action-D-0} to act with the last factor of \eqref{mon-11} on this state. It follows from Corollary~\ref{cor-prodM=0} that only the direct term of the action of the $D$-operator can finally lead to a non-zero contribution, so that the result of this action that we have to take into account  is
\begin{multline}   \label{act-mon11-D}
   \frac{\mathsf{d}_{n}^{(R)}}{d(\xi _{n}^{(1)})}\, 
   \prod_{b=1}^{M}\frac{\sinh (\mu_{b}-\xi _{n}^{(1)})}{\sinh (\mu _{b}-\xi _{n}^{(0)})}\
   \underline{B}_M(\{\mu_i\}_{i=1}^M| x_n-2,z_n)\, \ket{\eta,y_n+1}
   \\
   =   \frac{\mathsf{d}_{n}^{(R)}}{d(\xi _{n}^{(1)})}\, 
   \prod_{b=1}^{M}\frac{\sinh (\mu_{b}-\xi _{n}^{(1)})}{\sinh (\mu _{b}-\xi _{n}^{(0)})}\
   \underline{B}_M(\{\mu_i\}_{i=1}^M| x_{n+1}-1,z_{n+1}+1)\, \ket{\eta,y_{n+1}},
\end{multline}
with 
\begin{align}
   \mathsf{d}_{n}^{(R)}
   &=\frac{e^{-(z_n+M-1)\eta}-e^{-(y_n-N)\eta}}{e^{\eta/2}}\,
   \frac{\sinh(\eta \frac{y_n-x_n-M+2}{2})}{\sinh(\eta\frac{y_n-x_n+M+2}{2})}
   \nonumber\\
   &=\frac{e^{-(a_n-b_n-1+M)\eta}-e^{-(a_n+b_n-1-M-2)\eta}}{e^{\eta/2}}\,
   \frac{\sinh(\eta d_n)}{\sinh(\eta(d_n+M))},
   \nonumber\\
   &=\frac{2\sinh(\eta(b_n-M-1))}{e^{(a_n-2)\eta+\eta/2}}\,
   \frac{\sinh(\eta d_n)}{\sinh(\eta(d_n+M))}.
\end{align}
%
%
We now use the induction formula to compute the action of $\underline{M}^{(n+1,m)}_{\boldsymbol{\epsilon},\boldsymbol{\epsilon'}}(\alpha,\beta,x)$ on this state, which gives
\begin{multline} \label{D+BloA}
  \frac{\mathsf{d}_{n}^{(R)}}{d(\xi _{n}^{(1)})}\, 
   \prod_{b=1}^{M}\frac{\sinh (\mu_{b}-\xi _{n}^{(1)})}{\sinh (\mu _{b}-\xi _{n}^{(0)})}\
  \sum_{\mathsf{B}_{\boldsymbol{\epsilon},\boldsymbol{\epsilon'}}^{(n+1)}}
   \mathcal{F}_{\mathsf{B}_{\boldsymbol{\epsilon},\boldsymbol{\epsilon'}}^{(n+1)}}(\{\mu_j\}_{j=1}^{M},\{\xi_j^{(1)}\}_{j=n+1}^{m}\mid \alpha,\beta,x)\
   \\
   \times
   \underline{B}_{M+\tilde m_{n+1}}   (\{\mu_j\}_{j\in\mathsf{A}_{\boldsymbol{\epsilon},\boldsymbol{\epsilon'}}^{(n+1)}} | x_{n+1}-1,\tilde{z}_{n+1})\, \ket{\eta,\tilde{y}_{n+1}}.
\end{multline}
We finally act with the last $A$-operator on this state using \eqref{Action-A-0}.
%
%
We have
\begin{multline}
  A(\xi_n^{(1)} |x_{n+1},\tilde{y}_{n+1}-N+M+\tilde m_{n+1})\,
  \underline{B}_{M+\tilde m_{n+1}}   (\{\mu_j\}_{j\in\mathsf{A}_{\boldsymbol{\epsilon},\boldsymbol{\epsilon'}}^{(n+1)}} | x_{n+1}-1,\tilde{z}_{n+1})\, \ket{\eta,\tilde{y}_{n+1}}
  \\
  =\mathsf{a}_n^{(L)}
  \hspace{-3mm}
  \sum_{a\in\mathsf{A}^{(n+1)}_{\boldsymbol{\epsilon,\epsilon'}}\cup\{M+m+1-n\}} \hspace{-7mm}   
  d(\mu_a)\,\sinh(\xi_n^{(1)}-\mu_a+\eta(\bar{b}_n-1))\,
  \frac{\prod_{j\in\mathsf{A}^{(n+1)}_{\boldsymbol{\epsilon,\epsilon'}} }
  \sinh(\mu_j-\mu_a-\eta)}
  {\prod_{\substack{ j\in\mathsf{A}^{(n+1)}_{\boldsymbol{\epsilon,\epsilon'}}\cup\{M+m+1-n\} \\ j\neq a}}
  \sinh(\mu_j-\mu_a)}
  \\
  \times
  \underline{B}_{M+\tilde m_{n+1}}    (\{\mu_j\}_{j\in\mathsf{A}_{\boldsymbol{\epsilon},\boldsymbol{\epsilon'}}^{(n+1)}\cup\{M+m+1-n\}}\setminus\{\mu_a\} | x_{n+1},\tilde{z}_{n+1}+2)\, \ket{\eta,\tilde{y}_{n+1}-1}
\end{multline}
in which we have used that
\begin{equation}
   \frac{\tilde y_{n+1}-N-\tilde z_{n+1}-1+(M+\tilde m_{n+1})}{2}
   =b_n+2\tilde m_{n+1}-1=\bar b_n-1,
\end{equation}
and where we have defined
\begin{align}
  \mathsf{a}_n^{(L)} 
  &=\frac{e^{-(x_{n+1}-M-\tilde m_{n+1} )\eta}-e^{-\tilde{y}_{n+1}\eta}}
              {e^{\eta/2}\, \sinh(\eta\frac{\tilde y_{n+1}-N-\tilde z_{n+1}-1-M-\tilde m_{n+1}}{2})}
  \nonumber\\
  &= \frac{e^{-(c_n-d_n-M-\tilde m_n)\eta}-e^{-(c_n+d_n+M+\tilde m_n)\eta}}{e^{\eta/2}\, \sinh(\eta(\bar b_n-1-M-\tilde m_n))}
  \nonumber\\
  &= \frac{2\sinh(\eta(d_n+M+\tilde m_n))}{e^{c_n\eta}\, e^{\eta/2}\, \sinh(\eta(\bar b_n-1-M-\tilde m_n))}
  \nonumber\\
  &= \frac{2\sinh(\eta(d_n+M+\tilde m_n))}{e^{c_n\eta}\, e^{\eta/2}\, \sinh(\eta(b_n-1-M+\tilde m_n))}.
\end{align}
Note that here
\begin{align}
  &\tilde m_{n+1} =\tilde m_n,\quad
  x_{n+1}=x_n-1,\quad
  \tilde z_{n+1}= \tilde z_n-1,\quad
  \tilde y_{n+1}= \tilde y_n+1.
\end{align}
Note also that, with our notations, 
\begin{alignat}{2}
   &s_{(n)}=s_{(n+1)}, \qquad  & &s_{(n)}^{\prime }=s_{(n+1)}^{\prime}+1,
   \\
   &\mathsf{B}^{(n)}_{\boldsymbol{\epsilon,\epsilon'}}=\mathsf{B}^{(n+1)}_{\boldsymbol{\epsilon,\epsilon'}}\cup\{ \text{\textsc{b}}_{s_{(n)}+s_{(n)}^{\prime }}^{( n) } \},
   \qquad
   & &
   \mathsf{A}^{(n)}_{\boldsymbol{\epsilon,\epsilon'}}=\mathsf{A}^{(n+1)}_{\boldsymbol{\epsilon,\epsilon'}}\cup\{M+m+1-n\}\setminus\{ \text{\textsc{b}}_{s_{(n)}+s_{(n)}^{\prime }}^{( n) } \},
\end{alignat}
with
\begin{align}
   &\text{\textsc{b}}_{l}^{(n)} =\text{\textsc{b}}_{l}^{(n+1)}\qquad\,\,
   \forall l\in \{1,\ldots ,s_{(n)}+s_{(n)}^{\prime }-1\}, 
   \\
   &\text{\textsc{b}}_{s_{(n)}+s_{(n)}^{\prime }}^{( n) } \in 
   \mathsf{A}^{(n+1)}_{\boldsymbol{\epsilon,\epsilon'}}\cup\{M+m+1-n\}
   =\{1,\ldots ,M+m+1-n\}\setminus \{\text{\textsc{b}}_{1}^{( n)},\ldots ,\text{\textsc{b}}_{s_{(n)}+s_{(n)}^{\prime }-1}^{( n)}\},
\end{align}
being $i_p^{(n)}=i_p^{(n+1)}$ ($1\le p\le s_{(n)}+s'_{(n)}-1$) and $i_{s_{(n)}+s_{(n)}^{\prime }}^{\left( n\right) }=n$.
Hence, gathering all contributions, we can rewrite the action of \eqref{mon-11} on \eqref{BulkBLS}
as
\begin{multline}
   \frac{e^{\eta(c_n-a_n+1)}}{\sinh(\eta d_n)}\,
   \frac{\mathsf{a}^\mathrm{(tot)}_n}{d(\xi _{n}^{(1)})}\, 
   \prod_{b=1}^{M}\frac{\sinh (\mu_{b}-\xi _{n}^{(1)})}{\sinh (\mu _{b}-\xi _{n}^{(0)})}\
  \sum_{\mathsf{B}_{\boldsymbol{\epsilon},\boldsymbol{\epsilon'}}^{(n)}}
   \mathcal{F}_{\mathsf{B}_{\boldsymbol{\epsilon},\boldsymbol{\epsilon'}}^{(n+1)}}(\{\mu_j\}_{j=1}^{M},\{\xi_j^{(1)}\}_{j=n+1}^{m} \mid \alpha,\beta,x)\
   \\
   \times
     d(\mu_{\text{\textsc{b}}_{s_{(n)}+s_{(n)}^{\prime }}^{( n) }} )\,
     \sinh(\xi_n^{(1)}-\mu_{\text{\textsc{b}}_{s_{(n)}+s_{(n)}^{\prime }}^{( n) }}\hspace{-2mm}-\eta(1-\bar{b}_n))\,
     \frac{  \prod_{ j\in\mathsf{A}^{(n+1)}_{\boldsymbol{\epsilon,\epsilon'}} }  \sinh(\mu_j-\mu_{\text{\textsc{b}}_{s_{(n)}+s_{(n)}^{\prime }}^{( n) }}\hspace{-2mm}-\eta)}
  {   \prod_{ j\in\mathsf{A}^{(n)}_{\boldsymbol{\epsilon,\epsilon'}} }  \sinh(\mu_j-\mu_{\text{\textsc{b}}_{s_{(n)}+s_{(n)}^{\prime }}^{( n) }})}
  \\
  \times
  \underline{B}_{M+\tilde m_n}   (\{\mu_j\}_{j\in\mathsf{A}_{\boldsymbol{\epsilon},\boldsymbol{\epsilon'}}^{(n)} } | x_n-1,\tilde{z}_{n})\, \ket{\eta,\tilde{y}_{n}},
\end{multline}
with
\begin{align}\label{coeff-atot}
  \mathsf{a}^\mathrm{(tot)}_n
  =  \frac{e^{\eta (c_n+a_n-1)}\, \mathsf{d}_{n}^{(R)}\,  \mathsf{a}_n^{(L)}}{4\sinh(\eta d_{n})}
  &=\frac{\sinh(\eta(d_n+M+\tilde m_n))\,\sinh(\eta(b_n-M-1))}
   {\sinh(\eta(d_n+M))\,\sinh(\eta(b_n-M-1+\tilde m_n))}\nonumber\\
  &=\frac{\sinh(\eta(d_{n-1}+M+\tilde m_n+1))\,\sinh(\eta(b_{n-1}-M))}
   {\sinh(\eta(d_{n-1}+M+1))\,\sinh(\eta(b_{n-1}-M+\tilde m_n))}.
\end{align}
We therefore obtain the following recursion relation:
\begin{multline}\label{rec-F-11}
  \mathcal{F}_{\mathsf{B}_{\boldsymbol{\epsilon},\boldsymbol{\epsilon'}}^{(n)}}(\{\mu_j\}_{j=1}^{M},\{\xi_j^{(1)}\}_{j=n}^{m} \mid \alpha,\beta,x)
    = 
    \frac{e^{\eta(c_n-a_n+1)}\,\mathsf{a}^\mathrm{(tot)}_n}{\sinh(\eta d_n)}\,
    \frac{d(\mu_{\text{\textsc{b}}_{s_{(n)}+s_{(n)}^{\prime }}^{( n) }} )}{d(\xi _{n}^{(1)})}\,
    \prod_{b=1}^{M}\frac{\sinh (\mu_{b}-\xi _{n}^{(1)})}{\sinh (\mu _{b}-\xi _{n}^{(0)})}\,
    \\
    \times
    \sinh(\xi_n^{(1)}-\mu_{\text{\textsc{b}}_{s_{(n)}+s_{(n)}^{\prime }}^{( n) }}\hspace{-2mm}-\eta(1-\bar{b}_n)) \,
    \frac{  \prod_{ j\in\mathsf{A}^{(n+1)}_{\boldsymbol{\epsilon,\epsilon'}} }  \sinh(\mu_j-\mu_{\text{\textsc{b}}_{s_{(n)}+s_{(n)}^{\prime }}^{( n) }}\hspace{-2mm}-\eta)}
  {   \prod_{ j\in\mathsf{A}^{(n)}_{\boldsymbol{\epsilon,\epsilon'}} }  \sinh(\mu_j-\mu_{\text{\textsc{b}}_{s_{(n)}+s_{(n)}^{\prime }}^{( n) }})}
  \\
  \times
  \mathcal{F}_{\mathsf{B}_{\boldsymbol{\epsilon},\boldsymbol{\epsilon'}}^{(n+1)}}(\{\mu_j\}_{j=1}^{M},\{\xi_j^{(1)}\}_{j=n+1}^{m} \mid \alpha,\beta,x),
\end{multline}
which, using that
\begin{align*}
  & \hspace{-1mm}\prod_{ j\in\mathsf{A}^{(n+1)}_{\boldsymbol{\epsilon,\epsilon'}} } \hspace{-1mm} 
     \sinh(\mu_j-\mu_{\text{\textsc{b}}_{s_{(n)}+s_{(n)}^{\prime }}^{( n) }}\hspace{-3mm}-\eta)
  =\frac{\prod\limits_{j=1}^M\sinh(\mu_j-\mu_{\text{\textsc{b}}_{s_{(n)}+s_{(n)}^{\prime }}^{( n) }}\hspace{-3mm}-\eta)\,
            \prod\limits_{j=n+1}^m \sinh(\xi_j^{(1)}-\mu_{\text{\textsc{b}}_{s_{(n)}+s_{(n)}^{\prime }}^{( n) }}\hspace{-3mm}-\eta)}
           {\prod\limits_{j=1}^{s_{(n)}+s'_{(n)}-1}\hspace{-2mm}\sinh( \mu_{\text{\textsc{b}}_j^{(n)}}-\mu_{\text{\textsc{b}}_{s_{(n)}+s_{(n)}^{\prime }}^{( n) }}\hspace{-3mm}-\eta)},
           \\
   &\prod_{ j\in\mathsf{A}^{(n)}_{\boldsymbol{\epsilon,\epsilon'}} }  
     \sinh(\mu_j-\mu_{\text{\textsc{b}}_{s_{(n)}+s_{(n)}^{\prime }}^{( n) }}\hspace{-1mm})
   =\frac{\prod\limits_{\substack{j=1 \\ j\neq \text{\textsc{b}}_{s_{(n)}+s'_{(n)}} }}^M \hspace{-5mm}
             \sinh(\mu_j-\mu_{\text{\textsc{b}}_{s_{(n)}+s_{(n)}^{\prime }}^{( n) }} \hspace{-1mm} )\,
            \prod\limits_{\substack{ j=n \\ j\neq M+m+1- \text{\textsc{b}}_{s_{(n)}+s'_{(n)}}^{(n)}  }}^m 
            \hspace{-11mm}
            \sinh(\xi_j^{(1)}-\mu_{\text{\textsc{b}}_{s_{(n)}+s_{(n)}^{\prime }}^{( n) }}\hspace{-1mm})}
           {\prod\limits_{j=1}^{s_{(n)}+s'_{(n)}-1}\hspace{-2mm}\sinh( \mu_{\text{\textsc{b}}_j^{(n)}}-\mu_{\text{\textsc{b}}_{s_{(n)}+s_{(n)}^{\prime }}^{( n) }}\hspace{-1mm})},
\end{align*}
can be rewritten as
\begin{multline}\label{rec-F-11bis}
  \mathcal{F}_{\mathsf{B}_{\boldsymbol{\epsilon},\boldsymbol{\epsilon'}}^{(n)}}(\{\mu_j\}_{j=1}^{M},\{\xi_j^{(1)}\}_{j=n}^{m} \mid \alpha,\beta,x)
    = 
    \frac{e^{\eta(c_n-a_n+1)}\,\mathsf{a}^\mathrm{(tot)}_n}{\sinh(\eta d_n)}\, 
    \frac{d(\mu_{\text{\textsc{b}}_{s_{(n)}+s_{(n)}^{\prime }}^{( n) }} )\, 
            \frac{\prod_{j=1}^M\sinh(\mu_j-\mu_{\text{\textsc{b}}_{s_{(n)}+s_{(n)}^{\prime }}^{( n) }}\hspace{-3mm}-\eta)}{\prod_{\substack{j=1 \\ j\neq \text{\textsc{b}}_{s_{(n)}+s'_{(n)}}^{(n)} }}^M \hspace{-4mm}
             \sinh(\mu_j-\mu_{\text{\textsc{b}}_{s_{(n)}+s_{(n)}^{\prime }}^{( n) }})}
             }
             {d(\xi _{n}^{(1)})\,
    \prod_{b=1}^{M}\frac{\sinh (\mu _{b}-\xi _{n}^{(0)})}{\sinh (\mu_{b}-\xi _{n}^{(1)})}}\,
    \\
    \times
    \sinh(\xi_n^{(1)}-\mu_{\text{\textsc{b}}_{s_{(n)}+s_{(n)}^{\prime }}^{( n) }}\hspace{-2mm}-\eta(1-\bar{b}_n)) \,
    \frac{ \prod\limits_{j=n+1}^m \sinh(\xi_j^{(1)}-\mu_{\text{\textsc{b}}_{s_{(n)}+s_{(n)}^{\prime }}^{( n) }}\hspace{-2mm}-\eta)}
  { \hspace{-2mm} \prod\limits_{\substack{ j=n \\ j\neq M+m+1- \text{\textsc{b}}_{s_{(n)}+s'_{(n)}}^{(n)}  }}^m 
            \hspace{-11mm}
            \sinh(\xi_j^{(1)}-\mu_{\text{\textsc{b}}_{s_{(n)}+s_{(n)}^{\prime }}^{( n) }}\hspace{-1mm})}
  \\
  \times
     \prod\limits_{j=1}^{s_{(n)}+s'_{(n)}-1}\hspace{-1mm}
    \frac{\sinh(\mu_{\text{\textsc{b}}_j^{(n)}}- \mu_{\text{\textsc{b}}_{s_{(n)}+s_{(n)}^{\prime }}^{( n) }}\hspace{-1mm})}{\sinh(\mu_{\text{\textsc{b}}_j^{(n)}}- \mu_{\text{\textsc{b}}_{s_{(n)}+s_{(n)}^{\prime }}^{( n) }}\hspace{-3mm}-\eta)}\ \
  \mathcal{F}_{\mathsf{B}_{\boldsymbol{\epsilon},\boldsymbol{\epsilon'}}^{(n+1)}}(\{\mu_j\}_{j=1}^{M},\{\xi_j^{(1)}\}_{j=n+1}^{m} \mid \alpha,\beta,x),
\end{multline}
and gives the result once we notice that
\begin{equation}
   f^{(n,m)}_{\boldsymbol{\epsilon},\boldsymbol{\epsilon'}}(\alpha,\beta,x)
   = \mathsf{a}^\mathrm{(tot)}_n\,
   f^{(n+1,m)}_{\boldsymbol{\epsilon},\boldsymbol{\epsilon'}}(\alpha,\beta,x).
\end{equation}

B) Let us now consider the case $(\epsilon _{n},\epsilon _{n}^{\prime})=(1,2)$, which means that we have to act on \eqref{BulkBLS} with the monomial
\begin{multline}   \label{mon-12}
    \frac{e^{2\eta c_n}}{4\sinh^{2}(\eta d_{n})}\, 
    B(\xi_n^{(1)} | c_n-d_n, \bar{a}_n-\bar{b}_n)\ 
    \underline{M}^{(n+1,m)}_{\boldsymbol{\epsilon},\boldsymbol{\epsilon'}}(\alpha,\beta,x)\
    \frac{D(\xi_n^{(0)} | c_n+d_n-1, a_n-b_n-1)}{\det_q M(\xi_n)}
    \\
 =    \frac{e^{2\eta c_n}}{4\sinh^{2}(\eta d_{n})}\,
    B(\xi_n^{(1)} |x_{n+1},\tilde{z}_{n+1}-1)\
    \underline{M}^{(n+1,m)}_{\boldsymbol{\epsilon},\boldsymbol{\epsilon'}}(\alpha,\beta,x)\ 
    \frac{D(\xi_n^{(0)} | y_n-M,z_n-1)}{a(\xi_n^{(0)})\, d(\xi_n^{(1)})}    ,
\end{multline}
in which we have used \eqref{id:a-b}, \eqref{id:c+d}, \eqref {id:c-d} and
\begin{equation}\label{id:abar-bbar}
   \bar a_n-\bar b_n=z_{n+1}-1-2\tilde m_{n+1} 
   =\tilde z_{n+1}-1.
\end{equation}
As in the previous case, only the direct term \eqref{act-mon11-D} from the action of the $D$-operator contributes to the final result, and the action of $\underline{M}^{(n+1,m)}_{\boldsymbol{\epsilon},\boldsymbol{\epsilon'}}(\alpha,\beta,x)$ on this contribution gives \eqref{D+BloA}.
Note that here
\begin{align}
  &\tilde m_{n+1} =\tilde m_n-1,\quad
  x_{n+1}=x_n-1,\quad
  \tilde z_{n+1}= \tilde z_n-1,\quad
  \tilde y_{n+1}= \tilde y_n.
\end{align}
Note also that, with our notations, 
\begin{alignat}{2}
   &s_{(n)}=s_{(n+1)}, \qquad & &s_{(n)}^{\prime }=s_{(n+1)}^{\prime},
   \\
   &\mathsf{B}^{(n)}_{\boldsymbol{\epsilon,\epsilon'}}=\mathsf{B}^{(n+1)}_{\boldsymbol{\epsilon,\epsilon'}},\qquad
  & &\mathsf{A}^{(n)}_{\boldsymbol{\epsilon,\epsilon'}}=\mathsf{A}^{(n+1)}_{\boldsymbol{\epsilon,\epsilon'}}\cup\{M+m+1-n\},
\end{alignat}
with $i_p^{(n)}=i_p^{(n+1)}$, $1\le p\le s_{(n)}+s'_{(n)}$.
Gathering all contributions, we can therefore rewrite the action of \eqref{mon-12} on \eqref{BulkBLS}
as
\begin{multline} 
   \frac{e^{2\eta c_n}}{4\sinh^{2}(\eta d_{n})}\,   \frac{\mathsf{d}_{n}^{(R)}}{d(\xi _{n}^{(1)})}\, 
   \prod_{b=1}^{M}\frac{\sinh (\mu_{b}-\xi _{n}^{(1)})}{\sinh (\mu _{b}-\xi _{n}^{(0)})}\
  \sum_{\mathsf{B}_{\boldsymbol{\epsilon},\boldsymbol{\epsilon'}}^{(n+1)}}
   \mathcal{F}_{\mathsf{B}_{\boldsymbol{\epsilon},\boldsymbol{\epsilon'}}^{(n+1)}}(\{\mu_j\}_{j=1}^{M},\{\xi_j^{(1)}\}_{j=n+1}^{m} \mid \alpha,\beta,x)\
   \\
   \times
   B(\xi_n^{(1)} |x_{n}-1,\tilde{z}_{n}-2)\,
   \underline{B}_{M+\tilde m_n-1}   (\{\mu_j\}_{j\in\mathsf{A}_{\boldsymbol{\epsilon},\boldsymbol{\epsilon'}}^{(n+1)}} | x_{n}-2,\tilde{z}_{n}-1)\, \ket{\eta,\tilde{y}_{n}}.
\end{multline}
We therefore obtain the following recursion relation:
\begin{multline}
  \mathcal{F}_{\mathsf{B}_{\boldsymbol{\epsilon},\boldsymbol{\epsilon'}}^{(n)}}(\{\mu_j\}_{j=1}^{M},\{\xi_j^{(1)}\}_{j=n}^{m} \mid \alpha,\beta,x)
  = \frac{e^{\eta (c_n-a_n+1)}\, }{\sinh(\eta d_{n})}\, 
    \frac{\mathsf{b}_{n}^\mathrm{(tot)}}{d(\xi _{n}^{(1)})}\,\prod_{b=1}^{M}\frac{\sinh (\mu_{b}-\xi _{n}^{(1)})}{\sinh (\mu _{b}-\xi _{n}^{(0)})}\,
  \\
  \times
   \mathcal{F}_{\mathsf{B}_{\boldsymbol{\epsilon},\boldsymbol{\epsilon'}}^{(n+1)}}(\{\mu_j\}_{j=1}^{M},\{\xi_j^{(1)}\}_{j=n+1}^{m} \mid \alpha,\beta,x),
\end{multline}
with
\begin{align}
  \mathsf{b}_{n}^\mathrm{(tot)}
  =\frac{e^{\eta (c_n+a_n-1)}\, \mathsf{d}_{n}^{(R)}}{4\sinh(\eta d_{n})}
  &= \frac{e^{(c_n+1)\eta-\eta/2}\sinh(\eta(b_n-M-1))}{2\sinh(\eta(d_n+M))} \nonumber\\
  &= \frac{e^{\frac\eta 2(x+\alpha+\beta+N)-M\eta}\,\sinh(\eta(b_{n-1}-M))}{2\sinh(\eta(d_{n-1}+M+1))},
\end{align}
which gives the result once we notice that
\begin{equation}
   f^{(n,m)}_{\boldsymbol{\epsilon},\boldsymbol{\epsilon'}}(\alpha,\beta,x)
   = \mathsf{b}^\mathrm{(tot)}_n\,
   f^{(n+1,m)}_{\boldsymbol{\epsilon},\boldsymbol{\epsilon'}}(\alpha,\beta,x).
\end{equation}

C) Let us now consider the case $(\epsilon _{n},\epsilon _{n}^{\prime})=(2,2)$. We act on \eqref{BulkBLS} with the monomial: 
\begin{multline}   \label{mon-22}
    \frac{e^{2\eta c_n}}{4\sinh^{2}(\eta d_{n})}\, 
    D(\xi_n^{(1)} | c_n+d_n, \bar{a}_n-\bar{b}_n)\ 
    \underline{M}^{(n+1,m)}_{\boldsymbol{\epsilon},\boldsymbol{\epsilon'}}(\alpha,\beta,x)\
    \frac{A(\xi_n^{(0)} | c_n-d_n-1, a_n+b_n-1)}{\det_q M(\xi_n)}
    \\
 =    \frac{e^{2\eta c_n}}{4\sinh^{2}(\eta d_{n})}\,
    D(\xi_n^{(1)} |y_{n+1}-M,\tilde{z}_{n+1}-1)\
    \underline{M}^{(n+1,m)}_{\boldsymbol{\epsilon},\boldsymbol{\epsilon'}}(\alpha,\beta,x)\ 
    \frac{A(\xi_n^{(0)} | x_n,y_n-N+M)}{a(\xi_n^{(0)})\, d(\xi_n^{(1)})}    .
\end{multline}
From \eqref{Action-A-0}, and since $d(\xi_n^{(0)})=0$, the only non-zero-contributions from the action of the $A$-operator on the state \eqref{BulkBLS} come from indirect actions, so that it produces
\begin{multline}
   A(\xi_n^{(0)} | x_n,y_n-N+M)\, \underline{B}_M(\{\mu_i\}_{i=1}^M | x_n-1,z_n)\,\ket{\eta,y_n}
   =\mathsf{a}_n^{(R)} \sum_{\text{\textsc{b}}=1}^M d(\mu_\text{\textsc{b}})\,
    \frac{\prod_{k=1}^M \sinh (\mu_j-\mu_\text{\textsc{b}}-\eta)}{\prod_{\substack{j=1 \\ j\neq \text{\textsc{b}}}}^{M} \sinh (\mu_j-\mu_\text{\textsc{b}})}
   \\
   \times
   \frac{\sinh(\xi_n^{(0)}-\mu_\text{\textsc{b}}+\eta b_n)}{\sinh(\xi_n^{(0)}-\mu_\text{\textsc{b}})}
   \
  \underline{B}_M ( \{\mu_j\}_{\substack{ j=1 \\ j\neq {\text{\textsc{b}}}}}^{M}\cup\{\xi_n^{(0)}\}| x_n, z_n +1)\,
  \ket{\eta ,y_n-1} , 
\end{multline}
in which we have used that
\begin{equation}
  \frac{y_n-N-z_n-1+M}{2}=b_n,
\end{equation}
and defined
\begin{align}
    \mathsf{a}_n^{(R)} 
    &=  \frac{ e^{-(x_n-M)\eta }-e^{-y_n\eta }}{e^{\eta /2} \, \sinh(\eta\frac{y_n-z_n-N-M-1}{2})}
    \nonumber\\
    &= \frac{e^{-(c_n-d_n-(-1)^{\epsilon_n}-M)\eta}-e^{-(c_n+d_n+M+(-1)^{\epsilon_n})\eta}}{e^{\eta /2} \, \sinh(\eta(b_n-M))}
    \nonumber\\
    &=\frac{2\sinh(\eta(d_n+1+M))}{e^{\eta c_n+\eta/2}\, \sinh(\eta(b_n-M))}.
\end{align}
Noticing that
\begin{align}
  &x_n=x_{n+1}-1,\quad
  z_n=z_{n+1}-1,\quad
  y_n-1=y_{n+1},\quad
  \tilde{m}_{n+1}=\tilde{m}_n.
\end{align}
we can use the induction hypothesis so as to compute the action of $\underline{M}^{(n+1,m)}_{\boldsymbol{\epsilon},\boldsymbol{\epsilon'}}(\alpha,\beta,x)$ on the above state. 
We remark that, since $d(\xi_n^{(0)})=0$, the action of the monomial $\underline{M}^{(n+1,m)}_{\boldsymbol{\epsilon},\boldsymbol{\epsilon'}}(\alpha,\beta,x)$ cannot result into a replacement of the argument $\xi_n^{(0)}$ in the product of $B$-operators, so that we can make the following identification:
\begin{equation}
     \text{\textsc{b}}_{j+1}^{(n)}=\text{\textsc{b}}_{j}^{(n+1)}\quad
     \forall j\in \{1,\ldots,s_{(n)}+s_{(n)}^{\prime }-1\},
     \qquad\text{and}\qquad
     \text{\textsc{b}}_{1}^{(n)}=\text{\textsc{b}},
\end{equation}
with here 
\begin{alignat}{2}
   &s_{(n)}=s_{(n+1)}+1, \qquad & &s_{(n)}^{\prime }=s_{(n+1)}^{\prime },\\
   &\mathsf{B}_{\boldsymbol{\epsilon},\boldsymbol{\epsilon'}}^{(n)}
   =\mathsf{B}_{\boldsymbol{\epsilon},\boldsymbol{\epsilon'}}^{(n+1)}\cup\{ \text{\textsc{b}}_{1}^{(n)} \},
   \qquad & &
   \mathsf{A}_{\boldsymbol{\epsilon},\boldsymbol{\epsilon'}}^{(n)}
   =\mathsf{A}_{\boldsymbol{\epsilon},\boldsymbol{\epsilon'}}^{(n+1)}\cup\{M+m+1-n\}\setminus\{\text{\textsc{b}}_{1}^{(n)} \}.
\end{alignat}
Hence it gives
\begin{multline}\label{act-MA}
  \underline{M}^{(n+1,m)}_{\boldsymbol{\epsilon},\boldsymbol{\epsilon'}}(\alpha,\beta,x)\,
  A(\xi_n^{(0)} | x_n,y_n-N+M)\, \underline{B}_M(\{\mu_i\}_{i=1}^M | x_n-1,z_n)\,\ket{\eta,y_n}
  \\
   =\mathsf{a}_n^{(R)} 
   \sum_{\mathsf{B}_{\boldsymbol{\epsilon},\boldsymbol{\epsilon'}}^{(n)}}
    d(\mu_{\text{\textsc{b}}_1^{(n)}})\,
   \frac{\sinh(\xi_n^{(1)}-\mu_{\text{\textsc{b}}_1^{(n)}}+\eta (b_n+1))}{\sinh(\xi_n^{(1)}-\mu_{\text{\textsc{b}}_1^{(n)}}+\eta)}\,    
     \frac{\prod_{j=1}^M \sinh (\mu_j-\mu_{\text{\textsc{b}}_1^{(n)}}-\eta)}{ \prod_{\substack{j=1 \\ j\neq \text{\textsc{b}}_1^{(n)}}}^{M}\sinh (\mu_j-\mu_{\text{\textsc{b}}_1^{(n)}})} 
   \\
   \times
   \mathcal{F}_{\mathsf{B}_{\boldsymbol{\epsilon},\boldsymbol{\epsilon'}}^{(n+1)}}(\{\mu_j\}_{j=1}^{M}\setminus\{\mu_{\text{\textsc{b}}_1^{(n)}}\}\cup\{\xi_n^{(0)}\},\{\xi_j^{(1)}\}_{j=n+1}^{m} \mid \alpha,\beta,x)\
   \\
   \times
     \underline{B}_{M+\tilde m_{n+1}}   (\{\mu_j\}_{\substack{j=1 \\ j\notin\mathsf{B}_{\boldsymbol{\epsilon},\boldsymbol{\epsilon'}}^{(n)}}}^{M+m-n} 
     \cup\{\xi_n^{(0)}\} | x_{n+1}-1,\tilde z_{n+1})\, \ket{\eta,\tilde y_{n+1}}.
\end{multline}
Note that we have
\begin{align}
   &\{\mu_j\}_{\substack{j=1 \\ j\notin\mathsf{B}_{\boldsymbol{\epsilon},\boldsymbol{\epsilon'}}^{(n)}}}^{M+m-n} 
     \cup\{\xi_n^{(0)}\}
     =\{\mu_j\}_{j\in\mathsf{A}_{\boldsymbol{\epsilon},\boldsymbol{\epsilon'}}^{(n)} }\setminus\{\xi_n^{(1)}\}\cup\{\xi_n^{(0)}\}.
\end{align}
It remains to act with the last $D$-operator on this state. From \eqref{MM=0-1}, the only non-zero contribution of this action comes from the indirect term replacing $\xi_n^{(0)}$ by $\xi_n^{(1)}$ in the product of $B$:
\begin{multline}
   D(\xi_n^{(1)} |y_{n+1}-M,\tilde{z}_{n+1}-1)\
   \underline{B}_{M+\tilde m_{n+1}}   (\{\mu_j\}_{j\in\mathsf{A}_{\boldsymbol{\epsilon},\boldsymbol{\epsilon'}}^{(n)} }\cup\{\xi_n^{(0)}\}\setminus\{\xi_n^{(1)}\} | x_{n+1}-1,\tilde{z}_{n+1})\, \ket{\eta,\tilde {y}_{n+1}}
   \\
   = - \mathsf{d}_{n}^{(L) }\, a(\xi_n^{(0)})
   \hspace{-5mm}
   \prod_{\substack{j\in\mathsf{A}_{\boldsymbol{\epsilon},\boldsymbol{\epsilon'}}^{(n)} \\ j\not= M+m+1-n} }
   \hspace{-5mm}
   \frac{\sinh(\xi_n^{(0)}-\mu_j-\eta)}{\sinh(\xi_n^{(0)}-\mu_j)}\
   \underline{B}_{M+\tilde m_{n}}   (\{\mu_j\}_{j\in\mathsf{A}_{\boldsymbol{\epsilon},\boldsymbol{\epsilon'}}^{(n)} } | x_{n}-1,\tilde{z}_n)\, \ket{\eta,\tilde{y}_{n}},
\end{multline}
with
\begin{align}
  \mathsf{d}_{n}^{(L) }
  &= \frac{e^{-(\tilde z_{n+1}-1+M+\tilde m_n)\eta}-e^{-(\tilde y_{n+1}-N)\eta}}{e^{\eta/2}}\,
  \frac{\sinh(-\eta+\eta\frac{\tilde y_{n+1}-x_{n+1}-M-\tilde m_{n+1}+2}{2})}{\sinh(\eta\frac{\tilde y_{n+1}+M+\tilde m_{n+1}-x_{n+1}+2}{2})}
  \nonumber\\
  &=\frac{2\sinh(\eta(b_n-M+\tilde m_n))\,\sinh(\eta d_n)}{e^{\eta/2}\,e^{\eta\bar a_n}\,\sinh(\eta(d_n+M+\tilde m_n+1))}.
\end{align}
Hence, gathering all contributions, we obtain that the action of \eqref{mon-22} on \eqref{BulkBLS} results in
\begin{multline}
  -\frac{e^{2\eta c_n}\, \mathsf{a}_n^{(R)}\,\mathsf{d}_n^{(L)}}{4\sinh^{2}(\eta d_{n})\, d(\xi_n^{(1)})}\,
   \sum_{\mathsf{B}_{\boldsymbol{\epsilon},\boldsymbol{\epsilon'}}^{(n)}}
    d(\mu_{\text{\textsc{b}}_1^{(n)}})\,
      \frac{\sinh(\xi_n^{(1)}-\mu_{\text{\textsc{b}}_1^{(n)}}+\eta (b_n+1))}{\sinh(\xi_n^{(1)}-\mu_{\text{\textsc{b}}_1^{(n)}}+\eta)}\,    
     \frac{\prod_{j=1}^M \sinh (\mu_j-\mu_{\text{\textsc{b}}_1^{(n)}}-\eta)}{ \prod_{\substack{j=1 \\ j\neq \text{\textsc{b}}_1^{(n)}}}^{M}\sinh (\mu_j-\mu_{\text{\textsc{b}}_1^{(n)}})} 
   \\
   \times
   \mathcal{F}_{\mathsf{B}_{\boldsymbol{\epsilon},\boldsymbol{\epsilon'}}^{(n+1)}}(\{\mu_j\}_{j=1}^{M}\setminus\{\mu_{\text{\textsc{b}}_1^{(n)}}\}\cup\{\xi_n^{(0)}\},\{\xi_j^{(1)}\}_{j=n+1}^{m} \mid \alpha,\beta,x)\
   \\
   \times\hspace{-2mm}
    \prod_{\substack{j\in\mathsf{A}_{\boldsymbol{\epsilon},\boldsymbol{\epsilon'}}^{(n)} \\ j\not= M+m+1-n} }
   \hspace{-4mm}\frac{\sinh(\xi_n^{(1)}-\mu_j)}{\sinh(\xi_n^{(1)}-\mu_j+\eta)}\
   \underline{B}_{M+\tilde m_{n}}   (\{\mu_j\}_{j\in\mathsf{A}_{\boldsymbol{\epsilon},\boldsymbol{\epsilon'}}^{(n)} } | x_{n}-1,\tilde{z}_{n})\, \ket{\eta,\tilde{y}_{n}}.
\end{multline}
We therefore obtain the following recursion relation:
\begin{multline}\label{rec-F-22}
  \mathcal{F}_{\mathsf{B}_{\boldsymbol{\epsilon},\boldsymbol{\epsilon'}}^{(n)}}(\{\mu_j\}_{j=1}^{M},\{\xi_j^{(1)}\}_{j=n}^{m} \mid \alpha,\beta,x)
  = -\frac{e^{2\eta c_n}\, \mathsf{a}_n^{(R)}\,\mathsf{d}_n^{(L)}}{4\sinh^{2}(\eta d_{n})}\,
  \frac{d(\mu_{\text{\textsc{b}}_1^{(n)}})}{d(\xi _{n}^{(1)})}\,
      \frac{\sinh(\xi_n^{(1)}-\mu_{\text{\textsc{b}}_1^{(n)}}+\eta (b_n+1))}{\sinh(\xi_n^{(1)}-\mu_{\text{\textsc{b}}_1^{(n)}}+\eta)}\,    
      \\
      \times
     \frac{\prod_{j=1}^M \sinh (\mu_j-\mu_{\text{\textsc{b}}_1^{(n)}}-\eta)}{ \prod_{\substack{j=1 \\ j\neq \text{\textsc{b}}_1^{(n)}}}^{M}\sinh (\mu_j-\mu_{\text{\textsc{b}}_1^{(n)}})} 
   \prod_{\substack{j\in\mathsf{A}_{\boldsymbol{\epsilon},\boldsymbol{\epsilon'}}^{(n)} \\ j\not= M+m+1-n} }
   \hspace{-2mm}
   \frac{\sinh(\xi_n^{(1)}-\mu_j)}{\sinh(\xi_n^{(1)}-\mu_j+\eta)}
   \\
   \times
   \mathcal{F}_{\mathsf{B}_{\boldsymbol{\epsilon},\boldsymbol{\epsilon'}}^{(n+1)}}(\{\mu_j\}_{j=1}^{M}\setminus\{\mu_{\text{\textsc{b}}_1^{(n)}}\}\cup\{\xi_n^{(0)}\},\{\xi_j^{(1)}\}_{j=n+1}^{m} \mid \alpha,\beta,x) .
\end{multline}
Using moreover that
\begin{multline}\label{DevF}
  \mathcal{F}_{\mathsf{B}_{\boldsymbol{\epsilon},\boldsymbol{\epsilon'}}^{(n+1)}}(\{\mu_j\}_{j=1}^{M}\setminus\{\mu_{\text{\textsc{b}}_1^{(n)}}\}\cup\{\xi_n^{(0)}\},\{\xi_j^{(1)}\}_{j=n+1}^{m} \mid \alpha,\beta,x)
  \\
  =\hspace{-1mm}\prod_{j=2}^{s_{(n)}+s'_{(n)}}\frac{\sinh(\xi_n^{(0)}-\mu_{\text{\textsc{b}}_j^{(n)}}-\eta)\, \sinh(\mu_{\text{\textsc{b}}_1^{(n)}}-\mu_{\text{\textsc{b}}_j^{(n)}})}{\sinh(\xi_n^{(0)}-\mu_{\text{\textsc{b}}_j^{(n)}})\, \sinh(\mu_{\text{\textsc{b}}_1^{(n)}}-\mu_{\text{\textsc{b}}_j^{(n)}}-\eta)}
  \prod_{j=n+1}^m\frac{\sinh(\xi_n^{(0)}-\xi_j^{(1)})\,\sinh(\mu_{\text{\textsc{b}}_1^{(n)}}-\xi_j^{(1)}-\eta)}{\sinh(\xi_n^{(0)}-\xi_j^{(1)}-\eta)\,\sinh(\mu_{\text{\textsc{b}}_1^{(n)}}-\xi_j^{(1)})}
  \\
  \times
   \mathcal{F}_{\mathsf{B}_{\boldsymbol{\epsilon},\boldsymbol{\epsilon'}}^{(n+1)}}(\{\mu_j\}_{j=1}^{M},\{\xi_j^{(1)}\}_{j=n+1}^{m} \mid \alpha,\beta,x),
\end{multline}
and that
\begin{multline}\label{Dev-D}
      \prod_{\substack{j\in\mathsf{A}_{\boldsymbol{\epsilon},\boldsymbol{\epsilon'}}^{(n)} \\ j\not= M+m+1-n} }
   \hspace{-2mm}
   \frac{\sinh(\xi_n^{(1)}-\mu_j)}{\sinh(\xi_n^{(1)}-\mu_j+\eta)}
   =\prod_{k=1}^M\frac{\sinh(\mu_k-\xi_n^{(1)})}{\sinh(\mu_k-\xi_n^{(1)}-\eta)}\,
   \prod_{j=n+1}^m \frac{\sinh(\xi_n^{(1)}-\xi_j^{(1)})}{\sinh(\xi_n^{(1)}-\xi_j^{(1)}+\eta)}
   \\
   \times
   \prod_{j=1}^{s_{(n)}+s_{(n)}'}\frac{\sinh(\mu_{\text{\textsc{b}}_j^{(n)}}-\xi_n^{(1)}-\eta)}{\sinh(\mu_{\text{\textsc{b}}_j^{(n)}}-\xi_n^{(1)})},
\end{multline}
the relation \eqref{rec-F-22} becomes
\begin{multline}\label{rec-F-22-bis}
  \mathcal{F}_{\mathsf{B}_{\boldsymbol{\epsilon},\boldsymbol{\epsilon'}}^{(n)}}(\{\mu_j\}_{j=1}^{M},\{\xi_j^{(1)}\}_{j=n}^{m} \mid \alpha,\beta,x)
  = 
 - \frac{e^{\eta (c_n-a_n+1)}\, }{\sinh(\eta d_{n})}\,\mathsf{d}_n^\mathrm{(tot)}\,
  \frac{d(\mu_{\text{\textsc{b}}_1^{(n)}})\,    
          \frac{\prod_{j=1}^M \sinh (\mu_j-\mu_{\text{\textsc{b}}_1^{(n)}}-\eta)}
                 { \prod_{\substack{j=1 \\ j\neq \text{\textsc{b}}_1^{(n)}}}^{M}\sinh (\mu_j-\mu_{\text{\textsc{b}}_1^{(n)}})} }
       {d(\xi _{n}^{(1)})\, \prod_{k=1}^M\frac{\sinh(\mu_k-\xi_n^{(1)}-\eta)}{\sinh(\mu_k-\xi_n^{(1)})}}\,
         \\
         \times
         \frac{\sinh(\mu_{\text{\textsc{b}}_1^{(n)}}-\xi_n^{(1)}-\eta (b_n+1))}
                {\sinh(\mu_{\text{\textsc{b}}_1^{(n)}}-\xi_n^{(1)})}\,
         \prod_{j=n+1}^m\hspace{-1mm}\frac{\sinh(\mu_{\text{\textsc{b}}_1^{(n)}}-\xi_j^{(1)}-\eta)}
                                          {\sinh(\mu_{\text{\textsc{b}}_1^{(n)}}-\xi_j^{(1)})}\
         \prod_{j=2}^{s_{(n)}+s'_{(n)}}\hspace{-1mm}
         \frac{\sinh(\mu_{\text{\textsc{b}}_1^{(n)}}-\mu_{\text{\textsc{b}}_j^{(n)}})}
                { \sinh(\mu_{\text{\textsc{b}}_1^{(n)}}-\mu_{\text{\textsc{b}}_j^{(n)}}-\eta)}
     \\
       \times
   \mathcal{F}_{\mathsf{B}_{\boldsymbol{\epsilon},\boldsymbol{\epsilon'}}^{(n+1)}}(\{\mu_j\}_{j=1}^{M},\{\xi_j^{(1)}\}_{j=n+1}^{m} \mid \alpha,\beta,x),
\end{multline}
in which
\begin{align}\label{d_n-tot}
  \mathsf{d}_n^\mathrm{(tot)}
  &=\frac{e^{\eta (c_n+a_n-1)}\, \mathsf{a}_n^{(R)}\,\mathsf{d}_n^{(L)}}{4\sinh(\eta d_{n})}
  =\frac{ \sinh(\eta(d_n+1+M))\, \sinh(\eta(b_n-M+\tilde m_n))}{\sinh(\eta(d_n+M+\tilde m_n+1))\, \sinh(\eta(b_n-M))}=[\mathsf{a}_{n+1}^\mathrm{(tot)}]^{-1}\nonumber\\
  &=\frac{\sinh(\eta(d_{n-1}+M))\, \sinh(\eta(b_{n-1}-M+\tilde m_n-1))}{\sinh(\eta(d_{n-1}+M+\tilde m_n))\, \sinh(\eta(b_{n-1}-M-1))}=[\mathsf{a}_{n-1}^\mathrm{(tot)}]^{-1},
\end{align}
which gives  the result once we notice that $i_{p+1}^{(n)}=i_{p}^{(n+1)}$ ($1\le p\le s_{(n)}+s_{(n)}'-1$) and $i_1^{(n)}=n$ and that
\begin{equation}
   f^{(n,m)}_{\boldsymbol{\epsilon},\boldsymbol{\epsilon'}}(\alpha,\beta,x)
   = \mathsf{d}^\mathrm{(tot)}_n\,
   f^{(n+1,m)}_{\boldsymbol{\epsilon},\boldsymbol{\epsilon'}}(\alpha,\beta,x).
\end{equation}

D) Let us now consider the case $(\epsilon _{n},\epsilon _{n}^{\prime})=(2,1)$. 
Then we have to act on the state \eqref{BulkBLS} with the monomial
\begin{multline}\label{mon-21}
       \frac{e^{2\eta c_n}}{4\sinh^{2}(\eta d_{n})}\, 
    C(\xi_n^{(1)} | c_n+d_n, \bar{a}_n+\bar{b}_n)\ 
    \underline{M}^{(n+1,m)}_{\boldsymbol{\epsilon},\boldsymbol{\epsilon'}}(\alpha,\beta,x)\
    \frac{A(\xi_n^{(0)} | c_n-d_n-1, a_n+b_n-1)}{\det_q M(\xi_n)}
    \\
 =    \frac{e^{2\eta c_n}}{4\sinh^{2}(\eta d_{n})}\,
    C(\xi_n^{(1)} |y_{n+1}-M,\tilde{y}_{n+1}-N+(M+\tilde m_{n+1}))\
    \underline{M}^{(n+1,m)}_{\boldsymbol{\epsilon},\boldsymbol{\epsilon'}}(\alpha,\beta,x)\ 
    \\
    \times
    \frac{A(\xi_n^{(0)} | x_n,y_n-N+M)}{a(\xi_n^{(0)})\, d(\xi_n^{(1)})}    .
\end{multline}
The action \eqref{act-MA} is similar as in the previous case, except that here we have to make the following identification:
\begin{equation}
     \text{\textsc{b}}_{j+1}^{(n)}=\text{\textsc{b}}_{j}^{(n+1)}\quad
     \forall j\in \{1,\ldots,s_{(n+1)}+s_{(n+1)}^{\prime }=s_{(n)}+s_{(n)}^{\prime }-2\},
     \qquad\text{and}\qquad
     \text{\textsc{b}}_{1}^{(n)}=\text{\textsc{b}},
\end{equation}
with here 
\begin{alignat}{2}
   &s_{(n)}=s_{(n+1)}+1, \qquad & &s_{(n)}^{\prime }=s_{(n+1)}^{\prime }+1,\\
   &\mathsf{B}_{\boldsymbol{\epsilon},\boldsymbol{\epsilon'}}^{(n)}
   =\mathsf{B}_{\boldsymbol{\epsilon},\boldsymbol{\epsilon'}}^{(n+1)}\cup\{ \text{\textsc{b}}_{1}^{(n)},\text{\textsc{b}}_{s_{(n)}+s'_{(n)}}^{(n)} \},
   \qquad & &
   \mathsf{A}_{\boldsymbol{\epsilon},\boldsymbol{\epsilon'}}^{(n)}
   =\mathsf{A}_{\boldsymbol{\epsilon},\boldsymbol{\epsilon'}}^{(n+1)}\cup\{M+m+1-n\}\setminus\{\text{\textsc{b}}_{1}^{(n)}, \text{\textsc{b}}_{s_{(n)}+s'_{(n)}}^{(n)}\}.
\end{alignat}
Setting also
\begin{equation}
   \hat{\mathsf{B}}_{\boldsymbol{\epsilon},\boldsymbol{\epsilon'}}^{(n)}
   =\mathsf{B}_{\boldsymbol{\epsilon},\boldsymbol{\epsilon'}}^{(n+1)}\cup\{ \text{\textsc{b}}_{1}^{(n)}\}
   =\{\text{\textsc{b}}_{j}^{( n)}\}_{j=1}^{s_{(n)}+s_{(n)}^{\prime }-1},
\end{equation}
we therefore can write
\begin{multline}\label{act-MA-21}
  \underline{M}^{(n+1,m)}_{\boldsymbol{\epsilon},\boldsymbol{\epsilon'}}(\alpha,\beta,x)\,
  A(\xi_n^{(0)} | x_n,y_n-N+M)\, \underline{B}_M(\{\mu_i\}_{i=1}^M | x_n-1,z_n)\,\ket{\eta,y_n}
  \\
   =\mathsf{a}_n^{(R)} \,
   \sum_{\hat{\mathsf{B}}_{\boldsymbol{\epsilon},\boldsymbol{\epsilon'}}^{(n)}}
    d(\mu_{\text{\textsc{b}}_1^{(n)}})\,
    \frac{\sinh(\xi_n^{(1)}-\mu_{\text{\textsc{b}}_1^{(n)}}+\eta (b_n+1))}{\sinh(\xi_n^{(1)}-\mu_{\text{\textsc{b}}_1^{(n)}}+\eta)}\,    
     \frac{\prod_{j=1}^M \sinh (\mu_j-\mu_{\text{\textsc{b}}_1^{(n)}}-\eta)}{ \prod_{\substack{j=1 \\ j\neq \text{\textsc{b}}_1^{(n)}}}^{M}\sinh (\mu_j-\mu_{\text{\textsc{b}}_1^{(n)}})} 
   \\
   \times
   \mathcal{F}_{\mathsf{B}_{\boldsymbol{\epsilon},\boldsymbol{\epsilon'}}^{(n+1)}}(\{\mu_j\}_{j=1}^{M}\setminus\{\mu_{\text{\textsc{b}}_1^{(n)}}\}\cup\{\xi_n^{(0)}\},\{\xi_j^{(1)}\}_{j=n+1}^{m} \mid \alpha,\beta,x)\
   \\
   \times
     \underline{B}_{M+\tilde m_{n+1}}   (\{\mu_j\}_{\substack{j=1 \\ j\notin\hat{\mathsf{B}}_{\boldsymbol{\epsilon},\boldsymbol{\epsilon'}}^{(n)}}}^{M+m-n} 
     \cup\{\xi_n^{(0)}\} | x_{n+1}-1,\tilde z_{n+1})\, \ket{\eta,\tilde y_{n+1}}.
\end{multline}
Using \eqref{Alternative-q-det}, which can be rewritten as
\begin{multline}
   C(\xi_n-\eta/2|y,w)\, B(\xi_n+\eta/2|x-1,z)
   \\
   =-\frac{e^{\eta\frac{x'+y'}{2}}\,\sinh(\eta \frac{y-x}{2})}{e^{\eta\frac{x+y}{2}}\sinh(\eta \frac{y'-x'}{2})}\,
   A(\xi_n-\eta/2|x',w)\, D(\xi_n+\eta/2|y'-1, z),
\end{multline}
we obtain that
\begin{align}
  &C(\xi_n^{(1)} |y_{n+1}-M,\tilde{y}_{n+1}-N+(M+\tilde m_{n+1}))\,
  \nonumber\\
  &\qquad\times
   \underline{B}_{M+\tilde m_{n+1}}   (\{\mu_j\}_{\substack{j=1 \\ j\notin\hat{\mathsf{B}}_{\boldsymbol{\epsilon},\boldsymbol{\epsilon'}}^{(n)}}}^{M+m-n} 
     \cup\{\xi_n^{(0)}\} | x_{n+1}-1,\tilde z_{n+1})\, \ket{\eta,\tilde y_{n+1}}
  \nonumber\\
  &=C(\xi_n^{(1)} |y_{n+1}-M,\tilde{y}_{n+1}-N+(M+\tilde m_{n+1}))\,
      B(\xi_n^{(0)}|x_{n+1}-1,\tilde z_{n+1})\,
      \nonumber\\
   &\qquad\times
    \underline{B}_{M+\tilde m_{n+1}-1}   (\{\mu_j\}_{\substack{j=1 \\ j\notin\hat{\mathsf{B}}_{\boldsymbol{\epsilon},\boldsymbol{\epsilon'}}^{(n)}}}^{M+m-n} 
     | x_{n+1}-2,\tilde z_{n+1}+1)\, \ket{\eta,\tilde y_{n+1}}
     \nonumber\\
   &=-\frac{\sinh(\eta \frac{y_{n+1}-M-x_{n+1}}{2})}{\sinh(\eta \frac{y_{n+1}-M-x_{n+1}+4}{2})}\,
   A(\xi_n^{(1)}|x_{n+1}-2,\tilde{y}_{n+1}-N+M+\tilde m_{n+1})\, 
   \nonumber\\
   &\qquad\times
   D(\xi_n^{(0)}|y_{n+1}-M+1, \tilde z_{n+1})\,
    \underline{B}_{M+\tilde m_{n+1}-1}   (\{\mu_j\}_{\substack{j=1 \\ j\notin\hat{\mathsf{B}}_{\boldsymbol{\epsilon},\boldsymbol{\epsilon'}}^{(n)}}}^{M+m-n} 
     | x_{n+1}-2, \tilde z_{n+1}+1)\, \ket{\eta,\tilde y_{n+1}}     
     \nonumber\\
   &=-\frac{\sinh(\eta d_n)}{\sinh(\eta (d_n+2))}\,
   A(\xi_n^{(1)}|x_n-1,\tilde{y}_n-N+M+\tilde m_n+1)\, 
   \nonumber\\
   &\qquad\times
   D(\xi_n^{(0)}|y_n-M, \tilde z_n-1)\
    \underline{B}_{M+\tilde m_n}   (\{\mu_j\}_{\substack{j=1 \\ j\notin\hat{\mathsf{B}}_{\boldsymbol{\epsilon},\boldsymbol{\epsilon'}}^{(n)}}}^{M+m-n} 
     | x_n-1, \tilde z_n)\, \ket{\eta,\tilde y_n}
     \nonumber\\
   &=-\frac{\sinh(\eta d_n)}{\sinh(\eta (d_n+2))}\,
   A(\xi_n^{(1)}|x_n-1,\tilde{y}_n-N+M+\tilde m_n+1)\, 
   \nonumber\\
   &\qquad\times
   \mathsf{\hat{d}}_{n}^{(L)}\, a(\xi _{n}^{(0)})
   \prod_{\substack{ j=1 \\ j\notin \hat{\mathsf{B}}_{\boldsymbol{\epsilon},\boldsymbol{\epsilon'}}^{(n)} }}^{M+m-n}
   \frac{\sinh (\mu_j-\xi _{n}^{(1)})}{\sinh (\mu_j-\xi _{n}^{(0) })}\
    \underline{B}_{M+\tilde m_n}   (\{\mu_j\}_{\substack{j=1 \\ j\notin\hat{\mathsf{B}}_{\boldsymbol{\epsilon},\boldsymbol{\epsilon'}}^{(n)}}}^{M+m-n} 
     | x_n-2, \tilde z_n-1)\, \ket{\eta,\tilde y_n+1},
\end{align}
in which we have used that, here, $\tilde m_{n+1}=\tilde m_n+1$ and 
\begin{equation}
   x_{n+1}=x_n+1,\quad
   y_{n+1}=y_n-1,\quad
   \tilde y_{n+1}=\tilde y_n,\quad
   z_{n+1}=z_n+1,\quad
   \tilde z_{n+1}=\tilde z_n-1,
\end{equation}
and that, thanks to \eqref{MM=0-1}, the only non-zero contribution comes from the direct action of $D$.
Here we have defined
\begin{align}
  \mathsf{\hat{d}}_{n}^{(L)}
  &= \frac{e^{-(\tilde z_n-1+M+\tilde m_n)\eta}-e^{-(\tilde y_n-N)\eta}}{e^{\eta/2}}\,
        \frac{\sinh(\eta\frac{\tilde y_n-x_n-M-\tilde m_n+2}{2})}{\sinh(\eta \frac{\tilde y_n-x_n+M+\tilde m_n+2}{2})}
        \nonumber\\
  &=\frac{2\sinh(\eta(b_n+1+\tilde m_n-M))}{e^{(a_n-2+1/2)\eta}}\,
  \frac{\sinh(\eta(d_n+2))}{\sinh(\eta(d_n+2+M+\tilde m_n))}.
\end{align}
Then, the action of the last $A$-operator computed in $\xi _{n}^{(1)}$ can be both direct and
indirect and it generates a further sum over the index
\begin{equation}
         \text{\textsc{b}}_{s_{(n)}+s'_{(n)}}^{(n)}\in 
         \{1,\ldots ,M+m+1-n\}\setminus \mathsf{\hat{B}}_{\boldsymbol{\epsilon},\boldsymbol{\epsilon'}}^{(n)},
\end{equation}
so that we get
\begin{align}
  &C(\xi_n^{(1)} |y_{n+1}-M,\tilde{y}_{n+1}-N+(M+\tilde m_{n+1}))\,
  \nonumber\\
  &\qquad\times
   \underline{B}_{M+\tilde m_{n+1}}   (\{\mu_j\}_{\substack{j=1 \\ j\notin\hat{\mathsf{B}}_{\boldsymbol{\epsilon},\boldsymbol{\epsilon'}}^{(n)}}}^{M+m-n} 
     \cup\{\xi_n^{(0)}\} | x_{n+1}-1,\tilde z_{n+1})\, \ket{\eta,\tilde y_{n+1}}
  \nonumber\\
   &=-\frac{\sinh(\eta d_n)}{\sinh(\eta (d_n+2))}\,
   \mathsf{\hat{d}}_{n}^{(L)}\, a(\xi _{n}^{(0)})
   \prod_{\substack{ j=1 \\ j\notin \hat{\mathsf{B}}_{\boldsymbol{\epsilon},\boldsymbol{\epsilon'}}^{(n)} }}^{M+m-n}
   \frac{\sinh (\mu_j-\xi _{n}^{(1)})}{\sinh (\mu_j-\xi _{n}^{(0) })}\
   \sum_{\text{\textsc{b}}_{s_{(n)}+s'_{(n)}}^{(n)}} \mathsf{\hat{a}}_{n}^{(L)}\,
   d(\mu_{\text{\textsc{b}}_{s_{(n)}+s'_{(n)}}^{(n)}})
   \nonumber\\
   &\qquad\times
   \sinh(\xi_n^{(1)}-\mu_{\text{\textsc{b}}_{s_{(n)}+s'_{(n)}}^{(n)}}+\eta(\bar b_n-1))\,
   \frac{\prod_{\substack{j=1 \\ j\notin \mathsf{\hat B}_{\boldsymbol{\epsilon},\boldsymbol{\epsilon'}}^{(n)} }}^{M+m-n} \sinh( \mu_j-\mu_{\text{\textsc{b}}_{s_{(n)}+s'_{(n)}}^{(n)}}-\eta)}{\prod_{\substack{j=1 \\ j\notin \mathsf{B}_{\boldsymbol{\epsilon},\boldsymbol{\epsilon'}}^{(n)} }}^{M+m-n+1} \sinh(\mu_j- \mu_{\text{\textsc{b}}_{s_{(n)}+s'_{(n)}}^{(n)}})}
  \nonumber\\
  &\qquad\times
   \underline{B}_{M+\tilde m_n}   (\{\mu_j\}_{\substack{j=1 \\ j\notin\mathsf{B}_{\boldsymbol{\epsilon},\boldsymbol{\epsilon'}}^{(n)}}}^{M+m-n+1}  | x_n-1,\tilde z_n)\, \ket{\eta,\tilde y_n},
\end{align}
with
\begin{align}
  \mathsf{\hat{a}}_{n}^{(L)}
  &=\frac{e^{-(x_n-1-M-\tilde m_n)\eta}-e^{-(\tilde y_n+1-N)\eta}}{e^{\eta/2}\,\sinh(\eta\frac{\tilde y_n-N-\tilde z_n +1-M-\tilde m_n}{2})}
  \nonumber\\
  &=\frac{2\sinh(\eta(d_n+2+M+\tilde m_n))}{e^{c_n\eta+\eta/2}\,\sinh(\eta(b_n+1-M+\tilde m_n))}.
\end{align}
Setting
\begin{align}
  \mathsf{c}_{n}^{(L)}
  &=\frac{\sinh(\eta d_n)}{\sinh(\eta (d_n+2))}\,
   \mathsf{\hat{d}}_{n}^{(L)}\,\mathsf{\hat{a}}_{n}^{(L)}
   \nonumber\\
   &=   \frac{4\sinh(b_n+1+\tilde m_n-M)}{e^{(a_n+c_n-1)\eta}}\,
  \frac{\sinh(\eta d_n)}{\sinh(\eta(\bar b_n+1-M-\tilde m_n))}
  \nonumber\\
  &= \frac{4\sinh(\eta d_n)}{e^{(a_n+c_n-1)\eta}},
\end{align}
and gathering all contributions, we obtain that the action of the monomial \eqref{mon-21} on the state \eqref{BulkBLS} results into
\begin{multline}
  -\frac{e^{2\eta c_n}\,  \mathsf{a}_n^{(R)} \,\mathsf{c}_{n}^{(L)}}{4\sinh^{2}(\eta d_{n})}\,
  \sum_{\mathsf{B}_{\boldsymbol{\epsilon},\boldsymbol{\epsilon'}}^{(n)}}
    \frac{d(\mu_{\text{\textsc{b}}_1^{(n)}})}{d(\xi_n^{(1)})}\,
    \frac{\sinh(\xi_n^{(1)}-\mu_{\text{\textsc{b}}_1^{(n)}}+\eta (b_n+1))}{\sinh(\xi_n^{(1)}-\mu_{\text{\textsc{b}}_1^{(n)}}+\eta)}\,    
     \frac{\prod_{j=1}^M \sinh (\mu_j-\mu_{\text{\textsc{b}}_1^{(n)}}-\eta)}{ \prod_{\substack{j=1 \\ j\neq \text{\textsc{b}}_1^{(n)}}}^{M}\sinh (\mu_j-\mu_{\text{\textsc{b}}_1^{(n)}})} 
   \\
   \times
   \mathcal{F}_{\mathsf{B}_{\boldsymbol{\epsilon},\boldsymbol{\epsilon'}}^{(n+1)}}(\{\mu_j\}_{j=1}^{M}\setminus\{\mu_{\text{\textsc{b}}_1^{(n)}}\}\cup\{\xi_n^{(0)}\},\{\xi_j^{(1)}\}_{j=n+1}^{m} \mid \alpha,\beta,x)\
   \\
   \times
   d(\mu_{\text{\textsc{b}}_{s_{(n)}+s'_{(n)}}^{(n)}}\!)\,
   \sinh(\xi_n^{(1)}-\mu_{\text{\textsc{b}}_{s_{(n)}+s'_{(n)}}^{(n)}}\hspace{-2mm}-\eta(1-\bar b_n))\,
   \frac{\prod_{\substack{j=1 \\ j\notin \mathsf{\hat B}_{\boldsymbol{\epsilon},\boldsymbol{\epsilon'}}^{(n)} }}^{M+m-n} 
   \sinh( \mu_j-\mu_{\text{\textsc{b}}_{s_{(n)}+s'_{(n)}}^{(n)}}\hspace{-2mm}-\eta)
   \frac{\sinh (\mu_j-\xi _{n}^{(1)})}{\sinh (\mu_j-\xi _{n}^{(0) })}
   }{\prod_{\substack{j=1 \\ j\notin \mathsf{B}_{\boldsymbol{\epsilon},\boldsymbol{\epsilon'}}^{(n)} }}^{M+m-n+1} \sinh( \mu_j-\mu_{\text{\textsc{b}}_{s_{(n)}+s'_{(n)}}^{(n)}}\hspace{-2mm})}
  \\
   \times
   \underline{B}_{M+\tilde m_n}   (\{\mu_j\}_{\substack{j=1 \\ j\notin\mathsf{B}_{\boldsymbol{\epsilon},\boldsymbol{\epsilon'}}^{(n)}}}^{M+m-n+1}  | x_n-1,\tilde z_n)\, \ket{\eta,\tilde y_n}.
\end{multline}
We therefore obtain the following recursion relation
\begin{multline}\label{rec-F-21}
  \mathcal{F}_{\mathsf{B}_{\boldsymbol{\epsilon},\boldsymbol{\epsilon'}}^{(n)}}(\{\mu_j\}_{j=1}^{M},\{\xi_j^{(1)}\}_{j=n}^{m} \mid \alpha,\beta,x)
  = -\frac{e^{2\eta c_n}\, \mathsf{a}_n^{(R)}\,\mathsf{c}_n^{(L)}}{4\sinh^{2}(\eta d_{n})}\,
  \\
  \times
  \frac{d(\mu_{\text{\textsc{b}}_1^{(n)}})}{d(\xi _{n}^{(1)})}\,
    \frac{\sinh(\xi_n^{(1)}-\mu_{\text{\textsc{b}}_1^{(n)}}+\eta (b_n+1))}{\sinh(\xi_n^{(1)}-\mu_{\text{\textsc{b}}_1^{(n)}}+\eta)}\,    
     \frac{\prod_{j=1}^M \sinh (\mu_j-\mu_{\text{\textsc{b}}_1^{(n)}}-\eta)}{ \prod_{\substack{j=1 \\ j\neq \text{\textsc{b}}_1^{(n)}}}^{M}\sinh (\mu_j-\mu_{\text{\textsc{b}}_1^{(n)}})} 
     \\
     \times
   d(\mu_{\text{\textsc{b}}_{s_{(n)}+s'_{(n)}}^{(n)}}\hspace{-1mm})\,
   \sinh(\xi_n^{(1)}-\mu_{\text{\textsc{b}}_{s_{(n)}+s'_{(n)}}^{(n)}}\hspace{-3mm}-\eta(1-\bar b_n))\,
   \frac{\prod\limits_{j=1}^{M+m-n}  \sinh( \mu_j-\mu_{\text{\textsc{b}}_{s_{(n)}+s'_{(n)}}^{(n)}}\hspace{-3mm}-\eta)\,\frac{\sinh (\mu_j-\xi _{n}^{(1)})}{\sinh (\mu_j-\xi _{n}^{(0) })}}
   {\prod\limits_{j=1}^{s_{(n)}+s'_{(n)}-1}\hspace{-3mm} \sinh( \mu_{\text{\textsc{b}}_j}-\mu_{\text{\textsc{b}}_{s_{(n)}+s'_{(n)}}^{(n)}}\hspace{-3mm}-\eta)\,\frac{\sinh (\mu_{\text{\textsc{b}}_j}-\xi _{n}^{(1)})}{\sinh (\mu_{\text{\textsc{b}}_j}-\xi _{n}^{(0) })}}
   \\
   \times
   \frac{\prod\limits_{j=1}^{s_{(n)}+s'_{(n)}-1}\hspace{-3mm} \sinh( \mu_{\text{\textsc{b}}_j}-\mu_{\text{\textsc{b}}_{s_{(n)}+s'_{(n)}}^{(n)}}\hspace{-1mm})}
   {\prod\limits_{\substack{j=1 \\ j\not= \text{\textsc{b}}_{s_{(n)}+s'_{(n)}}^{(n)} }}^{M+m-n+1} \hspace{-3mm}\sinh( \mu_j-\mu_{\text{\textsc{b}}_{s_{(n)}+s'_{(n)}}^{(n)}}\hspace{-1mm})}\
      \mathcal{F}_{\mathsf{B}_{\boldsymbol{\epsilon},\boldsymbol{\epsilon'}}^{(n+1)}}(\{\mu_j\}_{j=1}^{M}\setminus\{\mu_{\text{\textsc{b}}_1^{(n)}}\}\cup\{\xi_n^{(0)}\},\{\xi_j^{(1)}\}_{j=n+1}^{m} \mid \alpha,\beta,x) ,
\end{multline}
which can equivalently be rewritten, using \eqref{DevF}, as
\begin{multline}
  \mathcal{F}_{\mathsf{B}_{\boldsymbol{\epsilon},\boldsymbol{\epsilon'}}^{(n)}}(\{\mu_j\}_{j=1}^{M},\{\xi_j^{(1)}\}_{j=n}^{m} \mid \alpha,\beta,x)
  = -\frac{e^{\eta (c_n-a_n+1)}\, }{\sinh(\eta d_{n})}\,\mathsf{c}_n^\mathrm{(tot)}\,
  \\
  \times
  \frac{d(\mu_{\text{\textsc{b}}_1^{(n)}})\, 
   \frac{\prod_{j=1}^M \sinh (\mu_j-\mu_{\text{\textsc{b}}_1^{(n)}}-\eta)}{ \prod_{\substack{j=1 \\ j\neq \text{\textsc{b}}_1^{(n)}}}^{M}\sinh (\mu_j-\mu_{\text{\textsc{b}}_1^{(n)}})} \
   d(\mu_{\text{\textsc{b}}_{s_{(n)}+s'_{(n)}}^{(n)}})\,
      \frac{\prod_{j=1}^{M}  \sinh( \mu_j-\mu_{\text{\textsc{b}}_{s_{(n)}+s'_{(n)}}^{(n)}}\hspace{-2mm}-\eta)} {\prod_{\substack{j=1 \\ j\not= \text{\textsc{b}}_{s_{(n)}+s'_{(n)}}^{(n)} }}^M \sinh( \mu_j-\mu_{\text{\textsc{b}}_{s_{(n)}+s'_{(n)}}^{(n)}}\hspace{-1mm})}
     }{d(\xi _{n}^{(1)})\,\prod_{k=1}^M\frac{\sinh(\mu_k-\xi_n^{(0)})}{\sinh(\mu_k-\xi_n^{(1)})}}
  \\
  \times
    \prod_{j=2}^{s_{(n)}+s'_{(n)}}\frac{ \sinh(\mu_{\text{\textsc{b}}_1^{(n)}}-\mu_{\text{\textsc{b}}_j^{(n)}})}{\sinh(\mu_{\text{\textsc{b}}_1^{(n)}}-\mu_{\text{\textsc{b}}_j^{(n)}}-\eta)}\
    \prod\limits_{j=1}^{s_{(n)}+s'_{(n)}-1}\hspace{-2mm}
    \frac{ \sinh(\mu_{\text{\textsc{b}}_j}- \mu_{\text{\textsc{b}}_{s_{(n)}+s'_{(n)}}^{(n)}}\hspace{-1mm})}{\sinh( \mu_{\text{\textsc{b}}_j}-\mu_{\text{\textsc{b}}_{s_{(n)}+s'_{(n)}}^{(n)}}\hspace{-2mm}-\eta)}
  \\
  \times
   \frac{\sinh(\mu_{\text{\textsc{b}}_1^{(n)}}-\xi_n^{(1)}-\eta (b_n+1))}{\sinh(\mu_{\text{\textsc{b}}_1^{(n)}}-\xi_n^{(1)})}\,
     \prod_{j=n+1}^m\frac{\sinh(\mu_{\text{\textsc{b}}_1^{(n)}}-\xi_j^{(1)}-\eta)}{\sinh(\mu_{\text{\textsc{b}}_1^{(n)}}-\xi_j^{(1)})}
     \\
     \times
   \sinh(\xi_n^{(1)}-\mu_{\text{\textsc{b}}_{s_{(n)}+s'_{(n)}}^{(n)}}\hspace{-2mm}-\eta(1-\bar b_n))\,
   \frac{\prod\limits_{j=M+1}^{M+m-n}  \sinh( \mu_j-\mu_{\text{\textsc{b}}_{s_{(n)}+s'_{(n)}}^{(n)}}\hspace{-2mm}-\eta)}
     {\prod\limits_{\substack{j=M+1 \\ j\not= \text{\textsc{b}}_{s_{(n)}+s'_{(n)}}^{(n)} }}^{M+m-n+1} \sinh( \mu_j-\mu_{\text{\textsc{b}}_{s_{(n)}+s'_{(n)}}^{(n)}}\hspace{-1mm})}\
     \frac{\sinh(\mu_{\text{\textsc{b}}_{s_{(n)}+s'_{(n)}}^{(n)}}\hspace{-2mm}-\xi_n^{(1)})}{\sinh(\mu_{\text{\textsc{b}}_{s_{(n)}+s'_{(n)}}^{(n)}}\hspace{-2mm}-\xi_n^{(0)})}
  \\
  \times
   \mathcal{F}_{\mathsf{B}_{\boldsymbol{\epsilon},\boldsymbol{\epsilon'}}^{(n+1)}}(\{\mu_j\}_{j=1}^{M},\{\xi_j^{(1)}\}_{j=n+1}^{m} \mid \alpha,\beta,x),
\end{multline}
with
\begin{align}
   \mathsf{c}_n^\mathrm{(tot)}
  &=\frac{e^{\eta (c_n+a_n-1)}\, \mathsf{a}_n^{(R)}\,\mathsf{c}_n^{(L)}}{4\sinh(\eta d_{n})}
  = \frac{2\sinh(\eta(d_n+1+M))}{e^{\eta c_n+\eta/2}\, \sinh(\eta(b_n-M))}
    =[\mathsf{b}_{n+1}^\mathrm{(tot)}]^{-1}\nonumber\\
  &= \frac{2\sinh(\eta(d_{n-1}+M))}{e^{\frac\eta 2(x+\alpha+\beta+N)-M\eta}\, \sinh(\eta(b_{n-1}-M-1))}
  =[\mathsf{b}_{n-1}^\mathrm{(tot)}]^{-1}.
\end{align}
Noticing that $i_1^{(n)}=n$, $i_{p+1}^{(n)}=i_p^{(n+1)}$ ($1\le p\le s_{(n)}+s'_{(n)}-2$) and $i_{s_{(n)}+s'_{(n)}}^{(n)}=n$, and that
\begin{equation}
   f^{(n,m)}_{\boldsymbol{\epsilon},\boldsymbol{\epsilon'}}(\alpha,\beta,x)
   = \mathsf{c}^\mathrm{(tot)}_n\,
   f^{(n+1,m)}_{\boldsymbol{\epsilon},\boldsymbol{\epsilon'}}(\alpha,\beta,x),
\end{equation}
we get the result.
\end{proof}

The previous result implies the following one for the action of the generic
element of the basis \eqref{Local-Basis} on the gauged
Bethe-like bulk states \eqref{bulk-state}.

\begin{theorem}\label{th-act-local-bulk}
For any given $m$-tuples $\boldsymbol{\epsilon}\equiv(\epsilon_1,\ldots,\epsilon_m), \boldsymbol{\epsilon'}\equiv(\epsilon'_1,\ldots,\epsilon'_m)\in \{1,2\}^{m}$ 
and gauge parameters $\alpha,\beta,x$, the action of the generic element $\prod_{n=1}^{m}E_{n}^{\epsilon _{n}^{\prime },\epsilon _{n}}(\xi_{n}|(a_{n},b_{n}),(\bar{a}_{n},\bar{b}_{n}))$ of the basis \eqref{Local-Basis}, with
\begin{alignat}{2}
   &a_n =\alpha+1,\qquad & 
   &b_n=\beta-\sum_{r=1}^n (-1)^{\epsilon_r},   \\
   &\bar{a}_n =\alpha-1,  
   \qquad &
   &\bar{b}_n=\beta+\sum_{r=n+1}^{m}(-1)^{\epsilon'_r}-\sum_{r=1}^{m}(-1)^{\epsilon_r},
\end{alignat}
on the gauged Bethe-like bulk states \eqref{bulk-state} is
\begin{multline}\label{act-local-op-bulk}
  \prod_{n=1}^m E_n^{\epsilon'_n,\epsilon_n}(\xi_n | (a_n ,b_n),(\bar{a}_n,\bar{b}_n))\,
   \underline{B}_M(\{\mu_i\}_{i=1}^M | x-1,\alpha-\beta)\, \ket{\eta,\alpha +\beta +N-M-1}
  \\
  =\sum_{\mathsf{B}_{\boldsymbol{\epsilon,\epsilon'}}} 
    \mathcal{F}_{\mathsf{B}_{\boldsymbol{\epsilon,\epsilon'}}}(\{\mu_j\}_{j=1}^M,\{\xi_j^{(1)}\}_{j=1}^m\mid\alpha,\beta,x)\
    \\
    \times
  \underline{B}_{M+\tilde m_{\boldsymbol{\epsilon,\epsilon'}}}(\{\mu_j\}_{j\in\mathsf{A}_{\boldsymbol{\epsilon,\epsilon'}}} | x-1,\alpha-\beta-2\tilde m_{\boldsymbol{\epsilon,\epsilon'}})\, 
  \ket{\eta,\alpha +\beta +N-M-1+\tilde m_{\boldsymbol{\epsilon,\epsilon'}}}
\end{multline}
in which we have defined $\mu_{M+j}=\xi_{m+1-j}^{(1)}$ for $j\in\{1,\ldots,m\}$ and
\begin{equation}
   \tilde m_{\boldsymbol{\epsilon,\epsilon'}}=\sum_{r=1}^m(\epsilon'_r-\epsilon_r)=m-(s+s').
\end{equation}
In \eqref{act-local-op-bulk}, the sum runs over all possible sets of integers $\mathsf{B}_{\boldsymbol{\epsilon,\epsilon'}}=\{\text{\textsc{b}}_1,\ldots,\text{\textsc{b}}_{s+s'}\}$ whose elements satisfy the conditions
\begin{equation}\label{beta-cond}
  \begin{cases}
  \text{\textsc{b}}_{p}\in \{1,\ldots ,M\}\setminus \{\text{\textsc{b}}_1,\ldots ,\text{\textsc{b}}_{p-1}\}\qquad 
  & \text{for}\quad 0<p\leq s, \\ 
  \text{\textsc{b}}_{p}\in \{1,\ldots ,M+m+1-i_{p}\}\setminus \{\text{\textsc{b}}_{1},\ldots ,\text{\textsc{b}}_{p-1}\}\qquad  
  & \text{for}\quad s<p\leq s+s',
\end{cases}
\end{equation}
whereas
\begin{equation}\label{def-setA}
  \mathsf{A}_{\boldsymbol{\epsilon,\epsilon'}}\equiv \{ \text{\textsc{a}}_a,\ldots,\text{\textsc{a}}_{M+ \tilde m_{\boldsymbol{\epsilon,\epsilon'}}}\}=\{1,\ldots,M+m\}\setminus  \mathsf{B}_{\boldsymbol{\epsilon,\epsilon'}}.
\end{equation}
Finally,
\begin{align}\label{Def-F_B}
     &\mathcal{F}_{\mathsf{B}_{\boldsymbol{\epsilon},\boldsymbol{\epsilon'}}}(\{\mu_j\}_{j=1}^{M},\{\xi_j^{(1)}\}_{j=1}^{m} \mid \alpha,\beta,x)
     = \prod_{n=1}^m\frac{e^\eta}{\sinh(\eta b_n)}\,
     f^{(1,m)}_{\boldsymbol{\epsilon},\boldsymbol{\epsilon'}}(\alpha,\beta,x)\, 
     \nonumber\\
     &\qquad\times
     \frac{\prod\limits_{j=1}^{s+s'} \bigg[ d(\mu_{\text{\textsc{b}}_j}) \,           
             \frac{\prod_{k=1}^M\sinh(\mu_k-\mu_{\text{\textsc{b}}_j}-\eta)}
                    {\prod_{\substack{k=1 \\ k\neq \text{\textsc{b}}_j}}^M\sinh(\mu_k-\mu_{\text{\textsc{b}}_j})}\bigg] }
           {\prod\limits_{j=1}^m \bigg[ d(\xi_j^{(1)}) \,
            \prod\limits_{k=1}^M\frac{\sinh(\mu_k-\xi_j^{(1)}-\eta)}{\sinh(\mu_k-\xi_j^{(1)})}
            \bigg]}  
           \,
     \prod_{1\leq i<j\leq s+s'}
     \frac{\sinh (\mu _{\text{\textsc{b}}_i}-\mu_{\text{\textsc{b}}_j})}
            {\sinh (\mu _{\text{\textsc{b}}_i}-\mu_{\text{\textsc{b}}_j}-\eta )}
     \nonumber\\
     &\qquad\times
     \prod_{p=1}^{s}\vast[ 
     \sinh(\xi_{i_p}^{(1)}-\mu_{\text{\textsc{b}}_p}+\eta(1+b_{i_p}))\,
     \frac{\prod_{k=i_p+1}^m \sinh(\mu_{\text{\textsc{b}}_p}-\xi_k^{(1)}-\eta)}
            {\prod_{k=i_p}^m \sinh(\mu_{\text{\textsc{b}}_p}-\xi_k^{(1)})} 
     \vast]    
     \nonumber\\
     &\qquad\times
     \prod_{p=s+1}^{s+s'}\vast[ 
     \sinh(\xi_{i_p}^{(1)}-\mu_{\text{\textsc{b}}_p}-\eta(1-\bar b_{i_p}))\,
     \frac{\prod_{k=i_p+1}^m \sinh(\xi_k^{(1)}-\mu_{\text{\textsc{b}}_p}-\eta)}
            {\prod_{\substack{k=i_p \\ k\not= M+m+1-{\text{\textsc{b}}_p}}}^m \sinh(\xi_k^{(1)}-\mu_{\text{\textsc{b}}_p})} 
     \vast]      ,     
\end{align}
where
\begin{equation}\label{Def-f}
  f^{(1,m)}_{\boldsymbol{\epsilon},\boldsymbol{\epsilon'}}(\alpha,\beta,x)
   = \begin{cases}
     \displaystyle \prod\limits_{k=1}^{\tilde m_{\boldsymbol{\epsilon,\epsilon'}}}\frac{e^{(\frac{x+\alpha+\beta+N}2-M)\eta}\,\sinh(\eta(\beta+k-M-1))}{2\sinh(\eta(\frac{-x+\alpha+\beta+N-1}2+k))} &\text{if}\quad \tilde m_{\boldsymbol{\epsilon,\epsilon'}}>0,\\
     1  &\text{if}\quad \tilde m_{\boldsymbol{\epsilon,\epsilon'}}=0,\\
    \displaystyle \prod\limits_{k=1}^{|\tilde m_{\boldsymbol{\epsilon,\epsilon'}}|}\frac{2\sinh(\eta(\frac{-x+\alpha+\beta+N-1}2-k+1))}{e^{(\frac{x+\alpha+\beta+N}2-M)\eta}\,\sinh(\eta(\beta-k-M))} &\text{if}\quad \tilde m_{\boldsymbol{\epsilon,\epsilon'}}<0.
   \end{cases}
\end{equation}
In all these expressions we have used the notations \eqref{def-s-s'-i_p}.
\end{theorem}

\begin{proof}
This is a direct consequence of the reconstruction \eqref{reconst-basis} and of Proposition~\ref{prop-rec-action}.
\end{proof}

Note that the expression of this action appears as the direct gauge generalization of (5.11)-(5.13) of \cite{KitKMNST07}, once we take into account the change $\eta$ into $-\eta$ and our slightly different definition of the inhomogeneity parameters with respect to \cite{KitKMNST07} which result here in a dependence of the formula into $\xi_n^{(1)}$, $1\le n\le m$, instead of $\xi_n$, $1\le n\le m$, in (5.13) of \cite{KitKMNST07}.

\subsection{Action of local operators on gauged boundary states}
\label{sec-act-boundary}

We now express, for the very specific boundary conditions that we consider in the framework of this paper, corresponding to the choice \eqref{Special-K+} of the boundary matrix, the action of the generic element of the basis of local operators \eqref{Local-Basis} on a generic gauged boundary Bethe-type state of the form
\begin{equation} \label{BoundarySS}
   \underline{\widehat{\mathcal B}}_{-,M}(\{\mu_i\}_{i=1}^M|\alpha-\beta+1)\, 
   \ket{\eta,\alpha+\beta+N-1-2M}.
\end{equation}
The result of this action is a direct consequence of Theorem~\ref{th-act-local-bulk} and of the boundary-bulk decomposition \eqref{Boundary-bulk-Bethe}:

\begin{theorem}
Let $\prod_{n=1}^{m}E_{n}^{\epsilon _{n}^{\prime },\epsilon _{n}}(\xi
_{n}|(a_{n},b_{n}),(\bar{a}_{n},\bar{b}_{n}))$ be the generic element of the
basis $\left( \ref{Local-Basis}\right) $ of local operators on the first m
sites of the chain, where we have defined
\begin{alignat}{2}
   &a_n=\alpha +1,\qquad & &b_n=\beta -\sum_{r=1}^n(-1)^{\epsilon_r},\\
   &\bar{a}_n =\alpha -1,\qquad & &\bar{b}_n=\beta+\sum_{r=n+1}^{m}(-1)^{\epsilon_r^{\prime }}-\sum_{r=1}^{m}(-1)^{\epsilon_r}.
\end{alignat}
Then, its action on the boundary separate states \eqref{BoundarySS} reads
\begin{multline}\label{act-boundary}
   \prod_{n=1}^m E_n^{\epsilon'_n,\epsilon_n}(\xi_n|(a_n,b_n),(\bar{a}_n,\bar{b}_n))\,
   \underline{\widehat{\mathcal B}}_{-,M}(\{\mu_i\}_{i=1}^M|\alpha-\beta+1)\, 
   \ket{\eta,\alpha+\beta+N-1-2M}
   \\
   = \sum_{\mathsf{B}_{\boldsymbol{\epsilon,\epsilon'}}} 
    \mathcal{\bar F}_{\mathsf{B}_{\boldsymbol{\epsilon,\epsilon'}}}(\{\mu_j\}_{j=1}^M,\{\xi_j^{(1)}\}_{j=1}^m | \beta)\
    \\
    \times
    \underline{\widehat{\mathcal B}}_{-,M+\tilde m_{\boldsymbol{\epsilon,\epsilon'}}}(\{\mu_i\}_{\substack{i=1\\ i\notin\mathsf{B}_{\boldsymbol{\epsilon,\epsilon'}}}}^{M+m}|\alpha-\beta+1-2\tilde m_{\boldsymbol{\epsilon,\epsilon'}})\, 
   \ket{\eta,\alpha+\beta+N-1-2M},
\end{multline}
where we have defined $\mu _{M+j}=\xi _{m+1-j}^{ (1) }$ for $j\in \{1,\ldots ,m\}$, and
\begin{equation}
   \tilde m_{\boldsymbol{\epsilon,\epsilon'}}=\sum_{r=1}^m(\epsilon'_r-\epsilon_r)=m-(s+s').
\end{equation}
The sum runs over all possible sets of integers $\mathsf{B}_{\boldsymbol{\epsilon,\epsilon'}}=\{\text{\textsc{b}}_1,\ldots,\text{\textsc{b}}_{s+s'}\}$ whose elements satisfy the conditions
\begin{equation} 
  \begin{cases}
  \text{\textsc{b}}_{p}\in \{1,\ldots ,M\}\setminus \{\text{\textsc{b}}_1,\ldots ,\text{\textsc{b}}_{p-1}\}\qquad 
  & \text{for}\quad 0<p\leq s, \\ 
  \text{\textsc{b}}_{p}\in \{1,\ldots ,M+m+1-i_{p}\}\setminus \{\text{\textsc{b}}_{1},\ldots ,\text{\textsc{b}}_{p-1}\}\qquad  
  & \text{for}\quad s<p\leq s+s',
\end{cases}
\end{equation}
and 
\begin{align}\label{act-bound-coeff}
&\mathcal{\bar F}_{\mathsf{B}_{\boldsymbol{\epsilon,\epsilon'}}}(\{\mu_j\}_{j=1}^M,\{\xi_j^{(1)}\}_{j=1}^m | \beta) 
=(-1)^{ (N+1)\tilde m_{\boldsymbol{\epsilon,\epsilon'}}} 
  e^{\eta \tilde m_{\boldsymbol{\epsilon,\epsilon'}}(\beta+ \tilde m_{\boldsymbol{\epsilon,\epsilon'}})}
  \prod_{n=1}^m\frac{e^\eta}{\sinh(\eta b_n)}\,
  \sum_{\sigma _{\alpha _{+}}=\pm }
  \frac{\prod_{j=1}^{s+s'}d(\mu_{\text{\textsc{b}}_{j}}^{\sigma })}{\prod_{j=1}^{m}d(\xi _{j}^{(1)})}\ 
  \notag \\
&\qquad \times 
  \frac{H_{\sigma _{\alpha_+}}(\{\mu _{\alpha_+}\})}{H_{1}(\{\xi _{\gamma_+}^{(1)}\})}\,
  \prod_{i\in \alpha _-}\prod_{\epsilon =\pm }\left\{
  \prod_{j\in \alpha_+}\frac{\sinh (\mu_j^\sigma +\epsilon \mu_i+\eta )}{\sinh (\mu_j^\sigma +\epsilon \mu_i)}
  \prod_{j\in\gamma_+}\frac{\sinh (\xi_j^{(1)}+\epsilon \mu_i)}{\sinh (\xi_j^{(0)}+\epsilon \mu_i)}\right\}  
  \notag \\
&\qquad \times
  \prod_{i\in \alpha_+}\left\{ \prod_{j\in \gamma_+}\frac{\sinh (\xi_j^{(1)}-\mu_i^{\sigma })}{\sinh(\xi_j^{(0) }-\mu_i^\sigma)}\ 
  \frac{\prod_{j\in\alpha_+}\sinh (\mu_j^\sigma -\mu_i^\sigma-\eta )}{\prod_{j\in \alpha_+\setminus\{i\}}\sinh (\mu_j^\sigma-\mu_i^\sigma)}\right\}  
  \prod_{1\leq i<j\leq s+s'}\frac{\sinh (\mu _{\text{\textsc{b}}_i}^\sigma-\mu_{\text{\textsc{b}}_j}^\sigma )}{\sinh (\mu _{\text{\textsc{b}}_i}^\sigma-\mu_{\text{\textsc{b}}_j}^\sigma -\eta )}  
   \notag \\
    &\qquad\times
     \prod_{p=1}^{s}\vast[ 
     \sinh(\xi_{i_p}^{(1)}-\mu_{\text{\textsc{b}}_p}^\sigma+\eta(1+b_{i_p}))\,
     \frac{\prod_{k=i_p+1}^m \sinh(\mu_{\text{\textsc{b}}_p}^\sigma-\xi_k^{(1)}-\eta)}
            {\prod_{k=i_p}^m \sinh(\mu_{\text{\textsc{b}}_p}^\sigma-\xi_k^{(1)})} 
     \vast]    
     \nonumber\\
     &\qquad\times
     \prod_{p=s+1}^{s+s'}\vast[ 
     \sinh(\xi_{i_p}^{(1)}-\mu_{\text{\textsc{b}}_p}^\sigma-\eta(1-\bar b_{i_p}))\,
     \frac{\prod_{k=i_p+1}^m \sinh(\xi_k^{(1)}-\mu_{\text{\textsc{b}}_p}^\sigma-\eta)}
            {\prod_{\substack{k=i_p \\ k\not= M+m+1-{\text{\textsc{b}}_p}}}^m \sinh(\xi_k^{(1)}-\mu_{\text{\textsc{b}}_p}^\sigma)} 
     \vast]   .   
\end{align}
Here, the sum is performed over all $\sigma_{j}\in \{+,-\}$ for $j\in \alpha _{+}$, we have defined $\mu _{i}^{\sigma}=\sigma _{i}\mu _{i}$ for $i\in \mathsf{B}_{\boldsymbol{\epsilon,\epsilon'}}$, with $\sigma_{i}=1$ if $i>M$, and 
\begin{alignat}{2}
& \alpha _{+}=\mathsf{B}_{\boldsymbol{\epsilon,\epsilon'}}\cap \{1,\ldots ,M\}, & 
& \alpha_{-}=\{1,\ldots ,M\}\setminus \alpha _{+}, \\
& \gamma _{-}=\{M+m+1-j\}_{j\in \mathsf{B}_{\boldsymbol{\epsilon,\epsilon'}}\cap \{N+1,\ldots,N+m\}},
\quad & & \gamma _{+}=\{1,\ldots ,m\}\setminus \gamma _{-}.
\end{alignat}
The function $H_{\sigma }^{\mathcal{B}_{+}}(\{\lambda \})$ is the
coefficient appearing in the boundary-bulk decomposition.
\end{theorem}

Once again, we can compare the expression of this action to the one obtained in the ungauged diagonal case (see Proposition~5.2 of \cite{KitKMNST07}): up to the change $\eta$ in $-\eta$, and considering our slightly different definition of the inhomogeneity parameters with respect to \cite{KitKMNST07}, \eqref{act-boundary}-\eqref{act-bound-coeff} appear as a direct generalization of (5.14)-(5.15) of \cite{KitKMNST07}.

\section{Correlation functions}
\label{sec-corr-fct}

Let us now derive the exact expressions for the correlation functions associated with
the local operators of the type
\begin{equation}\label{local-op}
\prod_{n=1}^{m}E_{n}^{\epsilon _{n}^{\prime },\epsilon _{n}}(\xi
_{n}|(a_{n},b_{n}),(\bar{a}_{n},\bar{b}_{n})),
\end{equation}
with $a_n,b_n,\bar a_n,\bar b_n$ given by \eqref{Gauge.Basis-1}-\eqref{Gauge.Basis-2} in terms of
\begin{align}
   &\eta\alpha=-\tau_--i\epsilon_{\varphi_-}\frac\pi 2, \qquad
   \eta\beta=\epsilon_{\varphi_-}(\varphi_--\psi_-)+i\frac\pi 2.\label{choice-alpha-beta}
\end{align}
For simplicity, we shall consider here only the local operators
satisfying the following constrain:
\begin{equation}
s+s^{\prime }=m\quad\text{ or equivalently }\quad  \tilde m_{\boldsymbol{\epsilon,\epsilon'}}=\sum_{r=1}^m(\epsilon'_r-\epsilon_r)=0.  \label{n_b=n_c}
\end{equation}
We recall that we use the notations \eqref{def-s-s'-i_p}.

\subsection{General strategy for the computation of matrix elements of local operators in the limit $\zeta_+=-\infty$}
\label{subsec-strategy-corr}

Let us here explain how to compute the mean value of some quasi-local operator $O_{1\to m}\in\End(\otimes_{n=1}^m\mathcal{H}_n)$,
\begin{equation}\label{mean-O}
 \moy{ O_{1\to m} }\equiv\frac{\bra{Q}\, O_{1\to m} \ket{Q} }{\moy{Q\,|\,Q}},
\end{equation}
in some eigenstate $\ket{Q}$ of the transfer matrix under the boundary conditions \eqref{Special-K+} which are out of the range of validity of Sklyanin's SoV framework. Here $\ket{Q}$ and $\bra{Q}$ denote the right and left eigenstates of $\mathcal{T}(\lambda)$ represented as the SoV states \eqref{eigenR} and \eqref{eigenL} (written in the new general SoV basis) for some $Q$ of the form \eqref{Q-form} with roots labelled by $\lambda_1,\ldots,\lambda_M$. The latter are admissible solutions of the Bethe equations issued from the homogeneous $TQ$-equation \eqref{hom-TQ}.

Since the SoV state $\ket{Q}$ and the boundary Bethe state $\underline{\widehat{\mathcal B}}_{-,M}(\{\lambda_i\}_{i=1}^M |\alpha-\beta+1)\,\ket{\eta,\alpha+\beta+N-2M-1}$ (with $\alpha$ and $\beta$ given in terms of the '-' boundary parameters by \eqref{Gauge-cond-A}-\eqref{Gauge-cond-B}) are eigenvectors of the transfer matrix with the same eigenvalue, they are collinear due to the simplicity of the transfer matrix spectrum. Hence we can write:
\begin{align}\label{mean-Q-1}
    \moy{ O_{1\to m} }
    &=\frac{\bra{Q}\, O_{1\to m} \, \underline{\widehat{\mathcal B}}_{-,M}(\{\lambda_i\}_{i=1}^M |\alpha-\beta+1)\,\ket{\eta,\alpha+\beta+N-2M-1} }{\bra{Q}\,\underline{\widehat{\mathcal B}}_{-,M}(\{\lambda_i\}_{i=1}^M |\alpha-\beta+1)\,\ket{\eta,\alpha+\beta+N-2M-1}}.
\end{align}
Expressing the quasi-local operator $O_{1\to m}$ on the basis \eqref{Local-Basis},
\begin{equation}
  O_{1\to m}=\sum_{\boldsymbol{\epsilon,\epsilon'}\in\{1,2\}} f_{O,\boldsymbol{\epsilon,\epsilon'}} \, \prod_{n=1}^m E_n^{\epsilon'_n,\epsilon_n}(\xi_n|(a_n,b_n),(\bar a_n,\bar b_n)),
\end{equation}
and using the action  \eqref{act-boundary} of these basis elements  on the boundary Bethe state $\underline{\widehat{\mathcal B}}_{-,M}(\{\lambda_i\}_{i=1}^M |\alpha-\beta+1)\,\ket{\eta,\alpha+\beta+N-2M-1}$, we can write
\begin{equation}
    \moy{ O_{1\to m} }
    =\sum_{\boldsymbol{\epsilon,\epsilon'}\in\{1,2\}} f_{O,\boldsymbol{\epsilon,\epsilon'}} \, 
    \sum_{\mathsf{B}_{\boldsymbol{\epsilon,\epsilon'}}} 
    \mathcal{\bar F}_{\mathsf{B}_{\boldsymbol{\epsilon,\epsilon'}}}(\{\lambda_j\}_{j=1}^M,\{\xi_j^{(1)}\}_{j=1}^m | \beta)\
    \mathrm{SP}_Q(\{\bar \lambda_i\}_{i\in \mathsf{A}_{\boldsymbol{\epsilon,\epsilon'}}}),
\end{equation}
%
%
%
where the set $\{\bar \lambda_i\}_{i\in \mathsf{A}_{\boldsymbol{\epsilon,\epsilon'}}}$ is a subset of the set $\{\lambda_i\}_{i=1}^M\cup\{\xi_n^{(1)}\}_{n=1}^m$ (see \eqref{def-setA}), 
%
%
and where $\mathrm{SP}_Q(\{\bar \lambda_i\}_{i\in \mathsf{A}_{\boldsymbol{\epsilon,\epsilon'}}})$ is the following ratio of scalar products:
\begin{multline}\label{ratio-SP}
  \mathrm{SP}_Q(\{\bar \lambda_i\}_{i\in \mathsf{A}_{\boldsymbol{\epsilon,\epsilon'}}})
  =\frac{\bra{Q}\, \underline{\widehat{\mathcal B}}_{-,M+\tilde m_{\boldsymbol{\epsilon,\epsilon'}}}(\{\bar \lambda_i\}_{i\in \mathsf{A}_{\boldsymbol{\epsilon,\epsilon'}}} |\alpha-\beta+1-2\tilde m_{\boldsymbol{\epsilon,\epsilon'}})\, 
   \ket{\eta,\alpha+\beta+N-1-2M}}{\bra{Q}\,\underline{\widehat{\mathcal B}}_{-,M}(\{\lambda_i\}_{i=1}^M |\alpha-\beta+1)\,\ket{\eta,\alpha+\beta+N-2M-1}}.
\end{multline}
The whole problem is now to compute the ratio of scalar products \eqref{ratio-SP} in a simple way.

To this aim, let us remark that the boundary conditions \eqref{Special-K+} can be reached as the $\varsigma_+\to -\infty$ limit of some continuous trajectories in the 3-dimensional space of the '+' boundary parameters $\varsigma_+, \kappa_+,\tau_+$ on which Sklyanin's SoV approach can be used. We can choose a particular trajectory on which the condition \eqref{+homo-cond} holds all along the trajectory for fixed $\alpha$ and $\beta$ given in terms of the (fixed) '-' boundary parameters by \eqref{Gauge-cond-A}-\eqref{Gauge-cond-B}. Along such a trajectory and for finite $\varsigma_+$, it follows from Proposition~\ref{Ref-States} that
\begin{equation}\label{prop-ref-states-traj}
   \ket{\eta,\alpha+\beta+N-1-2M}=\mathsf{c}^{(R,\varsigma_+)}_{M,\text{ref}}\,\ket{\Omega_{\alpha,\beta-2M+1}}^{(\varsigma_+)},
\end{equation}
Note that both the Sklyanin's reference state \eqref{ref-state-R} and the scalar factor in \eqref{prop-ref-states-traj} depend in a continuous way on the '+' boundary parameters, and therefore on the position on the trajectory, and we underline this dependance by a superscript $(\varsigma_+)$. In the same way, the boundary operators $\widehat{\mathcal B}_-$ depend continuously on the '+' boundary parameters, and we denote them along the trajectory by $\widehat{\mathcal B}_-^{(\varsigma_+)}$.

On the other hand, all along the trajectory and in the limit $\varsigma_+\to -\infty$, the transfer matrix is diagonalizable with simple spectrum, and its eigenvalues depend continuously on the '+' boundary parameters. Hence we have no crossing of levels along the trajectory and we can define in an unambiguous way the SoV eigenstate $\bra{Q^{(\varsigma_+)}}$ (still written in the new general SoV basis) that converges toward $\bra{Q}$. Once again, the state $\bra{Q^{(\varsigma_+)}}$ depends in a continuous way of the '+' boundary parameter along the trajectory\footnote{Note that, in the limit $\varsigma_+\to-\infty$, the degree of the polynomial $Q^{(\varsigma_+)}$ written as in \eqref{Q-form} may not be conserved. In other words, it may happen that some roots of $Q^{(\varsigma_+)}$ tend to infinity when $\varsigma_+\to -\infty$. Note also that $Q^{(\varsigma_+)}$ may no longer be solution of a homogeneous functional $TQ$-equation, but instead on the inhomogeneous one \eqref{inhom-TQ}.
However, the state \eqref{eigenL} itself is perfectly well defined and continuous all along the trajectory and in the limit $\varsigma_+\to-\infty$ by the condition
\begin{equation}
   \frac{Q^{(\varsigma_+)}(\xi_n^{(1)})}{Q^{(\varsigma_+)}(\xi_n^{(0)})}=\frac{\tau^{(\varsigma_+)}(\xi_n^{(0)})}{\mathbf{A}^{(\varsigma_+)}_{\boldsymbol{\varepsilon}}(\xi_n^{(0)})},
\end{equation}
being both the transfer matrix eigenvalue $\tau^{(\varsigma_+)}$ and the function $\mathbf{A}^{(\varsigma_+)}_{\boldsymbol{\varepsilon}}$ \eqref{DefFullA} continuous with respect to the '+' boundary parameters and well defined in the $\varsigma_+\to-\infty$ limit.
}, 
and we have from \eqref{prop-SoVeigen} that it is collinear to the corresponding eigenstate $_\mathrm{Sk}\bra{Q^{(\varsigma_+)}}$ defined in the Sklyanin's SoV framework as in \eqref{eigenL-Sk}.

Hence we can write
\begin{multline}\label{ratio-SP-2}
   \mathrm{SP}_Q(\{\bar \lambda_i\}_{i\in \mathsf{A}_{\boldsymbol{\epsilon,\epsilon'}}})
   \\
   =    \lim_{\varsigma_+\to -\infty}
    \frac{_\mathrm{Sk}\bra{Q^{(\varsigma_+)}}\, 
    \underline{\widehat{\mathcal B}}^{(\varsigma_+)}_{-,M+\tilde m_{\boldsymbol{\epsilon,\epsilon'}}}(\{\bar \lambda_i\}_{i\in \mathsf{A}_{\boldsymbol{\epsilon,\epsilon'}}} |\alpha-\beta+1-2\tilde m_{\boldsymbol{\epsilon,\epsilon'}})\, 
   \ket{\Omega_{\alpha,\beta-2M+1}}^{(\varsigma_+)}}
   {_\mathrm{Sk}\bra{Q^{(\varsigma_+)}}\,\underline{\widehat{\mathcal B}}^{(\varsigma_+)}_{-,M}(\{\lambda_i\}_{i=1}^M |\alpha-\beta+1)\,\ket{\Omega_{\alpha,\beta-2M+1}}^{(\varsigma_+)}}.
\end{multline}
Note here that we keep the roots $\lambda_i$ in the Bethe vectors entering the ratio of scalar products in the right hand side of \eqref{ratio-SP-2}  fixed along the trajectory, so that these Bethe vectors are {\em a priori} no longer eigenstates at finite $\varsigma_+$. 

The ratio of scalar products that appear in the right hand side of \eqref{ratio-SP-2} can now be computed in the framework of Sklyanin's SoV, at least for $\tilde m_{\boldsymbol{\epsilon,\epsilon'}}=0$ for which we can use the collinearity relations \eqref{separate-ABA-R}%
\footnote{In the case $\tilde m_{\boldsymbol{\epsilon,\epsilon'}}\not=0$, the action of the product  $\underline{\widehat{\mathcal B}}^{(\varsigma_+)}_{-,M+\tilde m_{\boldsymbol{\epsilon,\epsilon'}}}(\{\bar \lambda_i\}_{i\in \mathsf{A}_{\boldsymbol{\epsilon,\epsilon'}}} |\alpha-\beta+1-2\tilde m_{\boldsymbol{\epsilon,\epsilon'}})$ operators on the state $\ket{\Omega_{\alpha,\beta-2M+1}}^{(\varsigma_+)}$ produces a state written on a SoV basis with shifted gauge parameters, that needs to be reexpressed on the initial SoV basis with unshifted gauge parameters (we cannot directly use the orthogonality relation \eqref{orth-SoV-Sk}), and the study of the scalar products is {\em a priori} more complicated. We therefore do not consider this case in the framework of the present paper.}.
From now on, we therefore restrict our study to operators $O_{1\to m}$ in the sector $m_{\boldsymbol{\epsilon,\epsilon'}}=0$, i.e. which can be obtained as a linear combinations of quasi-local operators of the form \eqref{local-op} with condition  \eqref{n_b=n_c}.
Let us denote by $\bar Q_{\mathsf{A}_{\boldsymbol{\epsilon,\epsilon'}}}$ the polynomial of the form \eqref{Q-form} with $M+\tilde m_{\boldsymbol{\epsilon,\epsilon'}}=M$ roots labelled by $\bar \lambda_i$, $i\in  \mathsf{A}_{\boldsymbol{\epsilon,\epsilon'}}$. Recall that we denote by $Q$ the polynomial \eqref{Q-form} with roots labelled by $\lambda_1,\ldots,\lambda_M$. From \eqref{separate-ABA-R}, we can again transform \eqref{ratio-SP-2} as
\begin{equation}\label{ratio-SP-3}
   \mathrm{SP}_Q(\{\bar \lambda_i\}_{i\in \mathsf{A}_{\boldsymbol{\epsilon,\epsilon'}}})
   =    \lim_{\varsigma_+\to -\infty} 
    \frac{\mathsf{c}^{(R)}_{Q,\text{ABA}}}{\mathsf{c}^{(R)}_{\bar Q_{\mathsf{A}_{\boldsymbol{\epsilon,\epsilon'}}},\text{ABA}}}\
    \frac{_\mathrm{Sk}\moy{Q^{(\varsigma_+)}\, |\, \bar Q_{\mathsf{A}_{\boldsymbol{\epsilon,\epsilon'}}}}_\mathrm{Sk} }
   {_\mathrm{Sk}\moy{Q^{(\varsigma_+)}\,|\, Q}_\mathrm{Sk} },
\end{equation}
with
\begin{equation}
\frac{\mathsf{c}^{(R)}_{Q,\text{ABA}}}{\mathsf{c}^{(R)}_{\bar Q_{\mathsf{A}_{\boldsymbol{\epsilon,\epsilon'}}},\text{ABA}}}
=\prod_{j=1}^M\frac{a_{\boldsymbol{0}}(\bar \lambda_j)\, a_{\boldsymbol{0}}(-\bar\lambda_j)\,\sinh(2\bar\lambda_j-\eta)}{a_{\boldsymbol{0}}(\lambda_j)\, a_{\boldsymbol{0}}(-\lambda_j)\,\sinh(2\lambda_j-\eta)}.
\end{equation}
The ratio of scalar products of separate states in the right hand side of \eqref{ratio-SP-3} can, as usual, be expressed as a ratio of determinants:
\begin{equation}
      \frac{_\mathrm{Sk}\moy{Q^{(\varsigma_+)}\, |\, \bar Q_{\mathsf{A}_{\boldsymbol{\epsilon,\epsilon'}}}}_\mathrm{Sk} }
   {_\mathrm{Sk}\moy{Q^{(\varsigma_+)}\,|\, Q}_\mathrm{Sk} }
   =\frac{\det_{1\le i,j\le N}\left[\sum\limits_{h=0}^1 
             \Big(\! -\frac{\mathbf{a}^{(\varsigma_+)}_{\boldsymbol{\varepsilon}}(\xi_i+\frac\eta 2 )}{\mathbf{a}^{(\varsigma_+)}_{-\boldsymbol{\varepsilon}}(\xi_i+\frac\eta 2 )}
             \frac{(Q^{(\varsigma_+)} \bar Q_{\mathsf{A}_{\boldsymbol{\epsilon,\epsilon'}}})(\xi_i^{(1)})}{(Q^{(\varsigma_+)} \bar Q_{\mathsf{A}_{\boldsymbol{\epsilon,\epsilon'}}})(\xi_i^{(0)})}\Big)^{\!h}
             \Big(\frac{\cosh(2\xi_i^{(1-h)})}{2}\Big)^{\! j-1}\right]}
   {\det_{1\le i,j\le N}\left[\sum\limits_{h=0}^1 
             \Big(\! -\frac{\mathbf{a}^{(\varsigma_+)}_{\boldsymbol{\varepsilon}}(\xi_i+\frac\eta 2 )}{\mathbf{a}^{(\varsigma_+)}_{-\boldsymbol{\varepsilon}}(\xi_i+\frac\eta 2 )}
             \frac{(Q^{(\varsigma_+)} Q)(\xi_i^{(1)})}{(Q^{(\varsigma_+)} Q)(\xi_i^{(0)})}\Big)^{\!h}
             \Big(\frac{\cosh(2\xi_i^{(1-h)})}{2}\Big)^{\! j-1}\right]}.
\end{equation}
It is now easy to take the $\varsigma_+\to -\infty$ in the above expression, and we obtain
\begin{multline}\label{ratio-SP-4}
   \mathrm{SP}_Q(\{\bar \lambda_i\}_{i\in \mathsf{A}_{\boldsymbol{\epsilon,\epsilon'}}})
   =\prod_{j=1}^M
   \frac{d(\bar \lambda_j)\, d(-\bar\lambda_j)\,  \sinh(2\bar\lambda_j-\eta)}{d(\lambda_j)\, d(-\lambda_j)\,\sinh(2\lambda_j-\eta)}
   \\
   \times
  \frac{\det_{1\le i,j\le N}\left[\sum\limits_{h=0}^1 
             \Big(\! -\frac{\mathbf{a}_{\boldsymbol{\varepsilon}}(\xi_i+\frac\eta 2 )}{\mathbf{a}_{-\boldsymbol{\varepsilon}}(\xi_i+\frac\eta 2 )}
             \frac{(Q \bar Q_{\mathsf{A}_{\boldsymbol{\epsilon,\epsilon'}}})(\xi_i^{(1)})}{(Q \bar Q_{\mathsf{A}_{\boldsymbol{\epsilon,\epsilon'}}})(\xi_i^{(0)})}\Big)^{\!h}
             \Big(\frac{\cosh(2\xi_i^{(1-h)})}{2}\Big)^{\! j-1}\right]}
   {\det_{1\le i,j\le N}\left[\sum\limits_{h=0}^1 
             \Big(\! -\frac{\mathbf{a}_{\boldsymbol{\varepsilon}}(\xi_i+\frac\eta 2 )}{\mathbf{a}_{-\boldsymbol{\varepsilon}}(\xi_i+\frac\eta 2 )}
             \frac{(Q^2)(\xi_i^{(1)})}{(Q^2)(\xi_i^{(0)})}\Big)^{\!h}
             \Big(\frac{\cosh(2\xi_i^{(1-h)})}{2}\Big)^{\! j-1}\right]}.
\end{multline} 
Such a formula can be transformed similarly as in \cite{KitMNT18}, and we finally obtain
\begin{equation}\label{ratio-SP-5}
   \mathrm{SP}_Q(\{\bar \lambda_i\}_{i\in \mathsf{A}_{\boldsymbol{\epsilon,\epsilon'}}})
   =\prod_{j=1}^M
   \frac{ \sinh(2\lambda_j+\eta) }{\sinh(2\bar \lambda_j+\eta)}\,
   \frac{\widehat{V}(\lambda_1,\ldots,\lambda_M)}{\widehat{V}(\bar\lambda_1,\ldots,\bar\lambda_M)}\,
  \frac{\det_{1\le i,j\le N}\left[\mathcal{S}(\boldsymbol{\bar\lambda},\boldsymbol{\lambda})\right]}
   {\det_{1\le i,j\le N}\left[\mathcal{S}(\boldsymbol{\lambda},\boldsymbol{\lambda})\right]}.
\end{equation} 
in which $\mathcal{S}(\boldsymbol{\mu},\boldsymbol{\lambda})$, for $\boldsymbol{\mu}=(\mu_1,\ldots,\mu_M)$ and $\boldsymbol{\lambda}=(\lambda_1,\ldots,\lambda_M)$,  is the $M\times M$ matrix with elements
\begin{align}\label{mat-Slavnov}
  [\mathcal{S}(\boldsymbol{\mu},\boldsymbol{\lambda})]_{j,k}
  &=Q(\mu_j)\, \frac{\partial\tau_Q(\mu_j)}{\partial\lambda_k} \nonumber\\
  &=\mathbf{A}_{\boldsymbol{\varepsilon}}(\mu_j)\, Q(\mu_j-\eta)\, \big[ t(\mu_j+\lambda_k-\eta/2)-t(\mu_j-\lambda_k-\eta/2)\big]
  \nonumber\\
  &\hspace{2cm}
    -\mathbf{A}_{\boldsymbol{\varepsilon}}(-\mu_j)\, Q(\mu_j+\eta)\, \big[ t(\mu_j+\lambda_k+\eta/2)-t(\mu_j-\lambda_k+\eta/2)\big],
\end{align}
where $\tau_Q$ is the transfer matrix eigenvalue associated with the Bethe roots $\lambda_1,\ldots,\lambda_M$ labelling the roots of the polynomial $Q$, and
%
%
%
\begin{equation}\label{def-t}
   t(\lambda)=\frac{\sinh\eta}{\sinh(\lambda-\eta/2)\,\sinh(\lambda+\eta/2)}=\coth(\lambda-\eta/2)-\coth(\lambda+\eta/2).
\end{equation}
Note that \eqref{mat-Slavnov} is regular for $\mu_j=\lambda_j$ due to the fact that $\lambda_1,\ldots,\lambda_M$ satisfy the Bethe equations, and can be rewritten as
\begin{multline}
     [\mathcal{S}(\boldsymbol{\mu},\boldsymbol{\lambda})]_{j,k} \Big|_{\mu_j=\lambda_j}
     =-\mathbf{A}_{\boldsymbol{\varepsilon}}(-\lambda_j)\, Q(\lambda_j+\eta)\,
     \\
     \times
     \bigg\{2Ni\,\delta_{j,k}\,\Xi'_{\boldsymbol{\varepsilon},Q}(\lambda_j) 
     +2\pi i\big[ K(\lambda_j-\lambda_k)-K(\lambda_j+\lambda_k)\big]\bigg\},
\end{multline}
in which we have defined
\begin{align}
   &\Xi'_{\boldsymbol{\varepsilon},Q}(\mu)
   =\frac{i}{2N}\frac\partial{\partial\mu}\left(\log\frac{\mathbf{A}_{\boldsymbol{\varepsilon}}(-\mu)\, Q(\mu+\eta)}{\mathbf{A}_{\boldsymbol{\varepsilon}}(\mu)\, Q(\mu-\eta)}\right),
   \label{def-A_Q}\\
   &K(\lambda)=\frac{i\sinh(2\eta)}{2\pi\,\sinh(\lambda+\eta)\,\sinh(\lambda-\eta)}
   =\frac{i}{2\pi}\big[ t(\lambda+\eta/2)+t(\lambda-\eta/2)\big]. \label{def-K}
\end{align}
In particular, if $\{\lambda_j\}_{j=1}^M=\{\lambda_a\}_{a\in\alpha_-}\cup\{\lambda_b\}_{b\in\alpha_+}$ and $\{\bar\lambda_j\}_{j\in\mathbf{A}_{\boldsymbol{\epsilon,\epsilon'}}}=\{\lambda_a\}_{a\in\alpha_-}\cup\{\xi_{i_b}^{(1)}\}_{b\in\alpha_+}$, then
\begin{multline}\label{ratio-SP-part}
   \mathrm{SP}_Q(\{\bar \lambda_i\}_{i\in \mathsf{A}_{\boldsymbol{\epsilon,\epsilon'}}})
   =\prod_{b\in\alpha_+}
   \frac{  \sinh(2\lambda_b+\eta)\, A_{\boldsymbol{\varepsilon}}(-\xi_{i_b}^{(1)})\, Q(\xi_{i_b}^{(1)}+\eta)}{\sinh(2\xi_{i_b}^{(1)}+\eta)\,A_{\boldsymbol{\varepsilon}}(-\lambda_b)\, Q(\lambda_b+\eta)}
   \frac{}{}
   \\
   \times 
   \prod_{\substack{ a,b\in \alpha_+ \\ a<b}} \frac{\sinh^2\lambda_a-\sinh^2\lambda_b}{\sinh^2\xi_{i_a}^{(1)}-\sinh^2\xi_{i_b}^{(1)}}
   \prod_{\substack{ a\in \alpha_- \\ b\in\alpha_+}} \frac{\sinh^2\lambda_a-\sinh^2\lambda_b}{\sinh^2\lambda_a-\sinh^2\xi_{i_b}^{(1)}}\,
   \frac{\det_M \mathcal{M}(\boldsymbol{\bar \lambda,\lambda})}{\det_M \mathcal{N}(\boldsymbol{\lambda})},
\end{multline}
in which
\begin{align}
  &[\mathcal{N}(\boldsymbol{\lambda})]_{j,k}=2N\,\delta_{j,k}\,\Xi'_{\boldsymbol{\varepsilon},Q}(\lambda_j) 
     +2\pi \big[ K(\lambda_j-\lambda_k)-K(\lambda_j+\lambda_k)\big],
     \label{mat-N}\\
   &[\mathcal{M}(\boldsymbol{\bar \lambda,\lambda})]_{j,k}=\begin{cases}
        [\mathcal{N}(\boldsymbol{\lambda})]_{j,k} &\text{if } k\in\alpha_-,\\
        i[t(\xi_{i_k}-\lambda_j)-t(\xi_{i_k}+\lambda_j)] &\text{if } k\in\alpha_+.
        \end{cases}
        \label{mat-M}
\end{align}
Note that, in that case, the ratio of the determinants of $\mathcal{M}$ and $\mathcal{N}$ in \eqref{ratio-SP-part} reduces to the determinant of a matrix of size $|\alpha_+|$;
\begin{equation}\label{mat-reduced}
    \frac{\det_M \mathcal{M}(\boldsymbol{\bar \lambda,\lambda})}{\det_M \mathcal{N}(\boldsymbol{\lambda})}
    =\det_{a,b\in\alpha_+}\mathcal{R}_{a,b},\qquad
    \text{with}\quad \mathcal{R}_{a,b}=\sum_{k=1}^M \left[\mathcal{N}^{-1}\right]_{a,k}\mathcal{M}_{k,b}.
\end{equation}

\subsection{Expression of the correlation functions in the finite chain}
\label{subsec-corr-finite}

We can now give the exact expression of the mean value, in the eigenstate $\ket{Q}$ with Bethe roots $\lambda_1,\ldots,\lambda_M$, of the quasi-local operator \eqref{local-op} with condition \eqref{n_b=n_c} and under the special boundary conditions \eqref{Special-K+} corresponding to $\varsigma _{+}=-\infty$ on the site $N$. For $\ket{Q}$ being the ground state of the chain, this gives the zero-temperature correlation functions of all combinations of local operators in the sector \eqref{n_b=n_c}.

From the previous study, we therefore obtain the following multiple sum representation of these quantities: 

\begin{theorem}
Under the boundary conditions \eqref{Special-K+}, the
correlation functions of the quasi-local operators \eqref{local-op} satisfying the condition \eqref{n_b=n_c} can be written as
\begin{multline}\label{corr-finite}
   \moy{ \prod_{n=1}^m E_n^{\epsilon'_n,\epsilon_n}(\xi_n |(a_n,b_n),(\bar{a}_n,\bar{b}_n))} 
   \\
    =\sum_{\text{\textsc{b}}_{1}=1}^{M}\ldots \sum_{\text{\textsc{b}}_{s}=1}^{M}
      \sum_{\text{\textsc{b}}_{s+1}=1}^{M+m}\ldots \sum_{\text{\textsc{b}}_{m}=1}^{M+m}
      \frac{H_{\{\text{\textsc{b}}_{j}\}}(\{\lambda \}|\beta)}
      { \prod\limits_{1\leq l<k\leq m}\!\!\!\!\sinh (\xi_k-\xi_l)\prod\limits_{1\leq p\leq q\leq m}\!\!\!\!\sinh (\xi_p+\xi_q)},
\end{multline}
in which 
\begin{multline}\label{H-finite}
H_{\{\text{\textsc{b}}_j\}}(\{\lambda \}|\beta) 
=\prod_{n=1}^m\frac{e^\eta}{\sinh(\eta b_n)}
\sum\limits_{\sigma _{\text{\textsc{b}}_{j}}}\frac{(-1)^{s}
\prod\limits_{i=1}^{m}\sigma _{\text{\textsc{b}}_{i}}\prod\limits_{i=1}^{m}
\prod\limits_{j=1}^{m}\sinh (\lambda _{\text{\textsc{b}}_{i}}^{\sigma }+\xi _{j}+\eta/2)}
{\prod\limits_{1\leq i<j\leq m}\sinh (\lambda^\sigma_{\text{\textsc{b}}_i}-\lambda_{\text{\textsc{b}}_j}^\sigma-\eta )\sinh (\lambda_{\text{\textsc{b}}_{i}}^\sigma+\lambda_{\text{\textsc{b}}_{j}}^\sigma+\eta )}  
   \\
 \times\prod\limits_{p=1}^{s}\bigg\{
\sinh(\lambda _{\text{\textsc{b}}_p}^\sigma-\xi _{i_p}^{(1)}-\eta(1+b_{i_p}))
\prod\limits_{k=1}^{i_{p}-1}\sinh(\lambda _{\text{\textsc{b}}_{p}}^{\sigma }-\xi _{k}^{( 1)})\prod\limits_{k=i_{p}+1}^{m}\!\!\sinh (\lambda _{\text{\textsc{b}}_{p}}^{\sigma }-\xi _{k}^{\left( 0\right) })\bigg\} 
   \\
 \times \!\!\prod\limits_{p=s+1}^{m}\bigg\{
\sinh (\lambda _{\text{\textsc{b}}_{p}\,}^{\sigma }-\xi _{i_{p}}^{(1) }+\eta(1-\bar b _{i_{p}}))
\prod\limits_{k=1}^{i_{p}-1}\sinh (\lambda _{\text{\textsc{b}}_{p}}^\sigma -\xi _{k}^{(1)})
\prod\limits_{k=i_{p}+1}^{m}\!\!\sinh (\lambda _{\text{\textsc{b}}_{p}}^{\sigma }-\xi _{k}^{(1)}+\eta )\bigg\}
  \\
 \times\prod\limits_{k=1}^{m}
\frac{\sinh (\xi_k-\varsigma_+^{(D)})\sinh (\xi_k-\varsigma_-^{(D)})}{\sinh (\lambda _{\text{\textsc{b}}_k}^\sigma -\varsigma_+^{(D)}+\eta /2)\sinh
(\lambda _{\text{\textsc{b}}_k}^\sigma -\varsigma_-^{(D)}+\eta/2)} \
\det_{m}\Omega .
\end{multline}
where the parameters $\varsigma_\pm^{(D)}$ are given in terms of the boundary parameters at site 1 by \eqref{Id-boundary}.
%
%
Here the sum is performed over all $\sigma _{\text{\textsc{b}}_{j}}\in
\{+,-\} $ for \textsc{b}$_{j}\leq M$, and $\sigma _{\text{\textsc{b}}_{j}}=1$
for \textsc{b}$_{j}>M$, and the $m\times m$ matrix $\Omega $ reads 
\begin{alignat}{2}
& \Omega _{lk}=-\delta _{N+m+1-b_{l},k},\quad & & \text{for \textsc{b}}%
_{l}>M, \\
& \Omega _{lk}=\mathcal{R}_{\text{\textsc{b}}_{l},k},\quad & & \text{for 
\textsc{b}}_{l}\leq M,
\end{alignat}
in terms of the matrix $\mathcal{R}$ \eqref{mat-reduced}.
\end{theorem}


\subsection{Expression of the correlation functions in the half-infinite chain}

Let us now consider the thermodynamic limit $N\rightarrow \infty $ of these correlation functions, which can be obtained quite similarly as in \cite{KitKMNST07}. We recall here that the spectrum of the spin chain with boundary conditions \eqref{Special-K+} at site $N$ is isospectral to that of a chain with diagonal boundary conditions and boundary parameters $\varsigma_\pm^{(D)}$ \eqref{Id-boundary}, so that we can use the known description of the ground state in the diagonal case to derive the thermodynamic limit of the expression \eqref{corr-finite}.

The configuration for the Bethe roots for the ground state of the open spin chain with diagonal boundary conditions  in the thermodynamic limit $N\to\infty$ has been studied in \cite{SkoS95,KapS96}, see also \cite{GriDT19} for a more accurate description of the ground state in the massive regime.

In the thermodynamic limit $N\to\infty$, nearly all Bethe roots for the ground state of the spin chain condensate on an interval $(0,\Lambda)$ of the real axis (in the regime $|\Delta|<1$, for which $\Lambda=+\infty$) or of the imaginary axis (in the regime $\Delta>1$, for which $\Lambda=-i\pi/2$), with some density function $\rho(\lambda)$ solution of the following integral equation:
\begin{equation}\label{int-rho}
    \rho(\lambda)+\int_{0}^\Lambda \big[ K(\lambda-\mu)+K(\lambda+\mu)\big]\,\rho(\mu)\,d\mu=\frac{p'(\lambda)}\pi,
\end{equation}
which can be extended by parity on the whole interval $(-\Lambda,\Lambda)$ as
\begin{equation}\label{int-rho-ext}
    \rho(\lambda)+\int_{-\Lambda}^\Lambda K(\Lambda-\mu)\,\rho(\mu)\,d\mu=\frac{p'(\lambda)}\pi.
\end{equation}
Here $K$ is given by \eqref{def-K} and $p'$ is given in terms of $t$ \eqref{def-t} as
\begin{align}\label{def-p'}
   &
   p'(\lambda)=\frac{i\sinh\eta}{\sinh(\lambda+\eta/2)\,\sinh(\lambda-\eta/2)}=it(\lambda).
\end{align}
%
Explicitly, we have
\begin{equation}\label{rho}
   \rho(\lambda)
   =\left\lbrace\,
   \begin{array}{@{}l@{\quad}l@{\quad}l@{}}
   \displaystyle
   \frac{1}{\zeta\,\cosh(\pi\lambda/\zeta)} &\text{with} \ \ \zeta=i\eta>0  \  &\text{if}\ \ |\Delta|<1,
   \vspace{1mm}\\
   \displaystyle 
   \frac i\pi\frac{\vartheta'_1(0,q)}{\vartheta_2(0,q)}\frac{\vartheta_3(i\lambda,q)}{\vartheta_4(i\lambda,q)}  
                 &\text{with} \ \ q=e^\eta \ (\eta<0) \  &\text{if}\ \ \Delta>1.
   \end{array}\right.
\end{equation}
%
Furthermore, the set of Bethe roots for the ground state may contain some extra isolated complex roots, that we call boundary roots, and that may tend either to $\eta/2-\varsigma_+^{(D)}$ or to $\eta/2-\varsigma_-^{(D)}$ with exponentially small corrections in the large $N$ limit. We will denote such a boundary root by $\check\lambda_+$ in the former case and by $\check\lambda_-$ in the latter case. The presence of such a boundary root within the set of roots for the ground state depends on both boundary parameters $\varsigma_+^{(D)}$ and $\varsigma_-^{(D)}$ and on the parity of the number of sites $N$ of the chain (see \cite{GriDT19}). 
We refer to \cite{GriDT19} for a detailed study of the set of Bethe roots describing the ground state, and in particular of the cases in which it contains such a boundary root, in the regime $\Delta>1$.


In the expression of the correlation functions \eqref{corr-finite},  the sum over real Bethe roots\footnote{Here and in the following, the terms "real root" are used to designate a Bethe root $\lambda$ which belongs to the interval $(0,\Lambda)$, i.e. which is indeed real in the regime $|\Delta|<1$, but which is instead purely imaginary in the regime $\Delta>1$: in the latter case, we have to make an appropriate change of variable $\alpha=i\lambda$ to recover a real root.}, with indices $\text{\textsc b}_j$ running from 1 to $M$, become integrals over the density functions in the thermodynamic limit $N\rightarrow \infty $ according to the following rule:
\begin{equation}
    \frac{1}{N}\sum_{\substack{\text{\textsc b}_j=1 \\ \lambda_{\text{\textsc b}_j}\in(0,\Lambda)}}^{M}
    \sum_{\sigma _{\text{\textsc b}_{j}=\pm }}
    f(\lambda _{\text{\textsc b}_{j}}^{\sigma })
    \underset{N\rightarrow \infty }{\longrightarrow }
    \int\limits_{0}^{\Lambda }d\lambda _{j}\,
    \rho (\lambda_{j})\sum_{\sigma _{j}={\pm }} 
    f(\lambda _{j}^{\sigma})
    =\int\limits_{-\Lambda }^{\Lambda }d\lambda _{j}\,f(\lambda _{j})\,\rho(\lambda _{j}),
\end{equation}
while the sum over the \textsc{b}$_{j}>M$ can be written as contour
integrals thanks to the identity 
\begin{equation}
2i\pi \,\text{Res}\,\rho (\lambda -\xi )_{\,\vrule height13ptdepth1pt\>{\lambda =\xi +\eta/2}\!}=-2.
\end{equation}
The ratio of determinants can also be computed in the thermodynamic limit similarly as in \cite{KitKMNST07}. If $\lambda_j$ corresponds to a real root, we have
\begin{equation}\label{mat-N-thermo}
  [\mathcal{N}(\boldsymbol{\lambda})]_{j,k}=2\pi N\,\delta_{j,k}\,\Big[ \rho(\lambda_j)+O\Big(\frac{1}{N}\Big)\Big]
     +2\pi \big[ K(\lambda_j-\lambda_k)-K(\lambda_j+\lambda_k)\big],
\end{equation}
and, due to \eqref{int-rho-ext}, 
\begin{equation}
   \sum_{\substack{p=1 \\ \lambda_p\in(0,\Lambda)}}^M\mathcal{N}_{j,p}\,\frac{\rho(\lambda_p-\xi_k)-\rho(\lambda_p+\xi_k)}{2N\rho(\lambda_p)}
   \underset{N\to\infty}{\longrightarrow} i[t(\xi_k-\lambda_j)-t(\xi_k+\lambda_j)] .
\end{equation}
If instead $\lambda_j$ is a boundary root of the form $\lambda_j=\check\lambda_\sigma=\eta/2-\varsigma_\sigma^{(D)}+\check\epsilon_\sigma$ with $\check\epsilon_\sigma$ being an exponentially small correction in $N$, then we have\footnote{The factor 2 for $\varsigma_+^{(D)}=\varsigma_-^{(D)}$ is due to the fact that, in that case, the boundary root approaches a double pole of the Bethe equation, which corresponds to a doubling of the same term in \eqref{mat-N-thermo-BR}, see \cite{GriDT19}.}
\begin{equation}\label{mat-N-thermo-BR}
  [\mathcal{N}(\boldsymbol{\lambda})]_{j,k}=-\frac{i}{\check\epsilon_\sigma}\,
  \left[\delta_{j,k} \big(1+\delta_{\varsigma_+^{(D)},\varsigma_-^{(D)}}\big)+O(\check\epsilon_\sigma)\right].
\end{equation}
Hence, in the thermodynamic limit, the elements of the matrix \eqref{mat-reduced} are given by
\begin{align}\label{mat-R-thermo}
    \mathcal{R}_{a,b} \underset{N\to\infty}{\sim} \begin{cases}
    \displaystyle\frac{ i\pi\,\check\epsilon_\sigma\big[\rho(\lambda_a-\xi_{i_b})-\rho(\lambda_a+\xi_{i_b})\big]}{1+\delta_{\varsigma_+^{(D)},\varsigma_-^{(D)}}} \quad 
          &\text{if }\lambda_a=\check\lambda_\sigma,\vspace{1mm}\\
    \displaystyle\frac{\rho(\lambda_a-\xi_{i_b})-\rho(\lambda_a+\xi_{i_b})}{2N\rho(\lambda_a)} \quad
          &\text{if }\lambda_a\in(0,\Lambda).
    \end{cases}
\end{align}
Note that, if the boundary root $\check\lambda_\sigma$ belongs to the set of roots $\{\lambda_i\}_{i=1}^M\cap\{\lambda_{\text{\textsc b}_j}\}_{j=1}^m$, 
the corresponding row in $\mathcal{R}$ is proportional to $\check\epsilon_\sigma$ which is exponentially small in $N$. However, this exponentially small factor is compensated by the prefactor
\begin{equation}
   \frac{1}{\sinh(-\check\lambda_\sigma-\varsigma_\sigma^{(D)}+\eta/2)}\underset{N\to\infty}{\sim}-\frac 1{\check\epsilon_\sigma},
\end{equation}
so that the final contribution is of order 1. This contribution can be written as a contour integral around the point $\varsigma_\sigma^{(D)}-\eta/2$. 

Therefore, the following result holds:

\begin{theorem}
Under the boundary conditions \eqref{Special-K+}, the
correlation functions of the quasi-local operators \eqref{local-op} satisfying the condition \eqref{n_b=n_c} can be written in the thermodynamic limit as
\begin{multline}\label{result-thermo}
   \moy{ \prod_{n=1}^{m}E_{n}^{\epsilon _{n}^{\prime },\epsilon _{n}}(\xi_{n}|(a_{n},b_{n}),(\bar{a}_{n},\bar{b}_{n})) }
     =\prod_{n=1}^m\frac{e^\eta}{\sinh(\eta b_n)}\,
    \frac{(-1)^{s}}{\prod\limits_{j<i}\sinh (\xi _i-\xi_j)\prod\limits_{i\leq j}\sinh (\xi_i+\xi_j)}   \\
 \times \int_{\mathcal{C}}\prod_{j=1}^{s}d\lambda _{j}\ \int_{\mathcal{C}_{\boldsymbol{\xi}}}
\prod_{j=s+1}^{m}\!\!d\lambda _{j}\ 
H_{m}(\{\lambda _{j}\}_{j=1}^M;\{\xi _{k}\}_{k=1}^m)\ \det_{1\leq j,k\leq m}\big[\Phi(\lambda_j,\xi_k)\big],
\end{multline}
where we have denoted
\begin{equation}\label{mat-Phi}
\Phi(\lambda_j,\xi_k)=\frac{1}{2}\big[\rho (\lambda _{j}-\xi _{k})-\rho (\lambda_{j}+\xi _{k})\big],
\end{equation}
and
\begin{multline}\label{H-thermo}
H_{m}(\{\lambda _{j}\}_{j=1}^M;\{\xi _{k}\}_{k=1}^m)
 =\frac{\prod\limits_{j=1}^{m}\prod\limits_{k=1}^{m}\sinh (\lambda _{j}+\xi_k+\eta/2)}
   {\!\!\!\!\prod\limits_{1\leq i<j\leq m}\!\!\!\!\sinh (\lambda_i-\lambda_j-\eta )\,\sinh (\lambda_i+\lambda_j+\eta )}
    \\
 \times 
 \prod\limits_{p=1}^{s}\bigg\{
 \sinh (\lambda_p-\xi _{i_{p}}^{(1)}-\eta (1+b_{i_p}))
 \prod\limits_{k=1}^{i_{p}-1}\sin(\lambda _{p}-\xi _{k}^{(1)})
 \prod\limits_{k=i_{p}+1}^{m}\!\!\sinh (\lambda_{p}-\xi _{k}^{(1)}-\eta )\bigg\} 
 \\
 \times 
 \!\!\prod\limits_{p=s+1}^{m}\bigg\{
 \sinh (\lambda_p-\xi _{i_{p}}^{(1)}+\eta(1-\bar b_{i_p}))
 \prod\limits_{k=1}^{i_{p}-1}\sinh (\lambda _{p}-\xi _{k}^{(1)})
 \prod\limits_{k=i_{p}+1}^{m}\!\!\sinh (\lambda_{p}-\xi _{k}^{(1)}+\eta )\bigg\} 
   \\
   \times
   \prod\limits_{k=1}^{m}\frac{\sinh (\xi _{k}-\varsigma_{+}^{(D)})\,\sinh (\xi _{k}-\varsigma_{-}^{(D)})}{\sinh(\lambda _{k}-\varsigma_{+}^{(D)}+\eta /2)\, \sinh (\lambda _{k}-\varsigma_{-}^{(D)}+\eta /2)} ,
\end{multline}
in which we have used the identification \eqref{Id-boundary}.
The contours $\mathcal{C}$ is defined as
\begin{equation}\label{C-nude}
  \mathcal{C} =  [-\Lambda ,\Lambda ] 
\end{equation}
if the set of Bethe roots for the ground state does not contain any boundary root, and as
\begin{equation}\label{C-BR}
  \mathcal{C} =  [-\Lambda ,\Lambda ] \cup \Gamma^-(\varsigma_{\sigma}^{(D)}-\eta /2) 
\end{equation}
if the set of Bethe roots for the ground state contains the boundary root $\check\lambda_\sigma$.
The contour $\mathcal{C}_{\boldsymbol{\xi}}$ is defined as
\begin{equation}\label{C-xi}
   \mathcal{C}_{\boldsymbol{\xi}} =\mathcal{C}\cup \Gamma^- (\{\xi_k^{(1)}\}_{k=1}^m).
\end{equation}
Here $\Gamma^\pm (\varsigma_{\sigma}^{(D)}-\eta /2)$ (respectively $\Gamma^\pm (\{\xi _{k}^{(1)}\}_{k=1}^m)$)
surrounds the point $\varsigma_{\sigma}^{(D)}-\eta /2$ (respectively the points $\xi^{(1)} _{1},\ldots ,\xi^{(1)} _{m}$) with index $\pm1$, all other poles being outside.
\end{theorem}

We recall that a full description of the configuration of Bethe roots for the ground states in terms of the boundary parameters $\varsigma_{\sigma}^{(D)}$ has been done in \cite{GriDT19} in the regime $\Delta>1$.

It is interesting to compare the above result to the corresponding formulas (6.8)-(6.10) of \cite{KitKMNST07} that were obtained in the diagonal case.  To this aim, we can rewrite \eqref{result-thermo}-\eqref{H-thermo} in a closer form as (6.8)-(6.10) of \cite{KitKMNST07} by changing the sign of the integration variables $\lambda_j$, $1\le j\le m$, and by redefining new inhomogeneity parameters as
\begin{equation}\label{redef-xi}
   \tilde\xi_j=-\xi_j+\eta/2, \quad 1\le j\le m.
\end{equation}
It gives
\begin{multline}\label{result-thermo-bis}
   \moy{ \prod_{n=1}^{m}E_{n}^{\epsilon _{n}^{\prime },\epsilon _{n}}(\xi_{n}|(a_{n},b_{n}),(\bar{a}_{n},\bar{b}_{n})) }
     =\prod_{n=1}^m\frac{e^\eta}{\sinh(\eta b_n)}\,
    \frac{(-1)^{m-s}}{\prod\limits_{j<i}\sinh (\tilde\xi _i-\tilde\xi_j)\prod\limits_{i\leq j}\sinh (\tilde\xi_i+\tilde\xi_j-\eta)}   \\
 \times \int_{\mathcal{\tilde C}}\prod_{j=1}^{s}d\lambda _{j}\ \int_{\mathcal{\tilde C}_{\boldsymbol{\xi}}}
\prod_{j=s+1}^{m}\!\!d\lambda _{j}\ 
\tilde H_{m}(\{\lambda _{j}\}_{j=1}^M;\{\tilde \xi _{k}\}_{k=1}^m)\ \det_{1\leq j,k\leq m}\big[\tilde \Phi(\lambda_j,\tilde \xi_k)\big],
\end{multline}
where we have denoted
\begin{equation}\label{mat-Phi-tilde}
\tilde \Phi(\lambda_j,\xi_k)=\frac{1}{2}\big[\tilde \rho (\lambda _{j}-\tilde \xi _{k})-\tilde\rho (\lambda_{j}-\eta+\tilde\xi _{k})\big],
\end{equation}
with
\begin{equation}
    \tilde\rho(\lambda)=\rho(\lambda+\eta/2)
    =\left\lbrace\,
   \begin{array}{@{}l@{\quad}l@{\quad}l@{}}
   \displaystyle
   \frac{i}{\zeta\,\sinh(\pi\lambda/\zeta)} &\text{with} \ \ \zeta=i\eta>0  \  &\text{if}\ \ |\Delta|<1,
   \vspace{1mm}\\
   \displaystyle 
   -\frac 1\pi\frac{\vartheta'_1(0,q)}{\vartheta_2(0,q)}\frac{\vartheta_2(i\lambda,q)}{\vartheta_1(i\lambda,q)}  
                 &\text{with} \ \ q=e^\eta \ (\eta<0) \  &\text{if}\ \ \Delta>1.
   \end{array}\right.
\end{equation}
and
\begin{multline}\label{H-thermo-tilde}
\tilde H_{m}(\{\lambda _{j}\}_{j=1}^M;\{\tilde \xi _{k}\}_{k=1}^m)
 =\frac{\prod\limits_{j=1}^{m}\prod\limits_{k=1}^{m}\sinh (\lambda _{j}+\tilde \xi_k-\eta)}
   {\!\!\!\!\prod\limits_{1\leq i<j\leq m}\!\!\!\!\sinh (\lambda_i-\lambda_j+\eta )\,\sinh (\lambda_i+\lambda_j-\eta )}
    \\
 \times 
 \prod\limits_{p=1}^{s}\bigg\{
 \sinh (\lambda_p-\tilde\xi _{i_{p}}+\eta (1+b_{i_p}))
 \prod\limits_{k=1}^{i_{p}-1}\sin(\lambda _{p}-\tilde\xi _{k})
 \prod\limits_{k=i_{p}+1}^{m}\!\!\sinh (\lambda_{p}-\tilde\xi _{k}+\eta )\bigg\} 
 \\
 \times 
 \!\!\prod\limits_{p=s+1}^{m}\bigg\{
 \sinh (\lambda_p-\tilde\xi _{i_{p}}-\eta(1-\bar b_{i_p}))
 \prod\limits_{k=1}^{i_{p}-1}\sinh (\lambda _{p}-\tilde\xi _{k})
 \prod\limits_{k=i_{p}+1}^{m}\!\!\sinh (\lambda_{p}-\tilde\xi _{k}-\eta )\bigg\} 
   \\
   \times
   \prod\limits_{k=1}^{m}\frac{\sinh (\tilde\xi _{k}+\varsigma_{+}^{(D)}-\eta/2)\,\sinh (\tilde\xi _{k}+\varsigma_{-}^{(D)}-\eta/2)}{\sinh(\lambda _{k}+\varsigma_{+}^{(D)}-\eta /2)\, \sinh (\lambda _{k}+\varsigma_{-}^{(D)}-\eta /2)} .
\end{multline}
Here the contour $\mathcal{\tilde C}$ is defined as \eqref{C-nude}
if the set of Bethe roots for the ground state does not contain any boundary root, and as
\begin{equation}\label{C-BR-tilde}
  \mathcal{\tilde C} =  [-\Lambda ,\Lambda ] \cup \Gamma^+(\eta/2-\varsigma_{\sigma}^{(D)}) 
\end{equation}
if the set of Bethe roots for the ground state contains the boundary root $\check\lambda_\sigma$,
and the contour $\mathcal{\tilde C}_{\boldsymbol{\xi}}$ is defined as
\begin{equation}\label{C-xi-tilde}
   \mathcal{\tilde C}_{\boldsymbol{\xi}} =\mathcal{\tilde C}\cup \Gamma^+ (\{\tilde \xi_k\}_{k=1}^m).
\end{equation}
With this rewriting and this redefinition of the inhomogeneity parameters, we see a clear correspondence with the results (6.9)-(6.10) of \cite{KitKMNST07} for the diagonal case. The only difference comes from the boundary factor (last line of \eqref{H-thermo-tilde}) which here contains two factors and involves both parameters $\varsigma^{(D)}_+$ and $\varsigma^{(D)}_-$, the related fact that the residues at both boundary roots $\check\lambda_+$ and $\check\lambda_-$ may be involved through the definition of the integration contour \eqref{C-BR-tilde}, and an additional dependence in terms of the gauge parameters $b_n$ and $\bar b_n$, which are themselves defined in terms of the boundary parameters through the gauge condition \eqref{Gauge-cond-B}.
This strong similarity of the two results (the one of \cite{KitKMNST07} and the present one) is not surprising since, as we have shown, the chain that we consider here is isospectral to the one considered in \cite{KitKMNST07}, and the computation of the local operator actions and of the resulting scalar products are very similar in the two cases, except for the contribution of the gauge parameters involved in the definition of the eigenstates and in the local operators.
On the other hand, the fact that the two boundary parameters $\varsigma^{(D)}_+$ and $\varsigma^{(D)}_-$ and their corresponding boundary roots $\check\lambda_+$ and $\check\lambda_-$appear explicitly in the result\footnote{And not only indirectly through the definition of the ground state, see \cite{GriDT19}.} is also not surprising since, contrary to what happens in \cite{KitKMNST07}, these two boundary parameters both label the boundary field at site 1 through \eqref{Id-boundary}.

\begin{rem} It is interesting to consider the limit
\begin{equation}
   \psi_-\to +\infty \quad \text{with} \quad \varphi_-\sim \varsigma_-\ \text{finite},
\end{equation}
in which one recovers diagonal boundary conditions on both ends (in that case $\kappa_-\sim e^{-\psi_-}\to 0$). In Appendix~\ref{app-diag}, we explicitly show how the result of \cite{KitKMNST07} for the correlation functions can be inferred from \eqref{result-thermo-bis}-\eqref{H-thermo-tilde}.
\end{rem}

Finally, let us mention that, from such a formula, the homogeneous limit can be taken in both regimes of the chain completely similarly as in \cite{KitKMNST07} (see formulas (6.15) and (6.18) of that paper). 

\section{Conclusion}

Here we have shown how to compute correlation functions of a class of local operators for an open chain with unparallel boundary fields and with particular diagonal boundary conditions at the last site. The boundary field at the first site of the chain is here kept completely general.

Let us briefly summarize our method and results. The boundary matrix at site 1 can be made diagonal by means of the use of a vertex-IRF gauge transformation, and the spectrum and eigenstates of the corresponding transfer matrix can be described within the new SoV approach introduced in \cite{MaiN19}. These eigenstates can also be rewritten in the form of generalized Bethe states constructed by means of the gauge transformed boundary Yang-Baxter algebra. On the one hand we can use, as in \cite{KitKMNST07}, a boundary-bulk decomposition of the generalized Bethe states, as well as the solution of the bulk inverse problem, to compute the action of a basis of local operators on these generalized Bethe states, the main difference being that all this should be done in terms of the gauged transformed boundary and bulk Yang-Baxter operators. On the other hand, since our special boundary conditions at site $N$ can be obtained as some limit of some more general non-diagonal ones for which the generalized Sklyanin's SoV approach can be used, we can compute the resulting scalar products as some limit of the determinant representations obtained in \cite{KitMNT18} for the scalar products of separate states. As in  \cite{KitKMNST07}, we hence obtain correlation functions as multiple sums of determinants \eqref{corr-finite} and, in the half-infinite chain, as multiple integrals \eqref{result-thermo}. Let us also mention that, due to the fact the open chain that we consider is isospectral to a chain with diagonal boundary condition, the description of the ground state and computation of the scalar products have the same form than in  \cite{KitKMNST07}.
Hence, our final results are very similar to the ones of \cite{KitKMNST07}, except that our explicit final formula involves the description of the boundary field at site 1 through two boundary parameters (instead of only one in \cite{KitKMNST07}), and possibly through the contribution of two different isolated complex roots (boundary roots) converging towards these two boundary parameters. Finally, one should mention that the dependence into the boundary parameters at site 1 also appears in a slightly more intricate way via the gauge parameters \eqref{Gauge-cond-B}.

Let us underline once again that we have here only computed the correlation of a particular class of local operators, the ones which satisfy the condition \eqref{n_b=n_c}, i.e. which do not change the number of gauged $B$-operators when acting on bulk (or boundary) generalized gauged Bethe vectors, see \eqref{act-local-op-bulk}, \eqref{act-boundary}. The computation of completely general correlation functions indeed leads to more complicated types of scalar products, for which at the moment there exists no known simple determinant representation. We intend to consider this interesting but more delicate problem in a further study.

Finally, we expect that our study of correlation function can be adapted to the consideration of some other type of boundary conditions at site $N$, for instance non-diagonal ones with the constraint \eqref{homog-cond}. We also intend to consider such problems in a further publication.

\appendix

\section{On the bulk gauge Yang-Baxter algebra}
\label{app-gaugeYBbulk}

In this appendix, we list some useful properties satisfied by \eqref{gauge-M}.

\subsection{A useful identity}

From the definitions \eqref{gauge-M} and \eqref{gauge-Mhat}, one can show the following useful identity:
\begin{align}\label{Mhat-M}
\hat{M}(\lambda |(\alpha ,\beta ),(\gamma ,\delta ))
& =\frac{\sigma _{0}^{y} \left[ S^{t_{0}}(\lambda +\eta /2|\gamma ,\delta )\, M^{t_{0}}(-\lambda)\, \sigma _{0}^{y}\, S(\lambda +\eta /2|\alpha ,\beta )\, \sigma _{0}^{y}\right] \sigma _{0}^{y}}
{( -1) ^{N}\det S(\lambda +\eta /2|\gamma,\delta )} 
\nonumber\\
& =\frac{\sigma _{0}^{y}\left[ \sigma _{0}^{y}\, S^{t_{0}}(\lambda +\eta/2|\alpha ,\beta )\, \sigma _{0}^{y}\, M(-\lambda )\, S\,(\lambda +\eta /2|\gamma,\delta )\right] ^{t_{0}}\sigma _{0}^{y}}
{( -1) ^{N} \det S(\lambda +\eta /2|\gamma ,\delta )} 
\nonumber\\
& =\frac{\sigma _{0}^{y}\left[ S^{-1}(\lambda -\eta /2|\alpha -1,\beta)\, M(-\lambda )\, S(\lambda -\eta /2|\gamma -1,\delta )\right] ^{t_{0}}\sigma_{0}^{y}}
{( -1) ^{N}(\det S(\lambda +\eta /2|\alpha ,\beta))^{-1} \det S(\lambda +\eta /2|\gamma ,\delta )} 
\nonumber\\
& =( -1) ^{N}\, \frac{\det S(\lambda +\eta /2|\alpha ,\beta )}{ \det S(\lambda +\eta /2|\gamma ,\delta )}\,\sigma_{0}^{y}\, M^{t_{0}}(-\lambda |(\alpha -1,\beta ),(\gamma -1,\delta ))\, \sigma_{0}^{y},
\end{align}
or in components:
\begin{multline}\label{MhatM-comp}
 \hat{M}(\lambda |(\alpha ,\beta ),(\gamma ,\delta ))
 =( -1) ^{N} \, e^{-\eta(\alpha-\gamma)}\,\frac{\sinh(\eta\beta)}{\sinh(\eta\delta)}
 \\
 \times \left( 
\begin{array}{ll}
D(-\lambda |(\alpha -1,\beta ),(\gamma -1,\delta )) & -B(-\lambda |(\alpha
-1,\beta ),(\gamma -1,\delta )) \\ 
-C(-\lambda |(\alpha -1,\beta ),(\gamma -1,\delta )) & A(-\lambda |(\alpha
-1,\beta ),(\gamma -1,\delta ))%
\end{array}%
\right) .
\end{multline}

\subsection{Commutation relations}

The bulk monodromy matrix \eqref{bulk-mon} satisfies the following Yang-Baxter relation with the R-matrix \eqref{barR}:
\begin{equation}\label{YBbar}
  \bar R_{12}(\lambda-\mu)\, M_1(\lambda)\, M_2(\mu)
  =M_2(\mu)\, M_1(\lambda)\, \bar R_{12}(\lambda-\mu),
\end{equation}
Let us define
\begin{equation}\label{barS}
   \bar S(\lambda|\alpha,\beta)
   =\begin{pmatrix} 
     e^{\lambda+\eta(\beta+\alpha)} & e^{\lambda-\eta(\beta-\alpha)} \\
     1 & 1
     \end{pmatrix}
    =S(\lambda|-\alpha,-\beta)
\end{equation}
and
\begin{equation}
 \bar R^{\mathrm{SOS}}(\lambda |\beta )=%
\begin{pmatrix}
\sinh (\lambda -\eta ) & 0 & 0 & 0 \\ 
0 & \frac{\sinh (\eta (\beta +1))}{\sinh (\eta \beta )}\,\sinh \lambda & 
\frac{\sinh (\lambda -\eta \beta )}{\sinh (\eta \beta )}\,\sinh \eta & 0 \\ 
0 & -\frac{\sinh (\eta \beta +\lambda )}{\sinh (\eta \beta )}\,\sinh \eta & 
\frac{\sinh (\eta (\beta -1))}{\sinh (\eta \beta )}\,\sinh \lambda & 0 \\ 
0 & 0 & 0 & \sinh (\lambda -\eta )%
\end{pmatrix}%
.  \label{barR-SOS}
\end{equation}
the analogue of \eqref{mat-S} and \eqref{R-SOS} with $\eta$ replaced by $\bar\eta=-\eta$.
We have
\begin{equation}
 \bar R_{12}(\lambda -\mu )\,\bar S_{2}(-\mu |\alpha ,\beta )\, \bar S_{1}(-\lambda |\alpha
,\beta +\sigma _{2}^{z})
=\bar S_{1}(-\lambda |\alpha ,\beta )\,\bar S_{2}(-\mu |\alpha,\beta +\sigma _{1}^{z})\,
\bar R_{21}^{\mathrm{SOS}}(\lambda -\mu |\beta ),
\label{Vertex-IRF2-bar}
\end{equation}
which is the analogue of \eqref{Vertex-IRF2} with $\eta$ replaced by $\bar\eta=-\eta$. The latter relation implies that
\begin{multline}
 \bar R_{21}^{\mathrm{SOS}}(\lambda -\mu |-\delta )
 =S_{2}(-\mu-\eta/2 |\gamma,\delta -\sigma _{1}^{z})^{-1}\,S_{1}(-\lambda-\eta/2 |\gamma ,\delta )^{-1}\,
 \\
 \times\bar R_{12}(\lambda -\mu )\, S_{2}(-\mu-\eta/2 |\gamma ,\delta )\,  S_{1}(-\lambda-\eta/2 |\gamma,\delta -\sigma _{2}^{z}),
\label{Vertex-IRF2-bar2}
\end{multline}
so that, multiplying \eqref{YBbar} on the left by $S_{2}(-\mu-\eta/2 |\alpha,\beta -\sigma _{1}^{z})^{-1}\,S_{1}(-\lambda-\eta/2 |\gamma ,\delta )^{-1}$ and on the right by $S_{2}(-\mu-\eta/2 |\gamma ,\delta )\,  S_{1}(-\lambda-\eta/2 |\gamma,\delta -\sigma _{2}^{z})$, we obtain:
\begin{multline}\label{YBbar-gauge}
 \bar R_{21}^{\mathrm{SOS}}(\lambda -\mu |-\beta )\, 
 S_{1}(-\lambda-\eta/2 |\alpha,\beta -\sigma _{2}^{z})^{-1}\, S_{2}(-\mu-\eta/2 |\alpha ,\beta )^{-1}\, 
 M_1(\lambda)\, M_2(\mu)\, \\
 \times
 S_{2}(-\mu-\eta/2 |\gamma ,\delta )\,  S_{1}(-\lambda-\eta/2 |\gamma,\delta -\sigma _{2}^{z})
 = S_{2}(-\mu-\eta/2 |\alpha,\beta -\sigma _{1}^{z})^{-1}\,S_{1}(-\lambda-\eta/2 |\gamma ,\delta )^{-1}
 \\
 \times
 M_2(\mu)\, M_1(\lambda)\,
 S_{1}(-\lambda-\eta/2 |\gamma ,\delta )\, S_{2}(-\mu-\eta/2 |\gamma,\delta -\sigma _{1}^{z})\,
 \bar R_{21}^{\mathrm{SOS}}(\lambda -\mu |-\delta ).
\end{multline}
Notice that
\begin{equation}
 \bar R_{21}^{\mathrm{SOS}}(\lambda |-\beta )
=\bar R_{12}^{\mathrm{SOS}}(\lambda |\beta )
.  
\end{equation}
The commutation relations of the matrix elements of \eqref{gauge-M} can be deduced from \eqref{YBbar-gauge}. In particular, we have
\begin{align}
  &A(\lambda |(\alpha,\beta -1),(\gamma,\delta -1))\,
  A(\mu |(\alpha,\beta) ,(\gamma,\delta) )
  =A(\mu | (\alpha,\beta -1),(\gamma,\delta -1))\,
  A(\lambda |(\alpha,\beta) ,(\gamma,\delta) ) ,
  \\
  &B(\lambda |(\alpha,\beta -1),(\gamma,\delta +1))\,
  B(\mu |(\alpha,\beta) ,(\gamma,\delta ))
  =B(\mu |(\alpha ,\beta -1),(\gamma,\delta +1))\,
  B(\lambda |(\alpha,\beta) ,(\gamma,\delta) ) ,
  \\
  &C(\lambda |(\alpha,\beta +1),(\gamma,\delta -1))\,
  C(\mu |(\alpha ,\beta) ,(\gamma,\delta ))
  =C(\mu |(\alpha ,\beta +1),(\gamma ,\delta -1))\,
  C(\lambda |(\alpha,\beta) ,(\gamma ,\delta )) ,
  \\
  &D(\lambda |(\alpha,\beta +1),(\gamma,\delta +1))\,
  D(\mu |(\alpha ,\beta) ,(\gamma,\delta )) 
  =D(\mu |(\alpha ,\beta +1),(\gamma ,\delta +1))\,
  D(\lambda |(\alpha,\beta) ,(\gamma ,\delta) ).
\end{align}
We also have 
%
\begin{multline}\label{comm-AB-gauge}
   A(\mu|(\alpha,\beta-1),(\gamma,\delta+1))\, B(\lambda|(\alpha,\beta),(\gamma,\delta))
   \\
   =\frac{\sinh(\lambda-\mu-\eta)\,\sinh(\eta\delta)}{\sinh(\lambda-\mu)\,\sinh(\eta(\delta-1))}\,
   B(\lambda|(\alpha,\beta-1),(\gamma,\delta-1))\, A(\mu|(\alpha,\beta),(\gamma,\delta))
   \\
   -\frac{\sinh\eta\,\sinh(\lambda-\mu-\eta\delta)}{\sinh(\eta(\delta-1))\,\sinh(\lambda-\mu)}\,
   B(\mu|(\alpha,\beta-1),(\gamma,\delta-1))\, A(\lambda|(\alpha,\beta),(\gamma,\delta)),
\end{multline}
and
\begin{multline}\label{comm-DB-gauge}
   D(\lambda|(\alpha,\beta-1),(\gamma,\delta+1))\, B(\mu|(\alpha,\beta),(\gamma,\delta))
   \\
   =\frac{\sinh(\lambda-\mu-\eta)\,\sinh(\eta\beta)}{\sinh(\lambda-\mu)\,\sinh(\eta(\beta-1))}\,
   B(\mu|(\alpha,\beta+1),(\gamma,\delta+1))\, D(\lambda|(\alpha,\beta),(\gamma,\delta))
   \\
   +\frac{\sinh\eta\,\sinh(\lambda-\mu+\eta\beta)}{\sinh(\eta(\beta-1))\,\sinh(\lambda-\mu)}\,
   B(\lambda|(\alpha,\beta+1),(\gamma,\delta+1))\, D(\mu|(\alpha,\beta),(\gamma,\delta)).
\end{multline}
Using the notation \eqref{redef-gauge-op}, these commutation relations can be rewritten as
\begin{align}
  &A(\lambda |\alpha-\beta +1,\gamma+\delta -1)\,
  A(\mu |\alpha-\beta ,\gamma+\delta )
  =A(\mu | \alpha-\beta +1,\gamma+\delta -1)\,
  A(\lambda |\alpha-\beta ,\gamma+\delta ) ,
  \label{commAA}\\
  &B(\lambda |\alpha-\beta +1,\gamma-\delta -1)\,
  B(\mu |\alpha-\beta ,\gamma-\delta )
  =B(\mu |\alpha -\beta +1,\gamma-\delta -1)\,
  B(\lambda |\alpha-\beta ,\gamma-\delta ) ,
  \label{commBB}\\
  &C(\lambda |\alpha+\beta +1,\gamma+\delta -1)\,
  C(\mu |\alpha +\beta ,\gamma+\delta )
  =C(\mu |\alpha +\beta +1,\gamma +\delta -1)\,
  C(\lambda |\alpha+\beta ,\gamma +\delta ) ,
  \label{commCC}\\
  &D(\lambda |\alpha+\beta +1,\gamma-\delta -1)\,
  D(\mu |\alpha +\beta ,\gamma-\delta ) 
  =D(\mu |\alpha +\beta +1,\gamma -\delta -1)\,
  D(\lambda |\alpha+\beta ,\gamma -\delta ),
  \label{commDD}
\end{align}
i.e.
\begin{align}
  &A(\lambda |x +1,y -1)\, A(\mu |x ,y )  =A(\mu | x +1,y -1)\, A(\lambda |x ,y ) ,  
  \label{commAAbis}\\
  &B(\lambda |x +1,y -1)\, B(\mu |x ,y )  =B(\mu |x +1,y -1)\,  B(\lambda |x ,y ) ,
  \label{commBBbid}\\
  &C(\lambda |x +1,y -1)\,  C(\mu |x , y )  =C(\mu |x +1,y -1)\,  C(\lambda |x ,y ) ,
  \label{commCCbis}\\
  &D(\lambda |x +1,y -1)\,  D(\mu |x,y )   =D(\mu |x +1,y -1)\,  D(\lambda |x ,y ).
  \label{commDDbis}
\end{align}
and
\begin{multline}\label{comm-AB-bis}
   A(\mu|\alpha-\beta+1,\gamma+\delta+1)\, B(\lambda|\alpha-\beta,\gamma-\delta)
   \\
   =\frac{\sinh(\lambda-\mu-\eta)\,\sinh(\eta\delta)}{\sinh(\lambda-\mu)\,\sinh(\eta(\delta-1))}\,
   B(\lambda|\alpha-\beta+1,\gamma-\delta+1)\, A(\mu|\alpha-\beta,\gamma+\delta)
   \\
   -\frac{\sinh\eta\,\sinh(\lambda-\mu-\eta\delta)}{\sinh(\eta(\delta-1))\,\sinh(\lambda-\mu)}\,
   B(\mu|\alpha-\beta+1,\gamma-\delta+1)\, A(\lambda|\alpha-\beta,\gamma+\delta),
\end{multline}
\begin{multline}\label{comm-DB-bis}
   D(\lambda|\alpha+\beta-1,\gamma-\delta-1)\, B(\mu|\alpha-\beta,\gamma-\delta)
   \\
   =\frac{\sinh(\lambda-\mu-\eta)\,\sinh(\eta\beta)}{\sinh(\lambda-\mu)\,\sinh(\eta(\beta+1))}\,
   B(\mu|\alpha-\beta-1,\gamma-\delta-1)\, D(\lambda|\alpha+\beta,\gamma-\delta)
   \\
   +\frac{\sinh\eta\,\sinh(\lambda-\mu+\eta\beta)}{\sinh(\eta(\beta+1))\,\sinh(\lambda-\mu)}\,
   B(\lambda|\alpha-\beta-1,\gamma-\delta-1)\, D(\mu|\alpha+\beta,\gamma-\delta),
\end{multline}
i.e.
\begin{multline}\label{comm-AB-ter}
   A(\mu |x+1,y+1)\, B(\lambda |x,z)
   =\frac{\sinh(\mu-\lambda+\eta)\,\sinh(\eta\frac{y-z}2)}{\sinh(\mu-\lambda)\,\sinh(\eta(\frac{y-z}2-1))}\,
   B(\lambda|x+1,z+1)\, A(\mu|x,y)
   \\
   -\frac{\sinh\eta\,\sinh(\mu-\lambda+\eta\frac{y-z}2)}{\sinh(\eta(\frac{y-z}2-1))\,\sinh(\mu-\lambda)}\,
   B(\mu|x+1,z+1)\, A(\lambda|x,y),
\end{multline}
\begin{multline}\label{comm-DB-ter}
   D(\lambda|x-1,y-1)\, B(\mu|z,y)
   =\frac{\sinh(\lambda-\mu-\eta)\,\sinh(\eta\frac{x-z}2)}{\sinh(\lambda-\mu)\,\sinh(\eta(\frac{x-z}2+1))}\,
   B(\mu|z-1,y-1)\, D(\lambda|x,y)
   \\
   +\frac{\sinh\eta\,\sinh(\lambda-\mu+\eta\frac{x-z}2)}{\sinh(\eta(\frac{x-z}2+1))\,\sinh(\lambda-\mu)}\,
   B(\lambda|z-1,y-1)\, D(\mu|x,y).
\end{multline}

\subsection{Action on $\bra{0}$ and $\ket{0}$}

The vertex-IRF \eqref{Vertex-IRF2-bar2} can be easily extended to a transformation between bulk monodromy matrices:
\begin{multline}\label{rel-M-barMSOS}
  \bar M^\mathrm{SOS}(\lambda|-\beta)=\widetilde{S}_{1\ldots N}^{-1}(\{\xi\}|\alpha,\beta-\sigma_0^z)\,   S_0^{-1}(-\lambda-\eta/2|\alpha,\beta)\, M(\lambda)\,
  \\
  \times
 \widetilde{S}_{1\ldots N}(\{\xi\}|\alpha,\beta)\, S_0(-\lambda-\eta/2|\alpha,\beta-\mathrm{S}^z),
\end{multline}
in which
\begin{align}
  &\widetilde{S}_{1\ldots N}(\{\xi\}|\alpha,\beta)
  =\prod_{n=N\to 1} S_n\bigg(-\xi_n \, \Big|\, \alpha,\beta-\sum_{j=n+1}^N\sigma_j^z\bigg),
  \label{def-tilde-SN}
\end{align}
and
\begin{align}\label{barMSOS}
  \bar M^\mathrm{SOS}(\lambda|\beta)
  &=\prod_{n=N\to 1} \bar R^\mathrm{SOS}_{n0}\bigg(\lambda-\xi_n+\eta/2\,\Big|\, \beta+\sum_{j=n+1}^N \sigma_j^z\bigg),\nonumber\\
  &=\begin{pmatrix}
      \bar A^\mathrm{SOS}(\lambda|\beta) & \bar B^\mathrm{SOS}(\lambda|\beta) \\
      \bar C^\mathrm{SOS}(\lambda|\beta) & \bar D^\mathrm{SOS}(\lambda|\beta)
      \end{pmatrix}.
\end{align}
The relation \eqref{rel-M-barMSOS} can also be written in terms of the matrix \eqref{gauge-M} as
\begin{multline}\label{MSOS-Mgauge}
   \bar M^\mathrm{SOS}(\lambda|-\beta)
   =\widetilde{S}_{1\ldots N}^{-1}(\{\xi\}|\alpha,\beta-\sigma_0^z)\,
   \mathsf{S}_0(\alpha,\beta|\alpha',\beta')
   \\
   \times
   M(\lambda|(\alpha',\beta'),(\gamma,\delta))\,
   \widetilde{S}_{1\ldots N}(\{\xi\}|\alpha,\beta)\, \mathsf{S}_0(\gamma,\delta|\alpha,\beta-\mathrm{S}^z),
\end{multline}
in which
\begin{align}\label{S-1S}
   \mathsf{S}(\alpha,\beta|\alpha',\beta')
   &=S^{-1}(\lambda|\alpha,\beta)\, S(\lambda|\alpha',\beta')
   \nonumber\\
   &= \frac{e^{\eta\alpha}}{2\sinh(\eta\beta)}
   \begin{pmatrix}
     e^{\eta(\beta-\alpha)}-e^{-\eta(\alpha'+\beta')} & e^{\eta(\beta-\alpha)}- e^{\eta(\beta'-\alpha')} \\
     -e^{-\eta(\alpha+\beta)}+e^{-\eta(\alpha'+\beta')} & -e^{-\eta(\alpha+\beta)}+e^{\eta(\beta'-\alpha')}
   \end{pmatrix}.
\end{align}
Note that the matrix \eqref{S-1S} does not depend on the spectral parameter $\lambda$.

From \eqref{MSOS-Mgauge} and from the fact that
\begin{align}
  &\bar{A}^\mathrm{SOS}(\lambda|-\beta)\,\ket{0}=d(\lambda)\,\ket{0}, \\
  &\bar{D}^\mathrm{SOS}(\lambda|-\beta)\,\ket{0}=\frac{\sinh(\eta(\beta-N))}{\sinh(\eta\beta)}\, a(\lambda)\, \ket{0},\\
  &\bar{C}^\mathrm{SOS}(\lambda|-\beta)\,\ket{0}=0,
\end{align}
which follows from \eqref{barMSOS}, one obtains that
\begin{align}   
   &C(\lambda|(\alpha,\beta),(\alpha,\beta-N))\, \ket{\eta,\alpha+\beta}=0, \label{actC-0}\\
   &A(\lambda|(\alpha,\beta),(\alpha,\beta-N))\, \ket{\eta,\alpha+\beta}
   = d(\lambda)\,\ket{\eta,\alpha+\beta-1},\label{actA-0}\\
   &D(\lambda|(\alpha,\beta),(\alpha,\beta-N))\, \ket{\eta,\alpha+\beta}
   =\frac{\sinh(\eta(\beta-N))}{\sinh(\eta\beta)}\, a(\lambda)\,\ket{\eta,\alpha+\beta+1},\label{actD-0}
\end{align}
and more generally, by using the expression \eqref{S-1S}, that
\begin{align}
   &A(\lambda|(\alpha',\beta'),(\alpha,\beta-N))\, \ket{\eta,\alpha+\beta}
   =\frac{e^{\eta\alpha'}(e^{\eta(\beta'-\alpha')}-e^{-\eta(\alpha+\beta)})}{2\sinh(\eta\beta')}\, d(\lambda)\,\ket{\eta,\alpha+\beta-1},\label{actA-0gen}\\
   &D(\lambda|(\alpha,\beta),(\gamma,\delta))\, \ket{\eta,\alpha+\beta}
   =\frac{e^{\eta\alpha}(e^{\eta(\delta-\gamma)}-e^{-\eta(\alpha+\beta-N)})}{2\sinh(\eta\beta)}\, a(\lambda)\,\ket{\eta,\alpha+\beta+1},\label{actD-0gen}
\end{align}
%
in which
\begin{equation}
   \ket{\eta,\alpha+\beta}=\widetilde{S}_{1\ldots N}(\{\xi\}|\alpha,\beta)\ket{0}
\end{equation}
coincides with the definition \eqref{Right-C-ref}.
Similarly, from \eqref{MSOS-Mgauge} and from the fact that
\begin{align}
  &\bra{0}\,\bar{A}^\mathrm{SOS}(\lambda|-\beta)=d(\lambda)\,\bra{0} ,\\
  &\bra{0}\, \bar{D}^\mathrm{SOS}(\lambda|-\beta)=\frac{\sinh(\eta(\beta-N))}{\sinh(\eta\beta)}\, a(\lambda)\, \bra{0},\\
  &\bra{0}\,\bar{B}^\mathrm{SOS}(\lambda|-\beta)=0,
\end{align}
one obtains that
\begin{align}
  &\bra{\alpha-\beta+1,\eta}\, B(\lambda|(\alpha,\beta),(\alpha,\beta-N))=0,\label{0-actB}\\
  &\bra{\alpha-\beta+1,\eta}\, A(\lambda|(\alpha,\beta),(\gamma,\delta))
  =\frac{e^{\eta\alpha}(e^{\eta(\beta-N-\alpha)}-e^{-\eta(\gamma+\delta)})}{2\sinh(\eta\beta)}\, d(\lambda)\, \bra{\alpha-\beta,\eta},\label{0-actA}\\
  &\bra{\alpha-\beta+1,\eta}\, D(\lambda|(\alpha',\beta'),(\alpha,\beta-N))
  =\frac{e^{\eta\alpha'}(e^{\eta(\beta-\alpha)}-e^{-\eta(\alpha'+\beta')})}{2\sinh(\eta\beta')}\, a(\lambda)\, \bra{\alpha-\beta+2,\eta},\label{0-actD}
\end{align}
in which
\begin{equation}
   \bra{\alpha-\beta,\eta}=e^{-\sum_n\xi_n}\, N_{\alpha ,\beta ,\eta }\, \bra{0}\, \widetilde{S}_{1\ldots N}^{-1}(\{\xi\}|\alpha,\beta)
\end{equation}
coincides with the definition \eqref{Left-B-ref}.
With the normalization \eqref{redef-gauge-op}, \eqref{actC-0}-\eqref{actD-0gen} can be rewritten as
\begin{align}   
   &C(\lambda|\alpha+\beta,\alpha+\beta-N)\, \ket{\eta,\alpha+\beta}=0, \label{actC-0-op}\\
   &A(\lambda|\alpha'-\beta',\alpha+\beta-N)\, \ket{\eta,\alpha+\beta}
   =\frac{e^{\eta(\beta'-\alpha')}-e^{-\eta(\alpha+\beta)}}{e^{\eta/2}}\, d(\lambda)\,\ket{\eta,\alpha+\beta-1},\label{actA-0-op}\\
   &D(\lambda|\alpha+\beta,\gamma-\delta)\, \ket{\eta,\alpha+\beta}
   =\frac{e^{\eta(\delta-\gamma)}-e^{-\eta(\alpha+\beta-N)}}{e^{\eta/2}}\, a(\lambda)\,\ket{\eta,\alpha+\beta+1},\label{actD-0-op}
\end{align}
and \eqref{0-actB}-\eqref{0-actD} can be rewritten as
\begin{align}
  &\bra{\alpha-\beta+1,\eta}\, B(\lambda|\alpha-\beta,\alpha-\beta+N)=0,\label{0-actB-op}\\
  &\bra{\alpha-\beta+1,\eta}\, A(\lambda|\alpha-\beta,\gamma+\delta)
  =\frac{e^{\eta(\beta-N-\alpha)}-e^{-\eta(\gamma+\delta)}}{e^{\eta/2}}\, d(\lambda)\, \bra{\alpha-\beta,\eta},\label{0-actA-op}\\
  &\bra{\alpha-\beta+1,\eta}\, D(\lambda|\alpha'+\beta',\alpha-\beta+N)
  =\frac{e^{\eta(\beta-\alpha)}-e^{-\eta(\alpha'+\beta')}}{e^{\eta/2}}\, a(\lambda)\, \bra{\alpha-\beta+2,\eta}.\label{0-actD-op}
\end{align}

\subsection{Expression of the gauged operators in terms of a single operator}

From the definition \eqref{gauge-M}, all the elements of the gauged bulk monodromy matrix $M(\lambda |(\alpha ,\beta ),(\gamma ,\delta )) $ are linear combinations of the elements of $M(\lambda)$. More precisely, we have
\begin{equation}
\begin{pmatrix}
    A(\lambda |(\alpha ,\beta ),(\gamma ,\delta )) & B(\lambda |(\alpha ,\beta),(\gamma ,\delta )) \\ 
   C(\lambda |(\alpha ,\beta ),(\gamma ,\delta )) & D(\lambda |(\alpha ,\beta),(\gamma ,\delta ))
\end{pmatrix}
=\frac{\begin{pmatrix}
  -\Theta (\lambda |\alpha -\beta ,\gamma +\delta ) & -\Theta (\lambda |\alpha-\beta ,\gamma -\delta ) \\ 
   \Theta (\lambda |\alpha +\beta ,\gamma +\delta ) & \Theta (\lambda |\alpha+\beta ,\gamma -\delta )
\end{pmatrix} }{2 e^{-\lambda -\eta (\alpha +1/2)}\sinh \eta \beta },
\end{equation}
in which the operator $\Theta (\lambda |x,y)$ stands for the following $(x,y)$-dependant linear combination of the operators entries $A(\lambda)$, $B(\lambda)$, $C(\lambda)$, $D(\lambda)$ of $M(\lambda)$: 
%
\begin{equation}\label{def-Theta-op}
    \Theta (\lambda |x,y)
    =B(\lambda )+A(\lambda )\,e^{-\lambda-\eta/2-\eta y}-C(\lambda)\,e^{-2\lambda-\eta-\eta(y+x)}-D(\lambda )\, e^{-\lambda-\eta/2-\eta x}.
\end{equation}
It is indeed an interesting point that all these gauged generators can be written in terms of a single operator, function of three parameters $\lambda ,x,y$, and computed in different values
of its parameters. It follows in particular that:
%
\begin{align}
B(\lambda |(\alpha ,\beta ),(\gamma ,\delta ))& =A(\lambda |(\alpha ,\beta),(\gamma ,-\delta )), \\
C(\lambda |(\alpha ,\beta ),(\gamma ,\delta ))& =A(\lambda |(\alpha ,-\beta),(\gamma ,\delta )), \\
D(\lambda |(\alpha ,\beta ),(\gamma ,\delta ))& =A(\lambda |(\alpha ,-\beta),(\gamma ,-\delta )).
\end{align}
So that we can claim the following

\begin{lemma}
There is just one (three parameters) generator $\Theta (\lambda |x,y)$ of
the gauged Yang-Baxter algebra. The latter satisfies in particular the following two-shifts
commutation relations:
\begin{equation}\label{comm-Theta1}
\Theta (\lambda |x+1,y-1)\, \Theta (\mu |x,y)
=\Theta (\mu |x+1,y-1)\, \Theta(\lambda |x,y),
\end{equation}
and 
%
\begin{align}\label{comm-Theta2}
\Theta (\lambda |x+1,y+1)\, \Theta (\mu |x,z)
& =\frac{\sinh (\eta \,\frac{y-z}{2})\, \sinh (\lambda -\mu +\eta )}{\sinh (\lambda -\mu )\, \sinh (\eta (\frac{y-z}{2}-1))}\,
\Theta (\mu |x+1,z+1)\, \Theta (\lambda |x,y)  
\notag \\
& -\frac{\sinh \eta \, \sinh (\lambda -\mu +\frac{y-z}{2}\eta )}{\sinh (\lambda-\mu )\, \sinh (\eta (\frac{y-z}{2}-1))}\,
\Theta (\lambda |x+1,z+1)\, \Theta (\mu|x,y),
\end{align}
and
\begin{align}\label{comm-Theta3}
\Theta (\lambda |x-1,y-1)\, \Theta (\mu |w,y)
& =\frac{\sinh (\eta \frac{x-w}{2})\, \sinh (\lambda -\mu -\eta )}{\sinh (\lambda-\mu )\,\sinh (\eta (\frac{x-w}{2}+1)}\,
\Theta (\mu |w-1,y-1)\, \Theta (\lambda |x,y)  \notag \\
& +\frac{\sinh \eta \, \sinh (\lambda -\mu+\frac{x-w}{2}\eta )}{\sinh (\lambda-\mu )\, \sinh (\eta (\frac{x-w}{2}+1))}\, 
\Theta (\lambda |w-1,y-1)\,\Theta (\mu|x,y).
\end{align}
\end{lemma}

We have also the following lemma on the action of the gauged Yang-Baxter
generator on the reference states:

\begin{lemma}
The following identities hold:
\begin{align}
& \bra{ x+1,\eta }\, \Theta (\lambda |x,x+N)=0, 
\\
& \bra{x+1,\eta }\, \Theta (\lambda |x,y)=\frac{e^{-y\eta}-e^{-(x+N)\eta }}{e^{\lambda +\eta /2}}\,
    d(\lambda)\,\bra{x,\eta }, 
\\
& \bra{x+1,\eta}\, \Theta (\lambda |y,x+N)
       =\frac{e^{-x\eta }-e^{-y\eta }}{e^{\lambda +\eta /2}}\,
       a(\lambda)\,\bra{x+2,\eta},
\end{align}
and similarly,
\begin{align}
   &\Theta (\lambda |x,x-N)\, \ket{\eta ,x} =0,
   \label{act-Theta-r1}\\
   &\Theta (\lambda |y,x-N)\, \ket{\eta ,x}
    = \frac{e^{-x\eta }-e^{-y\eta }}{e^{\lambda +\eta /2}}\, d(\lambda)\,
    \ket{\eta ,x-1}, 
    \label{act-Theta-r2}\\
    &\Theta (\lambda |x,y)\, \ket{\eta ,x}
     = \frac{e^{-y\eta}-e^{-(x-N)\eta }}{e^{\lambda +\eta /2}}\, a(\lambda)\,
     \ket{\eta ,x+1}.
     \label{act-Theta-r3}
\end{align}
\end{lemma}

By using these identities, we can write the action of the gauge transformed
Yang-Baxter generator $\Theta (\lambda |x,y)$ on a Bethe-like vector of the
following form:
\begin{equation}\label{Bethe-st-Theta}
   \underline\Theta_M (\{\lambda _{a}\}_{a=1}^M |\alpha -\beta -1,\gamma-\delta +1-2M)\, 
   \ket{\eta ,\gamma +\delta +N},
\end{equation}
in which we have defined
%
\begin{align}\label{prod-ThetaM}
  \underline\Theta_M(\{\lambda_j\}_{j=1}^M|x-1,y+1)
  &=\prod_{j=1\to M}\Theta (\lambda _{j}|x -j, y +j)
  \nonumber\\
  &=  \Theta (\lambda _{1}|x-1,y +1)\, \Theta (\lambda_{2}|x -2,y +2)\cdots 
  \nonumber\\
  &\hspace{0.5cm}
  \cdots
  \Theta (\lambda_{M-1}|x -M+1, y+M-1)\,
  \Theta (\lambda_{M}|x -M,y+M).
\end{align}
Note that, due to the pseudo-commutativity \eqref{comm-Theta1}  of the gauged $\Theta $-operator,
the above operator product is
independent w.r.t. the order of the $\{\lambda _{i}\}_{i=1}^{M}$, i.e.
\begin{equation}
\prod_{j=1\to M}\Theta (\lambda _j |x -j,y +j)
=\prod_{j=1\to M}\Theta (\lambda _{\pi_{j}}|x -j,y +j),
\end{equation}
where $\pi $ is any permutation of the set $\{1,\ldots,M\}$, so that the notation \eqref{prod-ThetaM} is indeed coherent.
Then, the following lemma holds:

\begin{lemma}
\label{Action-Theta}
The actions of the gauged Yang-Baxter generator on a Bethe state of the form \eqref{Bethe-st-Theta} reads:
%
\begin{multline}\label{act-theta-Bethe1}
   \Theta (\lambda_{M+1} |\alpha -\beta ,\gamma +\delta+M)\,
   \underline\Theta_M(\{\lambda_j\}_{j=1}^M|\alpha-\beta-1,\gamma-\delta +1-2M)\,
   \ket{\eta ,\gamma +\delta +N}
   \\
   = \frac{e^{-(\gamma +\delta +N)\eta }-e^{-(\alpha -\beta -M)\eta }}{e^{\eta /2}}
   \sum_{a=1}^{M+1} d(\lambda _{a})\, e^{-\lambda_a}
   \frac{\sinh(\lambda_{M+1}-\lambda_a+\eta(\delta-1+\frac{3M}{2}))}{\sinh(\eta(\delta-1+\frac{M}{2}))}
   \\
   \times 
   \frac{\prod_{j=1 }^M\sinh(\lambda_a-\lambda_j+\eta)}{\prod_{\substack{j=1 \\ j\neq a}}^{M+1}\sinh(\lambda_a-\lambda_j)}\
      \underline\Theta_M(\{\lambda_j\}_{j=1}^{M+1}\setminus\{\lambda_a\}|\alpha -\beta,\gamma -\delta +2-2M)\,
   \ket{\eta ,\gamma+\delta +N-1} ,
\end{multline}
and
%
\begin{multline}\label{act-theta-Bethe2}
 \Theta (\lambda_{M+1} |\gamma +\delta +N-M,\gamma -\delta-2M)\,
   \underline\Theta_M(\{\lambda_j\}_{j=1}^M|\alpha-\beta-1,\gamma-\delta +1-2M)\,
   \ket{\eta ,\gamma +\delta +N}
   \\
 = 
  \frac{e^{-(\gamma -\delta -M)\eta }-e^{-(\gamma+\delta)\eta }}{e^{\eta /2}}
  \sum_{a=1}^{M+1} 
  a(\lambda_{a})\, e^{-\lambda_a}\,
  \frac{\sinh(\lambda_{M+1}-\lambda_a+\eta\frac{\gamma+\delta-\alpha+\beta+N-M+2}{2})}{\sinh(\eta\frac{\gamma+\delta-\alpha+\beta+N+M+2}{2})}
  \\
  \times
   \frac{\prod_{j=1}^M\sinh(\lambda_a-\lambda_j-\eta)}{\prod_{\substack{j=1 \\ j\neq a}}^{M+1}\sinh(\lambda_a-\lambda_j)}\
   \underline\Theta_M(\{\lambda_j\}_{j=1}^{M+1}\setminus\{\lambda_a\}|\alpha -\beta-2,\gamma -\delta -2M)\,
  \ket{\eta ,\gamma+\delta +N+1} .
\end{multline}
\end{lemma}

\begin{proof}
Due to \eqref{comm-Theta2}, one has
\begin{multline}
   \Theta (\lambda |\alpha -\beta ,\gamma +\delta+M)\,
   \underline\Theta_M(\{\lambda_j\}_{j=1}^M|\alpha-\beta-1,\gamma-\delta +1-2M)
   \\
   =\frac{\sinh(\eta(\delta-1+\frac{3M}{2}))}{\sinh(\eta(\delta-1+\frac{M}{2}))}\, 
   \prod_{j=1}^M\frac{\sinh(\lambda-\lambda_j+\eta)}{\sinh(\lambda-\lambda_j)}\,
   \underline\Theta_M(\{\lambda_j\}_{j=1}^M|\alpha-\beta,\gamma-\delta +2-2M)\,
   \Theta (\lambda |\alpha -\beta-M ,\gamma +\delta)
   \\
   -\sum_{n=1}^M 
   \frac{\sinh\eta\, \sinh(\lambda-\lambda_n+\eta(\delta-1+\frac{3M}{2})\eta)}{\sinh(\lambda-\lambda_n)\,\sinh(\eta(\delta-1+\frac{M}{2}))}
   \prod_{j\neq n}\frac{\sinh(\lambda_n-\lambda_j+\eta)}{\sinh(\lambda_n-\lambda_j)}\,
   \\
   \times
    \underline\Theta_M(\{\lambda_j\}_{j\neq n}\cup\{\lambda\}|\alpha-\beta,\gamma-\delta +2-2M)\,
   \Theta (\lambda_n |\alpha -\beta-M ,\gamma +\delta).
\end{multline}
the action on $\ket{\eta,\gamma+\delta+N}$ can then be computed by means of \eqref{act-Theta-r2}:
\begin{multline}
   \Theta (\lambda |\alpha -\beta ,\gamma +\delta+M)\,
   \underline\Theta_M(\{\lambda_j\}_{j=1}^M|\alpha-\beta-1,\gamma-\delta +1-2M)\,
   \ket{\eta,\gamma+\delta+N}
   \\
   =   \frac{e^{-(\gamma+\delta+N)\eta }-e^{-(\alpha -\beta-M)\eta }}{e^{\eta /2}\,\sinh(\eta(\delta-1+\frac{M}{2}))}
   \Bigg\{
   d(\lambda)\, e^{-\lambda}
   \sinh(\eta(\delta-1+\frac{3M}{2}))\, 
   \prod_{j=1}^M\frac{\sinh(\lambda-\lambda_j+\eta)}{\sinh(\lambda-\lambda_j)}\,
   \\
   \times
   \underline\Theta_M(\{\lambda_j\}_{j=1}^M|\alpha-\beta,\gamma-\delta +2-2M)\, \ket{\eta,\gamma+\delta+N-1}
   \\
   -\sum_{n=1}^M 
d(\lambda_n)\, e^{-\lambda_n}
   \frac{\sinh\eta\, \sinh(\lambda-\lambda_n+\eta(\delta-1+\frac{3M}{2})\eta)}{\sinh(\lambda-\lambda_n)}
   \prod_{j\neq n}\frac{\sinh(\lambda_n-\lambda_j+\eta)}{\sinh(\lambda_n-\lambda_j)}\,
   \\
   \times
    \underline\Theta_M(\{\lambda_j\}_{j\neq n}\cup\{\lambda\}|\alpha-\beta,\gamma-\delta +2-2M)\,
   \ket{\eta,\gamma+\delta+N-1}
   \Bigg\},
\end{multline}
which gives \eqref{act-theta-Bethe1}.
\eqref{act-theta-Bethe2} can be shown similarly by means of \eqref{comm-Theta3} and \eqref{act-Theta-r3}.
\end{proof}

Lemma~\ref{Action-Theta} can be rewritten in terms of the more usual notations $A$, $B$, $C$, $D$ \eqref{redef-gauge-op}, and provides the action of the gauged operators $A$ and $D$ on a bulk gauge Bethe state: 

\begin{lemma}
The actions of the diagonal gauge transformed Yang-Baxter generators on a gauged bulk Bethe state  read:

\begin{multline}
  A(\lambda_{M+1} |\alpha -\beta ,\gamma +\delta +M)\
  \underline{B}_M(\{\lambda_j\}_{j=1}^M|\alpha -\beta -1,\gamma -\delta +1-2M)\,
  \ket{\eta ,\gamma +\delta+N}
  \\
  = -\frac{e^{-(\gamma +\delta +N)\eta }-e^{-(\alpha -\beta-M)\eta }}{e^{\eta /2}}\,
  \sum_{a=1}^{M+1} 
  \frac{\sinh(\lambda_{M+1}-\lambda_a+\eta(\delta-1+\frac{3M}{2}))}{\sinh (\eta(\delta-1+\frac{M}{2}) }\,
  \frac{\prod_{b=1}^M \sinh (\lambda _{a}-\lambda _{b}+\eta)}{\prod_{\substack{b=1 \\ b\neq a}}^{M+1} \sinh (\lambda _{a}-\lambda _{b})} 
  \\
  \times
   d(\lambda_a) \
  \underline{B}_M ( \{\lambda_j\}_{j=1}^{M+1}\setminus\{\lambda_a\}|\alpha -\beta,\gamma -\delta +2-2M)\,
  \ket{\eta ,\gamma+\delta +N-1} ,  
\end{multline}
and
\begin{multline}
  D(\lambda_{M+1} |\gamma +\delta +N-M,\gamma -\delta-2M)\
  \underline{B}_M(\{\lambda_j\}_{j=1}^M |\alpha -\beta -1,\gamma-\delta +1-2M)\,
  \ket{\eta ,\gamma +\delta +N}
  \\
  =\frac{e^{-(\gamma -\delta -M)\eta }-e^{-(\gamma +\delta)\eta }}{e^{\eta /2}}\,
  \sum_{a=1}^{M+1}
  \frac{\sinh(\lambda_{M+1}-\lambda_a+\eta\frac{\gamma+\delta-\alpha+\beta+N-M+2}{2})}{\sinh(\eta\frac{\gamma+\delta-\alpha+\beta+N+M+2}{2})}\,
   \frac{\prod_{j=1}^M\sinh(\lambda_a-\lambda_j-\eta)}{\prod_{\substack{j=1 \\ j\neq a}}^{M+1}\sinh(\lambda_a-\lambda_j)} 
  \\
  \times a(\lambda_a)\
  \underline{B}_M(\{\lambda_j\}_{j=1}^{M+1}\setminus\{\lambda_a\}|\alpha -\beta -2,\gamma -\delta -2M)\,
  \ket{\eta ,\gamma +\delta+N+1} .  
\end{multline}
These actions can be rewritten as
\begin{align}
  &A(\lambda_{M+1} |x ,y +M)\ \underline{B}_M(\{\lambda_j\}_{j=1}^M|x -1,z +1)\,  \ket{\eta ,y+N}
  \nonumber\\  
  &\quad
  = \frac{e^{-(x-M)\eta }-e^{-(y+N)\eta }}{e^{\eta /2}}\,
  \sum_{a=1}^{M+1} d(\lambda_a)\,
  \frac{\sinh(\lambda_{M+1}-\lambda_a+\eta\frac{y-z-2+M}{2})}{\sinh (\eta\frac{y-z-2-M}{2}) }\,
  \frac{\prod_{b=1}^M \sinh (\lambda _{a}-\lambda _{b}+\eta)}{\prod_{\substack{b=1 \\ b\neq a}}^{M+1} \sinh (\lambda _{a}-\lambda _{b})} 
  \nonumber\\
  &\hspace{6cm}
  \times
  \underline{B}_M ( \{\lambda_j\}_{j=1}^{M+1}\setminus\{\lambda_a\}| x, z +2)\,
  \ket{\eta ,y +N-1} , 
  \nonumber\\
  &\quad
  = \frac{ e^{-(x-M)\eta }-e^{-(y+N)\eta }}{e^{\eta /2} }\,
  \sum_{a=1}^{M+1} d(\lambda_a)\, \frac{\sinh\eta}{\sinh (\eta\frac{z-y+2+M}{2})}
  \bigg[\frac{\sinh(\eta\frac{z-y+2-M}{2})}{\sinh\eta}\bigg]^{\delta_{a,M+1}} 
  \nonumber\\
  &\quad
  \times
  \prod_{\substack{b=1 \\ b\neq a}}^{M+1}\frac{\sinh(\lambda_a-\lambda_b+\eta_{M,a}(\frac{z-y+2-M}{2}))}{\sinh(\lambda_a-\lambda_b)}\,
  \underline{B}_M ( \{\lambda_j\}_{j=1}^{M+1}\setminus\{\lambda_a\}| x, z +2)\,
  \ket{\eta ,y +N-1} , 
  \label{Action-A-0}
\end{align}
and
\begin{align}
  &D(\lambda_{M+1} | y-M,z)\
  \underline{B}_M(\{\lambda_j\}_{j=1}^M |x -1,z +1)\,
  \ket{\eta ,y}
  \nonumber\\
  &\quad
  =\frac{e^{-(z+M)\eta }-e^{-(y-N)\eta }}{e^{\eta /2}}\,
  \sum_{a=1}^{M+1} a(\lambda_a)\,
  \frac{\sinh(\lambda_{M+1}-\lambda_a+\eta\frac{y-x-M+2}{2})}{\sinh(\eta\frac{y-x+M+2}{2})}\,
   \frac{\prod_{j=1}^M\sinh(\lambda_a-\lambda_j-\eta)}{\prod_{\substack{j=1 \\ j\neq a}}^{M+1}\sinh(\lambda_a-\lambda_j)} 
  \nonumber\\
  &\hspace{6cm}
  \times \
  \underline{B}_M(\{\lambda_j\}_{j=1}^{M+1}\setminus\{\lambda_a\}| x -2,z)\,
  \ket{\eta ,y+1} , 
  \nonumber\\
  &\quad
  =\frac{e^{-(z+M)\eta }-e^{-(y-N)\eta } }{e^{\eta /2}}\,
    \sum_{a=1}^{M+1} a(\lambda_a)\,
    \frac{\sinh\eta}{ \sinh(\eta\frac{x-y-M-2}{2})}
  \bigg[\frac{\sinh(\eta\frac{x-y+M-2}{2})}{\sinh\eta}\bigg]^{\delta_{a,M+1}}\,
  \nonumber\\
  &\quad
  \times
   \prod_{\substack{j=1 \\ j\neq a}}^{M+1}
   \frac{\sinh(\lambda_a-\lambda_j-\eta_{M,a}(\frac{y-x-M+2}{2}))}{\sinh(\lambda_a-\lambda_j)} \,
   \underline{B}_M(\{\lambda_j\}_{j=1}^{M+1}\setminus\{\lambda_a\}| x -2,z)\,
  \ket{\eta ,y+1} ,
   \label{Action-D-0}
\end{align}
with the notation
\begin{equation}
  \eta_{M,a}(x)=\big[ (1-\delta_{M+1,a})+x\, \delta_{M+1,a}   \big] \eta.
\end{equation}
\end{lemma}

\section{On the boundary-bulk decomposition}
\label{app-bound-bulk}

We recall that the elements of the gauge boundary monodromy matrix $\mathcal{U}_{-}(\lambda |\alpha ,\beta )$ can be expressed in terms of the elements of
the gauge bulk monodromy matrix by formula \eqref{bound-bulk-gauge}, which can also be written by components as
\begin{multline}
   \mathcal{U}_{-}(\lambda |\alpha ,\beta )
   =
     \frac{\left( -1\right) ^{N}e^{\eta (\gamma ^{\prime }+\alpha )}}{4\sinh \eta \delta ^{\prime }\sinh \eta \beta }     \,
     \begin{pmatrix}
     D(\lambda |\gamma +\delta -1,\alpha -\beta -1) & -B(\lambda |\gamma -\delta-1,\alpha -\beta -1) \\ 
    -C(\lambda |\gamma +\delta -1,\alpha +\beta -1) & A(\lambda |\gamma -\delta-1,\alpha +\beta -1)
     \end{pmatrix}
      \\
      \times 
      K_{-}(\lambda |(\gamma ,\delta ),(\gamma ^{\prime },\delta ^{\prime}))\,
      \begin{pmatrix}
      A(-\lambda |\gamma ^{\prime }-\delta ^{\prime },\alpha +\beta ) 
      & B(-\lambda|\gamma ^{\prime }-\delta ^{\prime },\alpha -\beta ) \\ 
      C(-\lambda |\gamma ^{\prime }+\delta ^{\prime },\alpha +\beta ) 
      & D(-\lambda|\gamma ^{\prime }+\delta ^{\prime },\alpha -\beta )
      \end{pmatrix} ,  
      \label{BB-Gauged-Dec}
\end{multline}
in terms of the gauged boundary matrix \eqref{gauge-K}. 
In this appendix, we prove several consequences of this boundary bulk decomposition.


Let us remark that the matrix $K_{-}(\lambda |(\gamma ,\delta ),(\gamma^{\prime },\delta ^{\prime }))$ can be made upper triangular identically w.r.t. $\lambda $ if
\begin{equation}
\gamma +\delta =\gamma ^{\prime }+\delta ^{\prime }.
\end{equation}
In that case,
\begin{equation}
   \left[ K_{-}(\lambda |(\gamma ,\delta ),(\gamma ^{\prime },\delta ^{\prime}))\,\right] _{21}=0
\end{equation}
is equivalent to
\begin{equation}
     \frac{\kappa _{+}}{\sinh \zeta _{+}}\left( \sinh (\eta (\gamma +\delta)+\tau _{+})+\sinh (\varphi _{+}+\psi _{+})\right) =0.
\end{equation}
Similarly, the matrix $K_{-}(\lambda |(\gamma ,\delta ),(\gamma ^{\prime},\delta ^{\prime }))$ can be made lower triangular identically w.r.t. $\lambda $ if
\begin{equation}
\gamma -\delta =\gamma ^{\prime }-\delta ^{\prime }.
\end{equation}
In that case,
\begin{equation}\label{K-lower-tr}
\left[ K_{-}(\lambda |(\gamma ,\delta ),(\gamma ^{\prime },\delta ^{\prime
}))\,\right] _{12}=0
\end{equation}
is equivalent to
\begin{equation}
\frac{\kappa _{+}}{\sinh \zeta _{+}}\left( \sinh (\eta (\gamma -\delta)+\tau _{+})+\sinh (\varphi _{+}+\psi _{+})\right) =0.
\end{equation}

\subsection{\label{Closed-Form-Ref}Construction of\ an $\mathcal{A}_{-}(%
\protect\lambda |\protect\alpha ,\protect\beta -1)$ and $\mathcal{D}_{-}(%
\protect\lambda |\protect\alpha ,\protect\beta +1)$ eigenvector}

Here, we construct an $\mathcal{A}_{-}(\lambda |\alpha ,\beta -1)$ right
eigenvector and a $\mathcal{D}_{-}(\lambda |\alpha ,\beta +1)$ left
eigenvector, for a special choice of the gauge parameters $(\alpha ,\beta )$.

\begin{lemma}\label{lemme-eigen-A}
Under the condition
\begin{equation}
\frac{\kappa _{+}}{\sinh \zeta _{+}}\big[ \sinh (\eta (\alpha +\beta
+N-1)+\tau _{+})+\sinh (\varphi _{+}+\psi _{+})\big] =0,  \label{Gauge-B+}
\end{equation}
the vector $\ket{\eta ,\alpha +\beta +N-1} $ is an $\mathcal{A}_{-}(\lambda |\alpha ,\beta -1)$ eigenvector:
\begin{multline}\label{eigen-A}
  \mathcal{A}_{-}(\lambda |\alpha ,\beta -1)\, \ket{\eta ,\alpha +\beta +N-1}
  \\
  =  (-1)^N a(\lambda )\,d(-\lambda )\, [ K_{-}(\lambda |(\alpha ,\beta+N-1),(\alpha ,\beta+N-1 ))] _{11} \, \ket{\eta ,\alpha +\beta +N-1} ,
\end{multline}
and is annihilated by $\mathcal{C}_{-}(\lambda |\alpha ,\beta -1)$:
\begin{equation}
\mathcal{C}_{-}(\lambda |\alpha ,\beta -1)\,\ket{\eta ,\alpha +\beta +N-1}=0.  \label{Annihilation-C}
\end{equation}
Under the condition
\begin{equation}\label{cond-eigen-D+}
\frac{\kappa _{+}}{\sinh \zeta _{+}}\big[ \sinh (\eta (\alpha +\beta
-N+1)+\tau _{+})+\sinh (\varphi _{+}+\psi _{+})\big] =0,
\end{equation}
the covector $\bra{\alpha +\beta -N+1,\eta }$ is a $\mathcal{D}_{-}(\lambda |\alpha ,\beta +1)$ eigencovector:
\begin{multline}\label{eigen-D}
  \bra{\alpha +\beta -N+1,\eta }\, \mathcal{D}_{-}(\lambda |\alpha ,\beta +1)
  \\
  =(-1)^N\, a(\lambda )\,d(-\lambda ) \,
  [ K_{-}(\lambda |(\alpha ,\beta-N+1 ),(\alpha ,\beta-N+1 ))]_{22}\,
 \bra{\alpha +\beta -N+1,\eta }.
\end{multline}
\end{lemma}

\begin{proof}
To prove \eqref{eigen-A} and \eqref{Annihilation-C}, let us use the boundary bulk decomposition of the gauged monodromy given in 
\eqref{BB-Gauged-Dec}, where we have fixed the internal parameters by
\begin{equation}
     \gamma =\gamma ^{\prime }=\alpha ,\qquad
     \delta ^{\prime }=\delta =\beta+N-1.
\end{equation}
Then under the conditions \eqref{Gauge-B+}, the matrix $K_{-}(\lambda |(\alpha ,\beta+N-1 ),(\alpha,\beta+N-1 ))$ is upper triangular.
Noticing moreover that
\begin{equation}
C(-\lambda |\alpha +\beta +N-1,\alpha +\beta -1)\, \ket{\eta ,\alpha +\beta+N-1}  = 0,
\end{equation}
we obtain that
\begin{multline}
   \mathcal{A}_{-}(\lambda |\alpha ,\beta -1)\, \ket{\eta ,\alpha +\beta +N-1}
    = \frac{(-1)^N\, e^{2\alpha\eta}\,  [ K_{-}(\lambda |(\alpha ,\beta +N-1 ),(\alpha,\beta +N-1 ))] _{11}}{4\sinh(\eta(\beta-1))\, \sinh(\eta(\beta+N-1))} \, 
   \\
    \times
    D(\lambda |\alpha +\beta+N-2,\alpha -\beta ) \, A(-\lambda |\alpha -\beta -N+1,\alpha+\beta -1)\,
    \ket{\eta ,\alpha +\beta +N-1} .  \label{Ex-A}
\end{multline}
By using the following identities:
\begin{multline}
A(-\lambda |\alpha -(\beta +N-1),\alpha +\beta -1) \,\ket{\eta ,\alpha +\beta+N-1}
\\
  = \frac{2\sinh \eta (\beta +N-1)}{e^{\eta (\alpha +1/2)}}\, d(-\lambda )\,\ket{\eta ,\alpha +\beta +N-2}  ,
\end{multline}
and
\begin{equation}
    D(\lambda |\alpha +\beta +N-2,\alpha -\beta )\, \ket{\eta ,\alpha +\beta+N-2}
     =\frac{2\sinh \eta (\beta -1)}{e^{\eta (\alpha -1/2)}}\, a(\lambda)\, \ket{\eta ,\alpha +\beta +N-1} ,
\end{equation}
we can compute the r.h.s. of \eqref{Ex-A}.
\eqref{Annihilation-C} is a consequence of the identity:
\begin{multline}
    C(\lambda |\alpha +\beta +N-2,\alpha +\beta -2)\, A(-\lambda |\alpha -\beta-N+1,\alpha +\beta -1)\,
    \ket{\eta ,\alpha +\beta +N-1}
    \\
     \propto  C(\lambda |\alpha +\beta +N-2,\alpha +\beta -2)\, 
     \ket{\eta ,\alpha+\beta +N-2} = 0.
\end{multline}

One can prove similarly, using \eqref{BB-Gauged-Dec} in the case \eqref{K-lower-tr}, that, under the condition
\begin{equation}\label{cond-eigen-covecA}
\frac{\kappa _{+}}{\sinh \zeta _{+}}\big[ \sinh (\eta (\alpha -\beta
-N+1)+\tau _{+})+\sinh (\varphi _{+}+\psi _{+})\big] =0,
\end{equation}
the covector $\bra{\alpha -\beta -N+1,\eta }$ is an $\mathcal{A}_{-}(\lambda |\alpha ,\beta -1)$ eigencovector:
\begin{multline}\label{eigen-covec-A}
  \bra{\alpha -\beta -N+1,\eta }\, \mathcal{A}_{-}(\lambda |\alpha ,\beta -1)
  \\
  =(-1)^N\, a(\lambda )\,d(-\lambda ) \,
  [ K_{-}(\lambda |(\alpha ,\beta+N-1 ),(\alpha ,\beta+N-1 ))]_{11}\,
 \bra{\alpha -\beta -N+1,\eta }.
\end{multline}
Since, for any gauge parameter $\beta'$, one has the identification
\begin{align}
   &\mathcal{A}_{-}(\lambda |\alpha ,\beta')=\mathcal{D}_{-}(\lambda |\alpha ,-\beta'),\\
   &  [ K_{-}(\lambda |(\alpha ,\beta' ),(\alpha ,\beta' ))]_{11}
   =  [ K_{-}(\lambda |(\alpha ,-\beta' ),(\alpha ,-\beta' ))]_{22},
\end{align}
this proves the second part of the lemma.
\end{proof}

Note that these vectors have been first constructed as {\em generalized reference states} for the gauge transformed boundary monodromy
matrix in the framework of a generalization of ABA for the open chain \cite{CaoLSW03}.
Here, they enter in our SoV description of the transfer matrix spectrum,
where no Ansatz on the eigenstate construction is done and where the complete spectrum
description is a built-in feature (see Proposition~\ref{Ref-States}).

\subsection{\label{app-boundary-bulk}Boundary-Bulk decomposition of boundary gauge Bethe states}

Here, we present the boundary-bulk decomposition of ABA states of the form
\begin{equation}\label{gauge-Bethe-bound}
   \underline{\widehat{\mathcal B}}_{-,M}(\{\lambda_i\}_{i=1}^M | z+1)\, \ket{\eta,y}
   \equiv \prod_{j=1\to M}\widehat{\mathcal B}_-(\lambda_j|z+2j-1)\,\ket{\eta,y}
\end{equation}
for the very specific boundary conditions that we consider in the framework of this paper, i.e. the one associated to the choice of the boundary matrix \eqref{Special-K+}.
Then under the gauge transformation for any choice $(\gamma ,\delta )$, we have
\begin{equation}
K_{-}(\lambda ,\varsigma _{+}=-\infty ,\kappa _{+},\tau _{+}|(\gamma ,\delta),(\gamma ,\delta ))
=e^{(\lambda -\eta /2)}%
\begin{pmatrix}
1 & 0 \\ 
0 & 1%
\end{pmatrix}%
_{0}.  \label{Special-K+gau}
\end{equation}

\begin{lemma}
The boundary operator $\widehat{\mathcal{B}}_{-}$ \eqref{Bhat} can be expressed in terms of the gauge bulk operators elements of \eqref{redef-gauge-op} as 
\begin{multline}
  \widehat{\mathcal{B}}_{-}(\lambda |\alpha -\beta )
   =\frac{(-1)^{N} e^{\eta (\gamma +\alpha -\beta )}}{4\sinh \eta (\delta +1)}
     \frac{\sinh (2\lambda -\eta )}{\sinh (2\lambda) }  
     \big[ B(-\lambda |\gamma -\delta -1,\alpha -\beta -1)\,
             D(\lambda |\gamma+\delta ,\alpha -\beta )  
             \\
     -B(\lambda |\gamma -\delta -1,\alpha -\beta -1)\,
      D(-\lambda |\gamma +\delta,\alpha -\beta )\big],
\end{multline}
for any gauge parameters $\alpha,\beta,\gamma,\delta$, which can be equivalently rewritten as
\begin{equation}\label{decomp-Bhat}
  \widehat{\mathcal{B}}_{-}(\lambda |z )
  =\frac{(-1)^{N+1} e^{\eta (\frac{x+y}{2}+z )}}{4\sinh (\eta \frac{y-x+2}2 )}\,
  \frac{\sinh (2\lambda -\eta )}{\sinh (2\lambda) } 
  \sum_{\sigma=\pm} \sigma B(\sigma \lambda | x-1,z-1)\, D(-\sigma\lambda|y,z),
\end{equation}
for any gauge parameter $z,x,y$.
\end{lemma}

\begin{proof}
By definition, it holds
\begin{multline}
   \widehat{\mathcal{B}}_{-}(\lambda |\alpha -\beta )
    =\frac{(-1) ^{N} e^{\eta (\gamma +\alpha -\beta )}}{4\sinh \eta \delta }
    \big[ D(\lambda |\gamma +\delta -1,\alpha -\beta -1)\,
            B(-\lambda |\gamma -\delta,\alpha -\beta ) 
             \\
   -B(\lambda |\gamma -\delta -1,\alpha -\beta -1)\,
    D(-\lambda |\gamma +\delta,\alpha -\beta ).
\end{multline}
We then use the commutation relation
\begin{multline}
   D(\lambda |\gamma +\delta -1,\alpha -\beta -1)\,
   B(-\lambda |\gamma -\delta,\alpha -\beta )
    \\
   = \frac{\sinh (\eta \delta)\,  \sinh (2\lambda -\eta )}{\sinh (2\lambda)\, \sinh \eta(\delta +1)}\,
     B(-\lambda |\gamma -\delta -1,\alpha -\beta -1)\,
     D(\lambda |\gamma+\delta ,\alpha -\beta ) 
      \\
     +\frac{\sinh \eta\, \sinh (2\lambda +\delta \eta )}{\sinh (2\lambda)\, \sinh\eta (\delta +1)}\,
     B(\lambda |\gamma -\delta -1,\alpha -\beta -1)\,
     D(-\lambda|\gamma +\delta ,\alpha -\beta ),
\end{multline}
and the following identity:
\begin{equation}
\frac{\sinh \eta \, \sinh (2\lambda +\delta \eta )}{\sinh (2\lambda)\, \sinh \eta(\delta +1)}-1
=-\frac{\sinh (\eta \delta)\, \sinh (2\lambda -\eta )}{\sinh \eta(\delta +1)\, \sinh (2\lambda) }.
\end{equation}
\end{proof}

It follows from this result that boundary gauge Bethe states of the form \eqref{gauge-Bethe-bound},
for any gauge parameters $z$ and $y$, can be expressed as linear combinations of bulk gauge Bethe states of the form
\begin{equation}\label{gauge-Bethe-bulk}
   \underline{B}_M(\{\sigma_i\lambda_i\}_{i=1}^M|x-1,z)\, \ket{\eta,y+M}
   \equiv \prod_{j=1\to M} B(\sigma_j\lambda_j|x-j,z+j-1)\, \ket{\eta,y+M}
\end{equation}
for any choice of the gauge parameter $x$, where $\sigma_i\in\{+,-\}$, $i=1,\ldots,M$.
More precisely, we can formulate the following result:

\begin{proposition}\label{prop-boundary-bulk}
Boundary gauge Bethe states of the form \eqref{gauge-Bethe-bound}, for any arbitrary set of spectral parameters $\{\lambda_1,\ldots,\lambda_M\}$ and any choice of gauge parameters $z$ and $y$, can be expressed as the following linear combination of bulk gauge Bethe states:
\begin{multline}\label{Boundary-bulk-Bethe}
   \underline{\widehat{\mathcal B}}_{-,M}(\{\lambda_i\}_{i=1}^M | z+1)\, \ket{\eta,y}
   = h_M(x,z,y)\, \\
   \times
   \sum_{\sigma_1=\pm,\ldots,\sigma_M=\pm} 
   H_{\sigma_1,\ldots,\sigma_M}(\{\lambda_i\}_{i=1}^M)\, 
   \underline{B}_M(\{\sigma_i\lambda_i\}_{i=1}^M|x-1,z)\, \ket{\eta,y+M},
\end{multline}
with
\begin{equation}\label{H_sigma}
   H_{\sigma_1,\ldots,\sigma_M}(\{\lambda_i\}_{i=1}^M)
   =\prod_{j=1}^M\left[ \sigma_j\, a(-\lambda_j^{(\sigma)})\,
                                    \frac{\sinh(2\lambda_j-\eta)}{\sinh(2\lambda_j)} \right]
     \prod_{1\le i<j\le M}\frac{\sinh(\lambda_i^{(\sigma)}+\lambda_j^{(\sigma)}+\eta)}{\sinh(\lambda_i^{(\sigma)}+\lambda_j^{(\sigma)})},
\end{equation}
in which we have used the shortcut notation $\lambda_j^{(\sigma)}\equiv \sigma_j\lambda_j$, $1\le j\le M$,
and where $h_M(x,z,y)$ is an overall non-zero constant depending exclusively on the gauge parameters $x,z,y$ and on the number $M$ of $\widehat{\mathcal B}_-$ operators. 
More precisely we have
\begin{align}\label{hM}
  h_M(x,z,y)
  &=(-1)^{MN}\, e^{M\eta\frac{x+z+N}2}\, \prod_{j=1}^M\frac{\sinh(\eta \frac{z-y+2j+N-1}2)}{2\sinh(\eta\frac{y-x+2j}2)}.
\end{align}
\end{proposition}

\begin{proof}
The proof can be done by induction on $M$ following the same lines as in the ungauged case \cite{KitKMNST07}.

For $M=1$, \eqref{Boundary-bulk-Bethe} is a direct consequence of  \eqref{decomp-Bhat} and of the action \eqref{actD-0-op} of the gauge operator $D$ on the gauge reference state.
Let us now assume \eqref{Boundary-bulk-Bethe} for a given $M$ and let us prove it for $M+1$.
Then, from \eqref{decomp-Bhat} and the recursion hypothesis, we can write
\begin{align}
   &\underline{\widehat{\mathcal B}}_{-,M+1}(\{\lambda_i\}_{i=1}^{M+1} | z+1)\, \ket{\eta,y}
   =\widehat{\mathcal{B}}_{-}(\lambda_{M+1} |z+1 )\,
        \underline{\widehat{\mathcal B}}_{-,M}(\{\lambda_i\}_{i=1}^M | z+3)\, \ket{\eta,y}
   \nonumber\\
      &\qquad
    = h_M(x,z+2,y)\, 
   \sum_{\sigma_1=\pm,\ldots,\sigma_M=\pm} 
   H_{\sigma_1,\ldots,\sigma_M}(\{\lambda_i\}_{i=1}^M)\, 
   \widehat{\mathcal{B}}_{-}(\lambda_{M+1} |z+1 )\,
   \nonumber\\
   & \qquad\quad 
   \times
   \underline{B}_M(\{\sigma_i\lambda_i\}_{i=1}^M|x-1,z+2)\, \ket{\eta,y+M},
   \nonumber\\   
   &\qquad
    = h_M(x,z+2,y)\,    \frac{(-1)^{N+1} e^{\eta (\frac{x+y}{2}+z+1 )}}{4\sinh \eta (\frac{y-x}2 +1)}
   \sum_{\sigma_1=\pm,\ldots,\sigma_M=\pm} 
   H_{\sigma_1,\ldots,\sigma_M}(\{\lambda_i\}_{i=1}^M)\, 
   \frac{\sinh (2\lambda_{M+1} -\eta )}{\sinh (2\lambda_{M+1}) } 
   \nonumber\\
   &\qquad\quad 
   \times 
   \sum_{\sigma=\pm} \sigma
   B(\sigma \lambda_{M+1} | x-1,z)\,
   D(-\sigma\lambda_{M+1}|y,z+1)\,
   \underline{B}_M(\{\sigma_i\lambda_i\}_{i=1}^M|x-1,z+2)\, \ket{\eta,y+M}.
\end{align}
From \eqref{Action-D-0}, it is easy to verify that the direct action of $D(\sigma \lambda_{M+1}| y, z+1)$
leads to the contribution
\begin{multline}
   h_{M+1}(x,z,y)\, \sum_{\sigma_1=\pm,\ldots,\sigma_{M+1}=\pm} 
   H_{\sigma_1,\ldots,\sigma_{M+1}}(\{\lambda_i\}_{i=1}^{M+1})\,
   \\
   \times
   B(\sigma_{M+1} \lambda_{M+1} | x-1,z)\,
    \underline{B}_M(\{\sigma_i\lambda_i\}_{i=1}^M|x-2,z+1)\, \ket{\eta,y+M+1}
\end{multline}
with
\begin{align}
   &h_{M+1}(x,z,y)=h_M(x,z+2,y)\,    
    \frac{(-1)^{N+1} e^{\eta (\frac{x+y}{2}+z+1 )}}{4\sinh(\eta\frac{y-x+2M+2}2)}\,
    \frac{e^{-(z+1+M)\eta}-e^{-(y+M-N)\eta}}{e^{\eta/2}}\nonumber\\
    &\hphantom{h_{M+1}(x,z,y)}=(-1)^N\, e^{\eta(\frac{x+z+N}2-M)}\,\frac{\sinh(\eta\frac{z-y+N+1}2)}{2\sinh(\eta\frac{y-x+2M=2}2)}\, h_M(x,z+2,y),
    \\
   &H_{\sigma_1,\ldots,\sigma_{M+1}}(\{\lambda_i\}_{i=1}^{M+1})
   =\sigma_{M+1}\, a(-\lambda_{M+1}^{(\sigma)})\,\frac{\sinh (2\lambda_{M+1} -\eta )}{\sinh (2\lambda_{M+1}) } \prod_{j=1}^M\frac{\sinh(-\lambda_{M+1}^{(\sigma)}-\lambda_j^{(\sigma)}-\eta)}{\sinh(-\lambda_{M+1}^{(\sigma)}-\lambda_j^{(\sigma)})}
   \nonumber\\
   &\hspace{7cm}\times
   H_{\sigma_1,\ldots,\sigma_M}(\{\lambda_i\}_{i=1}^M),
\end{align}
which gives the left hand side of \eqref{Boundary-bulk-Bethe} for $M+1$.
Hence, it remains to prove that the contributions due to the indirect terms in the action  of $D(\sigma \lambda_{M+1}| y, z+1)$ vanish. These terms lead to bulk gauge Bethe vectors of the form
\begin{equation}
   \underline{B}_{M+1}(\{\lambda_i^{(\sigma)}\}_{i=1}^{M}\cup\{\lambda_{M+1},-\lambda_{M+1}\}\setminus\{\lambda_a^{(\sigma)}\}|x-1,z)\,\ket{\eta,y+M+1}
\end{equation}
for some $a\in\{1,\ldots,M\}$, and it is easy to see that the coefficient of such a vector contains a factor of the form
\begin{align}
\sum_{\sigma _{a}=\pm 1,\sigma _{M+1}=\pm 1}
\frac{\sigma _{a}\sigma_{M+1}\sinh (\lambda _{a}^{( \sigma ) }+\lambda _{M+1}^{(\sigma ) }+\kappa)}{\sinh (\lambda _{a}^{( \sigma ) }+\lambda_{M+1}^{( \sigma ) })}
& =\frac{\sinh (\bar{\lambda}_{a,M+1}+\kappa)}{\sinh \bar{\lambda}_{a,M+1}}
   +\frac{\sinh (\bar{\lambda}_{a,M+1}-\kappa)}{\sinh \bar{\lambda}_{a,M+1}}  \notag \\
& -\frac{\sinh (\lambda _{a,M+1}+\kappa)}{\sinh \lambda _{a,M+1}}
   -\frac{\sinh (\lambda_{a,M+1}-\kappa)}{\sinh \lambda _{a,M+1}}  \label{Null-sum}
\end{align}
where
\begin{equation}
\kappa=\eta\frac{x-y-2}2. 
\end{equation}
This expression vanishes due to the fact that
\begin{equation}
\frac{\sinh (\lambda+\kappa)}{\sinh \lambda}+\frac{\sinh (\lambda-\kappa)}{\sinh \lambda}=2\cosh \kappa,
\end{equation}
which completes the proof.
\end{proof}

\section{The diagonal limit}
\label{app-diag}

In this appendix, we explicitly show how the expressions (6.8)-(6.10) of \cite{KitKMNST07} for the correlation function in the diagonal case can be inferred by taking the limit
\begin{equation}\label{limit-diag}
   \psi_-\to +\infty \quad \text{with} \quad \varphi_-\sim \varsigma_-\ \text{finite},
\end{equation}
in our result \eqref{result-thermo-bis}-\eqref{H-thermo-tilde}.
In this limit,  $\kappa_-\sim e^{-\psi_-}\to 0$, so that we indeed have diagonal boundary conditions at both ends of the chain. 

Let us fix for instance $\epsilon_{\varphi_-}=\epsilon_{\psi_-}=1$. From \eqref{choice-alpha-beta}, it follows that, in the limit \eqref{limit-diag},
\begin{equation}
    \eta\beta \sim \eta b_n \sim \eta\bar b_n \sim -\psi_- \to -\infty.
\end{equation}
Hence we have
\begin{align}
  &-2\sinh(\eta b_n)\,e^{-\eta}\, E_n^{1,1}(\xi_n|(a_n,b_n),(\bar{a}_n,\bar{b}_n)) 
     \sim e^{-\eta (\bar b_n-1)}\ E^{1,1}_n,\\
  &-2\sinh(\eta b_n)\,e^{-\eta}\, E_n^{1,2}(\xi_n|(a_n,b_n),(\bar{a}_n,\bar{b}_n)) 
     \sim e^{-\eta(b_n+\bar b_n+\alpha)-\xi_n}\ E^{1,2}_n,\\
  &-2\sinh(\eta b_n)\,e^{-\eta}\, E_n^{2,1}(\xi_n|(a_n,b_n),(\bar{a}_n,\bar{b}_n)) 
     \sim e^{\xi_n+\eta \alpha}\ E^{2,1}_n,\\
  &-2\sinh(\eta b_n)\,e^{-\eta}\, E_n^{2,2}(\xi_n|(a_n,b_n),(\bar{a}_n,\bar{b}_n)) 
     \sim e^{-\eta (b_n+1)}\ E_n^{2,2},
\end{align}
so that
\begin{multline}\label{limit-prod-op}
   \prod_{n=1}^m E_n^{\epsilon'_n,\epsilon_n} 
   \sim (-2)^m \prod_{n=1}^m\left[e^{-\eta}\sinh(\eta b_n)\right] \prod_{p=1}^s e^{\xi_{i_p}-\xi_p+\eta(b_{i_p}+1)} \prod_{p=s+1}^m e^{\xi_{i_p}-\xi_p+\eta(\bar b_{i_p}-1)}\ 
   \\
   \times
   \prod_{n=1}^{m}E_{n}^{\epsilon _{n}^{\prime },\epsilon _{n}}(\xi_{n}|(a_{n},b_{n}),(\bar{a}_{n},\bar{b}_{n}))
\end{multline}
in which we have used that $s+s'=m$.
On the other hand, from \eqref{H-thermo-tilde}, we get in the limit
\begin{multline}
\tilde H_{m}(\{\lambda _{j}\}_{j=1}^M;\{\tilde \xi _{k}\}_{k=1}^m)
 \sim\frac{\prod\limits_{j=1}^{m}\prod\limits_{k=1}^{m}\sinh (\lambda _{j}+\tilde \xi_k-\eta)}
   {\!\!\!\!\prod\limits_{1\leq i<j\leq m}\!\!\!\!\sinh (\lambda_i-\lambda_j+\eta )\,\sinh (\lambda_i+\lambda_j-\eta )}
    \\
 \times 
 \prod\limits_{p=1}^{s}\bigg\{-\frac{e^{-\lambda_p+\tilde\xi_{i_p}-\eta(1+b_{i_p})}}2
 \prod\limits_{k=1}^{i_{p}-1}\sin(\lambda _{p}-\tilde\xi _{k})
 \prod\limits_{k=i_{p}+1}^{m}\!\!\sinh (\lambda_{p}-\tilde\xi _{k}+\eta )\bigg\} 
 \\
 \times 
 \!\!\prod\limits_{p=s+1}^{m}\bigg\{ -\frac{e^{-\lambda_p+\tilde\xi_{i_p}-\eta(\bar b_{i_p}-1)}}2
 \prod\limits_{k=1}^{i_{p}-1}\sinh (\lambda _{p}-\tilde\xi _{k})
 \prod\limits_{k=i_{p}+1}^{m}\!\!\sinh (\lambda_{p}-\tilde\xi _{k}-\eta )\bigg\} 
   \\
   \times
   \prod_{k=1}^m e^{-\tilde\xi_k+\lambda_k}
   \prod\limits_{k=1}^{m}\frac{\sinh (\tilde\xi _{k}+\zeta_{-}-\eta/2)}{\sinh(\lambda _{k}+\zeta_{-}-\eta /2)},
\end{multline}
i.e.
\begin{multline}\label{limit-H}
  \tilde H_{m}(\{\lambda _{j}\}_{j=1}^M;\{\tilde \xi _{k}\}_{k=1}^m)
 \sim (-2)^{-m}\, \prod_{p=1}^s e^{-\tilde\xi_p+\tilde\xi_{i_p}-\eta(1+b_{i_p})}  \!\!\prod\limits_{p=s+1}^{m}e^{-\tilde\xi_p+\tilde\xi_{i_p}-\eta(\bar b_{i_p}-1)}\
 \\
 \times
 \tilde H_{m}^\text{diag}(\{\lambda _{j}\}_{j=1}^M;\{\tilde \xi _{k}\}_{k=1}^m),
\end{multline}
in which
\begin{multline}
\tilde H_{m}^\text{diag}(\{\lambda _{j}\}_{j=1}^M;\{\tilde \xi _{k}\}_{k=1}^m)
 \sim\frac{\prod\limits_{j=1}^{m}\prod\limits_{k=1}^{m}\sinh (\lambda _{j}+\tilde \xi_k-\eta)}
   {\!\!\!\!\prod\limits_{1\leq i<j\leq m}\!\!\!\!\sinh (\lambda_i-\lambda_j+\eta )\,\sinh (\lambda_i+\lambda_j-\eta )}
      \prod\limits_{k=1}^{m}\frac{\sinh (\tilde\xi _{k}+\zeta_{-}-\eta/2)}{\sinh(\lambda _{k}+\zeta_{-}-\eta /2)}
    \\
 \times 
 \prod\limits_{p=1}^{s}\bigg\{
 \prod\limits_{k=1}^{i_{p}-1}\sin(\lambda _{p}-\tilde\xi _{k})
 \prod\limits_{k=i_{p}+1}^{m}\!\!\sinh (\lambda_{p}-\tilde\xi _{k}+\eta )\bigg\} 
 \\
 \times 
 \!\!\prod\limits_{p=s+1}^{m}\bigg\{ 
 \prod\limits_{k=1}^{i_{p}-1}\sinh (\lambda _{p}-\tilde\xi _{k})
 \prod\limits_{k=i_{p}+1}^{m}\!\!\sinh (\lambda_{p}-\tilde\xi _{k}-\eta )\bigg\} .
\end{multline}
From \eqref{limit-prod-op} and \eqref{limit-H} and using \eqref{redef-xi}, we therefore obtain in the limit
\begin{multline}\label{limit-result}
   \moy{ \prod_{n=1}^{m}E_{n}^{\epsilon _{n}^{\prime },\epsilon _{n}} }
     =  \frac{(-1)^{m-s}}{\prod\limits_{j<i}\sinh (\tilde\xi _i-\tilde\xi_j)\prod\limits_{i\leq j}\sinh (\tilde\xi_i+\tilde\xi_j-\eta)}   \\
 \times \int_{\mathcal{\tilde C}}\prod_{j=1}^{s}d\lambda _{j}\ \int_{\mathcal{\tilde C}_{\boldsymbol{\xi}}}
\prod_{j=s+1}^{m}\!\!d\lambda _{j}\ 
\tilde H_{m}^\text{diag}(\{\lambda _{j}\}_{j=1}^M;\{\tilde \xi _{k}\}_{k=1}^m)\ \det_{1\leq j,k\leq m}\big[\tilde \Phi(\lambda_j,\tilde \xi_k)\big],
\end{multline}
in which $\tilde{\mathcal C}$ is defined by \eqref{C-nude} if the set of roots for the ground state does not contain the boundary root $\check\lambda_-$ converging towards $\eta/2-\zeta_-$, and by
\begin{equation}\label{C-BR-diag}
  \tilde{\mathcal C} =  [-\Lambda ,\Lambda ] \cup \Gamma^+(\eta/2-\varsigma_-) 
\end{equation}
if the set of Bethe roots for the ground state contains the boundary root $\check\lambda_-$, whereas $\mathcal{\tilde C}_{\boldsymbol{\xi}}$ is defined as in \eqref{C-xi-tilde}. This result coincides with (6.8)-(6.10) of \cite{KitKMNST07}.

\section*{Acknowledgments}

G. N. is supported by CNRS and Laboratoire de Physique, ENS-Lyon. 
V. T. is supported by CNRS.

\bibliographystyle{SciPost_bibstyle}

\end{document}